\documentclass[11pt,letterpaper]{article}

\usepackage[letterpaper, total={6.5in, 9in}]{geometry}
\usepackage{fourier}
\usepackage[utf8]{inputenc}
\usepackage{amstext,amssymb,amsopn,amsthm,mathtools,bbm,thm-restate}
\usepackage[colorlinks, allcolors=blue]{hyperref}
\usepackage{cite}
\usepackage{color}
\usepackage[font=small,skip=6pt]{caption}
\usepackage[font=small,skip=0pt]{subcaption}
\usepackage{bm}
\usepackage{dcolumn}

\renewenvironment{proof}[1][Proof]{\noindent\textit{#1. } }{\hfill$\square$}

\newtheoremstyle{theorem}{6pt}{6pt}{\rm}{}{\sffamily}{ }{ }{}
\theoremstyle{theorem}
\newtheorem{theorem}{\sc Theorem}[section]

\newtheoremstyle{lemma}{6pt}{6pt}{\rm}{}{\sffamily}{ }{ }{}
\theoremstyle{lemma}
\newtheorem{lemma}{\sc Lemma}[section]

\newtheorem{condition}[theorem]{Condition}

\newtheoremstyle{example}{6pt}{6pt}{\rm}{}{\sffamily}{ }{ }{}
\theoremstyle{example}

\newtheoremstyle{corollary}{6pt}{6pt}{\rm}{}{\sffamily}{ }{ }{}
\theoremstyle{corollary}
\newtheorem{corollary}{\sc Corollary}[section]

\newtheoremstyle{definition}{6pt}{6pt}{\rm}{}{\sffamily}{ }{ }{}
\theoremstyle{definition}
\newtheorem{definition}[theorem]{\sc Definition}

\newtheorem{proposition}[theorem]{\sc Proposition}

\newtheoremstyle{remark}{6pt}{6pt}{\rm}{}{\sffamily}{ }{ }{}
\theoremstyle{remark}
\newtheorem{remark}{\sc Remark}[section]

\newtheoremstyle{approximation}{6pt}{6pt}{\rm}{}{\sffamily}{ }{ }{}
\theoremstyle{approximation}

\newtheoremstyle{scheme}{6pt}{6pt}{\rm}{}{\sffamily}{ }{ }{}
\theoremstyle{scheme}

\usepackage{sectsty} 
\sectionfont{\rmfamily\bfseries\large} 
\subsectionfont{\rmfamily\normalsize} 
\subsubsectionfont{\rmfamily\it\normalsize}

\newtheorem{model}{Model}

\DeclareMathOperator*{\arginf}{arg\,inf}
\DeclarePairedDelimiterX{\norm}[1]{\lVert}{\rVert}{#1}

\newcommand{\E}{\mathbb{E}}
\newcommand{\Ex}{\mathbb{E}}
\newcommand{\Prob}{\mathbb{P}}
\newcommand{\R}{\mathbb{R}}
\newcommand{\Z}{\mathbb{Z}}
\newcommand{\ind}{\mathbf{1}}

\newcommand{\Cb}{\mathbf{C}}
\newcommand{\Lb}{\mathbf{L}}

\newcommand{\Bc}{\mathcal{B}}

\newcommand{\Sc}{S}

\newcommand{\edit}[1]{{\textcolor{black}{#1}}}
\newcommand{\Var}{\mathrm{Var}}
\newcommand{\transvar}{t}
\newcommand{\PsiConstant}{\mathbf{\Psi}}
\newcommand{\LambdaConstant}{\mathbf{\Lambda}}

\newcommand{\snr}{\text{SNR}}

\begin{document}

\title{\textbf{Wavelet invariants for statistically robust} \\ \textbf{multi-reference alignment}}


\author{{
		\sc Matthew Hirn}\\[2pt]
	\textit{Department of Computational Mathematics, Science, and Engineering,} \\ \textit{Department of Mathematics,} \\ \textit{Center for Quantum Computing, Science and Engineering,} \\ \textit{Michigan State Univeristy, East Lansing, MI}\\\texttt{mhirn@msu.edu}\\[6pt]
	{\sc and}\\[6pt]
	{\sc Anna Little}\\[2pt]
	\textit{Department of Computational Mathematics, Science, and Engineering,} \\ \textit{Michigan State Univeristy, East Lansing, MI}\\
	\texttt{littl119@msu.edu}}

\maketitle

\begin{abstract}
    {We propose a nonlinear, wavelet based signal representation that is translation invariant and robust to both additive noise and random dilations. Motivated by the multi-reference alignment problem and generalizations thereof, we analyze the statistical properties of this representation given a large number of independent corruptions of a target signal. We prove the nonlinear wavelet based representation uniquely defines the power spectrum but allows for an unbiasing procedure that cannot be directly applied to the power spectrum. After unbiasing the representation to remove the effects of the additive noise and random dilations, we recover an approximation of the power spectrum by solving a convex optimization problem, and thus \edit{reduce to a phase retrieval problem}. Extensive numerical experiments demonstrate the statistical robustness of this approximation procedure.}
    {Multi-reference alignment, method of invariants, wavelets, signal processing, wavelet scattering transform} 
\end{abstract}


\section{Introduction}


The goal in classic multi-reference alignment (MRA) is to recover a hidden signal $f: \R \rightarrow \R$ from a collection of noisy measurements. Specifically, the following data model is assumed.
\begin{model}[Classic MRA]
	\label{model:classicMRA}
		The \textit{classic MRA data model} consists of $M$ independent observations of a \edit{compactly supported, real-valued} signal $f \edit{\in \Lb^2(\R)}$:
		\begin{equation} 
		y_j(x) = f(x-t_j) + \varepsilon_j(x) \, , \quad 1 \leq j \leq M \, ,
		\end{equation}	
		where:
		\begin{itemize}
			\item[(i)] $\text{supp}(y_j)\subseteq [-\frac{1}{2},\frac{1}{2}]$ for $1 \leq j \leq M$.
			\item[(ii)] $\{t_j\}_{j=1}^M$ are independent samples of a random variable $t\in \R$.
			\item[(iii)] $\{\varepsilon_j(x)\}_{j=1}^M$ are independent white noise processes on $[-\frac{1}{2},\frac{1}{2}]$ with variance $\sigma^2$.
		\end{itemize}
\end{model}
The signal is thus subjected to both random translation and additive noise. 
The MRA problem arises in numerous applications, including structural biology \cite{theobald2012optimal, diamond1992multiple, scheres2005maximum, sadler1992shift, park2011stochastic, park2014assembly}, single cell genomic sequencing \cite{leggett2015nanook}, radar \cite{zwart2003fast, gil2005using}, crystalline simulations \cite{sonday2013noisy}, image registration \cite{foroosh2002extension, brown1992survey, robinson2007optimal}, and signal processing \cite{zwart2003fast}. It is a simplified model relevant for Cryo-Electron Microscopy (Cryo-EM), an imaging technique for molecules which achives near atomic resolution \cite{bartesaghi20152, sirohi20163, bendory2017bispectrum}. In this application one seeks to recover a three-dimensional reconstruction of the molecule from many noisy two-dimensional images/projections \cite{frank2006three}. Although MRA ignores the tomographic projection of Cryo-EM, investigation of the simplified model provides important insights. For example, \cite{OptConvRates_MRA, perry2017sample} investigate the optimal sample complexity for MRA and demonstrate that $M = \Theta(\sigma^6)$ is required to fully recover $f$ in the low signal-to-noise regime when the translation distribution is periodic; this optimal sample complexity is the same for Cryo-EM \cite{bandeira2017estimation, wein2018statistical}. Recent work has established an improved sample complexity of $M = \Theta(\sigma^4)$ for MRA when the translation distribution is aperiodic \cite{abbe2018multireference}, and this rate has been shown to also hold in the more complicated setting of Cryo-EM, if the viewing angles are nonuniformly distributed \cite{sharon2019method}. Problems closely related to Model \ref{model:classicMRA} include the heterogenous MRA problem, where the unknown signal $f$ is replaced with a template of $k$ unknown signals $f_1, \ldots, f_k$ \cite{sorzano2010clustering, ma2019heterogeneous, boumal2018heterogeneous, perry2017sample}, as well as multi-reference factor analysis, where the underlying (random) signal follows a low rank factor model and one seeks to recover its covariance matrix \cite{landa2019multi}.

Approaches for solving MRA generally fall into two categories: \textit{synchronization methods} and \edit{methods which estimate the signal directly, i.e. without estimating nuisance parameters}. Synchronization methods attempt to recover the signal by aligning the translations and then averaging. \edit{They include methods based on \textit{angular synchronization} \cite{singer2011angular, boumal2016nonconvex, perry2018message, chen2018projected, bandeira2017tightness, zhong2018near}, where for each pair of signals the best pairwise shift is computed and then the translations are estimated from this pairwise information \cite{bandeira201518}, and \textit{semi-definite programming} \cite{bandeira2020non, bandeira2014multireference, chen2014near, bandeira2016low}, which approximates the quasi-maximum likelihood estimator of the shifts by relaxing a nonconvex rank constraint. However} these methods fail in the low signal-to-noise regime. \edit{Methods which estimate the signal directly include both the \textit{method of moments} \cite{hansen1982large, kam1980reconstruction, sharon2019method} and \textit{expectation maximization, or EM-type, algorithms} \cite{dempster1977maximum, abbe2018multireference}; a number of EM-type algorithms have also been developed for the more complicated Cryo-EM problem \cite{dvornek2015subspaceem, punjani2017cryosparc}. An important special case of the method of moments is} the \textit{method of invariants}, \edit{which} seeks to recover $f$ by computing translation invariant features, and thus avoids aligning the translations. However the task is a difficult one, as a complete representation is needed to recover the signal, and yet the representation may be difficult to invert and corrupted by statistical bias. Generally the signal is recovered from translation invariant moments, which are estimated in the Fourier domain \cite{hansen1982large, collis1998higher}. Recent work \cite{bendory2017bispectrum, OptConvRates_MRA} utilizes such Fourier invariants (mean, power spectrum, and bispectrum), and recovers $\widehat{f}$ by solving a nonconvex optimization problem on the manifold of phases. 

\begin{figure}
	\centering
	\begin{subfigure}[b]{0.32\textwidth}
		\centering
		\includegraphics[width=.85\textwidth]{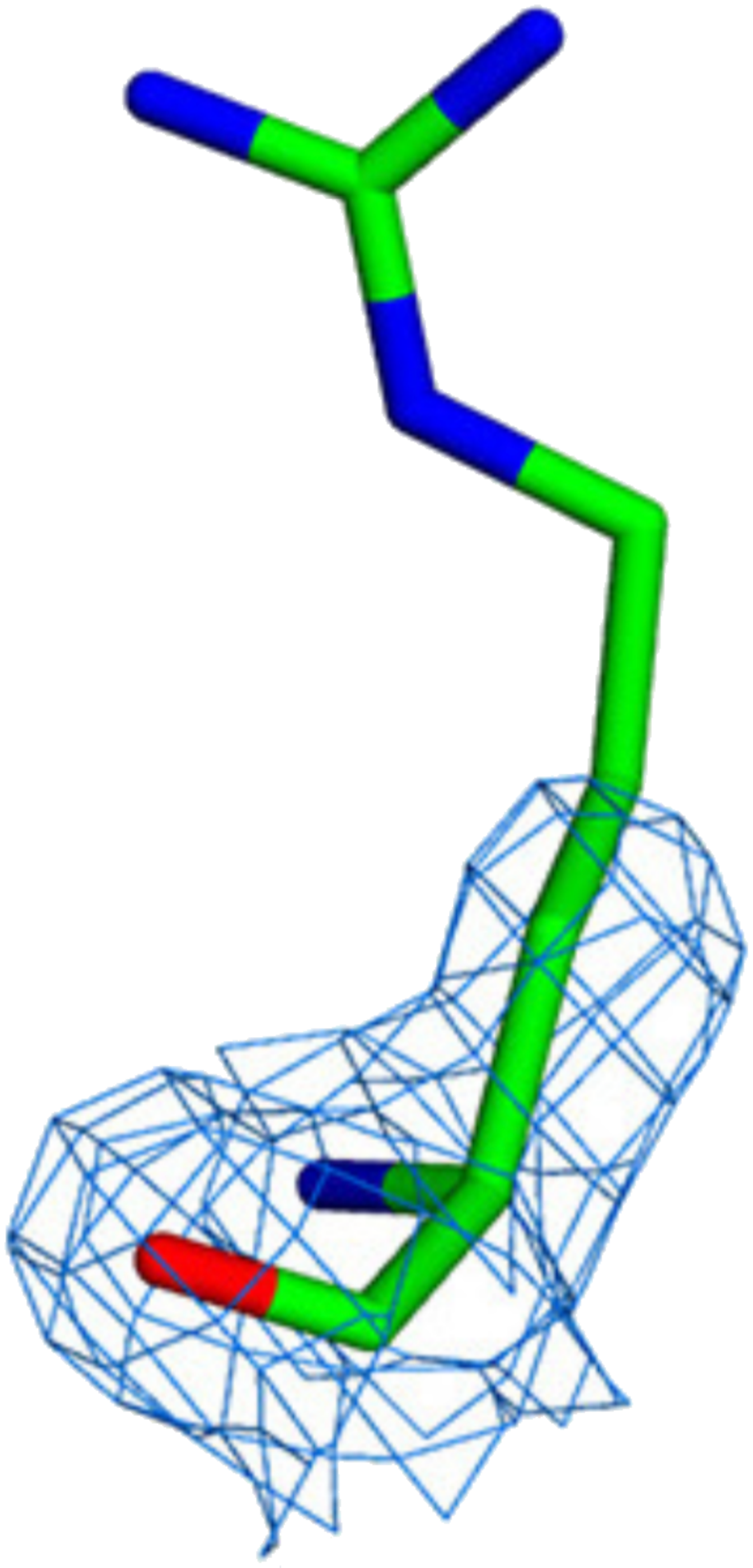}
		\caption{Molecule with flexible side chain.}
		\label{fig:molecule_original}
	\end{subfigure}
	\qquad
	\begin{subfigure}[b]{0.32\textwidth}
		\centering
		\includegraphics[width=.85\textwidth]{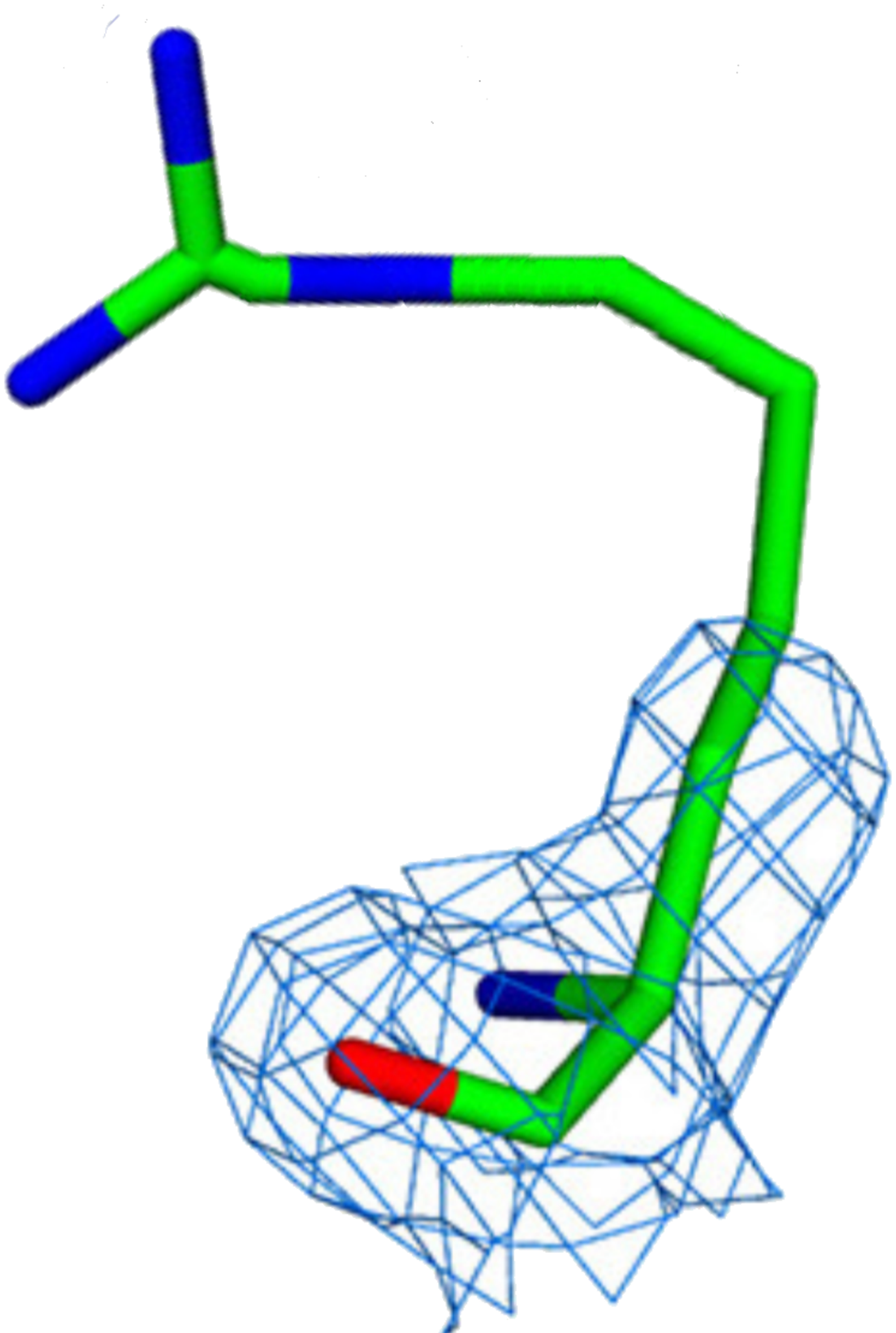}
		\caption{Diffeomorphism of Figure \ref{fig:molecule_original}.}
		\label{fig:molecule_jiggled}
	\end{subfigure}
	\caption{Dynamics arising from flexible regions in macromolecular structures \cite{palamini2016identifying}.}
	\label{fig:MoleculeFlexibleSidechain}
\end{figure}

\edit{Classic MRA however fails to capture many of the biological phenomena arising in molecular imaging, such as the random rotations of the molecules and the tomographic projection associated with the imaging of 3D objects. Another shortcoming is that the model fails to capture the dynamics which arise from flexible regions in macromolecular structures. These flexible regions are very important in structural biology, for example in understanding molecular interactions \cite{lim2002modular,ekman2005multi,levitt2009nature,forneris2012modular} and molecular recognition of epigenetic regulators of histone tails \cite{bowman2015post, desjarlais2016role, mcginty2016recognition}. The large scale dynamics of these regions makes imaging challenging \cite{villarreal2014cryoem}, and thus sample preparation in cryo-EM generally seeks to minimize these dynamics by focusing on well-folded macromolecules frozen in vitreous ice \cite{palamini2016identifying}. However this ``may severely impact ... the nature of the intrinsic dynamics and interactions displayed by macromolecules" \cite{palamini2016identifying}. Although modern cryo-EM is making great strides in understanding flexible systems \cite{fischer2015structure, bai2015sampling, fernandez2015cryo, merk2016breaking},} formulating \edit{models which are more capable of capturing the motions associated with the flexible regions of macromolecules could open the door to applying cryo-EM more broadly, i.e. to less well-folded macromolecules. Mathematically the motion of the flexible region can be modeled as a diffeomorphism. See Figure \ref{fig:MoleculeFlexibleSidechain}, which shows a molecule with a flexible side chain (\ref{fig:molecule_original}) and a diffeomorphism resulting from movement of the flexible region (\ref{fig:molecule_jiggled}). Figure \ref{fig:molecule_original} is taken from \cite{palamini2016identifying}, and Figure \ref{fig:molecule_jiggled} was obtained by deforming it.
}


This article \edit{thus} generalizes the classic MRA problem to include a random diffeomorphism. Specifically, we consider recovering a hidden signal $f: \R \rightarrow \R$ from
\begin{align*}
y_j(x) &= L_{\tau_j}f(x-t_j) + \varepsilon_j(x) \quad, \quad 1 \leq j \leq M,
\end{align*}
where $L_{\tau}$ is a dilation operator which dilates by a factor of $(1-\tau)$. The dilation operator $L_\tau$ is a simplified model for more general diffeomorphisms $L_{\zeta} f (x) = f(\zeta(x))$,  since in the simplest case when $\zeta (x)$ is affine, $L_{\zeta}$ simply translates and dilates $f$ (see Section \ref{sec: mra data models}). Dilations are also relevant for the analysis of time-warped audio signals, which can arise from the Doppler effect and in speech processing and bioacoustics. For example, \cite{omer2013estimation, omer2017time, meynard2018spectral} consider a stationary random signal $f(x)$ which is \textit{time-warped}, i.e. $D_\zeta f(x) = \sqrt{\zeta'(x)}f(\zeta(x))$, and use a maximum likelihood approach to estimate $\zeta$.  In \cite{clerc2002texture, clerc2003estimating},  a similar stochastic time warping model is analyzed using wavelet based techniques. The \edit{noisy dilation MRA} model considered here corresponds to the simplest case of time-warping, when $\zeta$ is an affine function. \edit{This special case is in fact very important in imaging applications  \cite{chandran1992position, capodiferro1987correlation, tsatsanis1990translation, hotta2001scale, martinec2007robust, robinson2007optimal}, where it is critical to compute features which are scale invariant, as objects are naturally dilated by the ``zoom" of an image. }


\edit{A new approach is needed to solve this more general MRA problem, as} Fourier invariants will fail, \edit{being} unstable to the action of diffeomorphisms, including dilations. The instability occurs in the high frequencies, where even a small diffeomorphism can significantly alter the Fourier modes. We instead propose $\Lb^2 (\R)$ wavelet coefficient norms as invariants, using a continuous wavelet transform. This approach is inspired by the invariant scattering representation of \cite{mallat:scattering2012}, which is provably stable to the actions of small diffeomorphisms. However here we replace local averages of the modulus of the wavelet coefficients with global averages (i.e. integrations) of the modulus squared, thus providing rigid invariants which can be statistically unbiased. Similar invariant coefficients have been utilized in a number of applications including predicting molecular properties \cite{eickenberg:3DSolidHarmonicScat2017, eickenberg:scatMoleculesJCP2018} and quantum chemical energies \cite{hirn:waveletScatQuantum2016}, and in microcanonical ensemble models for texture synthesis \cite{bruna:multiscaleMicrocanonical2018}. Recent work \cite{gao:graphScat2018} has also generalized such coefficients to graphs. 

\subsection{Notation}

The Fourier transform of a signal $f \in \Lb^1 (\R)$ is
\begin{equation*}
    \widehat{f} (\omega) = \int  f(x) e^{- i x \omega} \, dx \, .
\end{equation*}
\edit{We remind the reader that compactly supported $\Lb^2 (\R)$ functions are in $\Lb^1 (\R)$}. The power spectrum is the nonlinear transform $P : \Lb^2 (\R) \rightarrow \Lb^1 (\R)$ that maps $f$ to
\begin{equation*}
(Pf) (\omega) = |\widehat{f}(\omega)|^2, \quad \omega \in \R \, .
\end{equation*}
We denote $f(x) \leq C g(x)$ for some absolute constant $C$ by $f(x) \lesssim g(x)$. We also write $f(x)=O(g(x))$ if $|f(x)| \leq C g(x)$ 
for all $x\geq x_0$ for some constants $x_0, C>0$; \edit{$f(x)=o(g(x))$ denotes $f(x)/g(x) \rightarrow 0$ as $x\rightarrow \infty$}; $f(x)=\Theta(g(x))$ denotes $C_1 g(x)\leq |f(x)|\leq C_2 g(x)$ for all $x\geq x_0$ for some constants $x_0, C_1, C_2>0$. The minimum of $a$ and $b$ is denoted $a \wedge b$, \edit{and the maximum by $a\vee b$.}

\section{MRA models and the method of invariants}

Standard multi-reference alignment (MRA) models are generalized to models that include deformations of the underlying signal in Section \ref{sec: mra data models}. Section \ref{sec: invariant measurement models} reviews power spectrum invariants and introduces $\Lb^2 (\R)$ wavelet coefficient invariants. Theorem \ref{thm:WSC_PS_equivalence} proves wavelet coefficient invariants computed with a continuous wavelet transform and a suitable mother wavelet \edit{are equivalent to} the power spectrum, showing there is no information loss in the transition from one representation to the other. 

\subsection{MRA data models}
\label{sec: mra data models}

A standard multi-reference alignment (MRA) scenario considers the problem of recovering a signal $f \in \Lb^2 (\R)$ in which one observes random translations of the signal, each of which is corrupted by additive noise. The problem is particularly difficult when the signal to noise ratio is low, as registration methods become intractable. In  \cite{bendory2017bispectrum, OptConvRates_MRA, bendory2019multi, ma2019heterogeneous, singer2018mathematics, boumal2018heterogeneous}  the authors propose a method using Fourier based invariants, which are invariant to translations and thus eliminate the need to register signals. 

A more general MRA scenario incorporates random deformations of the signal $f$, which could be used to model underlying physical variability that is not captured by rigid transformations and additive noise models. For example \edit{\cite{bandeira2020non, bandeira2017estimation} consider a discrete signal $f$ corrupted by an arbitrary group action, \cite{zwart2003fast, hudson1993correlation} consider random deformations arising in RADAR, and} \cite{aizenbud2019rank} considers a generalization of MRA where signals are rescaled by random constants. Another natural mathematical model is small, random diffeomorphisms, which leads to observations of the form:
\begin{equation} \label{eqn: diffeomorphism model}
    y_j (x) = L_{\zeta_j} f (x - t_j) + \varepsilon_j (x), \quad 1 \leq j \leq M \, ,
\end{equation}
where $\zeta_j \in \Cb^1 (\R)$ is a random diffeomorphism, $t_j \in \R$ is a random translation, and the signals $\varepsilon_j(x)$ are independent white noise random processes. The transform $L_{\zeta}$ is the action of the diffeomorphism $\zeta$ on $f$,
\begin{equation*}
    L_{\zeta} f (x) = f(\zeta(x)) \, .
\end{equation*}
If $\| (\zeta^{-1})' \|_{\infty} < \infty$, then one can verify $L_{\zeta} : \Lb^2 (\R) \rightarrow \Lb^2 (\R)$. 

One of the keys to the Fourier invariant approach of \cite{bendory2017bispectrum, OptConvRates_MRA, bendory2019multi, ma2019heterogeneous, singer2018mathematics, boumal2018heterogeneous} is the authors can unbias the Fourier invariants of the noisy signals, thus allowing them to devise an unbiased estimator of the Fourier invariants of the signal $f$ (or a mixture of signals in the heterogeneous MRA case). For the diffeomorphism model \eqref{eqn: diffeomorphism model} this would require developing a procedure for unbiasing the (Fourier) invariants of $\{y_j\}_{j=1}^M$ against both additive noise and random diffeomorphisms. 

In order to get a handle on the difficulties associated with the proposed diffeomorphism model, in this paper we consider random dilations of the signal $f$, which corresponds to restricting the diffeomorphism to be of the form:
\begin{equation*}
\zeta (x) = \frac{x}{1-\tau}, \quad |\tau| \leq 1/2 \, .
\end{equation*}
Specifically, we assume the following \edit{noisy dilation MRA} model.
\begin{model}[\edit{Noisy dilation MRA} data model] 
\label{model:genMRA}
The \textit{\edit{noisy dilation MRA} data model} consists of $M$ independent observations of a \edit{compactly supported, real-valued} signal $f \edit{\in \Lb^2(\R)}$:
\begin{equation} \label{eqn: translation dilation noise model}
y_j (x) = L_{\tau_j}f(x - t_j) + \varepsilon_j(x) \, , \quad 1 \leq j \leq M \, ,
\end{equation}	
where $L_{\tau}$ is an $\Lb^1 (\R)$ normalized dilation operator, 
\begin{equation*}
L_{\tau}f(x) = (1 - \tau)^{-1} f \left( (1 - \tau)^{-1} x \right) \, .
\end{equation*}
In addition, we assume:
\begin{itemize}
	\item[(i)] $\text{supp}(y_j)\subseteq [-\frac{1}{2},\frac{1}{2}]$ for $1 \leq j \leq M$.
	\item[(ii)] $\{t_j\}_{j=1}^M$ are independent samples of a random variable $t\in \R$.
	\item[(iii)] $\{\tau_j\}_{j=1}^M$ are independent samples of a bounded, symmetric random variable $\tau$ satisfying:
	\begin{equation*}
	\tau \in \R\quad,\quad \Ex(\tau) = 0 \quad,\quad \Var(\tau) = \eta^2\quad,\quad |\tau| \leq 1/2.
	\end{equation*} 
	\item[(iv)] $\{\varepsilon_j(x)\}_{j=1}^M$ are independent white noise processes on $[-\frac{1}{2},\frac{1}{2}]$ with variance $\sigma^2$.
\end{itemize}
\end{model}	

\begin{remark}
	The interval $[-\frac{1}{2},\frac{1}{2}]$ is arbitrary and can be replaced with any interval of length 1. In addition, the spatial box size is arbitrary, i.e. $[-\frac{1}{2},\frac{1}{2}]$ can be replaced with $[-\frac{N}{2},\frac{N}{2}]$. All results still hold with $\sigma\sqrt{N}$ replacing $\sigma$ wherever it appears. 
\end{remark}		

Thus the hidden signal $f$ is supported on an interval of length $1$, and we observe $M$ independent instances of the signal that have been randomly translated, randomly dilated, and corrupted by additive white noise. 
We assume the hidden signal is real, but the proposed methods can also handle complex valued signals with minor modifications. Recall $\varepsilon(x)$ is a white noise process if $\varepsilon(x) = dB_x$, i.e. it is the derivative of a Brownian motion with variance $\sigma^2$.

While the \edit{noisy dilation MRA} model does not capture the full richness of the diffeomorphism model, it already presents significant mathematical difficulties. Indeed, as we show in Section \ref{sec: noisy dilation MRA model}, Fourier invariants, specifically the power spectrum, cannot be used to form accurate estimators under the action of dilations and random additive noise. The reason is that Fourier measurements are not stable to the action of small dilations (measured here by $|\tau|$), since the displacement of $\widehat{L_{\tau}f}(\omega)$ relative to $\widehat{f}(\omega)$ depends on $|\omega|$. Intuitively, high frequency modes are unstable, and yet high frequencies are often critical; for example removing high frequencies increases the sample complexity needed to distinguish between signals in a heterogeneous MRA model \cite{OptConvRates_MRA}. We thus replace Fourier based invariants with wavelet coefficient invariants, which are defined in Section \ref{sec: invariant measurement models}. As we show the wavelet invariants of the signal $f$ can be accurately estimated from wavelet invariants of the noisy signals $\{y_j\}_{j=1}^M$, with no information loss relative to the power spectrum of $f$.

For future reference we also define the following dilation MRA model, which includes random translations and random dilations but no additive noise. Thus Models \ref{model:classicMRA} and \ref{model:dilMRA} are both special cases of Model \ref{model:genMRA}.

\begin{model}[Dilation MRA data model] 
	\label{model:dilMRA}
	The \textit{dilation MRA data model} consists of $M$ independent observations of a \edit{compactly supported, real-valued} signal $f \edit{\in \Lb^2(\R)}$:
	\begin{equation} \label{eqn: translation dilation model}
	y_j (x) = L_{\tau_j}f(x - t_j)\, , \quad 1 \leq j \leq M \, ,
	\end{equation}	
	where $L_{\tau}$ is an $\Lb^1 (\R)$ normalized dilation operator, 
	\begin{equation*}
	L_{\tau}f(x) = (1 - \tau)^{-1} f \left( (1 - \tau)^{-1} x \right) \, .
	\end{equation*}
	In addition, we assume (i)-(iii) of Model \ref{model:genMRA}.
\end{model}	

\subsection{Method of invariants}
\label{sec: invariant measurement models}

We now discuss how invariant representations can be used to solve MRA data models, and introduce the wavelet invariants used in this article.

\subsubsection{Motivation and related work}
\label{sec: invariant mm motivation}

Let $T_t f(x) = f(x - t)$ denote the operator \edit{which translates by $t$} acting on a signal $f$. Invariant measurement models seek a representation $\Phi (f) \in \Bc$ in a Banach space $\Bc$ such that
\begin{equation} \label{eqn: translation invariance}
    \Phi (T_t f) = \Phi (f), \quad \forall \, t \in \R \, .
\end{equation}
In MRA problems, one additionally requires that
\begin{equation} \label{eqn: completeness}
    \Phi (f) = \Phi (g) \Longleftrightarrow g = T_t f \text{ for some } t \in \R \, .
\end{equation}
The first condition \eqref{eqn: translation invariance} removes the need to align random translations of the signal $f$, whereas the second condition \eqref{eqn: completeness} ensures that if one can estimate $\Phi (f)$ from the collection $\{\Phi (y_j)\}_{j=1}^M$, then one can recover an estimate of $f$ (up to translation) by solving
\begin{equation} \label{eqn: optimization}
    f^{\star} = \arginf_{g \in \Lb^1 \cap \Lb^2 (\R)} \| \Phi (g) - \Phi (f) \|_{\Bc}\, ,
\end{equation}
where $\| \cdot \|_{\Bc}$ is the Banach space norm.

When the observed signals $\{y_j\}_{j=1}^M$ are corrupted by more than just a random translation, though, as in Model \ref{model:genMRA}, estimating $\Phi (f)$ from $\{\Phi (y_j)\}_{j=1}^M$ is not always straightforward. Indeed, one would like to compute
\begin{equation} \label{eqn: biased rep}
    \overline{\Phi}_M(f) = \frac{1}{M} \sum_{j=1}^M \Phi (y_j) \, ,
\end{equation}
but the quantity $\overline{\Phi}_M(f)$ is not always an unbiased estimator of $\Phi (f)$, meaning that $\lim_{M \rightarrow \infty} \overline{\Phi}_M(f) \neq \Phi (f)$. In order to circumvent this issue, one must select a representation $\Phi$ such that
\begin{equation} \label{eqn: bias model}
    \E\, \Phi (y_j) = \Phi (f) + b_\Phi(f, \mathcal{M}) \, ,
\end{equation}
where $b_\Phi(f,\mathcal{M})$ is a bias term depending on the choice of $\Phi$, $f$, and the signal corruption model $\mathcal{M}$. If \eqref{eqn: bias model} holds and if we can compute a $\tilde{b}$ such that $\Ex\, \tilde{b}_\Phi(y_j, \mathcal{M}) = b_\Phi(f, \mathcal{M})+\delta$ for $|b_\Phi(f, \mathcal{M})| \gg |\delta|$, then one can amend \eqref{eqn: biased rep} to reduce the bias:
\begin{equation*}
\widetilde{\Phi}_M (f) = \frac{1}{M} \sum_{j = 1}^M ( \Phi (y_j) - \tilde{b}_\Phi(y_j, \mathcal{M})) \, ,
\end{equation*}	
in which case
\begin{equation*}
    \lim_{M \rightarrow \infty} \widetilde{\Phi}_M (f) = \Phi (f) + \delta
\end{equation*}
\edit{almost surely} by the law of large numbers. The main difficulty therefore is twofold. On the one hand, one must design a representation $\Phi$ that satisfies \eqref{eqn: translation invariance}, \eqref{eqn: completeness}, and \eqref{eqn: bias model} with a bias $b$ that can be estimated; on the other hand, the optimization \eqref{eqn: optimization} must be tractable. For random translation plus additive noise models (i.e., Model \ref{model:classicMRA}), the authors of \cite{bendory2017bispectrum, OptConvRates_MRA} describe a representation $\Phi$ based on Fourier invariants that satisfies the outlined requirements and for which one can solve \eqref{eqn: optimization} despite the optimization being non-convex. The Fourier invariants include $\widehat{f}(0)$ (i.e., the integral of $f$), the power spectrum of $f$, and the bispectrum of $f$. Each invariant captures successively more information in $f$. While $\widehat{f}(0)$ carries limited information, the power spectrum recovers \edit{the magnitude of the Fourier transform, namely it recovers the nonnegative, real-valued function $\rho(\omega)$ such that $\widehat{f}(\omega) = \rho(\omega)e^{i\theta(\omega)}$ but the phase information $\theta(\omega)$ is lost.} \edit{Since $\widehat{T_tf}(\omega)=e^{-i\omega t}\widehat{f}(\omega)$,} the power spectrum  is invariant to translations \edit{as} the Fourier modulus kills the phase factor induced by a translation $t$ of $f$.  \edit{However, it is in general not possible to recover a signal from its power spectrum, although in certain special cases the phase information} can be resolved; results along these lines are in the field of phase retrieval \cite{sun2017phaseless, cheng2017phaseless}. The bispectrum is also translation invariant and invertible so long as $\widehat{f}(\omega) \neq 0$ \cite{perry2017sample}. 

In Section \ref{sec: noisy dilation MRA model} we show that it is impossible to significantly reduce the power spectrum bias for Model \ref{model:genMRA}, which includes translations, dilations, and additive noise. We thus propose replacing the power spectrum with the $\Lb^2 (\R)$ norms of the wavelet coefficients of the signal $f$. These invariants satisfy \eqref{eqn: translation invariance} and \eqref{eqn: bias model} for Model \ref{model:genMRA}, and yield a convex formulation of \eqref{eqn: optimization}. They do not satisfy \eqref{eqn: completeness} for general $f \in \Lb^2 (\R)$, but Theorem \ref{thm:WSC_PS_equivalence} in Section \ref{sec: wavelet invariants} shows that knowing the wavelet invariants of $f$ is equivalent to knowing the power spectrum of $f$, which means that any phase retrieval setting in which recovery is possible will also be possible with the specified wavelet invariants. \edit{For example if the signal lives in a spline or shift invariant space in addition to being real-valued, then it can be recovered from its phaseless measurements \cite{sun2017phaseless, cheng2017phaseless}.}

\subsubsection{Wavelet invariants}
\label{sec: wavelet invariants}

We now define the wavelet invariants used in this article. A wavelet $\psi \in \Lb^2 (\R)$ is a waveform that is localized in both space and frequency and has zero average,
\begin{equation*}
    \int  \psi (x) \, dx = 0 \, .
\end{equation*}
\edit{Note throughout this article $\psi$ will always denote a wavelet in $\Lb^1\cap\Lb^2 (\R)$ with zero average, satisfying $\|\psi\|_2=1$ as well as the classic admissability condition $\int \frac{|\widehat{\psi}(\omega)|^2}{\omega}\ d\omega < \infty$.}
A dilation of the wavelet by a factor $\lambda \in (0, \infty)$ is denoted,
\begin{equation*}
\psi_{\lambda} (x) = \lambda^{1/2} \psi (\lambda x) \, ,
\end{equation*}
where the normalization guarantees that $\| \psi_{\lambda} \|_2 = \| \psi \|_2 = 1$. The continuous wavelet transform $W$ computes
\begin{equation*}
    Wf = \{ f \ast \psi_{\lambda} (x) : \lambda \in (0, \infty) \, , \, x \in \R \} \, .
\end{equation*}
The parameter $\lambda$ corresponds to a frequency variable. Indeed, if $\xi_0$ is the central frequency of $\psi$, the wavelet coefficients $f \ast \psi_{\lambda}$ recover the frequencies of $f$ in a band of size proportional to $\lambda$ centered at $\lambda \xi_0$. Thus high frequencies are grouped into larger packets, which we shall use to obtain a stable, invariant representation of $f$. 

The wavelet transform $Wf$ is equivariant to translations but not invariant. Integrating the wavelet coefficients over $x$ yields translation invariant coefficients, but they are trivial since $\int \psi_{\lambda} = 0$. We therefore compute $\Lb^2 (\R)$ norms in the $x$ variable, yielding the following nonlinear wavelet invariants:
\begin{definition}[Wavelet invariants]
	\label{def: wavelet invariants}
	The $\Lb^2$ wavelet invariants of a \edit{real-valued} signal $f\in \Lb^1\cap \Lb^2(\R)$ are given by
	\begin{equation} \label{eqn: first order scattering}
	(Sf) (\lambda) = \norm{ f \ast \psi_{\lambda} }_{\edit{2}}^2, \quad \lambda \in (0, \infty) \, ,
	\end{equation}
	where $\psi_{\lambda} (x) =\lambda^{1/2} \psi (\lambda x) $ are dilations of a mother wavelet $\psi$.
\end{definition}
\edit{Throughout this article $\psi$ can be taken as a Morlet wavelet, in which case $\psi$ is constructed to have frequency centered at $\xi$ by $\psi(x) = C_\xi \pi^{-1/4} e^{-x^2/2}(e^{i\xi x}-e^{-\xi^2/2})$ for $C_\xi =(1-e^{-\xi^2}-2e^{-3\xi^2/4})^{-1/2}$, but results hold more generally for what we refer to as $k$-admissible wavelets, where $k \geq 0$ is an even integer. See Appendix \ref{app:wavelet_admissability} for a precise description of this admissibility criteria.} The wavelet invariants can be expressed in the frequency domain as
\begin{align*}
(Sf) (\lambda) = \frac{1}{2\pi} \int |\widehat{f}(\omega)|^2 |\widehat{\psi}_\lambda(\omega)|^2\ d\omega\, ,
\end{align*}
which motivates the following definition of ``wavelet invariant derivatives."
\begin{definition}[Wavelet invariant derivatives]
	\label{def:WSCderiv}
	The $n$-th derivative of $(Sf)(\lambda)$ is defined as:
	\begin{align*}
	(Sf)^{(n)}(\lambda) &:= \frac{1}{2\pi} \int |\widehat{f}(\omega)|^2  \frac{d^{n}}{d\lambda^{n}} |\widehat{\psi}_\lambda(\omega)|^2 \ d\omega\, .
	\end{align*}
\end{definition}
\begin{remark}
Definition \ref{def: wavelet invariants} assumes $f \edit{:\R \rightarrow \R}$, which allows the wavelet $\psi$ to be either real or complex. Our results can easily be extended to complex $f$, but a strictly complex wavelet would be needed, with $Sf (\lambda)$ computed for all $\lambda \in (-\infty, \infty)\setminus {0}$.
\end{remark}
\begin{remark}
	\label{rmk:ComputationalCost}
	For a discrete signal of length $n$, computing the wavelet invariants via a continuous wavelet transform is $O(n^2)$, while computing the power spectrum is $O(n\log n)$. Thus one pays a computational cost to achieve greater stability with no loss of information. On the other hand, if wavelet invariants are computed for a dyadic wavelet transform (i.e. only for $O(\log n)$ $\lambda$'s), the computational cost is the same and stability is maintained, but more information is lost.
\end{remark}
\begin{remark}
	When $(Pf) (\omega) = |\widehat{f}(\omega)|^2$ is continuous, Definition \ref{def:WSCderiv} reduces to a normal derivative, i.e. one can check that $(Sf)^{(n)}(\lambda) =  \frac{d^{n}}{d\lambda^{n}} (Sf)(\lambda)$. However when $Pf$ is not continuous, in general $(Sf)^{(n)}(\lambda) \ne  \frac{d^{n}}{d\lambda^{n}} (Sf)(\lambda)$, and $(Sf)^{(n)}(\lambda)$ is more convenient for controling the error of the estimators proposed in this article.  Throughout this article, the notation $(Sf)^{(n)}(\lambda)$ will thus denote the derivative of Definition \ref{def:WSCderiv} and $\frac{d^{n}}{d\lambda^{n}} (Sf)(\lambda)$ will denote the standard derivative.   	
\end{remark}		 

\edit{Under mild conditions}, one can show that $S : \Lb^2 (\R) \rightarrow \Lb^1 \cap \Cb (0, \infty)$. The values $\lambda = 2^j$ for $j \in \Z$ correspond to rigid versions of first order $\Lb^2 (\R)$ wavelet scattering invariants \cite{mallat:scattering2012}. 
The continuous wavelet transform $Wf$ is extremely redundant; indeed, for suitably chosen mother wavelets the dyadic wavelet transform with $\lambda = 2^j$ for $j \in \Z$ is a complete representation of $f$. However, the corresponding operator $S$ restricted to $\lambda = 2^j$ \edit{is not invertible}. When one utilizes every frequency $\lambda \in (0, \infty)$, though, the resulting $\Lb^2 (\R)$ norms $(Sf) (\lambda) = \| f \ast \psi_{\lambda} \|_2^2$ uniquely determine the power spectrum of $f$, so long as the wavelet $\psi$ satisfies a type of independence condition. 

\begin{condition}
	\label{cond: linear indep wavelet}
	Define
	\begin{align*}
	|\widehat{\psi}_\lambda^+(\omega)|^2 &= \left(|\widehat{\psi}_\lambda(\omega)|^2 + |\widehat{\psi}_\lambda(-\omega)|^2\right)\cdot\ind(\omega \geq 0) \, .
	\end{align*}
	If for any finite sequence $\{\omega_i\}_{i=1}^n$ of distinct positive frequencies, the collection $\{|\widehat{\psi}^{+}_\lambda(\omega_i)|^2\}_{i=1}^n$ are linearly independent functions of $\lambda$, we say \edit{the wavelet} $\psi$ satifies the linear independence condition.
\end{condition}

\begin{remark}
	Condition \ref{cond: linear indep wavelet} is stated in terms of $|\widehat{\psi}_\lambda^+(\omega)|^2$ to avoid assumptions on whether $\psi$ is real or complex. When $\psi (x) \in \R$, $|\widehat{\psi}_\lambda^+(\omega)|^2 = 2|\widehat{\psi}_\lambda(\omega)|^2$ for $\omega \geq 0$. When $\psi$ is complex analytic, $|\widehat{\psi}_\lambda^+(\omega)|^2 = |\widehat{\psi}_\lambda(\omega)|^2$. When $\psi \in \mathbb{C}$ but not complex analytic, $|\widehat{\psi}_\lambda^+(\omega)|^2$ simply incorporates a reflection of $|\widehat{\psi}_\lambda(\omega)|^2$ about the origin. Since we assume $f(x) \in \R$, $|\widehat{\psi}_\lambda^+(\omega)|^2$ uniquely defines $(\Sc f)(\lambda)$, since  \\
	$(\Sc f)(\lambda) = \frac{1}{2\pi}\left\langle |\widehat{f}|^2, |\widehat{\psi}^+_{\lambda}|^2 \right\rangle$ by the Plancherel and Fourier convolution theorems.
\end{remark}

\begin{theorem}
	\label{thm:WSC_PS_equivalence}
Let $f,g \in \Lb^1 \cap \Lb^2 (\R)$ and \edit{assume} $\psi$ satisfies Condition \ref{cond: linear indep wavelet} \edit{and $\widehat{\psi}$ has compact support}. Then:
\begin{equation*}
    Sf = Sg \Longleftrightarrow Pf = Pg \, .
\end{equation*}
\end{theorem}

\begin{proof}
First assume $Pf = Pg$, which means $|\widehat{f}(\omega)|^2 = |\widehat{g}(\omega)|^2$ for almost every $\omega \in \R$. Using the Plancheral and Fourier convolution theorems,
\begin{align*}
    (Sf)(\lambda) = \int  |f \ast \psi_{\lambda}(x)|^2 \, dx &= \frac{1}{2\pi} \int  |\widehat{f}(\omega)|^2 |\widehat{\psi}_{\lambda}(\omega)|^2 \, d\omega \\
    &= \frac{1}{2\pi} \int  |\widehat{g}(\omega)|^2 |\widehat{\psi}_{\lambda}(\omega)|^2 \, d\omega = (Sg) (\lambda), \enspace \forall \, \lambda \in (0, \infty) \, .
\end{align*}

Now suppose $Sf = Sg$. Since $Sf$ and $Sg$ are continuous in $\lambda$, we have:
\begin{equation*}
    0 = (Sf)(\lambda) - (Sg)(\lambda) = \frac{1}{2\pi} \int  \left( |\widehat{f}(\omega)|^2 - |\widehat{g}(\omega)|^2 \right) |\widehat{\psi}_{\lambda} (\omega)|^2 \, d\omega, \enspace \forall \, \lambda \in (0, \infty) \, .
\end{equation*}
Since $f \in \Lb^1 \cap \Lb^2 (\R)$ we have $\widehat{f} \in \Lb^2 \cap \Lb^{\infty}(\R)$ and thus $Pf \in \Lb^1 \cap \Lb^{\infty}(\R)$. By interpolation we have $Pf \in \Lb^2 (\R)$, and the same for $Pg$. By applying Lemma \ref{lem:WSC_PS_equivalence} (stated below) with $p(\omega) = (Pf)(\omega)-(Pg)(\omega)$ \edit{(note $p$ is continuous since $f,g\in\Lb^1(\R)$)}, we conclude $Pf = Pg$ for almost every $\omega$.
\end{proof}

\begin{restatable}{lemma}{lemWSCPSequivalence}
	\label{lem:WSC_PS_equivalence}
Let $p \in  \Lb^2\edit{(\R)}$ \edit{be continuous and assume} $p(\omega)=p(-\omega)$, \edit{$\widehat{\psi}$ has compact support, and} Condition \ref{cond: linear indep wavelet}. Then 
\begin{align*}
\int  p(\omega) |\widehat{\psi}_\lambda(\omega)|^2\ d\omega &= 0\ \forall \lambda >0 \implies p=0\ \text{a.e.}
\end{align*}
\end{restatable}

The proof of Lemma \ref{lem:WSC_PS_equivalence} is in Appendix \ref{app:WSC_PS_equivalence}. We remark that \edit{many} wavelets satisfy Condition \ref{cond: linear indep wavelet} \edit{and have compactly supported Fourier transform,} so Theorem \ref{thm:WSC_PS_equivalence} is broadly applicable. \edit{For example, Proposition \ref{prop:linear_indep} below proves that any complex analytic wavelet with compactly supported Fourier transform satisfies Condition \ref{cond: linear indep wavelet}.} \edit{Morlet wavelets satisfy Condition \ref{cond: linear indep wavelet} (see Lemma \ref{lem:Morlet_lin_indep} in Appendix \ref{app:WSC_PS_equivalence}), but do not have compactly supported Fourier transform; however, $\widehat{\psi}$ does have fast decay for a Morlet wavelet and numerically we observe no issues. We also note, the assumption that $\widehat{\psi}$ has compact support in Theorem \ref{thm:WSC_PS_equivalence} can be removed if $f,g$ are bandlimited.} The following Proposition, proved in Appendix \ref{app:WSC_PS_equivalence}, gives some sufficient conditions guaranteeding Condition \ref{cond: linear indep wavelet}. 


\begin{restatable}{proposition}{proplinearindep}
	\label{prop:linear_indep}
	The following are sufficient to guarantee Condition \ref{cond: linear indep wavelet}:
	\begin{enumerate}
		\item[(i)] $|\widehat{\psi}(\omega)|^2$ has a compact support contained in the interval $[a,b]$, where $a$ and $b$ have the same sign, e.g., complex analytic wavelets with compactly supported Fourier transform.
		\item[(ii)] $|\widehat{\psi}(\omega)|^2 \in \Cb^{\infty} (\R)$ and there exists an $N$ such that all derivatives of order at least $N$ are nonzero at $\omega=0$, e.g., the Morlet wavelet.
	\end{enumerate}	
\end{restatable}


\edit{
\begin{remark} \label{rem: stability}
	In practice, $Pf, Sf$ are implemented as discrete vectors, and $Sf$ is obtained from $Pf$ via matrix multiplication, i.e. $Sf = F(Pf)$ for some real matrix $F$ with $F^TF$ strictly positive definite. Thus $\|Pf-Pg\|_2 \leq \sigma_{\min}^{-1}\|Sf-Sg\|_2$, where $\sigma_{\min}>0$ is the smallest singular value of the matrix $F$, and the spectral decay of $F$, which can be explicitly computed, thus determines the stability of the representation. The smoother the wavelet, the more rapidly the spectrum decays, since when $P\psi \in C^{p}$, $F^TF$ is defined by a $C^{p}$ kernel and thus has eigenvalues which decay like $o(1/n^{p+1})$ \cite{buescu2007eigenvalue}. There is thus a tradeoff between smoothness and stability. In this article we choose smoothness over stability, since smoothness is required for unbiasing noisy dilation MRA, and in our experiments the Morlet wavelet yielded the best results. We therefore invert the representation by solving an optimization problem which is initialized to be close to the desired solution (see Section \ref{sec:optimization}), and we avoid computing the pseudo-inverse of $F$, which is unstable for our smooth wavelet.  
\end{remark}}	

\section{Unbiasing for classic MRA}
\label{sec:AddNoise}

\edit{In this section} we consider the classic MRA model (Model \ref{model:classicMRA}). \edit{We} discuss unbiasing results for \edit{both} the power spectrum \edit{and} wavelet invariants, \edit{as well as} simulation results comparing the two methods.
\edit{In the following proposition we establish unbiasing results for the power spectrum by rederiving some results from \cite{bendory2017bispectrum}, extended to the continuum setting. The Proposition is proved in Appendix \ref{app:AddNoise}.}

\begin{restatable}{proposition}{propAddNoisePS}\label{prop:AddNoisePS}
	Assume Model \ref{model:classicMRA}.
	Define the following estimator of $(P f)(\omega)$:
	\begin{align*}
	(\widetilde{P f})(\omega) := \frac{1}{M} \sum_{j=1}^M (Py_j)(\omega)  - \sigma^2.
	\end{align*}
	Then with probability at least $1-1/t^2$,
	\begin{align}
	\label{equ:AddNoisePS}
	|(P f)(\omega)-(\widetilde{P f})(\omega)| &\leq \frac{2t\sigma}{\sqrt{M}}\left(\norm{f}_1+\sigma\right).
	\end{align}
\end{restatable}

We obtain an identical result for wavelet invariants (Proposition \ref{prop:AddNoiseWSC}) when signals are corrupted by additive noise only. \edit{See} Appendix \ref{app:AddNoise} \edit{for the proof}.

\begin{restatable}{proposition}{propAddNoiseWSC}\label{prop:AddNoiseWSC}
	Assume Model \ref{model:classicMRA}.
	Define the following estimator of $(\Sc f)(\lambda)$:
	\begin{align*}
	(\widetilde{\Sc f})(\lambda) := \frac{1}{M} \sum_{j=1}^M (\Sc y_j)(\lambda) - \sigma^2.
	\end{align*}
	Then with probability at least $1-1/t^2$,
	\begin{align}
	\label{equ:AddNoiseWSC}
	|(\Sc f)(\lambda)-(\widetilde{\Sc f})(\lambda)| &\leq \frac{2t\sigma}{\sqrt{M}}(\norm{f}_1+\sigma).
	\end{align}
\end{restatable}

As $M \rightarrow \infty$, the error of both the power spectrum and wavelet invariant estimators decays to zero at the same rate, and one can perfectly unbias both representations. As demonstrated in Section \ref{sec: noisy dilation MRA model}, this is not possible for \edit{noisy dilation MRA} (Model \ref{model:genMRA}), as there is a nonvanishing bias term. However a nonlinear unbiasing procedure on the wavelet invariants can significantly reduce the bias.  

We illustrate and compare additive noise unbiasing for power spectrum estimation using \edit{$(\widetilde{P f})$}, the power spectrum method of Proposition \ref{prop:AddNoisePS}, and \edit{$(\widetilde{\Sc f})$}, the wavelet invariant method of Proposition \ref{prop:AddNoiseWSC}. \edit{To approximate $(Pf)$ from the wavelet invariants $(\widetilde{\Sc f})$, we apply the convex optimization algorithm described in Section \ref{sec:optimization} to obtain $(\widetilde{P_{\Sc}f})$, the power spectrum approximation which best matches the wavelet invariants $(\widetilde{\Sc f})$. Thus throughout this article $(\widetilde{P_{\Sc}f})$ denotes a power spectrum estimator obtained by first unbiasing wavelet invariants and then running an optimization procedure, while $(\widetilde{Pf})$ denotes an estimator computed by directly unbiasing the power spectrum. Our simulations compare the $\Lb^2$ error of both of these estimators, i.e. we compare  $\|Pf-\widetilde{Pf}\|_2$ and $\|Pf-\widetilde{P_Sf}\|_2$.} 

Figure \ref{fig:AddNoiseCorruptedPS}  shows the uncorrupted power spectrum (red curve) of a medium frequency Gabor function ($f(x) = e^{-5x^2}\cos(16x)$), and the power spectrum after the signal is corrupted by additive noise with level $\sigma = 2^{-3}$ (blue curve); \edit{the signal-to-noise ratio ($\snr$) of the experiment is $0.56$ (see Section \ref{sec:signal_synthesis})}. Figure \ref{fig:AddNoiseErrorFixedSig} shows the $\Lb^2$ error of the power spectrum estimation for the two methods as a function of $\log_2(M)$ for a fixed \edit{$\snr$}, and Figure \ref{fig:AddNoiseL2ErrorFixedM} shows the $\Lb^2$ error as a function of $\log_2(\sigma)$ for a fixed $M$. The $\Lb^2$ errors for the two methods are similar; however, estimation via wavelet invariants is advantageous when the sample size $M$ is small or the additive noise level $\sigma$ is large. As $M$ becomes very large or $\sigma$ very small, the power spectrum method is preferable as the smoothing procedure of the wavelet invariants may numerically erase some extremely small scale features of the original power spectrum. 

\begin{figure}
	\centering
	\begin{subfigure}[b]{0.32\textwidth}
		\centering
		\includegraphics[width=.85\textwidth]{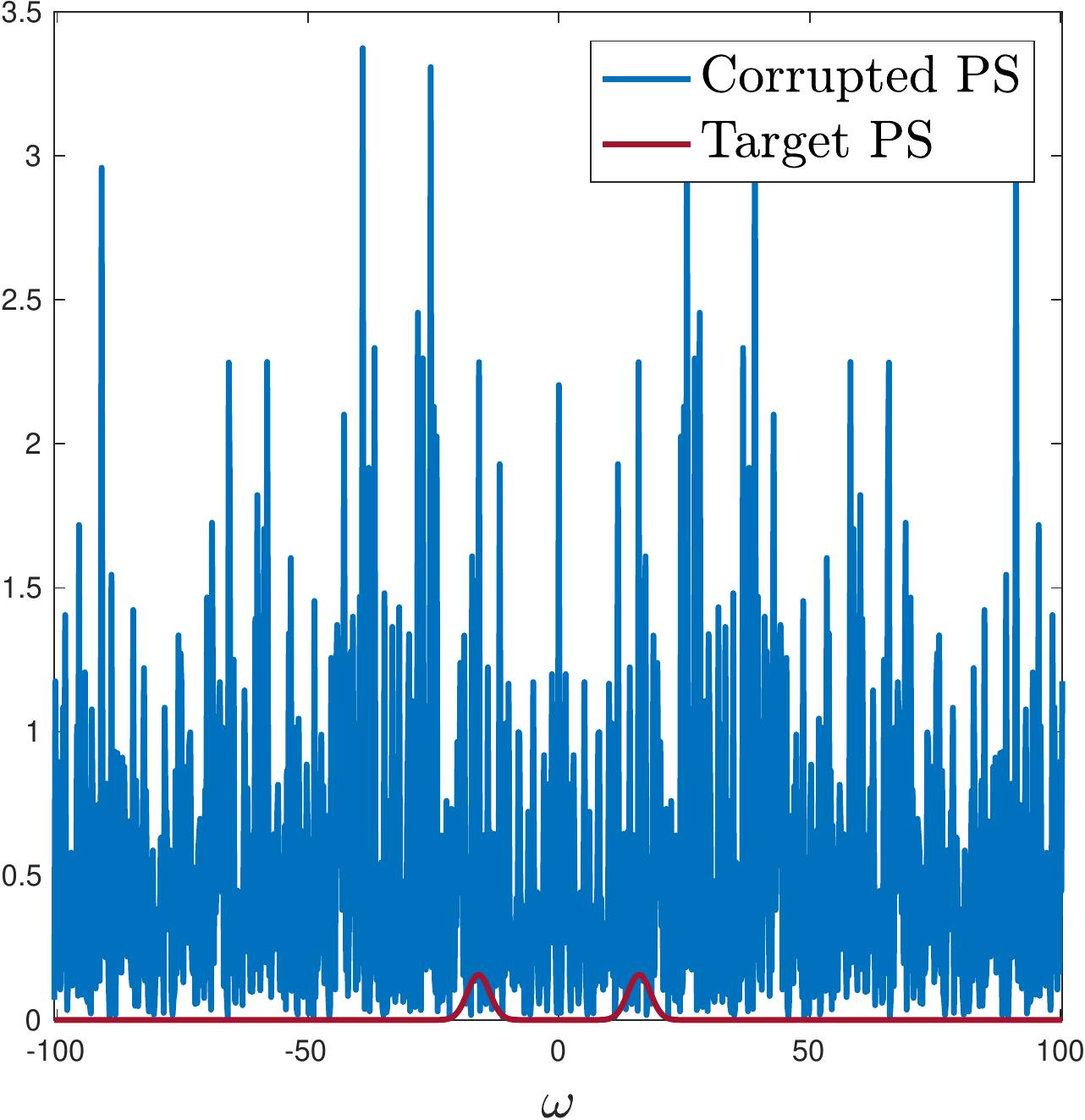}
		\caption{Noisy PS (\edit{$\snr=0.56$})}
		\label{fig:AddNoiseCorruptedPS}
	\end{subfigure}
	\hfill
		\begin{subfigure}[b]{0.32\textwidth}
		\centering
		\includegraphics[width=.85\textwidth]{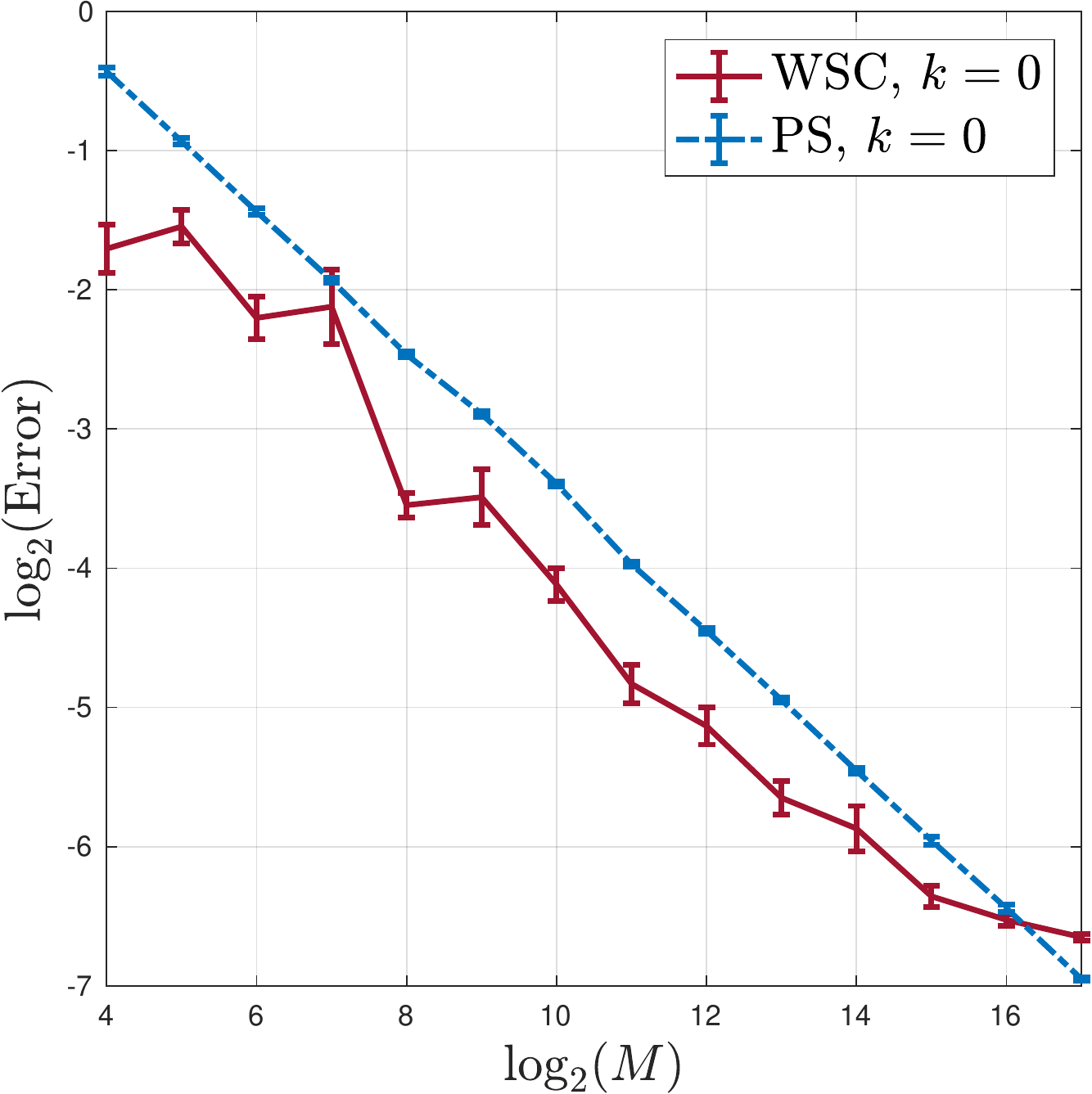}
		\caption{$\Lb^2$ error (\edit{$\snr = 0.56$})}
		\label{fig:AddNoiseErrorFixedSig}
	\end{subfigure}
	\hfill
	\begin{subfigure}[b]{0.32\textwidth}
		\centering
		\includegraphics[width=.85\textwidth]{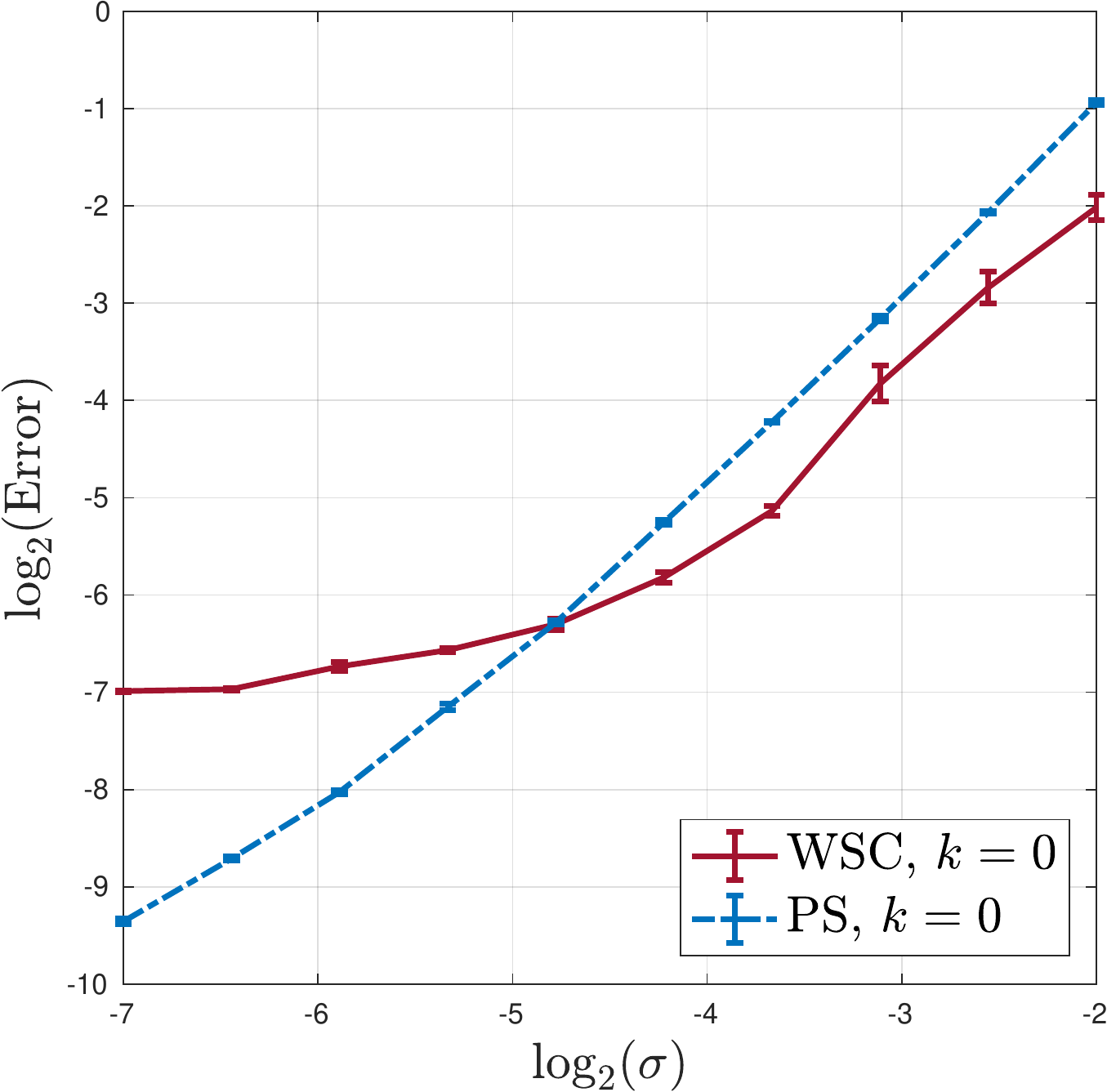}
		\caption{$\Lb^2$ error ($M = 500$)}
		\label{fig:AddNoiseL2ErrorFixedM}
	\end{subfigure}
	\caption{Simulation results for additive noise model for medium frequency Gabor $f(x) = e^{-5x^2}\cos(16x)$.}
	\label{fig:AddNoiseSim}
\end{figure}

\section{Unbiasing for dilation MRA}
\label{sec:dilationMRA}


In this section we analyze the dilation MRA model (Model \ref{model:dilMRA}).
We thus assume the signals have been randomly translated and dilated but there is no additive noise.

\edit{In fact there is a simple algorithm to recover $f$ under this model. Since $\norm{f_{\tau_j}}_2^2 = \norm{f}_2^2/(1-\tau_j)$, \\ $\frac{1}{M} \sum_{j=1}^M 1/\norm{f_{\tau_j}}_2^2$ is an unbiased estimator of $1/\norm{f}_2^2$, and so $\norm{f}_2^2$ can be accurately approximated. Once $\norm{f}_2^2$ is recovered, one can take any signal $y_j$ and dilate it so that $\norm{y_j}_2^2 = \norm{f}_2^2$, and the result will be an accurate approximation of the hidden signal $f$ for $M$ large. However, this approach collapses in the presence of even a small amount of additive noise. In the presence of additive noise, an alternative is to attempt a synchronization by centering each signal. The center $c_f$ of signal $f$ can be defined in the classical way by 
\begin{align*}
c_{f} &= \frac{1}{\norm{f}_2^2} \int x\, |f(x)|^2\ dx \, .
\end{align*}
Since the signals $y_j(x-(c_f+t_j))$ are perfectly aligned, one can thus attempt an alignment by defining $\widetilde{y}_j(x) = y_j(x-c_{y_j})$. However $c_{y_j}-(c_f+t_j) = O(\sigma\vee\sigma^2 + \eta)$, so significant errors arise in the synchronization which cannot be resolved by averaging. As our goal is ultimately to produce a method which can be extended to the noisy dilation MRA model, we abandon both the trivial solution (which cannot be extended to noisy dilation MRA) and the synchronization approach (which produces large errors), and explore a method based on empirical averages.}

\edit{We first observe that} random dilations cause $\frac{1}{M}\sum_{j=1}^M(P y_j)(\omega)$ and $\frac{1}{M}\sum_{j=1}^M(S y_j)(\lambda)$ to be biased estimators of $(Pf)(\omega)$ and $(\Sc f)(\lambda)$, and the bias for both is $O(\eta^2)$, where $\eta^2$ is the variance of the dilation distribution. However if the moments of the dilation distribution are known and $Pf, \Sc f$ are sufficiently smooth, one can apply an unbiasing procedure to the above estimators so that the resulting bias is $O(\eta^{k+2})$, where $k \geq 2$ is an even integer.

Throughout this section we assume $k\geq 2$ is an even integer, and define the constants $C_i$ from the first $k/2$ even moments of $\tau$ by $\Ex [\tau^i] = C_i \eta^i$ for $i=2,4,\ldots,k$. Note since we assume $\Ex[\tau^2]=\eta^2$, $C_2=1$. We define the constants $B_2,B_4, \ldots, B_k$ by solving
\begin{align}
\label{equ:MomentBasedConstants}
\frac{C_i}{i!} - \frac{B_2C_{i-2}}{(i-2)!} - \ldots - \frac{B_{i-2}C_2}{2!}-B_i &= 0
\end{align}
for $i=2,4,\ldots,k$; these constants are deterministic functions of the moments of $\tau$. A nonrecursive formula related to the Euler numbers can be derived which defines $B_i$ explicitly in terms of $C_2, \ldots, C_i$; however the recursive formula (\ref{equ:MomentBasedConstants}) is easier to implement numerically. 

We introduce two additional moment-based constants which are defined by the $C_i, B_i$ constants:
\begin{align}
\label{equ:T}
T &:= \max_{i = 0, 2,\ldots} C_i^{\frac{1}{i}} \\
\label{equ:E}
E &:= \max_{i=0,2,\ldots,k} \, \max_{j=0,\ldots,k+2-i} \left(\frac{T^j}{j!}|B_i|\right)^{\frac{1}{i+j}}\, ,
\end{align}
where $C_0, |B_0| = 1$, and when $i=j=0$ in (\ref{equ:E}), $\left(\frac{T^j}{j!}|B_i|\right)^{\frac{1}{i+j}}$ is replaced with 1. 

\begin{remark}
	Since the distribution of $\tau$ is bounded, we are guaranteed that $T< \infty$, and in general can consider both $T$ and $E$ to be $O(1)$ constants. For example for the uniform distribution, $T \leq \sqrt{3}$ and $|B_i| \leq \frac{|\text{Euler}(i)|}{i!} \leq 1$ which gives $E \leq \sqrt{3}$.
\end{remark}

We utilize the following two lemmas, which are proved in Appendix \ref{app:WSC}, to derive results for both the power spectrum and wavelet invariants.

\begin{restatable}{lemma}{lemRDGeneralUnbiasingMeanVar}\label{lem:RD_GeneralUnbiasing_MeanVar}
	Let $F_\lambda(\tau) = L((1-\tau)\lambda)$ for some function $L \in \Cb^{k+2}(0, \infty)$ and a random variable $\tau$ satisfying the assumptions of Section \ref{sec: mra data models}, and let $k\geq 2$ be an even integer. Assume \edit{there exist functions $\Lambda_i:\R\rightarrow\R$, \edit{$R:\R\rightarrow\R$}  such that}
	\begin{align*}
	|\lambda^iL^{(i)}(\lambda)| &\leq \Lambda_i(\lambda) \ \text{ for }\ 0 \leq i \leq k+2\quad,\quad \frac{\Lambda_{k+2}((1-\tau)\lambda)}{\Lambda_{k+2}(\lambda)}\leq R\edit{(\lambda)},
	\end{align*}
	and define the following estimator of $L(\lambda)$:
	\begin{align*}
	G_\lambda(\tau) &:= F_\lambda(\tau) - B_2\eta^2F''_\lambda(\tau)-B_4\eta^4F^{(4)}_\lambda(\tau)-\ldots -B_k\eta^kF^{(k)}_\lambda(\tau).
	\end{align*}
	Then $G_\lambda(\tau)$ satisfies
	\begin{align*}
	|\Ex\ G_\lambda(\tau) - L(\lambda)| &\lesssim k\edit{R(\lambda)}\Lambda_{k+2}(\lambda)(2E\eta)^{k+2} \\[.1cm]
	\Var\ G_\lambda(\tau) &\lesssim k^2\edit{R(\lambda)}^2\LambdaConstant(\lambda)^2
	\end{align*}
	where
	\begin{align*}
	\LambdaConstant(\lambda)^2 &:= \sum_{ \substack{ 0\leq i, j \leq k+2, i+j\geq 2}}\, \Lambda_{i}(\lambda) \Lambda_{j}(\lambda)(2E\eta)^{i+j}\, 
	\end{align*} 
	and $E$ is \edit{the} absolute constant defined in (\ref{equ:E}). 
\end{restatable}	

\begin{restatable}{lemma}{lemRDGeneralUnbiasingDevOfEstimator}\label{lem:RD_GeneralUnbiasing_DevOfEstimator}
	Let the assumptions and notation of Lemma \ref{lem:RD_GeneralUnbiasing_MeanVar} hold, and let $\tau_1, \ldots, \tau_M$ be independent. Define:
	\begin{align*}
	\widetilde{L}(\lambda) := \frac{1}{M} \sum_{j=1}^M G_\lambda(\tau_j).
	\end{align*}
	Then with probability at least $1-1/t^2$
	\begin{align*}
	| \widetilde{L}(\lambda) - L(\lambda)| &\lesssim k\edit{R(\lambda)}\left(\Lambda_{k+2}(\lambda)(2E\eta)^{k+2} +\frac{t\LambdaConstant(\lambda)}{\sqrt{M}}\right) \, .
	\end{align*}	
\end{restatable}	
The deviation of the estimator $\widetilde{L}(\lambda)$ from $L(\lambda)$ thus depends on two things: (1) the bias of the estimator which is $O(\eta^{k+2})$ and (2) the standard deviation of the estimator which is $O(\eta M^{-\frac{1}{2}})$, since $\LambdaConstant(\lambda)=O(\eta)$. 

\subsection{Power spectrum results for dilation MRA}
\label{sec:PSdilationMRA}

We now show how this unbiasing procedure based on both the moments of $\tau$ and the even derivatives of $P y$ can be used to obtain an estimator of $Pf$.

\begin{proposition}
	\label{prop:RandomDilations_PS}
	Assume Model \ref{model:dilMRA} and $Pf \in \Cb^{k+2} (\R)$.
	Define the following estimator of  $(P f)(\omega)$:	
	\begin{align*}
	(\widetilde{P f})(\omega) &:= \frac{1}{M} \sum_{j=1}^M \left[ (Py_j)(\omega) -B_2\eta^2\omega^2(Py_j)''(\omega) - \ldots - B_k\eta^k\lambda^k(Py_j)^{(k)}(\omega)\right]
	\end{align*}
	where the constants $B_i$ satisfy (\ref{equ:MomentBasedConstants}). Let: 
	\begin{align*}
	\Omega_i(\omega) &=  |\omega^i (Pf)^{i}(\omega)| \text{ for } 0 \leq i \leq k+2 \quad,\quad \edit{R(\omega)} = \max_{\tau} \frac{\Omega_{k+2}((1-\tau)\omega)}{\Omega_{k+2}(\omega)}.
	\end{align*}
	Then for all $\omega \ne 0$, with probability at least $1-1/t^2$, 
	\begin{align}
	| (\widetilde{P f})(\omega) - (P f)(\omega)| &\lesssim k\edit{R(\omega)}\left(\Omega_{k+2}(\omega)(2E\eta)^{k+2} +\frac{t\mathbf{\Omega}(\omega)}{\sqrt{M}}\right), \label{eqn: Ptilde estimate of P}
	\end{align}
	where 
	\begin{align*}
	\mathbf{\Omega}(\omega) &= \sum_{ \substack{ 0\leq i, j \leq k+2, i+j\geq 2}} \Omega_i(\omega)\Omega_j(\omega)(2E\eta)^{i+j} \,.
	\end{align*} 
\end{proposition}

\begin{proof}
	Since $P f$ is a translation invariant representation, we can ignore the translation factors $\{t_k\}_{k=1}^M$ and consider the model $y_j = L_{\tau_j} f$. In addition since $y_j(x) \in \R$, $(Py_j)(\omega) =	(Py_j)(-\omega)$ and it is sufficient to consider $\omega \in (0, \infty)$. Proposition \ref{prop:RandomDilations_PS} then follows directly from Lemma \ref{lem:RD_GeneralUnbiasing_DevOfEstimator} with $\lambda=\omega$, $L=Pf$ since $(Py_j)(\omega) = (Pf)((1-\tau_j)\omega)=F_\omega(\tau_j)$, \edit{$\Lambda_i = \Omega_i$, and $\Lambda = \Omega$}.
\end{proof}

We postpone a discussion of the shortcomings of Proposition \ref{prop:RandomDilations_PS} to Section \ref{sec:MotivatingExample}, where we compare the power spectrum and wavelet invariant results for dilation MRA.

\subsection{Wavelet invariant results for dilation MRA}
\label{sec:WSCdilationMRA}

We now apply the same unbiasing procedure to the wavelet invariants. Unlike for the power spectrum, where the error may depend on the frequency $\omega$ (see \eqref{eqn: Ptilde estimate of P} and Section \ref{sec:MotivatingExample}), the wavelet invariant error can be uniformly bounded independently of $\lambda$ with high probability. The following two Lemmas establish bounds on the derivatives of $(Sf)(\lambda)$ and are needed to prove Proposition \ref{prop:RandomDilations}; they are proved in Appendix \ref{app:WSC_props}.

\begin{restatable}{lemma}{lemWaveletScatteringDeriv}[Low Frequency Bound] 
	\label{lem:WaveletScatteringDeriv}
	Assume $P\psi \in \Cb^m (\R)$ and $f\in \Lb^1 (\R)$. Then the quantity $|\lambda^m (\Sc f)^{(m)}(\lambda)|$ can be bounded uniformly over all $\lambda$. Specifically:
	\begin{align*}
	|\lambda^m (Sf)^{(m)}(\lambda)| &\leq \Psi_m\norm{f}_1^2
	\end{align*}
	for $\Psi_m$ defined in (\ref{equ:Psik}). 	
\end{restatable}

\begin{restatable}{lemma}{lemWaveletScatteringDerivHighFreqDiffFunc}[High Frequency Bound for Differentiable Functions] 
	\label{lem:WaveletScatteringDerivHighFreq_DiffFunc}
	Assume $P\psi \in \Cb^m (\R)$, and $f' \in \Lb^1 (\R)$. Then the quantity $|\lambda^m{(\Sc f)}^{(m)}(\lambda)|$ can be bounded by:
	\begin{align*}
	|\lambda^m(Sf)^{(m)}(\lambda)|&\leq \frac{\Theta_m}{\lambda^2} \norm{f'}_1^2
	\end{align*}
	for $\Theta_m$ defined in (\ref{equ:Thetak}). 	
\end{restatable}

\edit{When $\psi$ is a Morlet wavelet or more generally when $\psi$ is $(k+2)$-admissable as described in Appendix \ref{app:wavelet_admissability},} these lemmas allow one to bound the error of the order $k$ wavelet invariant estimator for dilation MRA in terms of the following quantities:
\begin{align}
\label{equ:LambdaConstant}
\Lambda_i(\lambda) =  \Psi_{i} \norm{f}_1^2\,\wedge\, \frac{\Theta_{i}}{\lambda^2}\norm{f'}_1^2 \quad,\quad 	\LambdaConstant(\lambda)^2 =  \sum_{ \substack{ 0\leq i, j \leq k+2, i+j\geq 2}} \, \Lambda_i(\lambda)\Lambda_j(\lambda)(2E\eta)^{i+j} \, ,
\end{align}
where \edit{$\Psi_i, \Theta_i$ are defined in (\ref{equ:Psik}), (\ref{equ:Thetak}) and} $E$ is defined in (\ref{equ:E}).

\begin{proposition}
	\label{prop:RandomDilations}
	Assume Model \ref{model:dilMRA}, the notation in (\ref{equ:LambdaConstant}), and that $\psi$ is $(k+2)$-admissable.
	Define the following estimator of $(\Sc f)(\lambda)$:
	\begin{align*}
	(\widetilde{\Sc f})(\lambda) &:= \frac{1}{M} \sum_{j=1}^M \left[(S y_j)(\lambda) -B_2\eta^2\lambda^2(S y_j)''(\lambda) - \ldots - B_k\eta^k\lambda^k(S y_j)^{(k)}(\lambda)\right]
	\end{align*}
	 where the constants $B_i$ satisfy (\ref{equ:MomentBasedConstants}). Then with probability at least $1-1/t^2$,
	\begin{align*}
	| (\widetilde{\Sc f})(\lambda) - (\Sc f)(\lambda)| &\lesssim k\left(\Lambda_{k+2}(\lambda)(2E\eta)^{k+2} +\frac{t\LambdaConstant(\lambda)}{\sqrt{M}} \right)\, .
	\end{align*}

\end{proposition}

\begin{proof}
	Since $\Sc f$ is a translation invariant representation, we can ignore the translation factors $\{t_k\}_{k=1}^M$ and consider the model $y_j = L_{\tau_j}f$. Since $\psi$ is $k+2$-admissable, $\widehat{\psi} \in \Cb^{k+2} (\R)$ which guarantees $(Sf)(\lambda) \in \Cb^{k+2} (0, \infty)$. We note that since $f\in \Lb^1 (\R)$, $Pf$ is continuous, and the Leibniz integral rule guarantees that $(Sf)^{(n)}(\lambda) = \frac{d^n}{d\lambda^n}(Sf)(\lambda)$ for $1 \leq n \leq k+2$. By applying Lemma \ref{lem:WaveletScatteringDeriv}, we have $|\lambda^{i}({\Sc}f)^{(i)}(\lambda)| \leq \Psi_i \norm{f}_1^2$ for all $0 \leq i \leq k+2$, so that Lemma \ref{lem:RD_GeneralUnbiasing_DevOfEstimator} holds for $L(\lambda) = (\Sc f)(\lambda)$, $\Lambda_i(\lambda) = \Psi_i \norm{f}_1^2$, and $\edit{R(\lambda)}=1$. Now by applying Lemma \ref{lem:WaveletScatteringDerivHighFreq_DiffFunc}, we have $|\lambda^{i}({\Sc}f)^{(i)}(\lambda)| \leq \Theta_i \norm{f'}_1^2/\lambda^2$ for all $0 \leq i \leq k+2$, so that  Lemma \ref{lem:RD_GeneralUnbiasing_DevOfEstimator} also holds for $L(\lambda) = (\Sc f)(\lambda)$, $\Lambda_i(\lambda) = \Theta_i \norm{f'}_1^2/\lambda^2$, and $\edit{R(\lambda)}=4$ (note since $|\tau| \leq \frac{1}{2}$, $\Lambda_{k+2}((1-\tau)\lambda) / \Lambda_{k+2}(\lambda) \leq 4$). Thus Lemma \ref{lem:RD_GeneralUnbiasing_DevOfEstimator} in fact holds with $\Lambda_i(\lambda) = \left(\Psi_{i} \norm{f}_1^2\wedge \frac{\Theta_{i}}{\lambda^2}\norm{f'}_1^2\right)$; since $(\Sc y_j)(\lambda) = (\Sc f)((1-\tau_j)\lambda)=F_\lambda(\tau_j)$, we obtain Proposition \ref{prop:RandomDilations}. 	
\end{proof}

Since $\Lambda_{i}(\lambda) \leq \Psi_i \norm{f}_1^2$, Proposition \ref{prop:RandomDilations} guarantees that the error can be uniformly bounded independent of $\lambda$. In addition if the signal is smooth, the error for high frequency $\lambda$ will have the favorable scaling $\lambda^{-2}$. An important question in practice is how to choose $k$, i.e. what order wavelet invariant estimator minimizes the bias. 
Consider for example when $f' \notin \Lb^1 (\R)$, and $\Lambda_{k+2}(\lambda)=\Psi_{k+2}\norm{f}_1^2$. By using a second order estimator, we can decrease the bias from $O(\eta^2)$ to $O(\eta^4)$, and we can further decrease the bias to $O(\eta^6)$ by choosing $k=4$. However, $\Psi_k$ increases very rapidly in $k$. Indeed, as can be seen from (\ref{equ:Psik}), $\Psi_k$ increases like $k!$. Thus one possible heuristic (assuming $\eta$ is known) is to choose $k=\widetilde{k}$ where $\widetilde{k}$ minimizes the bias upper bound $k\Psi_{k+2}(2E\eta)^{k+2}$. 
Since $\Psi_{k}$ increases factorially, $\Psi_k \sim (Ck)^k$ for some constant $C$, and $\widetilde{k}+2$ will be inversely proportional to $\eta$, that is $(\widetilde{k}+2) \sim \eta^{-1}$. 
The following corollary of Proposition \ref{prop:RandomDilations} then holds for any $k \leq \widetilde{k}$.

\begin{corollary}
	\label{cor:RandomDilationWithBestk}
	Under the assumptions of Proposition \ref{prop:RandomDilations}, if $\Psi_i(2E\eta)^i$ is decreasing for $i \leq k+2$, then with probability at least $1-1/t^2$:
	\begin{align}
	| (\widetilde{\Sc f})(\lambda) - (\Sc f)(\lambda)| &\lesssim \norm{f}_1^2 \left(k\Psi_{k+2}(2E\eta)^{k+2} +\frac{tk^2\eta}{\sqrt{M}} \right)\, .
	\end{align}
	Similarly, if $\Theta_i(2E\eta)^i$ is decreasing for $i \leq k+2$, then with probability at least $1-1/t^2$:
	\begin{align}
	| (\widetilde{\Sc f})(\lambda) - (\Sc f)(\lambda)| &\lesssim \frac{\norm{f'}_1^2}{\lambda^2} \left(k\Theta_{k+2}(2E\eta)^{k+2} +\frac{tk^2\eta}{\sqrt{M}} \right)\, .
	\end{align}
\end{corollary}

\begin{remark}
	We observe that for a discrete lattice $I$ of $\lambda$ values, we can define the discrete 1-norm by $\| g \|_{\Lb^1(I)} = \sum_{\lambda\in I} |g(\lambda)| \ \Delta\lambda$. Assume the lattice has cardinality $n$, and that $\Psi_i(2E\eta)^i, \Theta_i(2E\eta)^i$ are decreasing for $i\leq k+2$. Applying Proposition \ref{prop:RandomDilations} with $t=\sqrt{n}s$ and a union bound over the lattice gives
	\begin{align*}
	\| \widetilde{\Sc f} - \Sc f \|_{\Lb^1(I)} &\lesssim k\left(\norm{f}_1^2\Psi_{k+2}+ \norm{f'}_1^2\Theta_{k+2}\right)(2E\eta)^{k+2}
	+ \frac{s\sqrt{n}k^2\eta}{\sqrt{M}}\left(\norm{f}_1^2+\norm{f'}_1^2\right)
	\end{align*}
	with probability at least $1-1/s^2$. When $n \ll M$, which is the context for MRA, the 1-norm of the error is $O(\eta^{k+2})$ as $M \rightarrow \infty$.
\end{remark}

%
%

\subsection{Comparison}
\label{sec:MotivatingExample}

Although Propositions \ref{prop:RandomDilations} and \ref{prop:RandomDilations_PS} at first glance appear quite similar, the wavelet invariant method has several important advantages over the power spectrum method, which we enumerate in the following remarks. 

\begin{remark}
	Proposition \ref{prop:RandomDilations} (wavelet invariants) applies to any signal satisfying $f \in \Lb^1(\R)$ but Proposition \ref{prop:RandomDilations_PS} requires $Pf \in \Cb^{k+2} (\R)$. Thus as $k$ is increased the power spectrum results apply to an increasingly restrictive function class. Furthermore, as discussed in Section \ref{sec: noisy dilation MRA model}, if the signal contains any additive noise, $Py_j$ is not even $\Cb^1$, which means the unbiasing procedure of Proposition \ref{prop:RandomDilations_PS} cannot be applied. On the other hand, by choosing $P\psi \in \Cb^{\infty} (\R)$, $Sf$ will inherit the smoothness of the wavelet, and the wavelet invariant results will hold for any $f \in \Lb^1(\R)$ and any $k$.
\end{remark}

\begin{remark}
    \label{rmk:PSbad}
    \edit{Since $(Pf_\tau)(\xi) = (Pf)((1-\tau)\xi)$, dilation will transport the frequency content at $\xi$ to $(1-\tau)\xi$, so that the displacement is $\tau\xi$. Thus when $\xi$ is very large, $|(Pf)(\xi) - (P f_\tau)(\xi)|$ can be large even for $\tau$ small. Because the wavelet invariants bin the frequency content, and these bins become increasingly large in the high frequencies, this does not occur for wavelet invariants. More specifically,}
	there is always a signal $f$ and frequency $\xi$ for which $|(Pf)(\xi) - (\widetilde{P f})(\xi)|$ is large regardless of $k$.
	Consider for example when $(Pf)(\omega)=e^{-(\omega-\xi)^2}$. Then $\Omega_k(\xi) \sim \xi^{k}$, and $|(Pf)(\xi) - (\widetilde{P f})(\xi)| \gtrsim 1$. 
	However for $M$ large enough, the order $k$ wavelet invariant estimator satisfies $|(\Sc f)(\lambda) - (\widetilde{\Sc f})(\lambda)|=O(k\Psi_{k+2}\eta^{k+2})$ for all $\lambda$. The wavelet invariants are thus stable for high frequency signals, where the power spectrum fails.  
\end{remark}

\begin{remark}
	For the wavelet invariants there will be a unique $\widetilde{k}$ which minimizes $k\Psi_{k+2}(2E\eta)^{k+2}$, and $\widetilde{k}$ does not depend on $\lambda$. Furthermore, $\widetilde{k}$ can be explicitly computed given the wavelet $\psi$ and moment constant $E$. On the other hand, the minimum of $k\Omega_{k+2}(\omega)(2E\omega)^{k+2}$ with respect to $k$ will depend on both the frequency $\omega$ and the signal $f$, so that $\widetilde{k}=\widetilde{k}(\omega, f)$, and it becomes unclear how to choose the unbiasing order. 
\end{remark}

\subsection{Simulation results for dilation MRA}
\label{sec:DilationNoiseSims}

We first illustrate the unbiasing procedure of Propositions \ref{prop:RandomDilations_PS} and \ref{prop:RandomDilations} for the high frequency signal $f(x)=e^{-5x^2}\cos(32x)$. \edit{Figure \ref{fig:HigherOrderUnbiasing} shows the power spectrum estimator $\widetilde{Pf}$ and the wavelet invariant estimator $\widetilde{P_Sf}$ for $k=0,2,4$ for both small and large dilations, where $\widetilde{P_Sf}$ denotes the combined wavelet invariant unbiasing plus optimization procedure (see Section \ref{sec:optimization}).} Higher order unbiasing is beneficial for both methods for small dilations, but fails for the power spectrum for large dilations. Both methods will of course fail for $\eta$ large enough, but for high frequency signals the power spectrum fails much sooner. 

\begin{figure}
	\centering
	\begin{subfigure}[b]{0.24\textwidth}
		\centering
		\includegraphics[width=\textwidth]{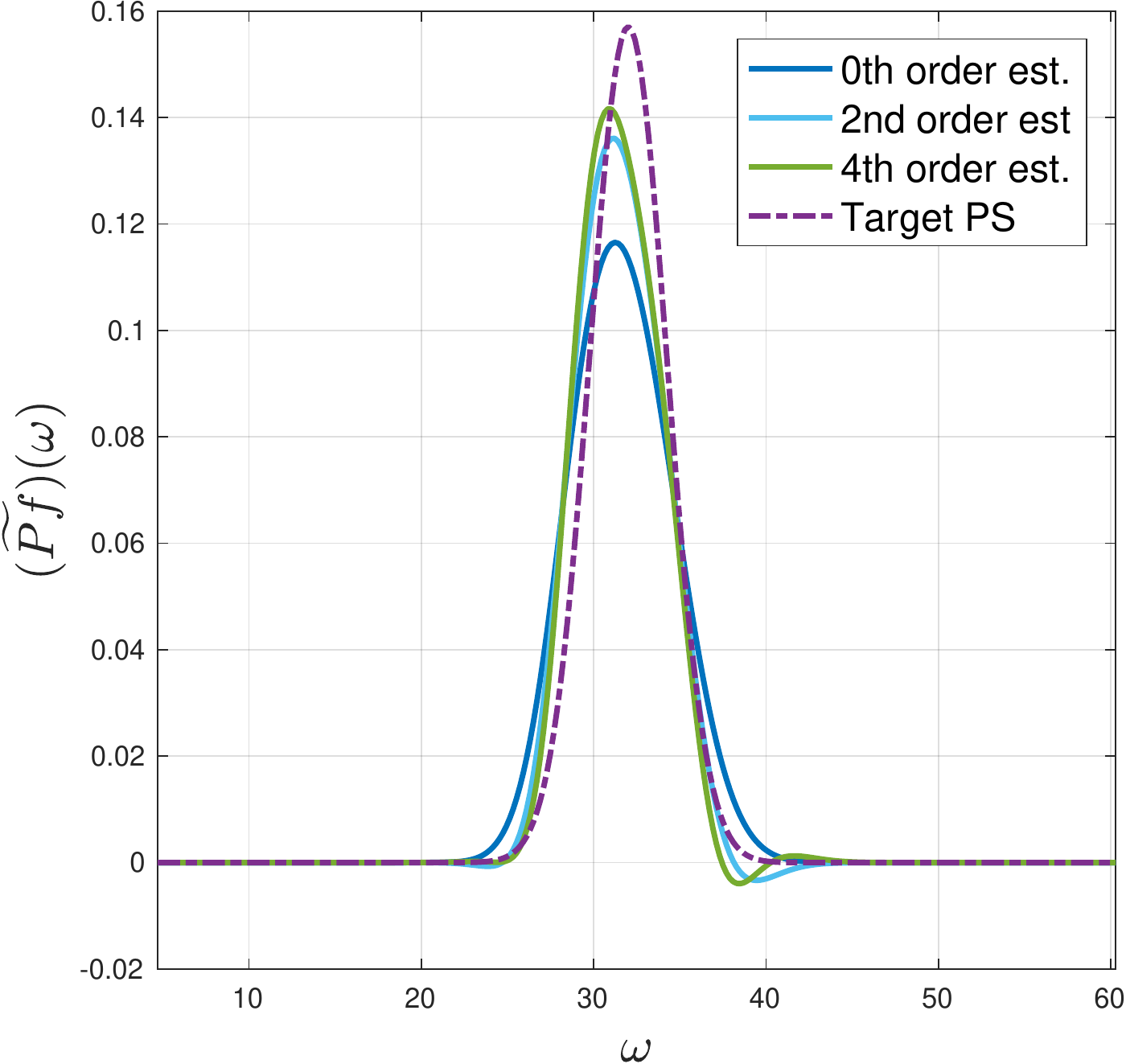}
		\caption{$f_3(x), \eta = 0.06$}
		\label{fig:PSUnbiasing_SmallDilations}
	\end{subfigure}
	\hfill
	\begin{subfigure}[b]{0.24\textwidth}
		\centering
		\includegraphics[width=\textwidth]{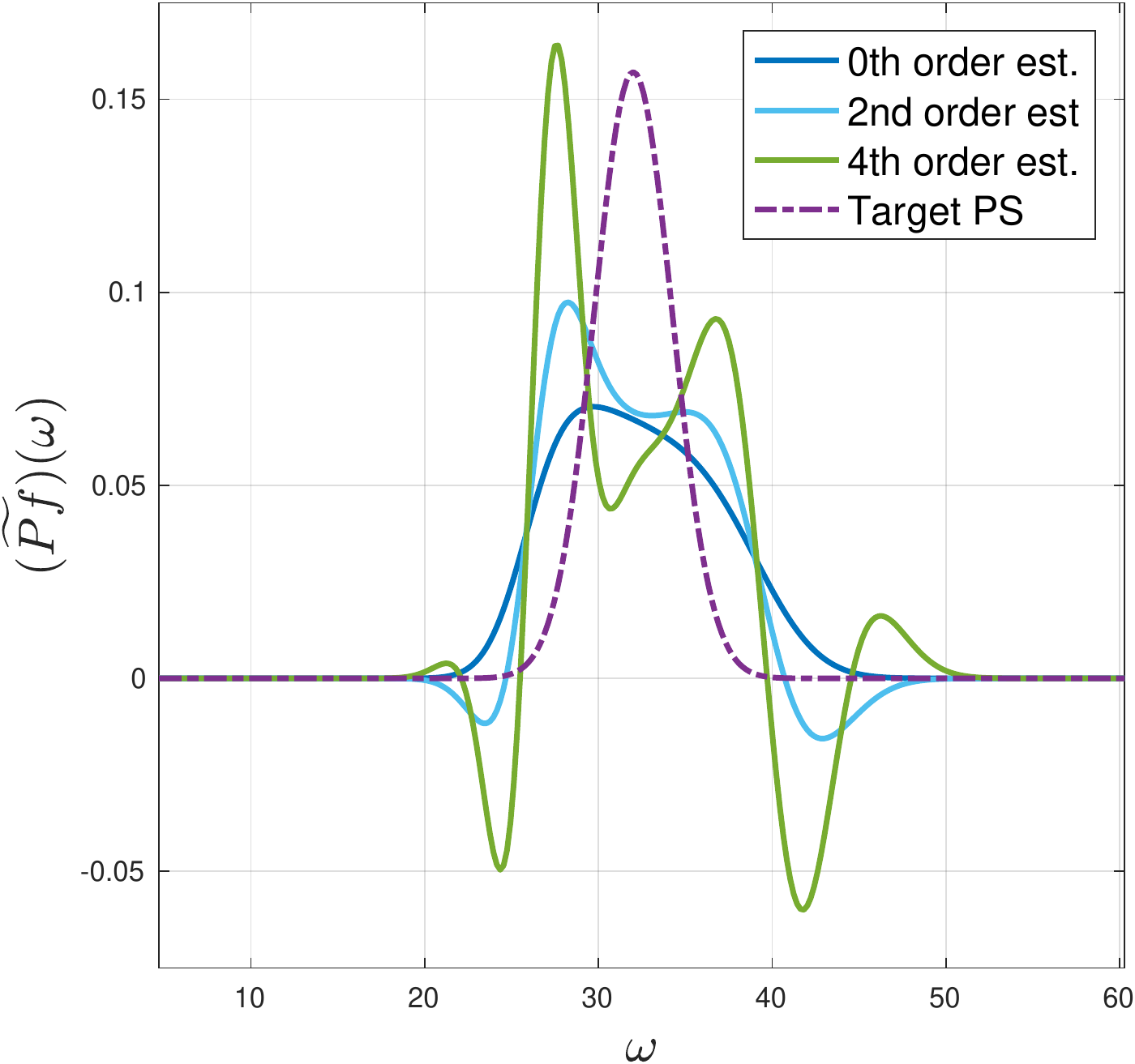}
		\caption{$f_3(x), \eta = 0.12$}
		\label{fig:PSUnbiasing_LargeDilations}
	\end{subfigure}
	\hfill
	\begin{subfigure}[b]{0.24\textwidth}
		\centering
		\includegraphics[width=\textwidth]{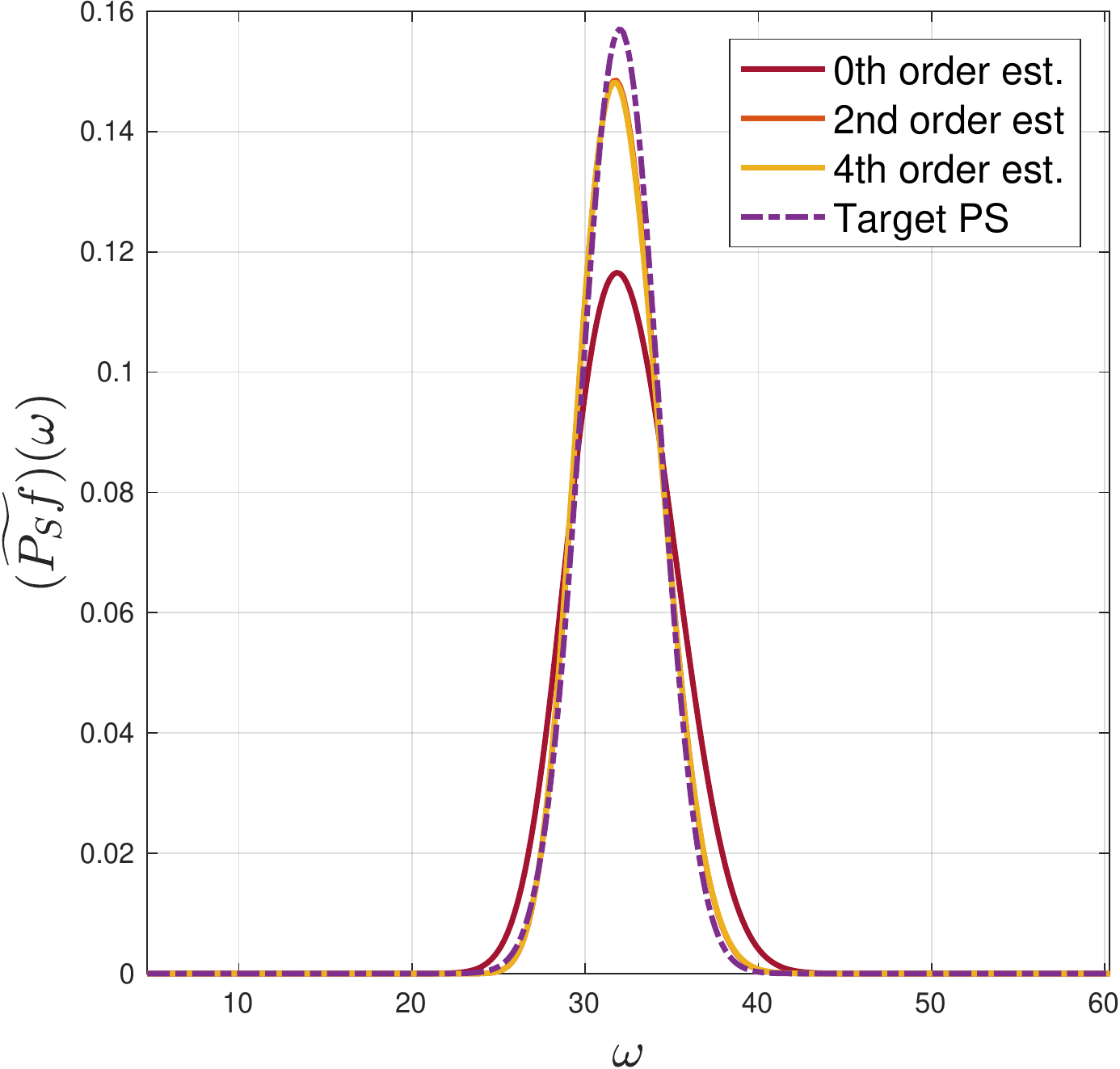}
		\caption{$f_3(x), \eta = 0.06$}
		\label{fig:WSCUnbiasing_SmallDilations}
	\end{subfigure}
	\hfill
	\begin{subfigure}[b]{0.24\textwidth}
		\centering
		\includegraphics[width=\textwidth]{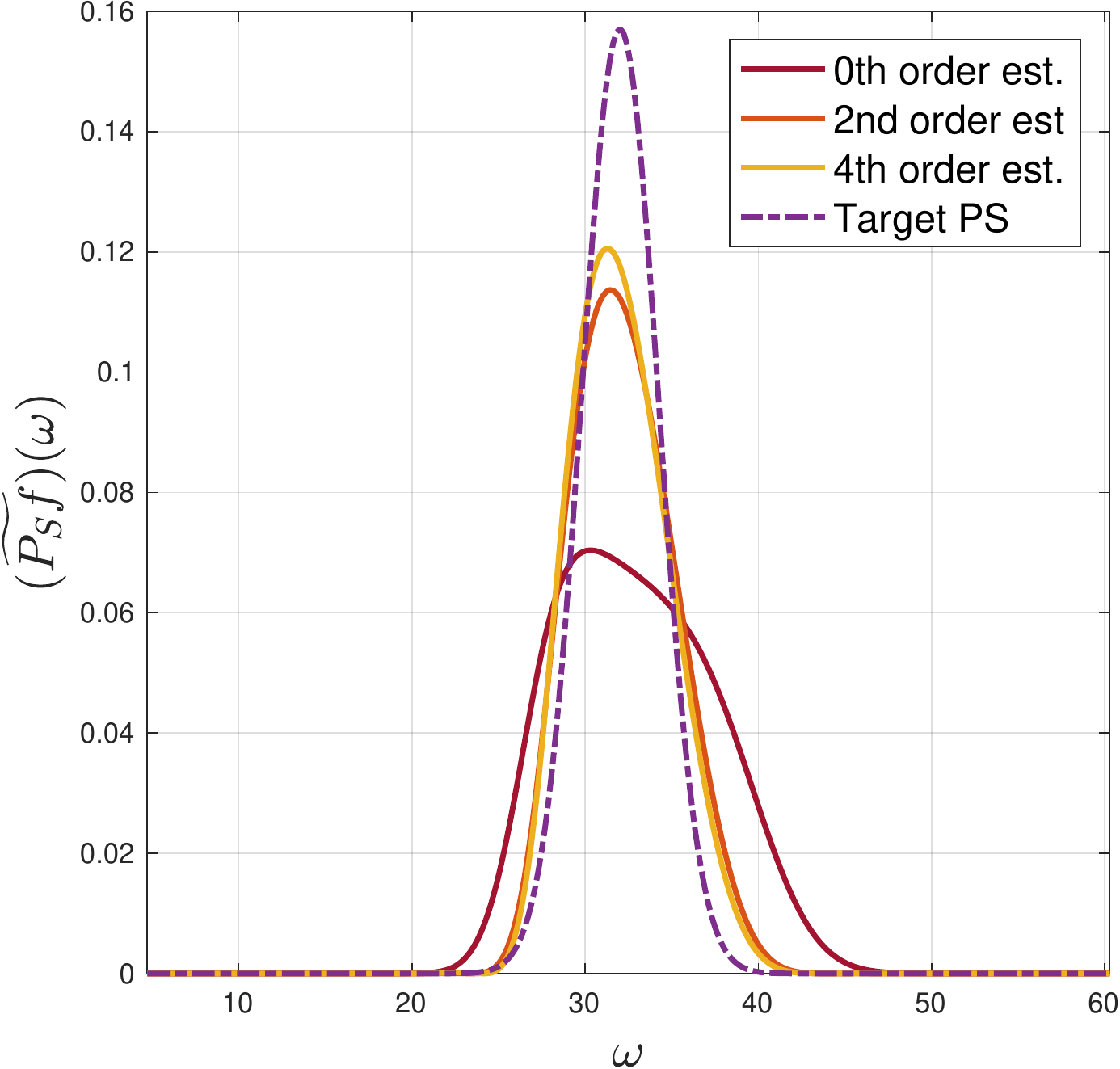}
		\caption{$f_3(x), \eta = 0.12$}
		\label{fig:WSCUnbiasing_LargeDilations}
	\end{subfigure}
	\caption{Order $k=0,2,4$ power spectrum estimators \edit{$\widetilde{Pf}$} (first two figures) and wavelet invariant estimators \edit{$\widetilde{P_Sf}$} (last two figures) for the signal $f_3(x) = e^{-5x^2} \cos (32 x)$. Figures \ref{fig:PSUnbiasing_SmallDilations} and \ref{fig:WSCUnbiasing_SmallDilations} show small dilations and Figures \ref{fig:PSUnbiasing_LargeDilations} and \ref{fig:WSCUnbiasing_LargeDilations} show large dilations.}
	\label{fig:HigherOrderUnbiasing}
\end{figure}

Next we compare \edit{$\|Pf-\widetilde{Pf}\|_2$ and $\|Pf-\widetilde{P_Sf}\|_2$}, the $\Lb^2$ error of estimating the power spectrum of the target signal via the power spectrum estimators of Proposition \ref{prop:RandomDilations_PS} and via the wavelet invariant estimators of Proposition \ref{prop:RandomDilations}, followed by a convex optimzation procedure. We consider order $k=0,2,4$ estimators for both the power spectrum and wavelet invariants on the following Gabor atoms of increasing frequency:
\begin{align*}
f_1(x) &= e^{-5x^2}\cos(8x) \\
f_2(x) &= e^{-5x^2}\cos(16x) \\
f_3(x) &= e^{-5x^2}\cos(32x).
\end{align*}
These functions satisfy $f = \text{Real}(h)$ where $(Ph)(\omega) = (\pi/5)e^{-(\omega-\xi)^2/10}$ for $\xi=8, 16, 32$, and thus exhibit the behavior described in Remark \ref{rmk:PSbad}. 

Simulation results are shown in Figure \ref{fig:DilationNoiseEmpirical}; the horizontal axis shows $\log_2(M)$ while the vertical axis shows $\log_2(\text{Error})$. For each value of $M$, the error was calculated for 10 independent simulations and then averaged. The unbiasing procedure of Propositions \ref{prop:RandomDilations_PS} and \ref{prop:RandomDilations} requires knowledge of the moments of the dilation distribution, but in practice these are unknown. Thus the first two even moments of the dilation distribution $(\eta^2, C_4\eta^4)$ were estimated empirically with the fourth order estimators described in Section \ref{sec:EmpMomentEst} (see Definition \ref{def:emp_dil_mom_est_dilMRA}). For the low frequency signal, the $4^{\text{th}}$ order power spectrum estimator was best for both small and large dilations, and is preferable due to the lower computational cost (see Remark \ref{rmk:ComputationalCost}). For the high frequency signal,  the $4^{\text{th}}$ order wavelet invariant estimator was best for large dilations and WSC $k=2$ and $k=4$ were best and equivalent for small dilations. For the medium frequency signal, the higher order power spectrum estimators were best for small dilations while the higher order wavelet invariant estimators were best for large dilations. Thus the simulation results confirm that the wavelet invariants will have an advantage over Fourier invariants when the signals are either high frequency or corrupted by large dilations. We remark that one obtains nearly identical error plots with oracle knowledge of the dilation moments, indicating that the empirical moment estimation procedure is highly accurate in the absencse of additive noise, even for small $M$ values.

\begin{figure}
	\centering
	\begin{subfigure}[b]{0.32\textwidth}
		\centering
		\includegraphics[width=.85\textwidth]{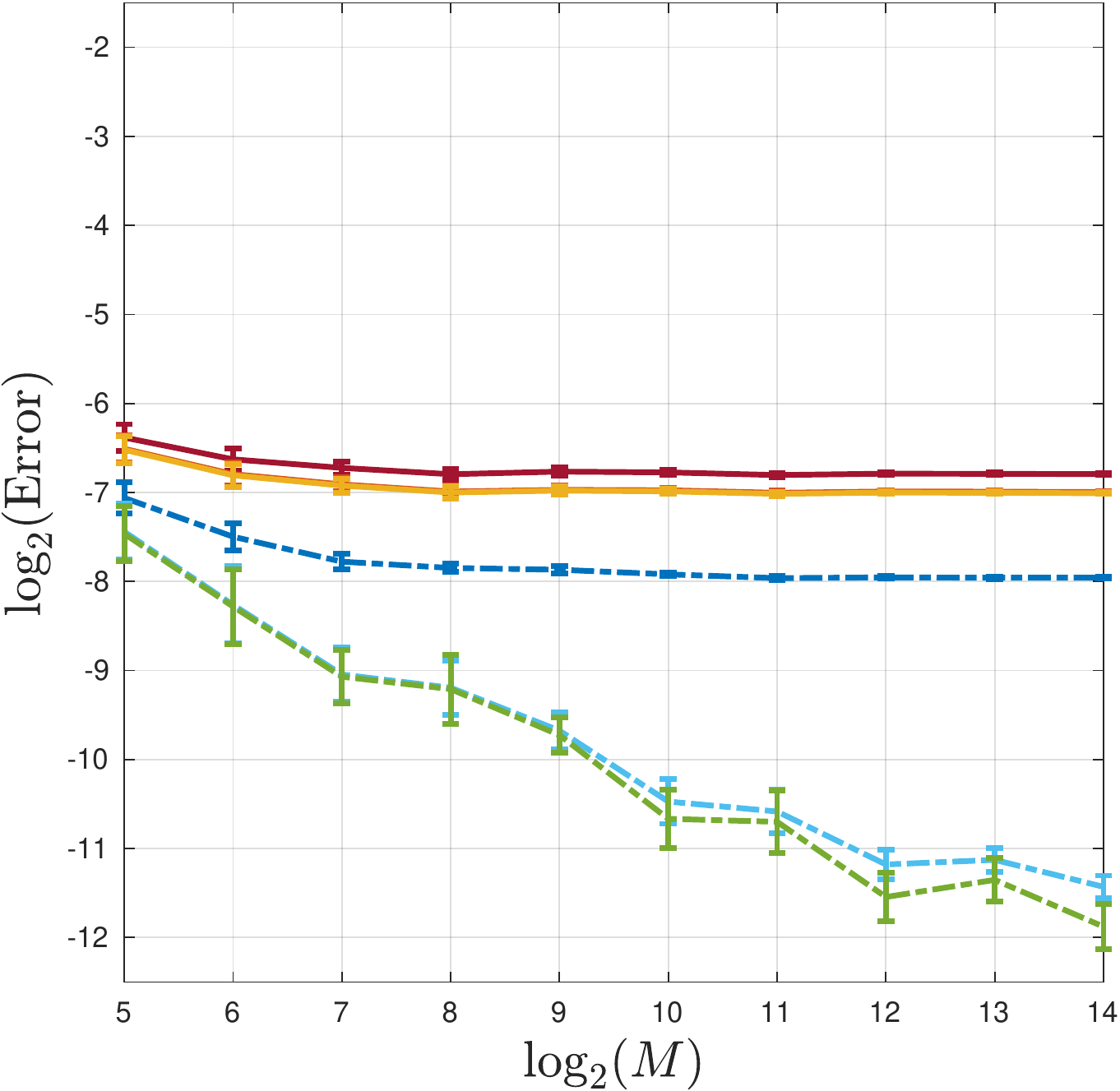}
		\caption{$f_1(x), \eta = 0.06$}
		\vspace*{.1cm}
		\label{fig:sim_E_1_low}
	\end{subfigure}
	\hfill
	\begin{subfigure}[b]{0.32\textwidth}
		\centering
		\includegraphics[width=.85\textwidth]{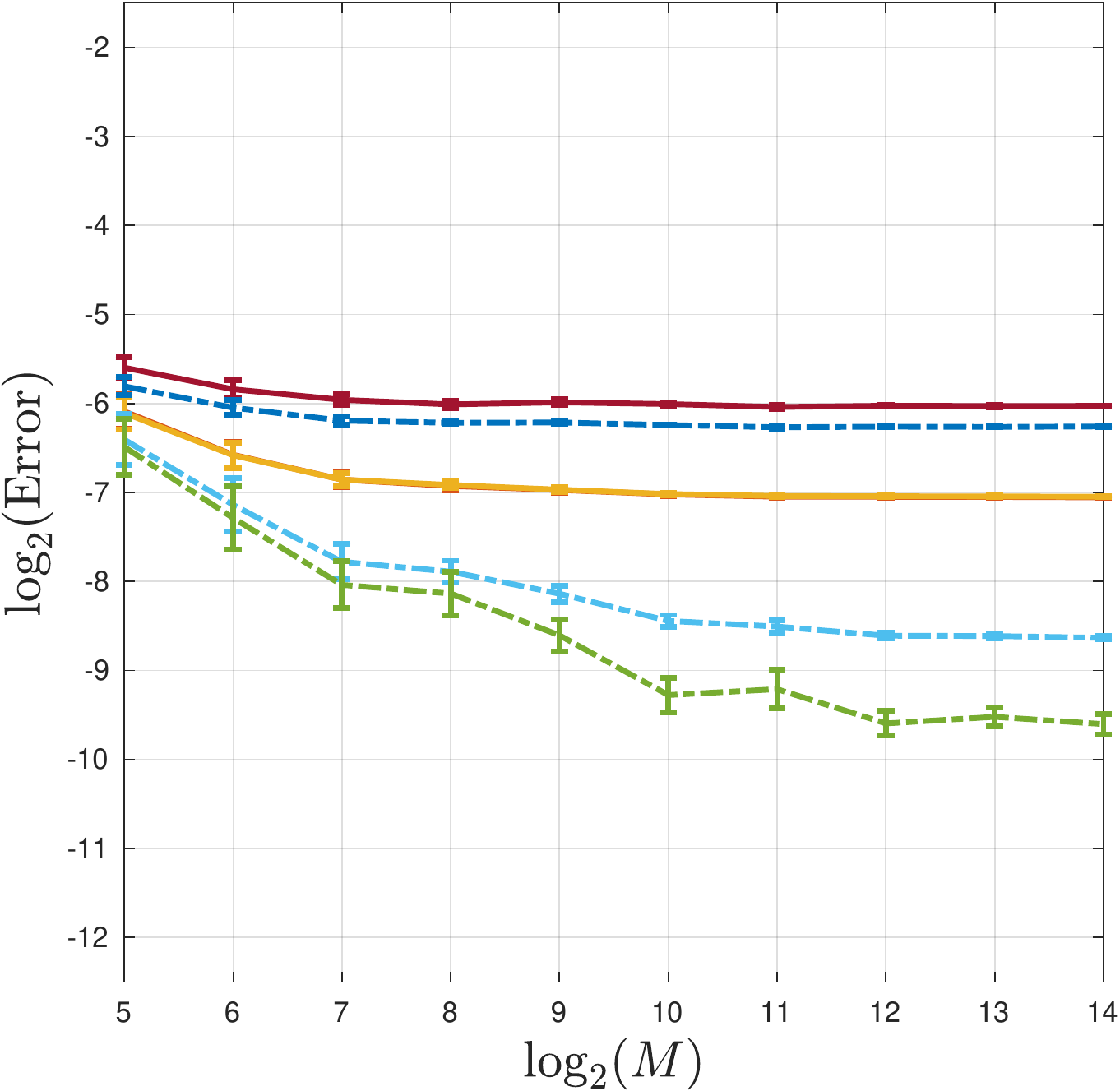}
		\caption{$f_2(x), \eta = 0.06$}
		\vspace*{.1cm}
		\label{fig:sim_E_1_med}
	\end{subfigure}
	\hfill
	\begin{subfigure}[b]{0.32\textwidth}
		\centering
		\includegraphics[width=.85\textwidth]{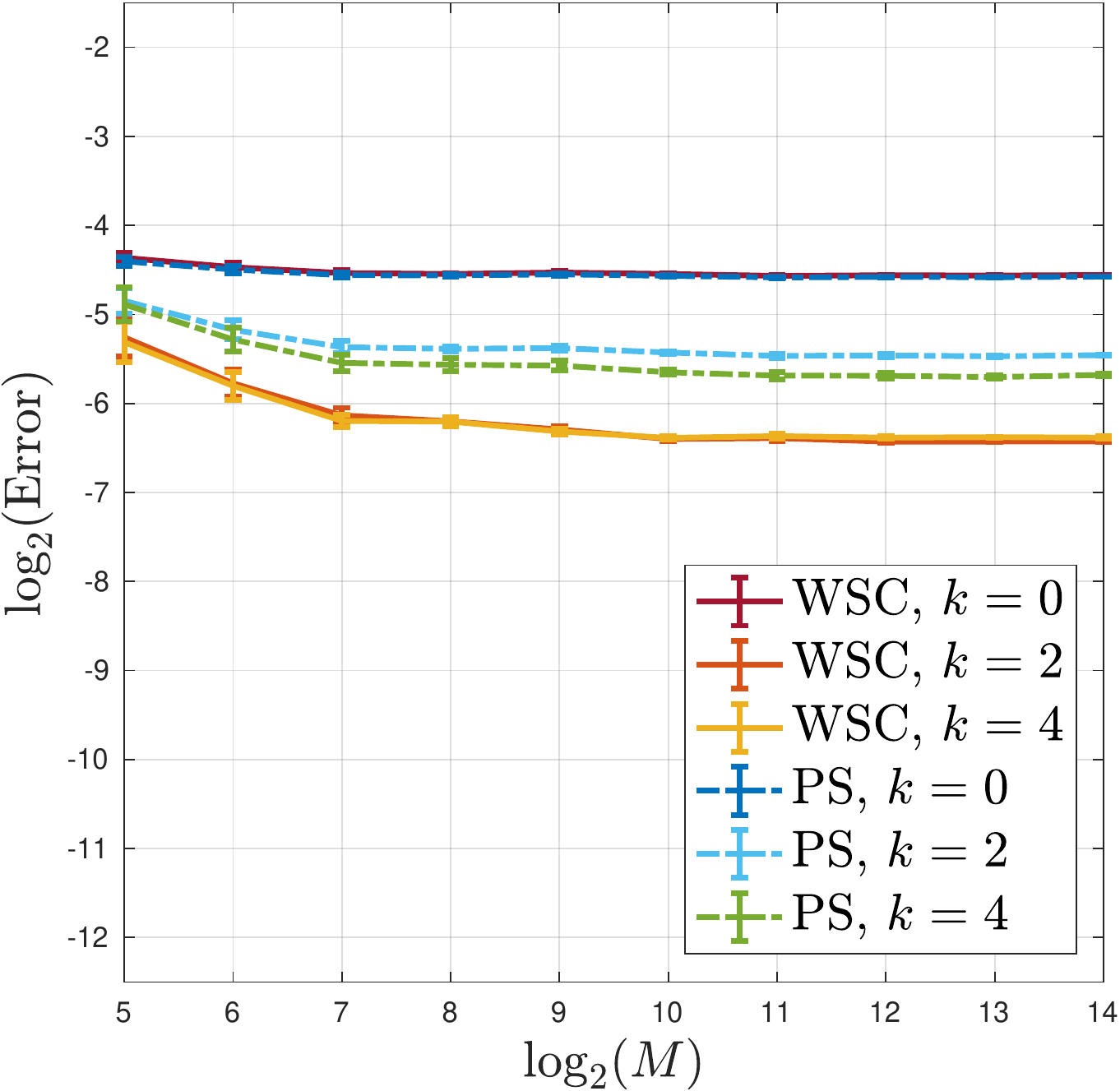}
		\caption{$f_3(x), \eta = 0.06$}
		\vspace*{.1cm}
		\label{fig:sim_E_1_high}
	\end{subfigure}
	\begin{subfigure}[b]{0.32\textwidth}
		\centering
		\includegraphics[width=.85\textwidth]{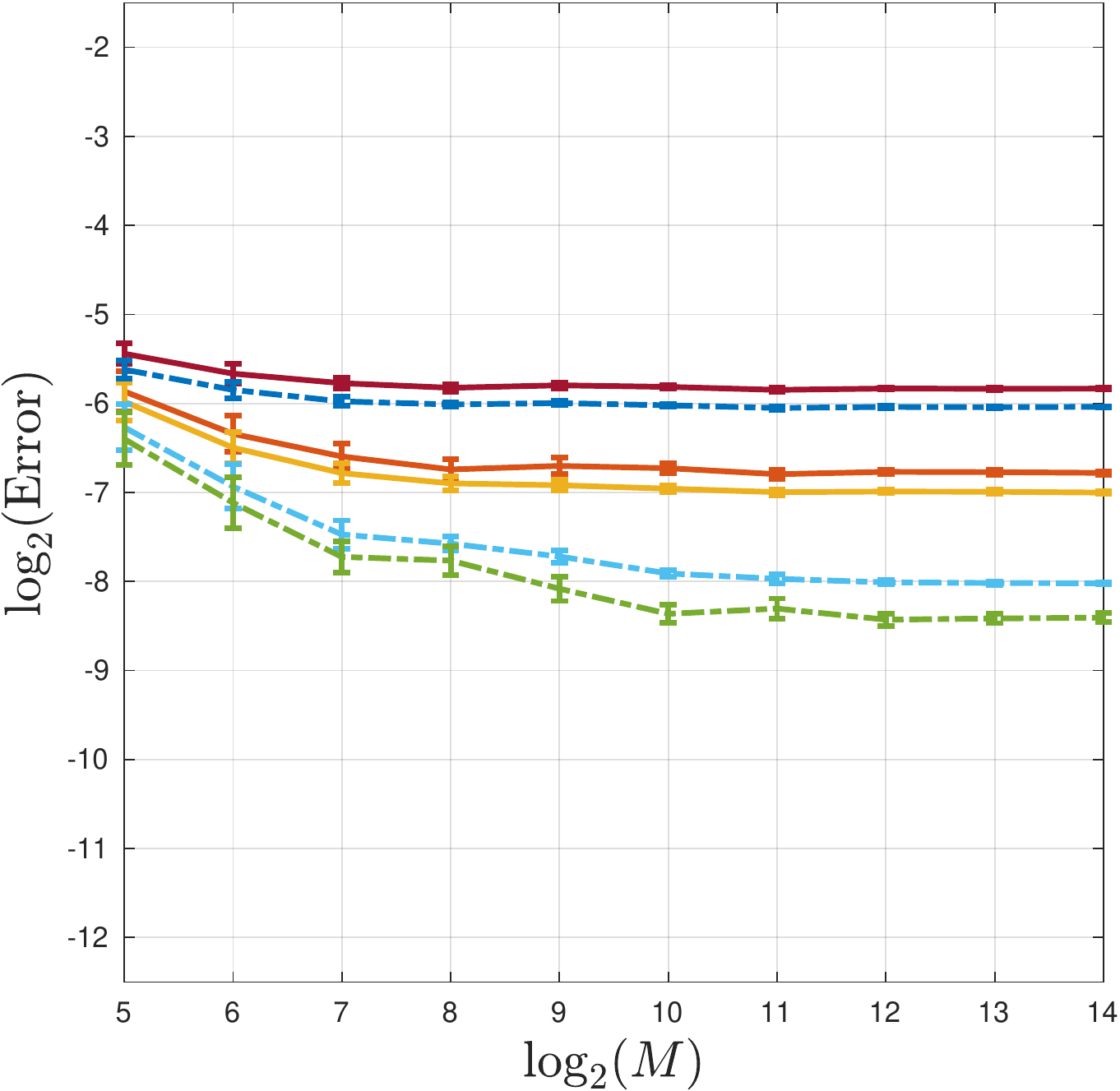}
		\caption{$f_1(x), \eta = 0.12$}
		\label{fig:sim_E_2_low}
	\end{subfigure}
	\hfill
	\begin{subfigure}[b]{0.32\textwidth}
		\centering
		\includegraphics[width=.85\textwidth]{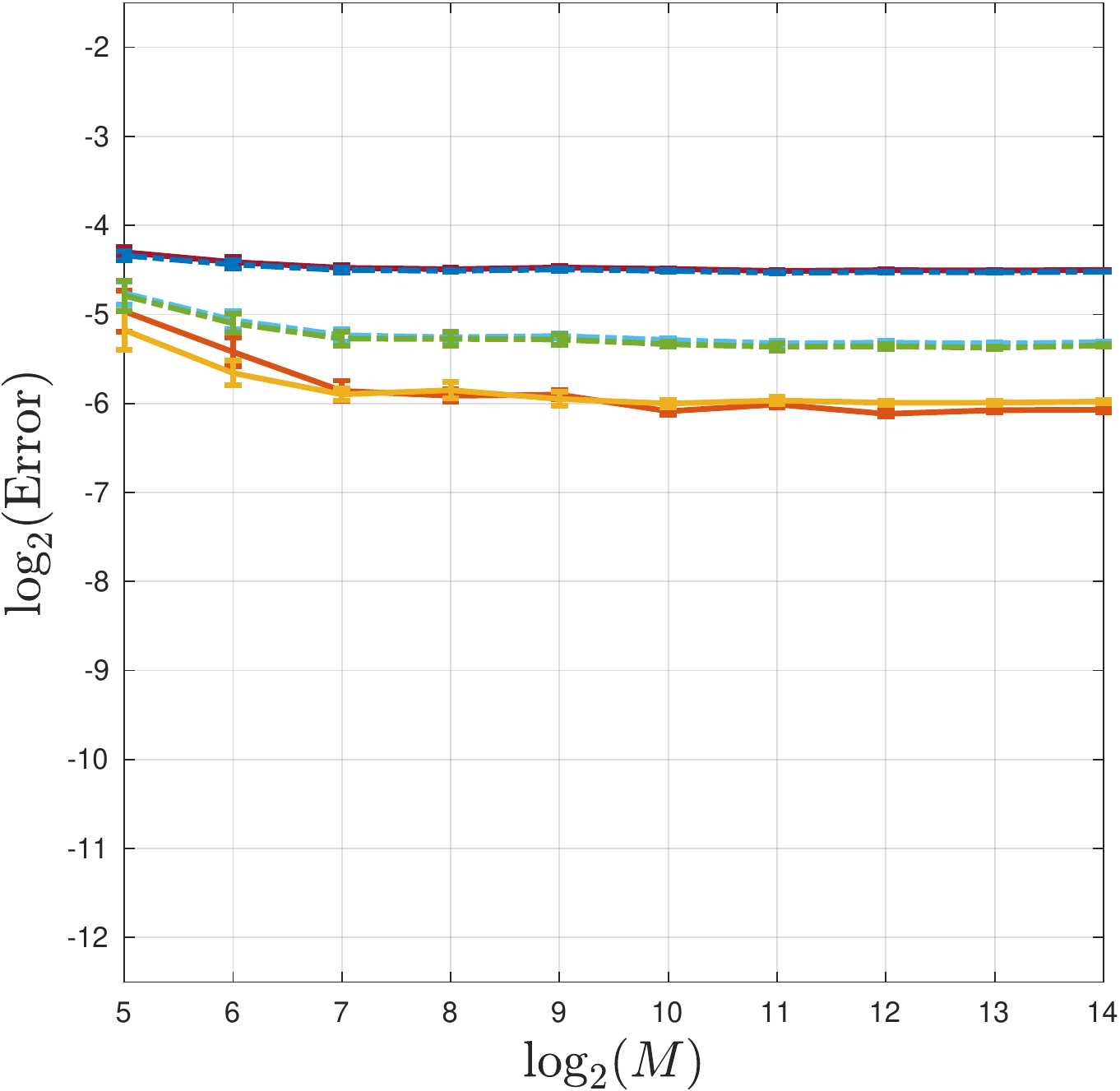}
		\caption{$f_2(x), \eta = 0.12$}
		\label{fig:sim_E_2_med}
	\end{subfigure}
	\hfill
	\begin{subfigure}[b]{0.32\textwidth}
		\centering
		\includegraphics[width=.85\textwidth]{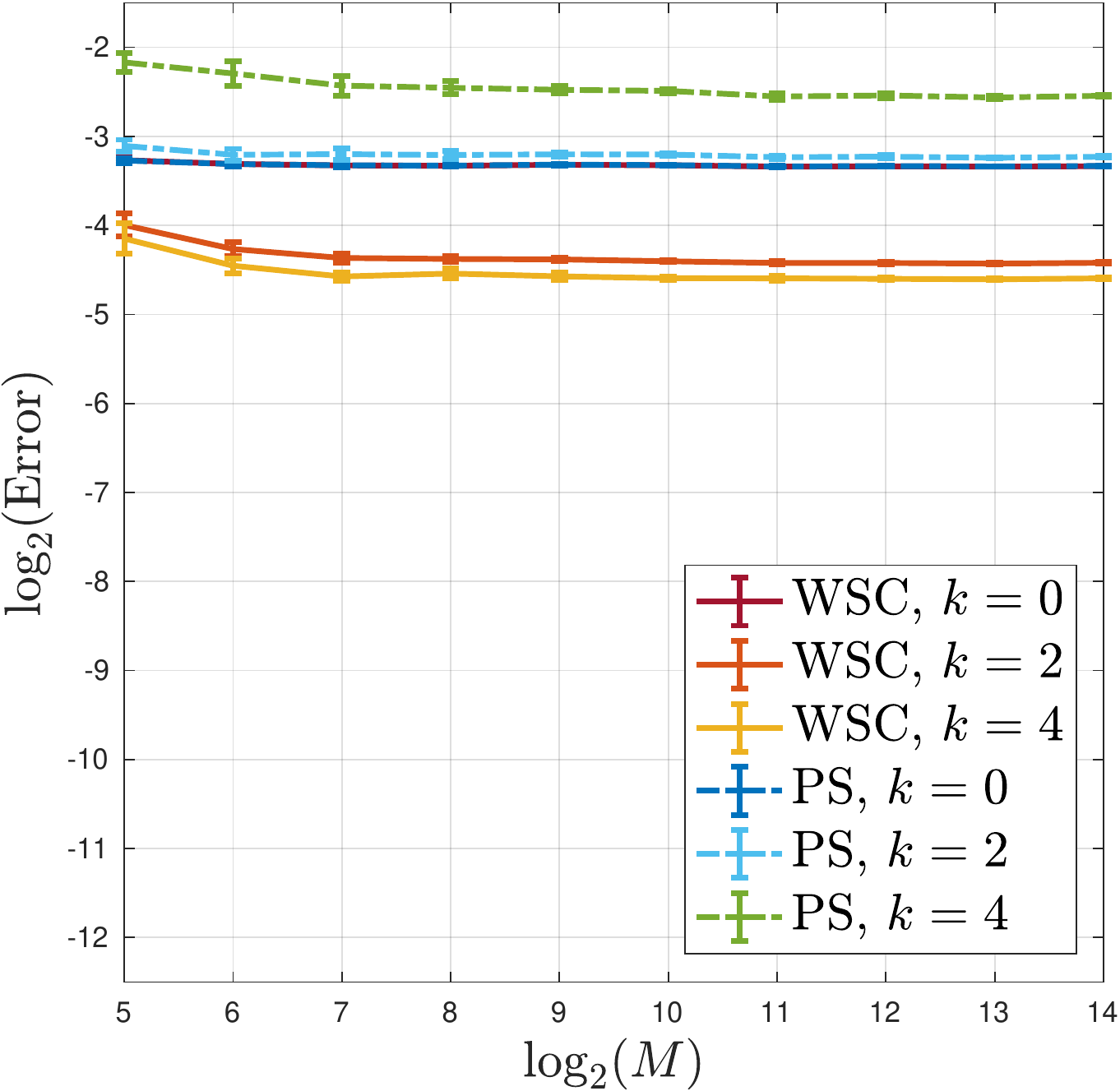}
		\caption{$f_3(x), \eta = 0.12$}
		\label{fig:sim_E_2_high}
	\end{subfigure}
	\caption{$\Lb^2$ error with standard error bars for dilation model (empirical moment estimation). Top row shows results for small dilations ($\eta=0.06$) and bottom row shows results for large dilations ($\eta=0.12$). First, second, third column shows results for low, medium, high frequency Gabor signals.  All plots have the same axis limits.
	}
	\label{fig:DilationNoiseEmpirical}
\end{figure}

\section{\edit{Noisy dilation MRA} model}
\label{sec: noisy dilation MRA model}

Finally, we consider the \edit{noisy dilation MRA} model (Model \ref{model:genMRA}) where signals are randomly translated and dilated and corrupted by additive noise. \edit{Section  \ref{sec:WSCnoisydilationMRA} gives unbiasing results for wavelet invariants and Section \ref{sec:GenMRANoiseSim} reports relevant simulations.}

\edit{\subsection{Wavelet inariant results for noisy dilation MRA}\label{sec:WSCnoisydilationMRA}}

To state Proposition \ref{prop:DilationAndAdditiveNoise} as succinctly as possible, we also define the following quantity
\begin{align}
\label{equ:PsiConstant}
\PsiConstant &:= \sum_{m=0,2,\ldots,k} \Psi_m(E\eta)^m\, ,
\end{align}
where $E$ is defined in (\ref{equ:E}) and $\Psi_m$ is defined in \eqref{equ:Psik}.

\begin{proposition}
	\label{prop:DilationAndAdditiveNoise}
	Assume Model \ref{model:genMRA} and that $\psi$ is $(k+2)$-admissable.  Define the following estimator of $(\Sc f)(\lambda)$:
	\begin{align*}
	(\widetilde{\Sc f})(\lambda) &:= \frac{1}{M}\sum_{j=1}^M \left[(S y_j)(\lambda) -B_2\eta^2\lambda^2(S y_j)''(\lambda)- \ldots - B_k\eta^k\lambda^k(S y_j)^{(k)}(\lambda)\right]- \sigma^2
	\end{align*}
	where the constants $B_i$ satisfy (\ref{equ:MomentBasedConstants}). 
	Then with probability at least $1-1/t^2$
	\begin{align}
	&\left|(\widetilde{\Sc f})(\lambda) - (\Sc f)(\lambda)\right|\lesssim k\Lambda_{k+2}(\lambda)(2E\eta)^{k+2} +\frac{t}{\sqrt{M}}\left[k\LambdaConstant(\lambda)+ \PsiConstant\sigma^2+\sqrt{\PsiConstant(\Lambda_0(\lambda)+\LambdaConstant(\lambda))}\sigma\right]\, , 	\label{equ:GenMRA_error}
	\end{align}
	where $E, \LambdaConstant(\lambda), \PsiConstant$ are as defined in (\ref{equ:E}), (\ref{equ:LambdaConstant}), (\ref{equ:PsiConstant}).
\end{proposition}
The following corollary is an immediate consequence of Proposition \ref{prop:DilationAndAdditiveNoise}.
\begin{corollary}
	\label{cor:GenMRAModel_bestk}
	Let the assumptions of Proposition \ref{prop:DilationAndAdditiveNoise} hold, and in addition assume $\Psi_i(2E\eta)^i$ is decreasing for $i\leq k+2$. Then with probability at least $1-1/t^2$
	\begin{align}
	\label{equ:gen_mra_error_bound}
	\left|(\widetilde{\Sc f})(\lambda) - (\Sc f)(\lambda)\right| &\lesssim  k\Psi_{k+2}(2E\eta)^{k+2}\norm{f}_1^2 + \frac{tk}{\sqrt{M}} \left[ k\eta \norm{f}_1^2 +\sigma\norm{f}_1+\sigma^2 \right]\, .
	\end{align} 
\end{corollary}
We remark that there are two components to the estimation error bounded by the right-hand side of (\ref{equ:gen_mra_error_bound}): the first two terms are the error due to dilation, as in Corollary \ref{cor:RandomDilationWithBestk} of Proposition \ref{prop:RandomDilations}, and the last two terms are the error due to additive noise, as given in Proposition \ref{prop:AddNoiseWSC}. Thus the wavelet invariant representation allows for a decomposition of the error of the \edit{noisy dilation MRA} model into the sum of the errors of the random dilation model and the additive noise model. This is possible because the representation inherits the differentiability of the wavelet, and is not possible when $P\psi \notin \Cb^k (\R)$, in which case the dilation unbiasing procedure has a more complicated effect on the additive noise. A result equivalent to Proposition \ref{prop:DilationAndAdditiveNoise} cannot be made for the power spectrum, because the nonlinear unbiasing procedure of Proposition \ref{prop:RandomDilations_PS} cannot be applied to the power spectra of signals from the \edit{noisy dilation MRA} corruption model, since they are not differentiable in the presence of additive noise.

\begin{proof}[Proof of Proposition \ref{prop:DilationAndAdditiveNoise}]
Since $\Sc f$ is a translation invariant representation, we can ignore the translation factors $\{t_j\}_{j=1}^M$ and consider the model $y_j = f_{\tau_j} + \varepsilon_j$.	For notational convenience, we define the following order $k$ derivative ``unbiasing" operator:
\begin{align}
\label{equ:unbiasing_operator}
A_\lambda g(\lambda):= g(\lambda) -B_2\eta^2\lambda^2\frac{d}{d\lambda^2}g(\lambda) - \ldots - B_k\eta^k\lambda^k\frac{d}{d\lambda^k}g(\lambda)
\end{align} 
which is defined on any function of $\lambda$, so that we can express our estimator by
\begin{align*}
(\widetilde{\Sc f})(\lambda) &= 
\frac{1}{M}\sum_{j=1}^M \left[\frac{1}{2\pi} \int |\widehat{y}_j(\omega)|^2 A_\lambda |\widehat{\psi}_\lambda(\omega)|^2\ d\omega\right] - \sigma^2 \\
&=\frac{1}{M}\sum_{j=1}^M \left[\frac{1}{2\pi} \int \left(|\widehat{f}_{\tau_j}(\omega)|^2+\widehat{f}_{\tau_j}(\omega)\overline{\widehat{\varepsilon}_j}(\omega)+\overline{\widehat{f}_{\tau_j}}(\omega)\widehat{\varepsilon}_j(\omega)+|\widehat{\varepsilon}_j(\omega)|^2\right) A_\lambda |\widehat{\psi}_\lambda(\omega)|^2\ d\omega \right] - \sigma^2 \, .
\end{align*}

We can thus decompose the error as follows:
\begin{align*}
&|(\widetilde{\Sc f})(\lambda) - (\Sc f)(\lambda)| \leq \underbrace{\left| \frac{1}{M}\sum_{j=1}^M \frac{1}{2\pi} \int \left(\widehat{f}_{\tau_j}(\omega)\overline{\widehat{\varepsilon}_j}(\omega)+\overline{\widehat{f}_{\tau_j}}(\omega)\widehat{\varepsilon}_j(\omega)\right) A_\lambda |\widehat{\psi}_\lambda(\omega)|^2\ d\omega\right|}_{\text{Cross Term Error}} \\
&\quad+\underbrace{\left| \frac{1}{M}\sum_{j=1}^M \frac{1}{2\pi} \int |\widehat{f}_{\tau_j}(\omega)|^2 A_\lambda |\widehat{\psi}_\lambda(\omega)|^2\ d\omega - (\Sc f)(\lambda)\right|}_{\text{Dilation Error}} + \underbrace{\left|  \frac{1}{M}\sum_{j=1}^M \frac{1}{2\pi} \int |\widehat{\varepsilon}_j(\omega)|^2 A_\lambda |\widehat{\psi}_\lambda(\omega)|^2\ d\omega- \sigma^2\right|}_{\text{Additive Noise Error}} \, .
\end{align*}

To bound the above terms we utilize the following two Lemmas, which are proved in Appendix \ref{app:gen_mra}. 

\begin{restatable}{lemma}{lemGenMRAAddNoise}
	\label{lem:GenMRA_AddNoise}
	Let the notation and assumptions of Proposition \ref{prop:DilationAndAdditiveNoise} hold, and let $A_\lambda$ be the operator defined in (\ref{equ:unbiasing_operator}).
	Then with probability at least $1-1/t^2$
	\begin{align*}
	\left|  \frac{1}{M}\sum_{j=1}^M \frac{1}{2\pi} \int |\widehat{\varepsilon}_j(\omega)|^2 A_\lambda |\widehat{\psi}_\lambda(\omega)|^2\ d\omega- \sigma^2\right| \leq \frac{2t\sqrt{k}\PsiConstant\sigma^2}{\sqrt{M}}\, .
	\end{align*}
\end{restatable}

\begin{restatable}{lemma}{lemGenMRACrossTerm}
	\label{lem:GenMRA_CrossTerm}
	Let the notation and assumptions of Proposition \ref{prop:DilationAndAdditiveNoise} hold, and let $A_\lambda$ be the operator defined in (\ref{equ:unbiasing_operator}).
	Then with probability at least $1-1/t^2$
	\begin{align*}
	\left| \frac{1}{M}\sum_{j=1}^M \frac{1}{2\pi} \int \left(\widehat{f}_{\tau_j}(\omega)\overline{\widehat{\varepsilon}_j}(\omega)+\overline{\widehat{f}_{\tau_j}}(\omega)\widehat{\varepsilon}_j(\omega)\right) A_\lambda |\widehat{\psi}_\lambda(\omega)|^2\ d\omega\right| &\lesssim \frac{t}{\sqrt{M}}\sqrt{\PsiConstant(\Lambda_0(\lambda)+\LambdaConstant(\lambda))}\sigma \,.
	\end{align*}
\end{restatable}

Applying Proposition \ref{prop:RandomDilations} to bound the dilation error, Lemma \ref{lem:GenMRA_AddNoise} to bound the additive noise error, and Lemma \ref{lem:GenMRA_CrossTerm} to bound the cross term error gives (\ref{equ:GenMRA_error}). 

\end{proof}

\subsection{Simulation results for \edit{noisy dilation MRA}}
\label{sec:GenMRANoiseSim}

We once again consider the Gabor atoms of varying frequency introduced in Section \ref{sec:DilationNoiseSims}, and compare the $\Lb^2$ error of estimating the power spectrum by (1) averaging the power spectra of the noisy signals, and applying additive noise unbiasing; this is the zero order power spectrum method (PS $k=0$), defined in Proposition \ref{prop:AddNoisePS}, and (2) by approximating the wavelet invariants by the estimators given in Proposition \ref{prop:DilationAndAdditiveNoise} for $k=0,2,4$, \edit{and then applying the optimization procedure described in Section \ref{sec:optimization}}; we refer to these methods as WSC $k=i$ for $i=0,2,4$. We emphazise that for the \edit{noisy dilation MRA} model, it is impossible to define higher order methods for the power spectrum.

We first consider the errors obtained given oracle knowledge of the noise moments, both additive and dilation. Results are shown in Figure \ref{fig:GenMRAModelOracle} for all parameter combinations resulting from $\sigma = 2^{-4}, 2^{-3}$ \edit{(giving $\snr=2.2, 0.56$)} and $\eta = 0.06, 0.12$. The horizontal axis shows $\log_2(M)$ and the vertical axis shows $\log_2(\text{Error})$; for each value of $M$, the error was calculated for 10 independent simulations and then averaged. For all simulations $\tau$ was given a uniform distribution, a challenging regime for dilations, and the sample size ranged over $16 \leq M \leq 131,072$. For the medium and high frequency signals, for large enough $M$, WSC $k=2$ and WSC $k=4$ have significantly smaller error than the order zero estimators, indicating that the nonlinear unbiasing procedure of Proposition \ref{prop:DilationAndAdditiveNoise} contributes a definitive advantage. For the high frequency signal and large $M$, the error using WSC $k=4$ is decreased by a factor of about 3 from the PS $k=0$ error. For small dilations ($\eta = 0.06$), there is not much of a difference in performance between WSC $k=2$ and WSC $k=4$, but the gap between these estimators widens for large dilations ($\eta=0.12$), as the fourth order correction becomes more important. For the low frequency signal under small dilations, PS $k=0$ achieves the smallest error for large $M$. However when $M$ is small or the dilations are large, the WSC estimators have the advantage for the low frequency signal as well, and WSC $k=4$ is once again the best estimator for large $M$.

\begin{figure}
	\centering
	\begin{subfigure}[b]{0.32\textwidth}
		\centering
		\includegraphics[width=.85\textwidth]{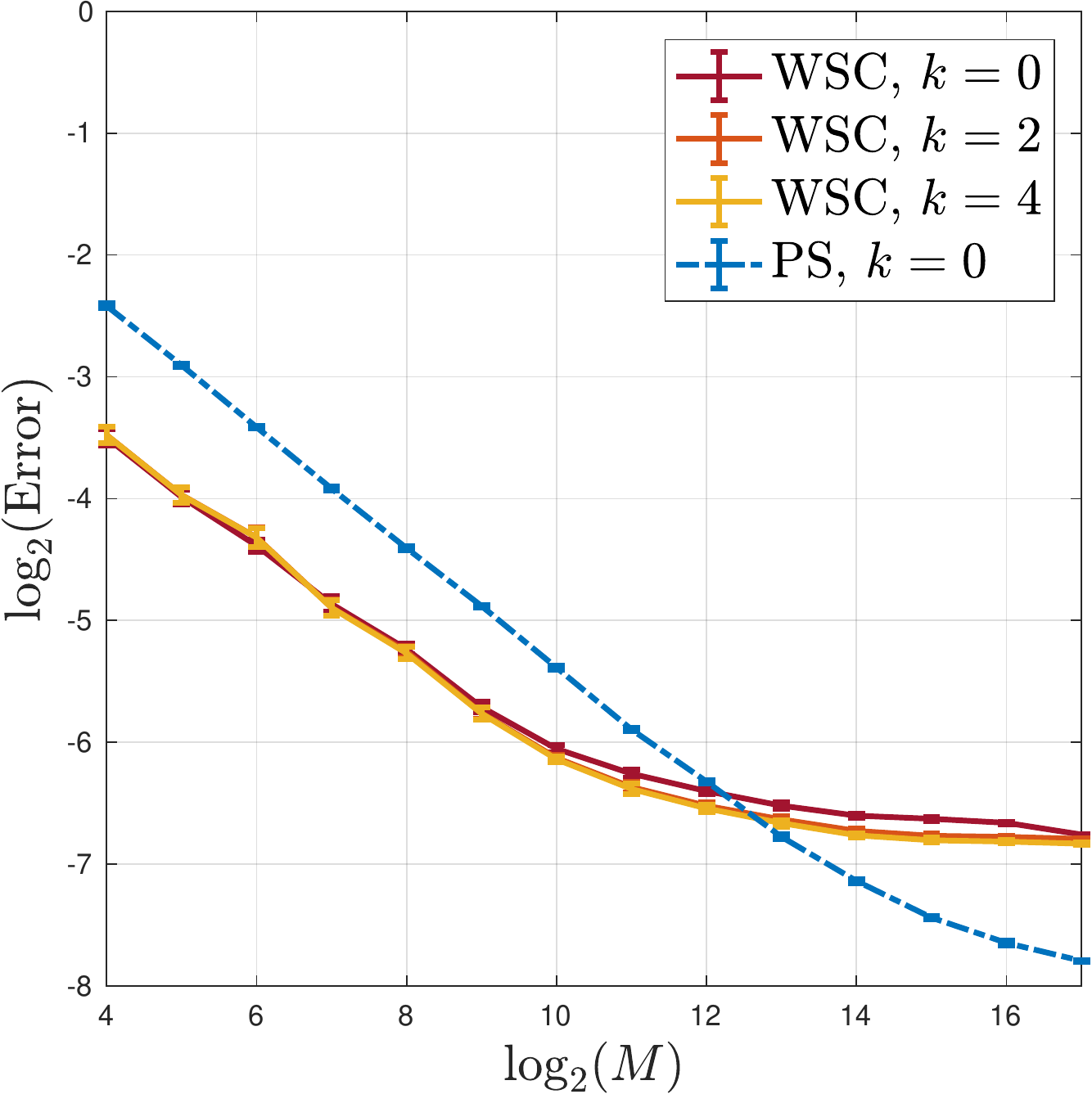}
		\caption{$f_1,\edit{\snr=2.2}, \eta=0.06$}
		\vspace*{.1cm}
		\label{fig:sim_O_3_low}
	\end{subfigure}
	\hfill
	\begin{subfigure}[b]{0.32\textwidth}
		\centering
		\includegraphics[width=.85\textwidth]{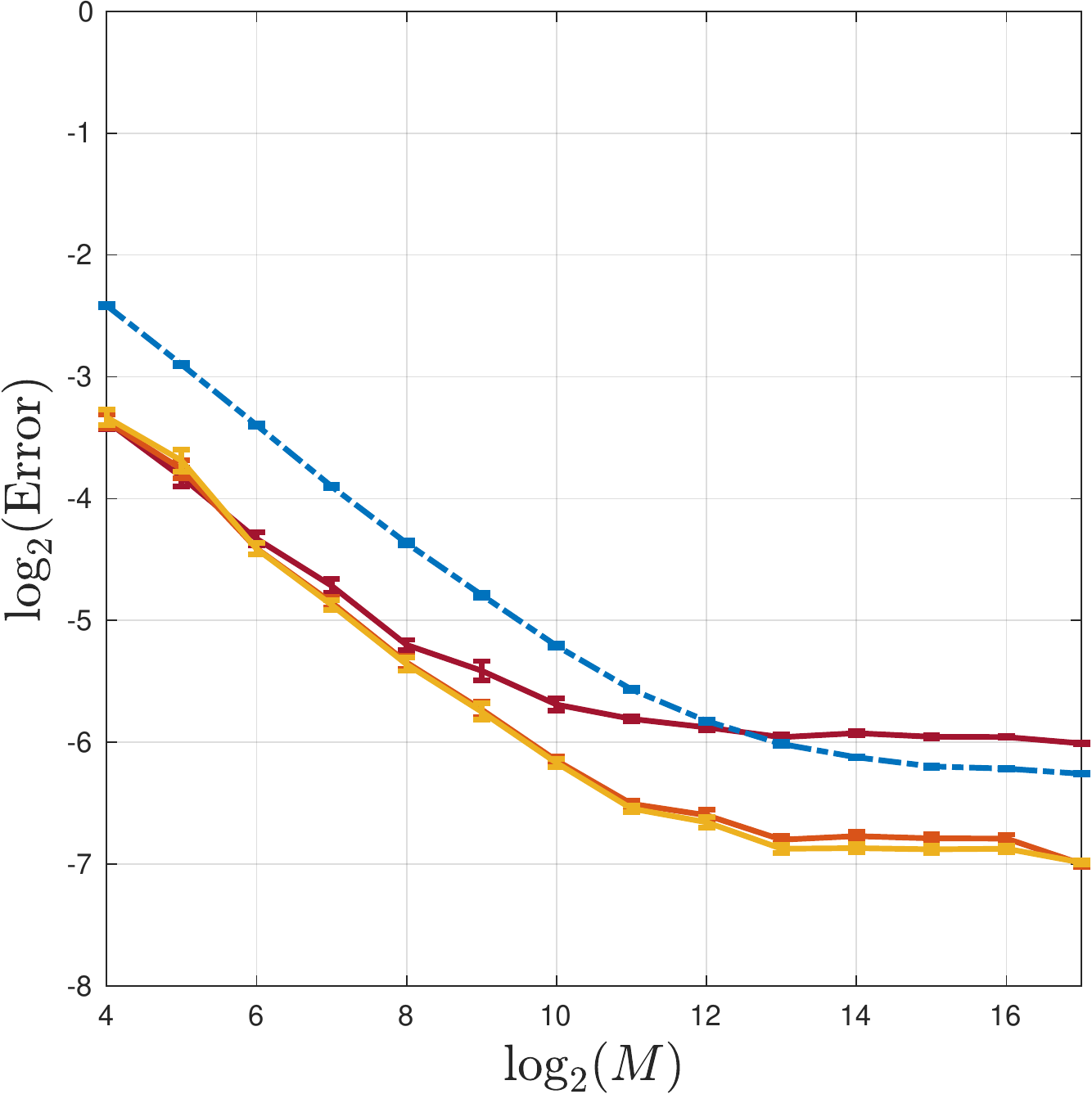}
		\caption{$f_2,\edit{\snr=2.2}, \eta=0.06$}
		\vspace*{.1cm}
		\label{fig:sim_O_3_med}
	\end{subfigure}
	\hfill
	\begin{subfigure}[b]{0.32\textwidth}
		\centering
		\includegraphics[width=.85\textwidth]{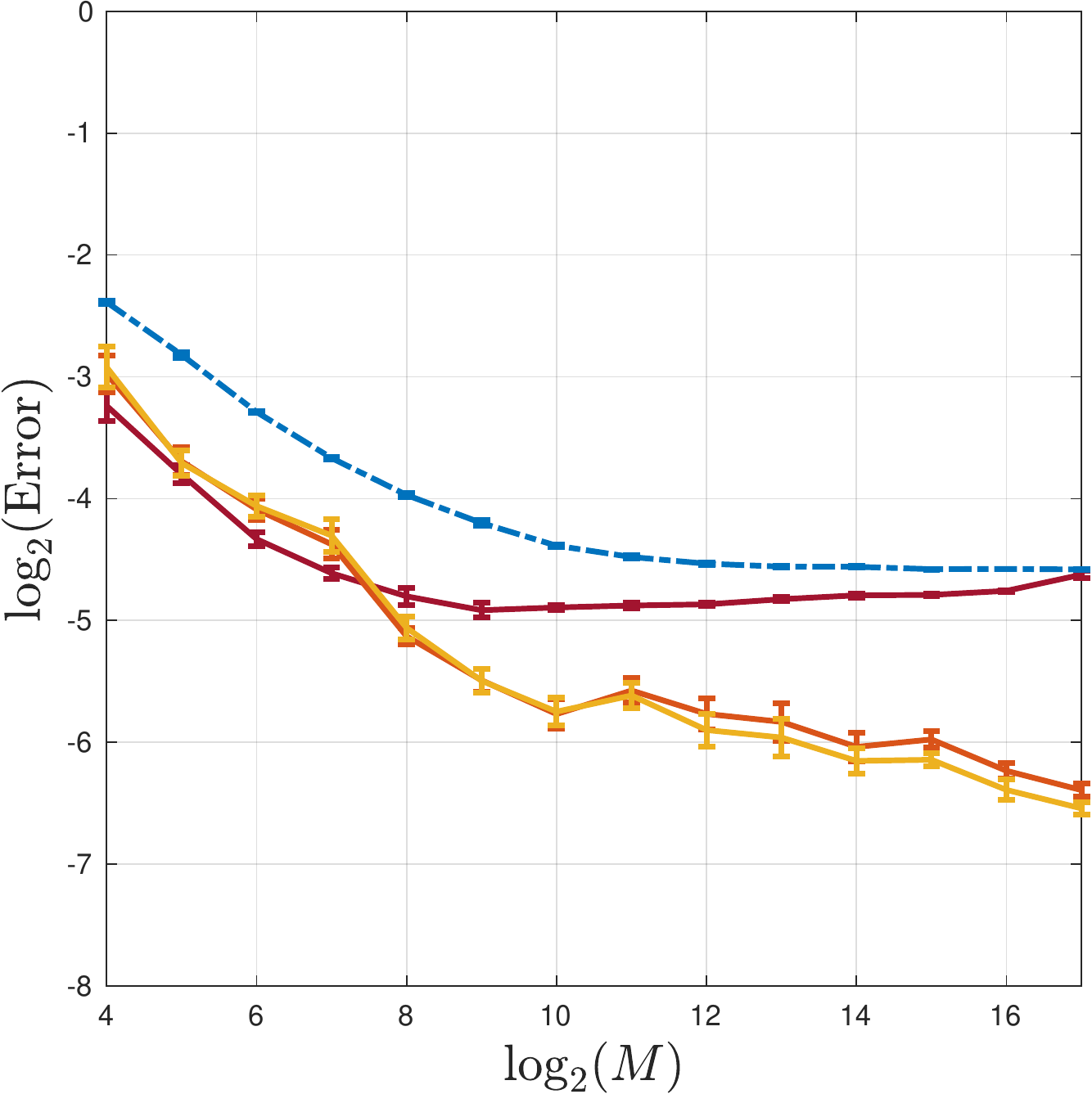}
		\caption{$f_3,\edit{\snr=2.2}, \eta=0.06$}
		\vspace*{.1cm}
		\label{fig:sim_O_3_high}
	\end{subfigure}
	\begin{subfigure}[b]{0.32\textwidth}
		\centering
		\includegraphics[width=.85\textwidth]{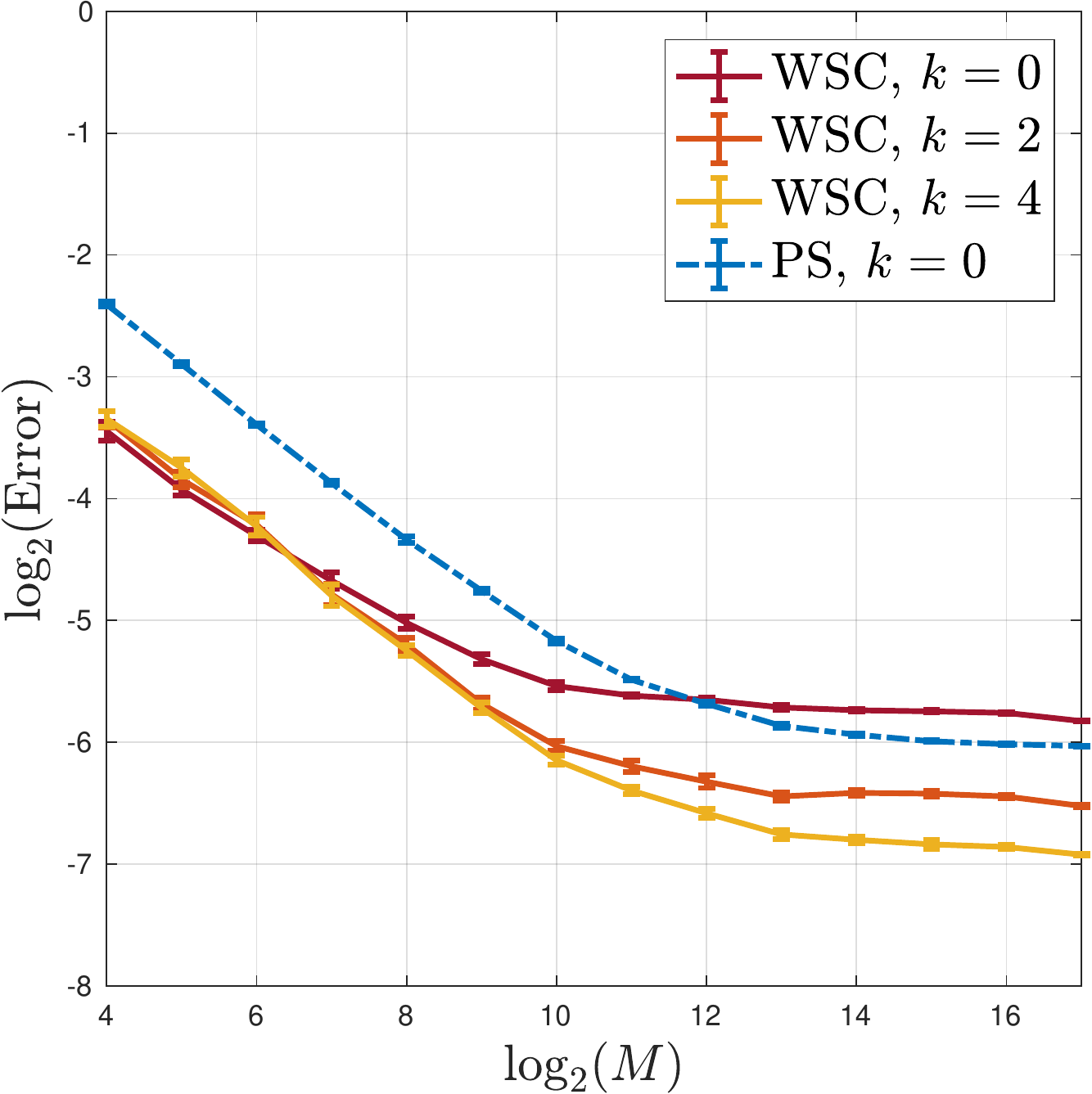}https://www.overleaf.com/project/5c9d8c94a843b632cccc13d1
		\caption{$f_1,\edit{\snr=2.2}, \eta=0.12$}
		\vspace*{.1cm}
		\label{fig:sim_O_4_low}
	\end{subfigure}
	\hfill
	\begin{subfigure}[b]{0.32\textwidth}
		\centering
		\includegraphics[width=.85\textwidth]{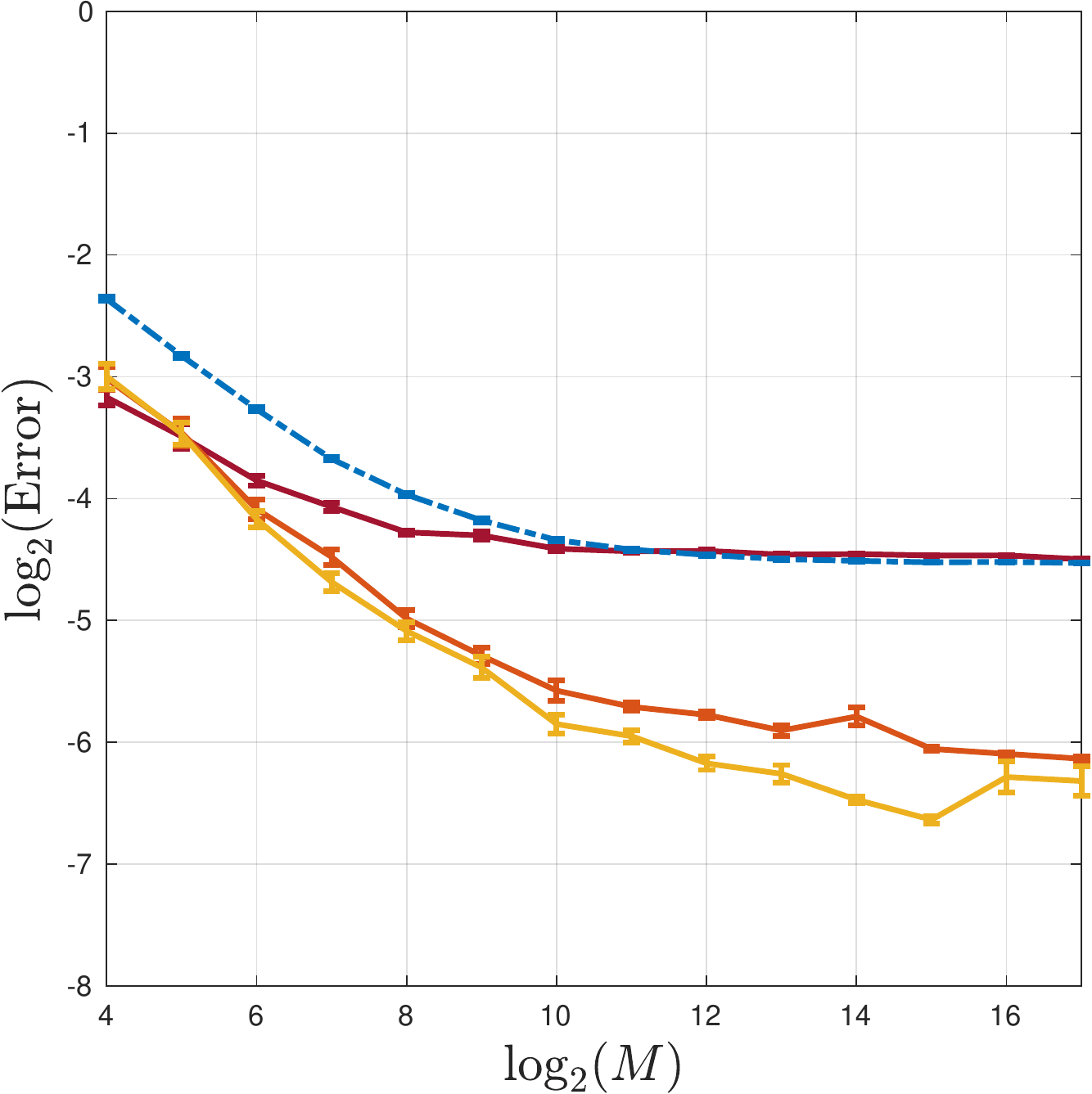}
		\caption{$f_2,\edit{\snr=2.2}, \eta=0.12$}
		\vspace*{.1cm}
		\label{fig:sim_O_4_med}
	\end{subfigure}
	\hfill
	\begin{subfigure}[b]{0.32\textwidth}
		\centering
		\includegraphics[width=.85\textwidth]{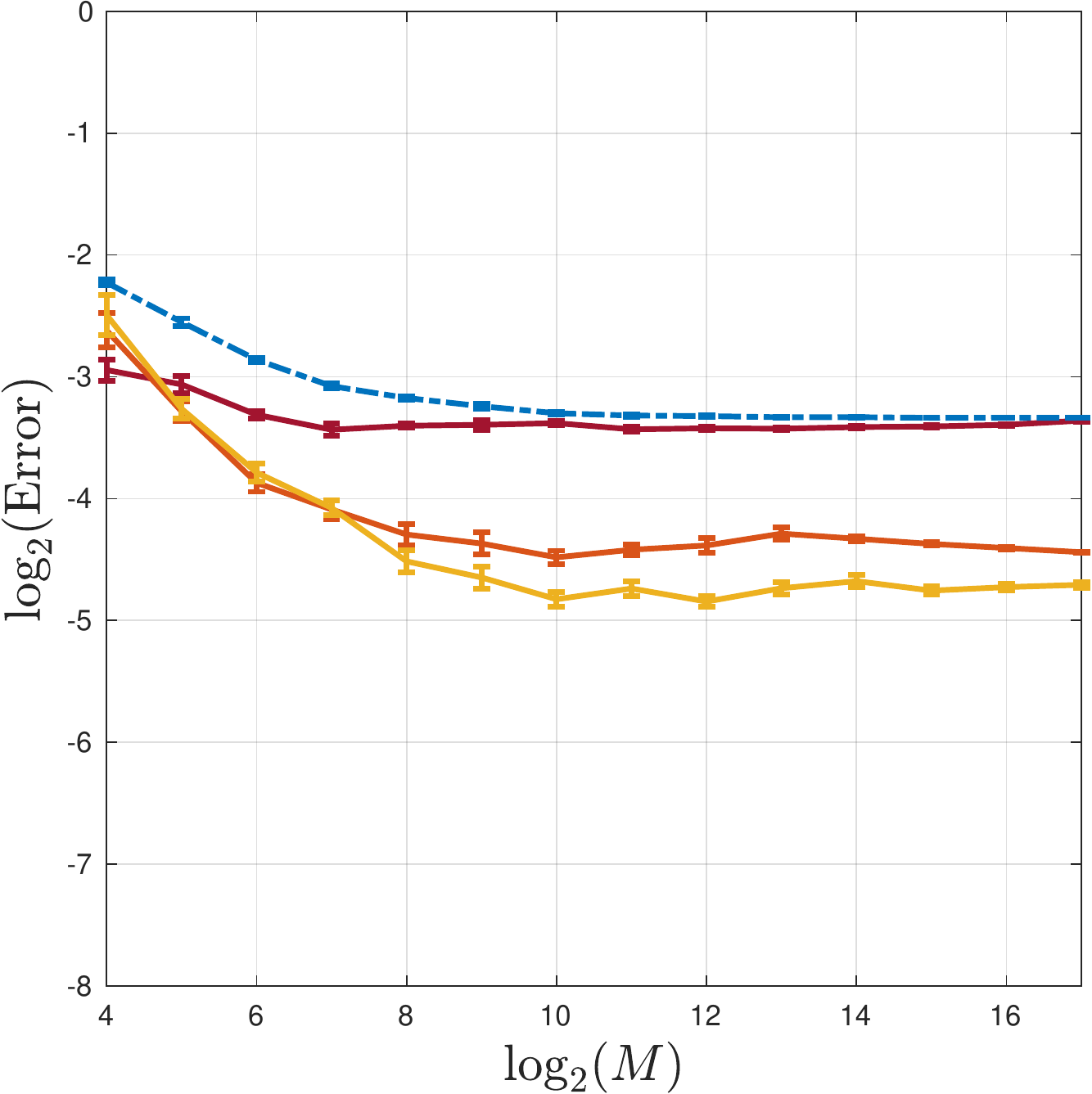}
		\caption{$f_3,\edit{\snr=2.2}, \eta=0.12$}
		\vspace*{.1cm}
		\label{fig:sim_O_4_high}
	\end{subfigure}
	\begin{subfigure}[b]{0.32\textwidth}
		\centering
		\includegraphics[width=.85\textwidth]{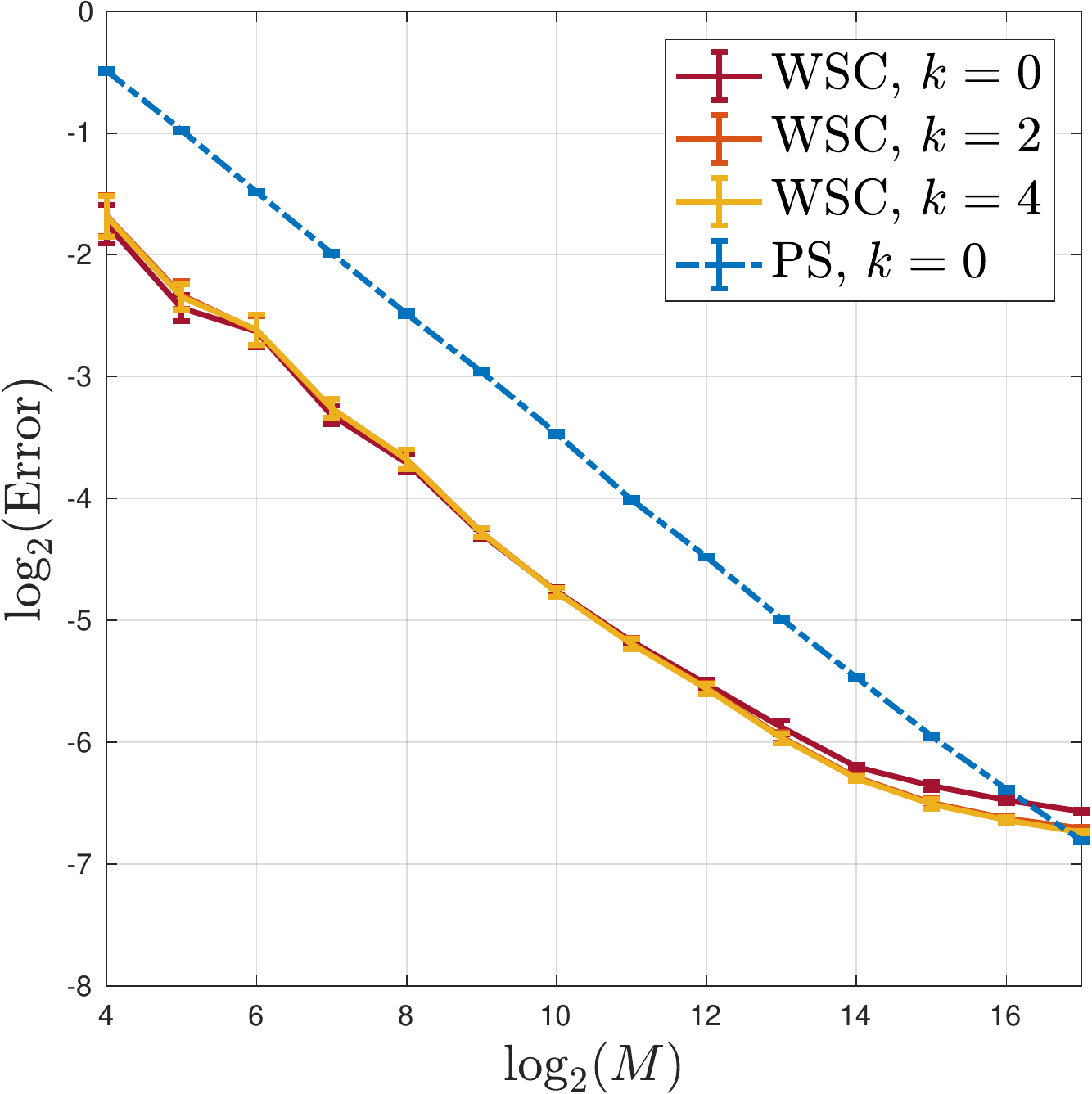}
		\caption{$f_1,\edit{\snr=0.56}, \eta=0.06$}
		\vspace*{.1cm}
		\label{fig:sim_O_5_low}
	\end{subfigure}
	\hfill
	\begin{subfigure}[b]{0.32\textwidth}
		\centering
		\includegraphics[width=.85\textwidth]{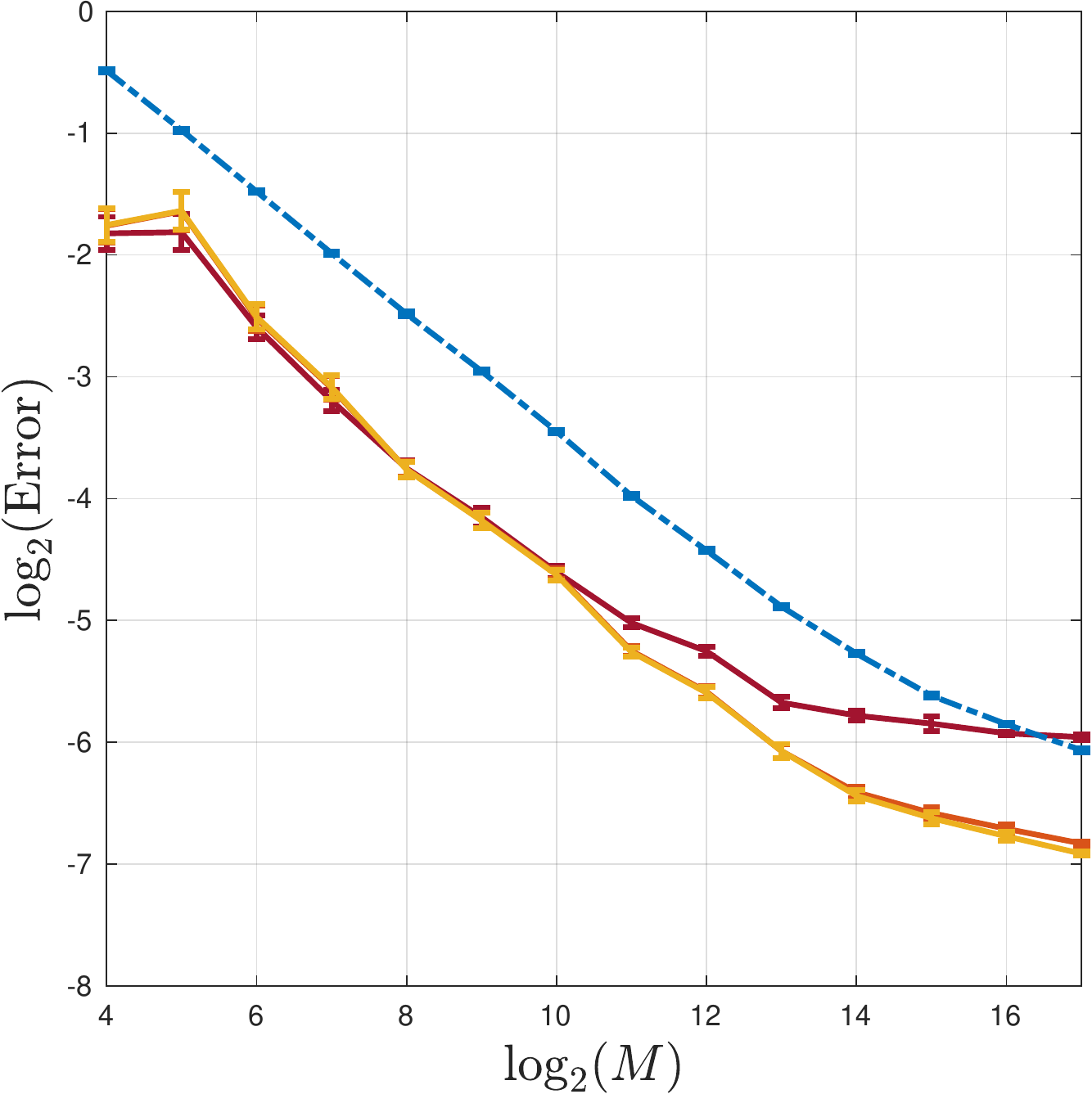}
		\caption{$f_2,\edit{\snr=0.56}, \eta=0.06$}
		\vspace*{.1cm}
		\label{fig:sim_O_5_med}
	\end{subfigure}
	\hfill
	\begin{subfigure}[b]{0.32\textwidth}
		\centering
		\includegraphics[width=.85\textwidth]{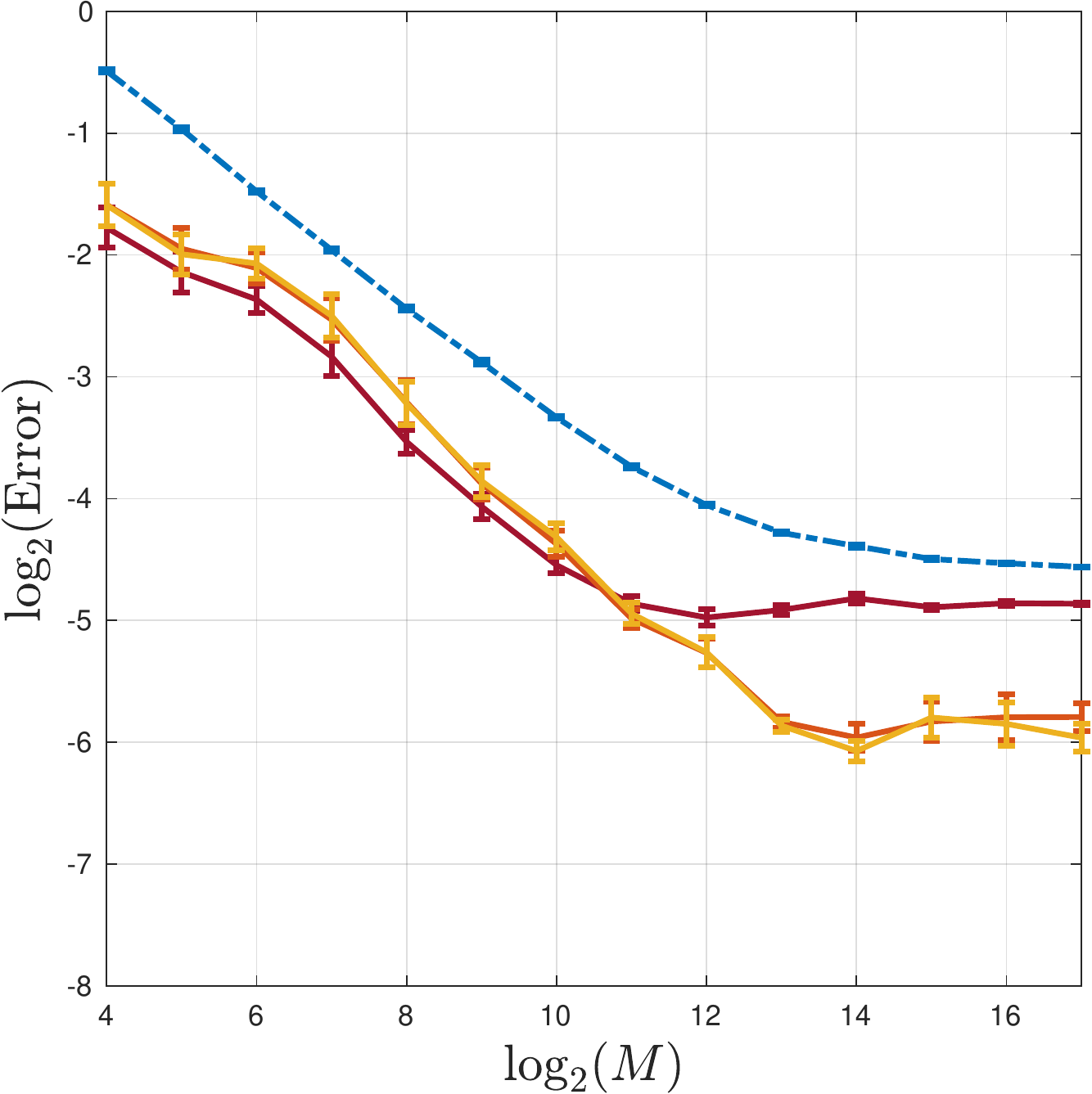}
		\caption{$f_3,\edit{\snr=0.56}, \eta=0.06$}
		\vspace*{.1cm}
		\label{fig:sim_O_5_high}
	\end{subfigure}
	\begin{subfigure}[b]{0.32\textwidth}
		\centering
		\includegraphics[width=.85\textwidth]{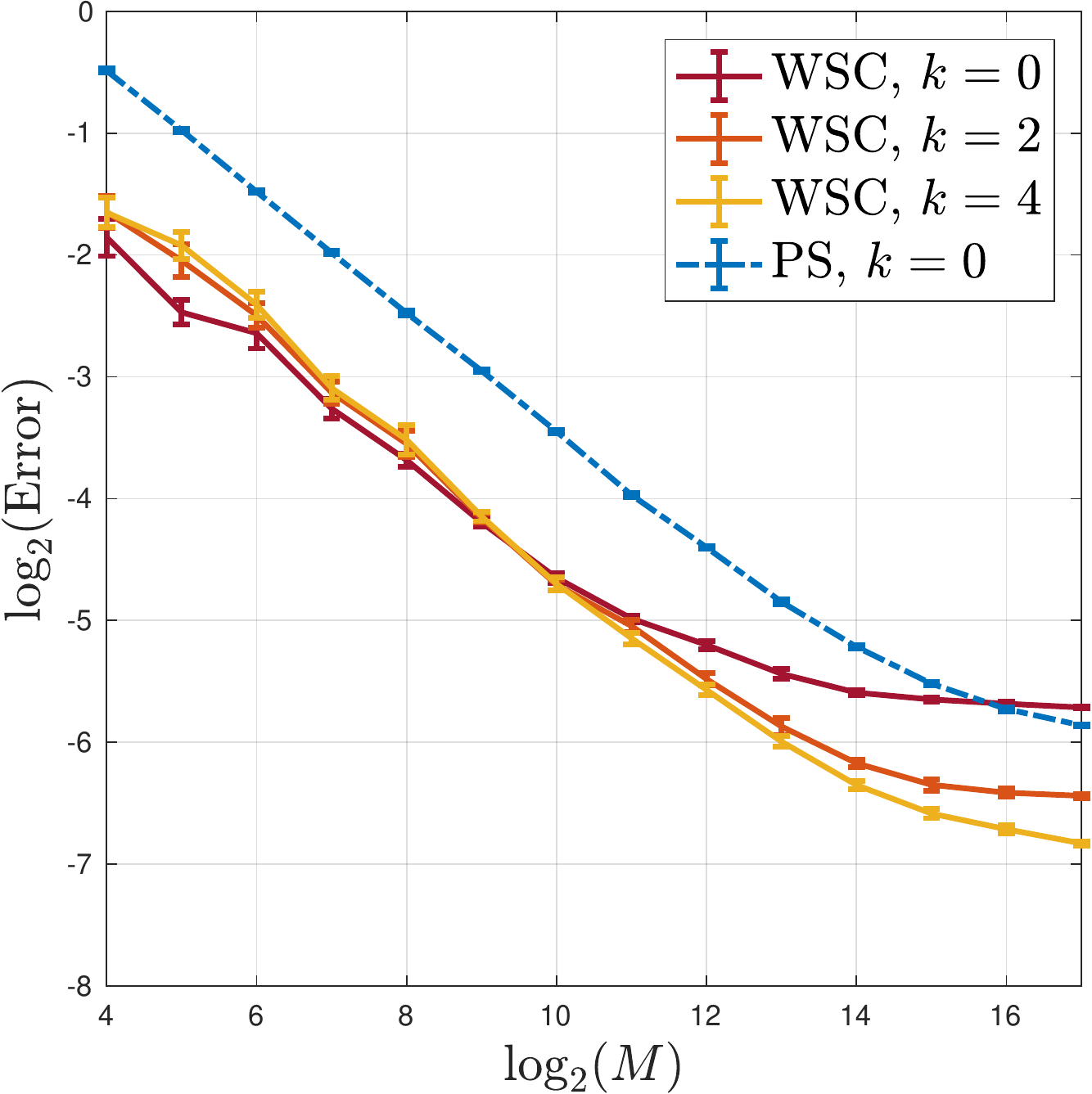}
		\caption{$f_1,\edit{\snr=0.56}, \eta=0.12$}
		\label{fig:sim_O_6_low}
	\end{subfigure}
	\hfill
	\begin{subfigure}[b]{0.32\textwidth}
		\centering
		\includegraphics[width=.85\textwidth]{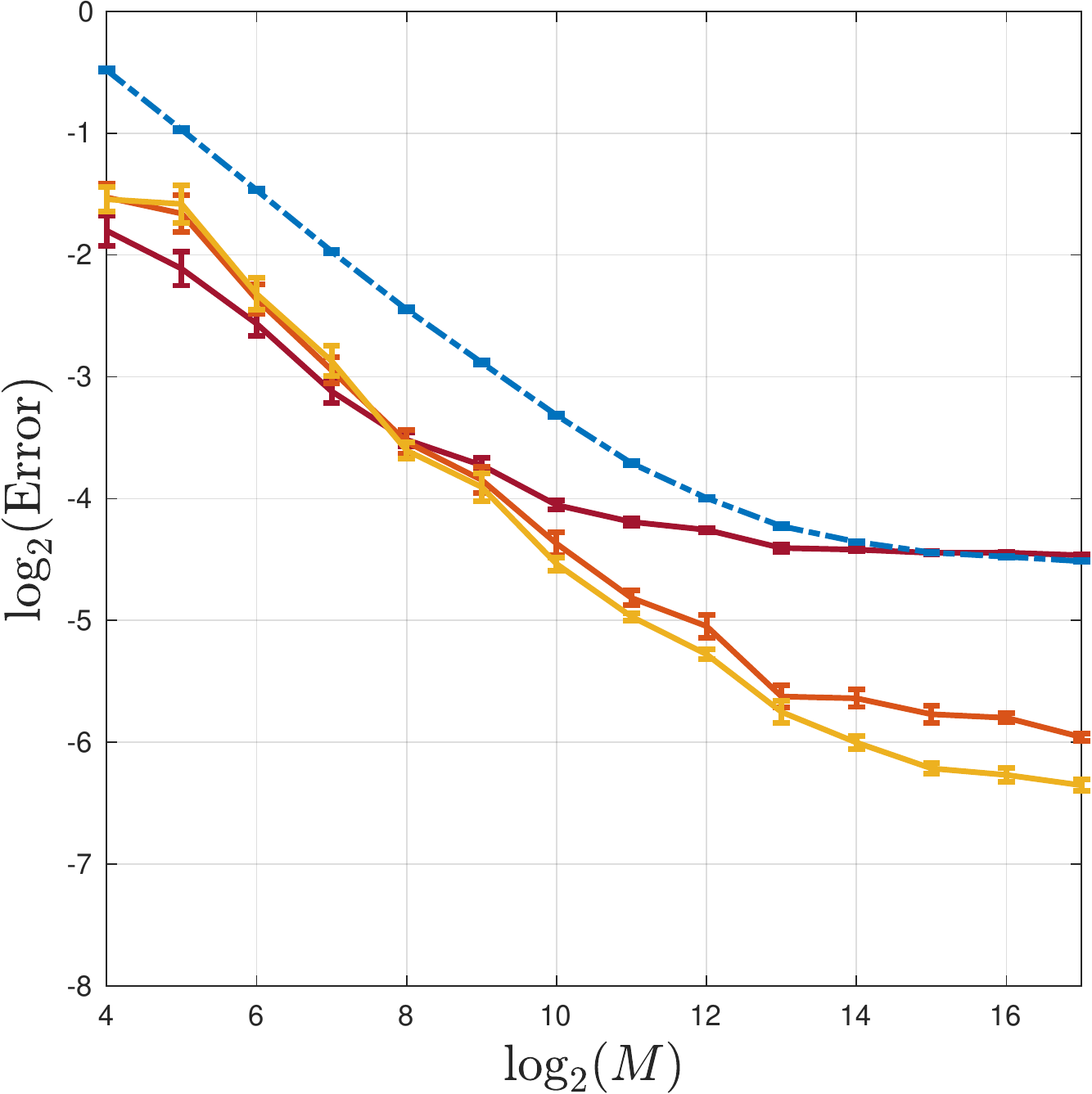}
		\caption{$f_2,\edit{\snr=0.56}, \eta=0.12$}
		\label{fig:sim_O_6_med}
	\end{subfigure}
	\hfill
	\begin{subfigure}[b]{0.32\textwidth}
		\centering
		\includegraphics[width=.85\textwidth]{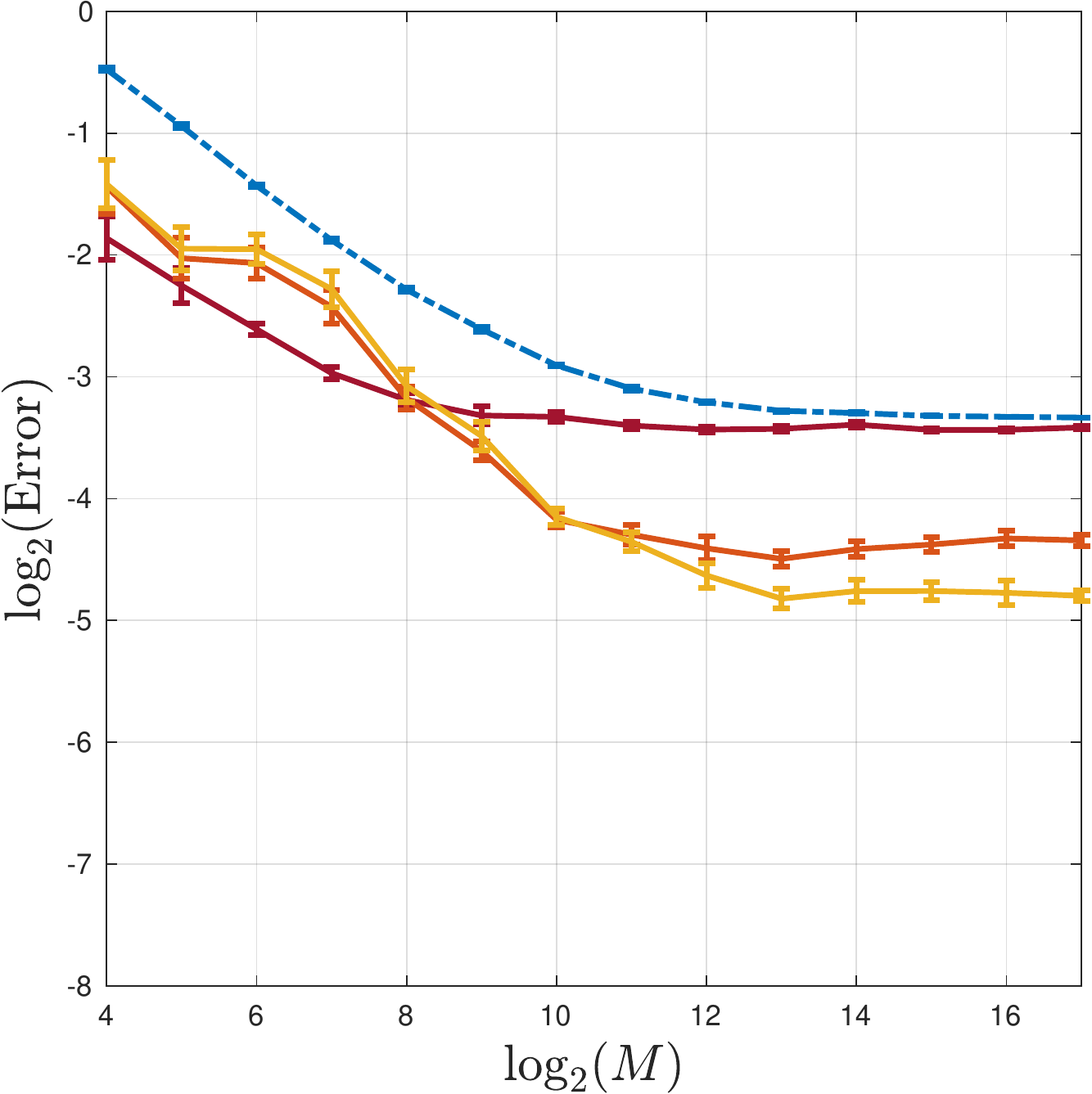}
		\caption{$f_3,\edit{\snr=0.56}, \eta=0.12$}
		\label{fig:sim_O_6_high}
	\end{subfigure}
	\caption{$\Lb^2$ error with standard error bars for \edit{noisy dilation MRA} model (oracle moment estimation). First, second, third column shows results for low, medium, high frequency Gabor signals.  All plots have the same axis limits.}
	\label{fig:GenMRAModelOracle}
\end{figure}

We note that although in general recovering the power spectrum is insufficient for recovering the signal, the signal can be recovered when \edit{$\widehat{f} (\omega) \in \R$ and} $\widehat{f}(\omega)\geq 0$ by taking the inverse Fourier transform of the root power spectrum. Figure \ref{fig:SignalRecovery} shows the approximate signals recovered by this procedure from PS $k=0$ (Figure \ref{fig:PSRecSignal}) and WSC $k=4$ (Figure \ref{fig:WSCRecSignal}) for the high frequency Gabor signal $f_3(x)$ (Figure \ref{fig:TargetSignal}). The WSC recovered signal is a much better approximation of the target signal. The recovered power spectra are shown in Figure \ref{fig:CompareRecPS}; PS $k=0$ is much flatter than the target power spectrum, while WSC $k=4$ is a good approximation of both the shape and height of the target power spectrum.

\begin{figure}
	\centering
	\begin{subfigure}[b]{0.24\textwidth}
		\centering
		\includegraphics[width=\textwidth]{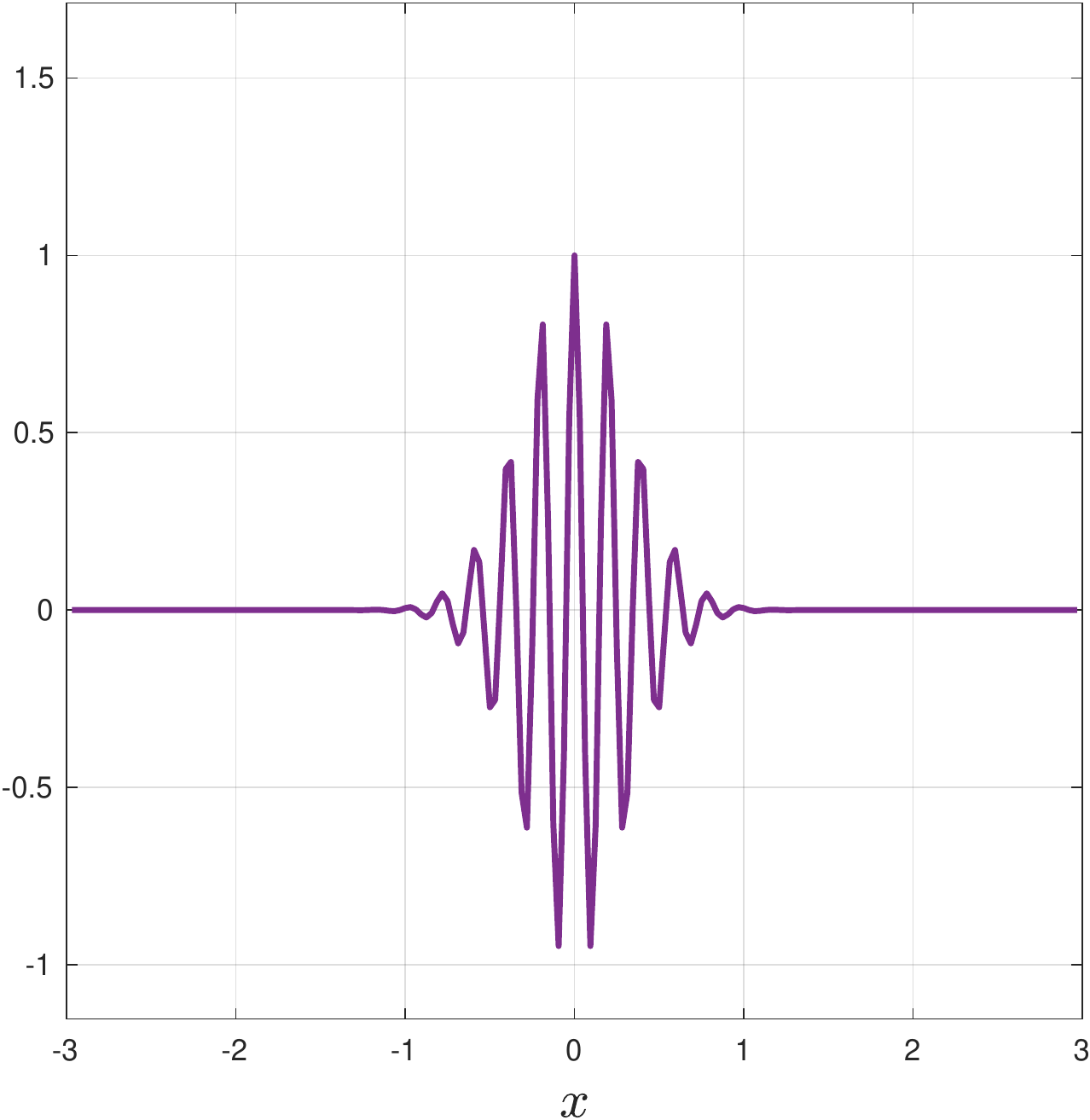}
		\caption{Target Signal}
		\label{fig:TargetSignal}
	\end{subfigure}
	\hfill
	\begin{subfigure}[b]{0.24\textwidth}
		\centering
		\includegraphics[width=\textwidth]{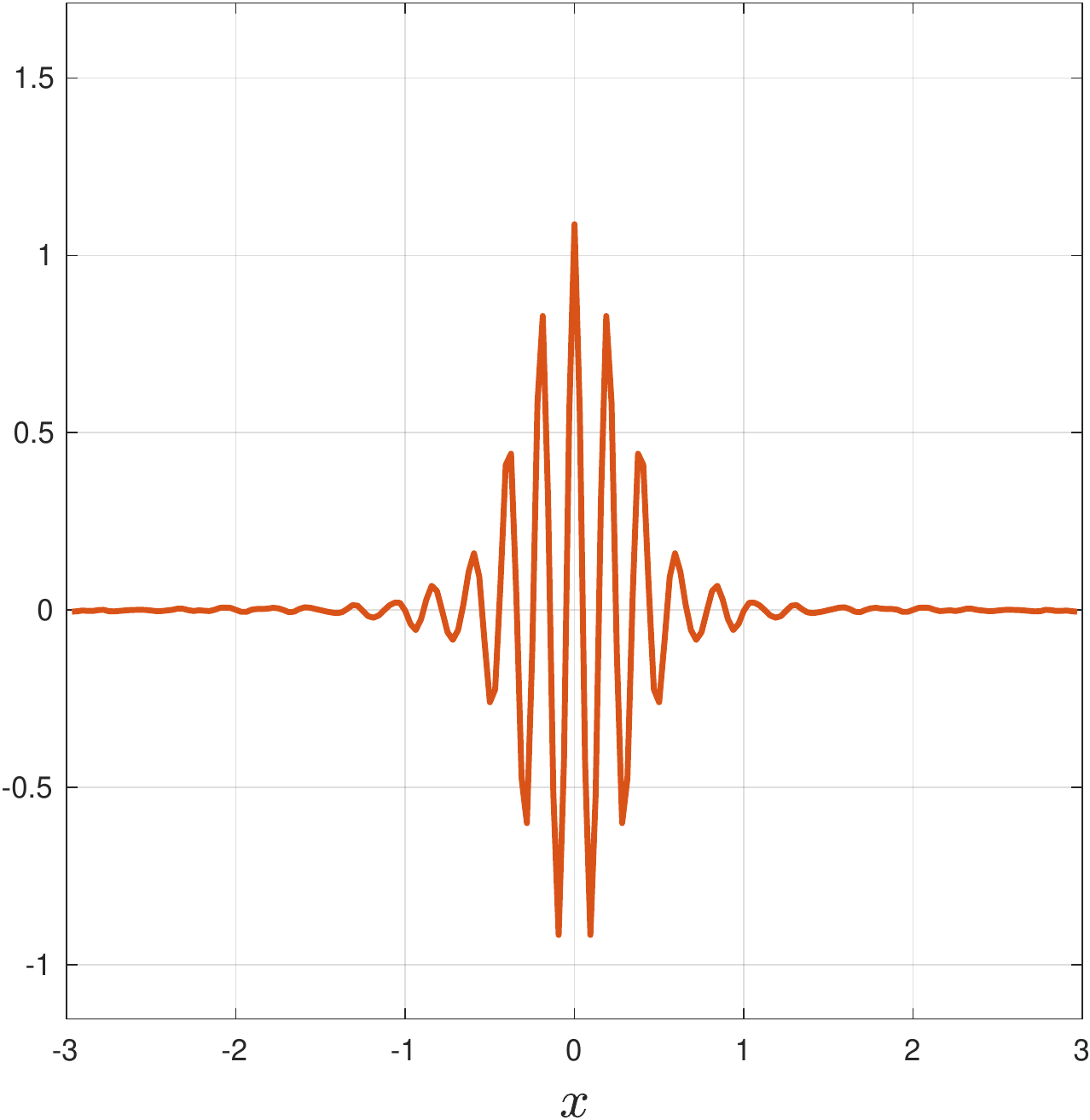}
		\caption{WSC Recovered}
		\label{fig:WSCRecSignal}
	\end{subfigure}
	\hfill
	\begin{subfigure}[b]{0.24\textwidth}
		\centering
		\includegraphics[width=\textwidth]{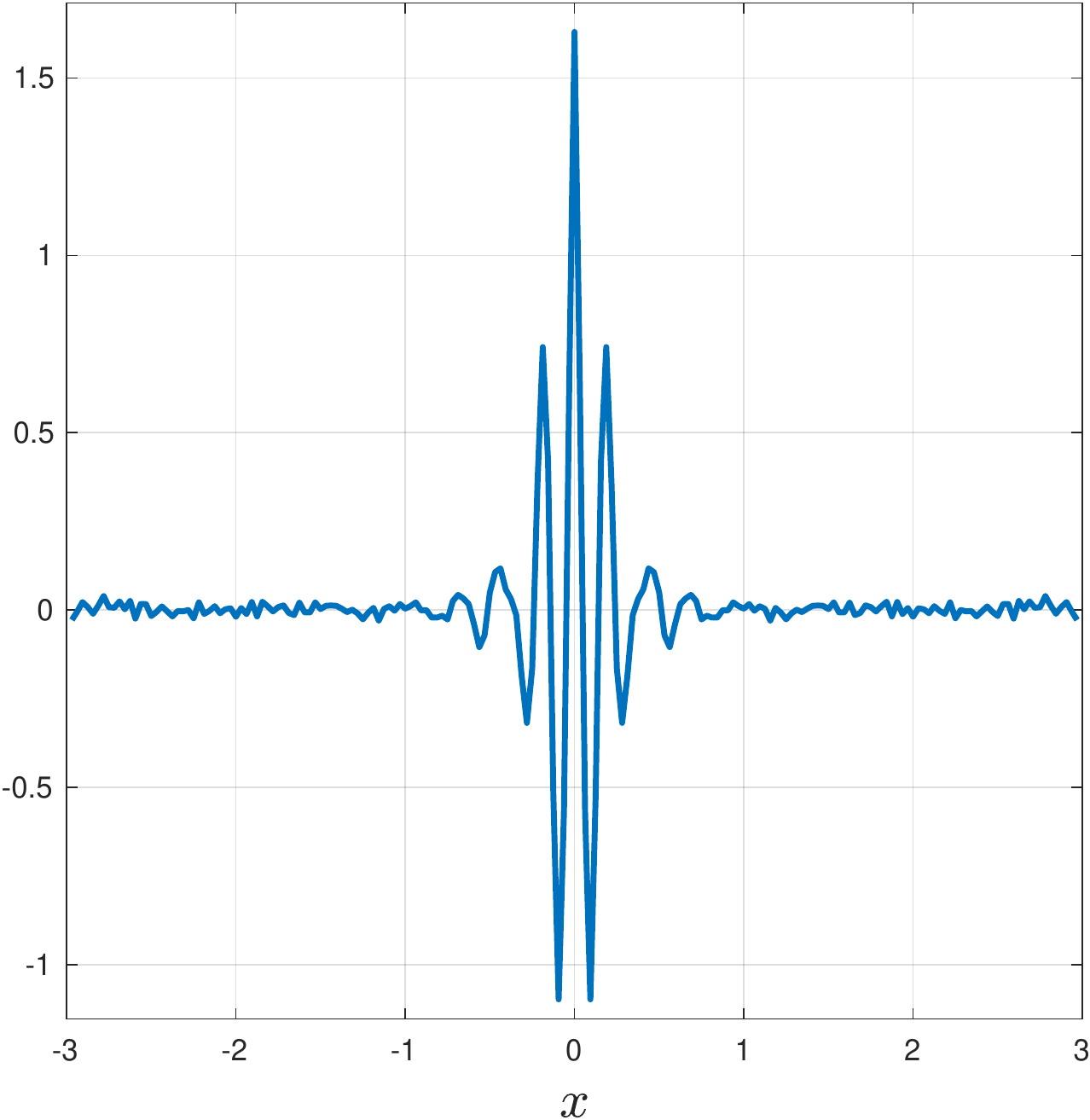}
		\caption{PS Recovered}
		\label{fig:PSRecSignal}
	\end{subfigure}
	\hfill
	\begin{subfigure}[b]{0.24\textwidth}
		\centering
		\includegraphics[width=\textwidth]{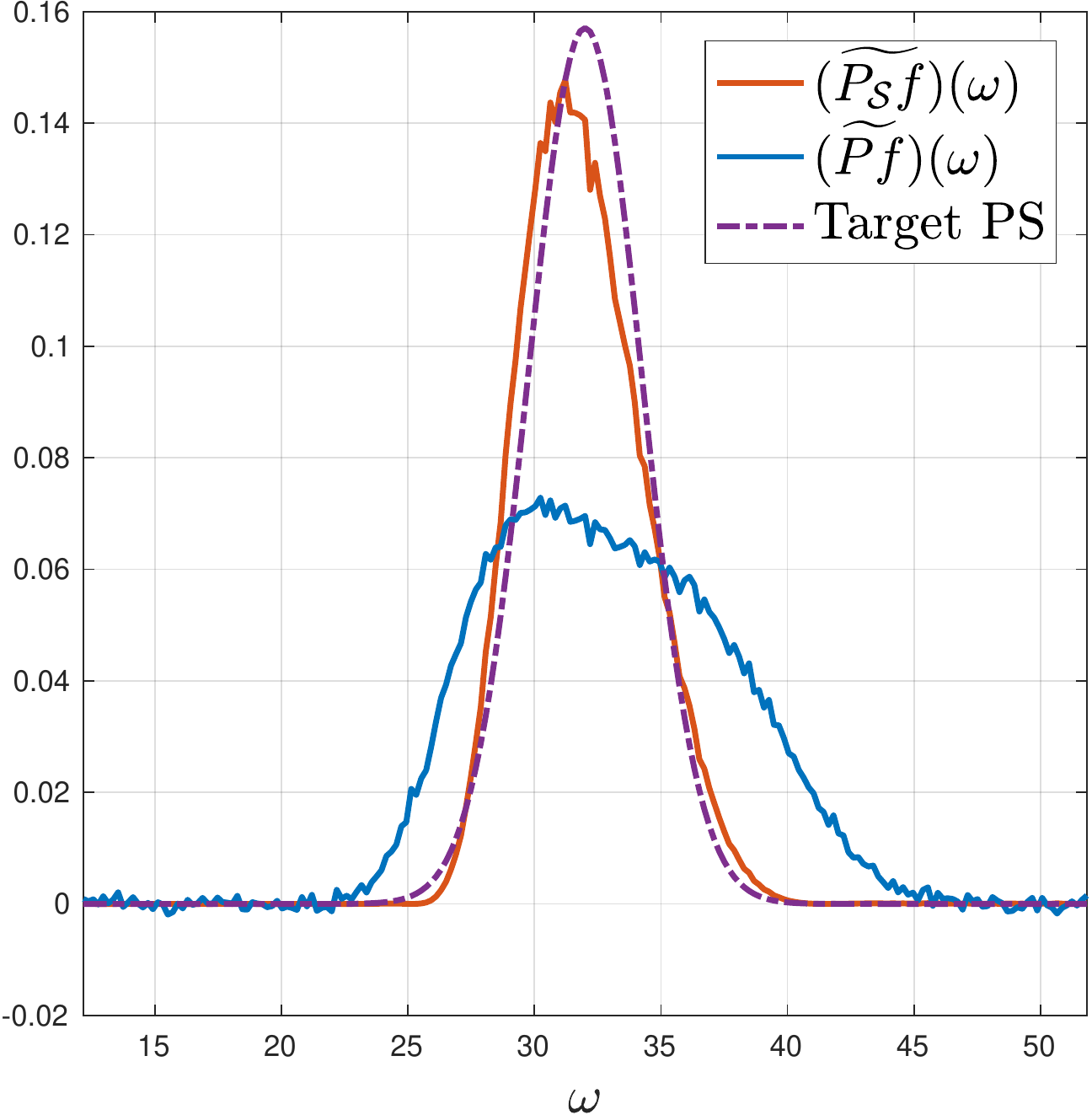}
		\caption{Recovered PS}
		\label{fig:CompareRecPS}
	\end{subfigure}
	\caption{Signal recovery results for $f_3(x) = e^{-5x^2} \cos (32 x)$ with $M=20,000$, $\eta=0.12$, $\edit{\snr=2.2}$. }
	\label{fig:SignalRecovery}
\end{figure}

Appendix \ref{app:numerical_implementation} outlines an empirical procedure for estimating the moments of $\tau$ in the special case when $t=0$ in the \edit{noisy dilation MRA} model (i.e., no random translations). All simulations reported in Figure \ref{fig:GenMRAModelOracle} are repeated (with minor modifications) with empirical additive and dilation moment estimation, and the results are reported in Figure \ref{fig:GenMRAModelEmpirical} of Appendix \ref{app:numerical_implementation}.

Appendix \ref{sec:additional_sim_results} contains additional simulation results for a variety of high frequency signals.

\edit{
\begin{remark} One could also solve noisy dilation MRA with an expectation-maximization (EM) algorithm. Appendix \ref{app:EMalg} describes how the method proposed in \cite{abbe2018multireference} can be extended to solve Model \ref{model:genMRA}. Althought EM algorithms provide a flexible tool for accurate parameter estimation in a variety of MRA models, the primary disadvantage is the high computational cost of each iteration. Each iteration costs $O(Mn^3)$, while wavelet invariant estimators can be computed in $O(Mn^2)$. In addition the statistical priors chosen may bias the signal reconstruction \cite{bendory2019single}, and the algorithm will generally only converge to a local maximum. In this article we thus explore whether it is possible to solve noisy dilation MRA more efficiently and accurately by nonlinear unbiasing procedures.
\end{remark}}

\section{Numerical implementation}
\label{sec:num_imp}

In this section we describe the numerical implementation of the proposed method used to generate the results reported in Sections \ref{sec:AddNoise}, \ref{sec:DilationNoiseSims}, and \ref{sec:GenMRANoiseSim}. Section \ref{sec:signal_synthesis} describes how signals were generated, and Sections \ref{sec:AddNoiseLevel} and  \ref{sec:EmpMomentEst} describe empirical procedures for estimating the additive noise level and the moments of the dilation distribution $\tau$. Finally, \edit{Section \ref{sec:derivatives} discusses how the derivatives used for unbiasing were computed, and} Section \ref{sec:optimization} describes the convex optimization algorithm used to recover $Pf$ from $Sf$. All simulations used a Morlet wavelet constructed with $\xi = 3\pi/4$.

\subsection{Signal \edit{generation and $\snr$}}
\label{sec:signal_synthesis}

All signals were defined on $[-N/4, N/4]$ and then padded with zeros to obtain a signal defined on \\$[-N/2, N/2]$; the additive noise was also defined on $[-N/2, N/2]$. Signals were sampled at a rate of $1/2^\ell$, thus resolving frequencies in the interval $[-2^\ell\pi, 2^{\ell}\pi]$ with a frequency sampling rate of $2\pi/N$. We used $N = 2^5$ and $\ell = 5$ in all experiments, keeping the box size and resolution fixed. \edit{For each experiment with hidden signal $f$, the SNR was calculated by
$\snr = \left(\frac{1}{N} \int_{-N/2}^{N/2} f(x)^2\ dx\right) / \sigma^2.$}

\subsection{Empirical estimation of additive noise level}
\label{sec:AddNoiseLevel}

The additive noise level $\sigma^2$ can be estimated from the mean vertical shift of the mean power spectrum $\frac{1}{M} \sum_{j=1}^M |\widehat{y}_j(\omega)|^2$ in the tails of the distribution. Specifically, for $\Sigma = [-2^\ell \pi, 2^\ell \pi] \setminus [-2^{\ell-1} \pi, 2^{\ell-1} \pi]$, we define
\begin{align*}
\widetilde{\sigma}^2 &= \frac{1}{|\Sigma|} \sum_{\omega \in \Sigma} \frac{1}{M} \sum_{j=1}^M |\widehat{y}_j(\omega)|^2.
\end{align*}
If we choose $\ell$ large enough so that the target signal frequencies are \edit{essentially} contained in the interval $[-2^{\ell-1}\pi, 2^{\ell-1}\pi]$, $|\widehat{y}_j(\omega)|^2 = |\widehat{\varepsilon}_j(\omega)|^2$ for $\omega \in \Sigma$, and this is a robust and unbiased estimation procedure since $\Ex|\widehat{\varepsilon}_j(\omega)|^2=\sigma^2$ by Lemma \ref{lem:PS_addnoise}. 

\subsection{Empirical moment estimation for dilation MRA}
\label{sec:EmpMomentEst}

Given the additive noise level, the moments of the dilation distribution $\tau$ for dilation MRA (Model \ref{model:dilMRA}) can be empirically estimated from the mean and variance of the random variables $\alpha_{m}(y_j)$ defined by
\begin{align}
\label{equ:alpha_m_def}
\alpha_{m}(y_j) &= \int_0^{2^{\ell}\pi} \omega^m |\widehat{y}_j(\omega)|^2\ d\omega
\end{align}
for integer $m \geq 0$. More specifically, we define the order $m$ squared coefficient of variation by
\begin{align}
\label{equ:CV_m_dilMRA}
CV_m &:= \frac{\Var[\alpha_{m}(y_j)]}{|\Ex[\alpha_{m}(y_j)]|^2} \, .
\end{align}
The following proposition guarantees that for $M$ large the second and fourth moments of the dilation distribution can be recovered from $CV_0, CV_1$. In fact one could continue this procedure for higher $m$ values, i.e. $\{CV_m\}_{m=0}^{k/2-1}$ will define estimators of the first $\frac{k}{2}$ even moments of $\tau$, accurate up to $O(\eta^{k+2})$, but for brevity we omit the general case.
\begin{proposition}
	\label{prop:emp_mom_est_dilMRA}
	Assume Model \ref{model:dilMRA} and $CV_0, CV_1$ defined by (\ref{equ:alpha_m_def}) and (\ref{equ:CV_m_dilMRA}). Then
	\begin{align*}
	CV_0 &=\eta^2+(3C_4-3)\eta^4+O(\eta^6) \\
	CV_1 &= 4\eta^2 + (25C_4-33)\eta^4+O(\eta^6) \, .
	\end{align*}	
\end{proposition}	

\begin{proof}
Since $y_j = L_{\tau_j}f(x-t_j)$, 
\begin{align*}
\alpha_{m}(y_j) &= \int_0^{2^{\ell}\pi} \omega^m |\widehat{f}((1-\tau_j)\omega)|^2\ d\omega \\
&= \int_0^{2^{\ell}\pi (1-\tau_j)} \frac{\xi^m}{(1-\tau_j)^m} |\widehat{f}(\xi)|^2\ \frac{d\xi}{(1-\tau_j)} \\
&= (1-\tau_j)^{-(m+1)} \alpha_m(f)\, ,
\end{align*}
where we assume we have choosen $\ell$ large enough so that the target signal frequencies are essentially supported in $[-2^{\ell-1} \pi, 2^{\ell-1} \pi]$. Thus:
\begin{align*}
CV_m &= \frac{ \Ex[\alpha_{m}(y_j)^2] - (\Ex[\alpha_{m}(y_j)])^2}{(\Ex[\alpha_{m}(y_j)])^2} = \frac{ \Ex[(1-\tau_j)^{-2(m+1)} ] }{(\Ex[(1-\tau_j)^{-(m+1)}])^2} - 1 \, .
\end{align*}
When $m=0$, we have
\begin{align*}
CV_0 &= \frac{ \Ex[(1-\tau_j)^{-2} ] }{(\Ex[(1-\tau_j)^{-1}])^2} - 1 \\
&= \frac{\Ex[1+2\tau+3\tau^2+4\tau^3+5\tau^4+O(\tau^5)]}{(\Ex[1+\tau+\tau^2+\tau^3+\tau^4+O(\tau^5) ])^2}-1 \\
&= \frac{1+3\eta^2+5C_4\eta^4+O(\eta^6)}{(1+\eta^2+C_4\eta^4+O(\eta^6))^2}-1 \\
&= \frac{1+3\eta^2+5C_4\eta^4+O(\eta^6))}{1+2\eta^2+(2C_4+1)\eta^4+O(\eta^6)}-1 \\
&=(1+3\eta^2+5C_4\eta^4+O(\eta^6))(1-2\eta^2+(3-2C_4)\eta^4+O(\eta^6))-1 \\
&=\eta^2+(3C_4-3)\eta^4+O(\eta^6)\, .
\end{align*}
When $m=1$,  we have
\begin{align*}
CV_1 &= \frac{ \Ex[(1-\tau_j)^{-4} ] }{(\Ex[(1-\tau_j)^{-2}])^2} - 1 \\
&= \frac{\Ex[1+4\tau+10\tau^2+20\tau^3+35\tau^4+O(\tau^5)]}{(\Ex[1+2\tau+3\tau^2+4\tau^3+5\tau^4+O(\tau^5)])^2} -1 \\
&= \frac{1+10\eta^2+35C_4\eta^4+O(\eta^6)}{(1+3\eta^2+5C_4\eta^4+O(\eta^6))^2} -1 \\
&= \frac{1+10\eta^2+35C_4\eta^4+O(\eta^6)}{(1+6\eta^2+(9+10C_4)\eta^4+O(\eta^6))} -1 \\
&=(1+10\eta^2+35C_4\eta^4+O(\eta^6))(1-6\eta^2+(27-10C_4)\eta^4+O(\eta^6)) - 1 \\
&= 4\eta^2 + (25C_4-33)\eta^4+O(\eta^6)\, .
\end{align*}
\end{proof}

We cannot compute $CV_m$ exactly, but by replacing $\Var, \Ex$ with their finite sample estimators, we obtain an approximate $\widetilde{CV}_m \rightarrow CV_m$ as $M \rightarrow \infty$. Motivated by Proposition \ref{prop:emp_mom_est}, we thus use $\widetilde{CV}_0, \widetilde{CV}_1$ to define estimators of $\eta^2$ and $C_4\eta^4$. 

\begin{definition}
	\label{def:emp_dil_mom_est_dilMRA}
	Assume Model \ref{model:dilMRA} and let $\widetilde{CV}_0, \widetilde{CV}_1$ be the empirical versions of (\ref{equ:CV_m_dilMRA}).
	Define the second order estimator of $\eta^2$ by $\widetilde{\eta}^2 = \widetilde{CV}_0.$
	Define the fourth order estimators of $(\eta^2, C_4\eta^4)$ by the unique positive solution $(\widetilde{\eta}^2, \widetilde{C}_4)$ of
	\begin{align*}
	\widetilde{CV}_0 &=\eta^2+(3C_4-3)\eta^4 \\
	\widetilde{CV}_1 &= 4\eta^2 + (25C_4-33)\eta^4.
	\end{align*}
\end{definition}

For \edit{noisy dilation MRA} (Model \ref{model:genMRA}), estimating the dilation moments is more difficult. We give a procedure for estimating the moments in the special case $t=0$ in Appendix \ref{app:numerical_implementation}. Empirical moment estimation procedures which are simultaneously robust to translations, dilations, and additive noise is an important area of future research.   

\edit{
\subsection{Derivatives}
\label{sec:derivatives}
All derivatives were approximated numerically using finite difference calculations. A 6$^{\text{th}}$ order finite difference approximation was used for second derivatives, and a 4$^{\text{th}}$ order finite difference approximation was used for fourth derivatives. This procedure was done on the empirical mean for each representation, not the individual signals. In fact since the wavelet is known, $\frac{d^{n}}{d\lambda^n}|\widehat{\psi}_\lambda(\omega)|^2$ could be computed analytically, and $(Sy_j)^{(n)}(\lambda)$ computed using Definition \ref{def:WSCderiv}. Thus error due to finite difference approximations could be avoided for wavelet invariant derivatives.
}

\subsection{Optimization}
\label{sec:optimization}

\edit{In this section we describe the convex optimization algorithm for computing $(\widetilde{P_{\Sc}f})$, the power spectrum approximation which best matches the wavelet invariants $(\widetilde{\Sc f})$.}
Since the wavelet invariants are only computed for $\lambda>0$, we also incorporate zero frequency information into the loss function via $(\widetilde{Pf})(0)$, an approximation of the power spectrum at frequency zero.     
For all of the examples reported in this article, the quasi-newton algorithm was used to solve an unconstrained optimization problem minimizing the following \edit{convex} loss function:
\begin{align*}
\text{loss}(\,\widehat{g}\,) := \sum_{\lambda} \left(\left\langle \widehat{g}^2, |\widehat{\psi}^+_{\lambda}|^2 \right\rangle - \widetilde{\Sc f}(\lambda)\right)^2 + \left(\widehat{g}(0)^2 -(\widetilde{Pf})(0) \right)^2 \, ,
\end{align*}
where
\begin{align*}
|\widehat{\psi}_\lambda^+(\omega)|^2 &= \left(|\widehat{\psi}_\lambda(\omega)|^2 + |\widehat{\psi}_\lambda(-\omega)|^2\right)\cdot\ind(\omega \geq 0) \, .
\end{align*}
Letting $\widehat{g}^*$ denote the minimizer of the above loss function, we then define $(\widetilde{P_{\Sc}f}):=\widehat{g}^*(\omega)^2$. Theorem \ref{thm:WSC_PS_equivalence} ensures that when the loss function is defined with the exact wavelet invariants $\Sc f$, it has a unique minimizer corresponding to $Pf$. Whenever $f(x) \in \R$, the symmetry of $(Pf)(\omega)$ ensures that $(\Sc f)(\lambda) = \left\langle |\widehat{f}|^2, |\widehat{\psi}^+_{\lambda}|^2 \right\rangle$, and thus it is sufficient to optimize over the nonnegative frequencies and then symmetrically extend the solution. Such a procedure ensures the output of the optimization algorithm is symmetric while avoiding adding constraints to the optimization. 
\edit{The algorithm was initialized using the mean power spectrum with additive noise unbiasing only, i.e. PS $k=0$.}
The optimization output does depend on various numerical tolerance parameters which were held fixed for all examples.

\edit{
\begin{remark}
    Alternatively, one can invert the representation by applying a pseudo-inverse with Tikhonov regularization. Specifically if $F$ is the matrix defining the wavelet invariants, so that $Sy=F(Py)$, then one can define $(\widetilde{P_{\Sc}f}) = (F^TF+\lambda I)^{-1}F^T(\widetilde{\Sc f})$. This procedure however requires careful selection of the hyper-parameter $\lambda$ and did not work as well as inverting via optimization in our experiments.
\end{remark}
}

\section{Conclusion}

This article considers a generalization of classic MRA which incorporates random dilations in addition to random translations and additive noise, and proposes solving the problem with a wavelet invariant representation. These wavelet invariants have several desirable properties over Fourier invariants which allow for the construction of unbiasing procedures which cannot be constructed for Fourier invariants. Unbiasing the representation is critical for high frequency signals, where even small diffeomorphisms cause a large perturbation. After unbiasing, the power spectrum of the target signal can be recovered from a convex optimization procedure. 

Several directions remain for further investigation, including extending results to higher dimensions and considering rigid transformations instead of translations. Such extensions could be especially relevant to image processing, where variations in the size of an object can be modeled as dilations. Incorporating the effect of tomographic projection would also lead to results more directly relevant to problems such as Cryo-EM. The tools of the present article, although significantly reducing the bias, do not allow for a completely unbiased estimator for \edit{noisy dilation MRA} due to the bad scaling of certain intrinsic constants. Thus an important open question is whether it is possible to define unbiased estimators for \edit{noisy dilation MRA} using a different approach. The \edit{noisy dilation MRA} model of this article corresponds to linear diffeomorphisms, and constructing unbiasing procedures which apply to more general diffeomorphisms is also an important future direction. In addition, one can construct wavelet invariants which characterize higher order auto-correlation functions such as the bispectrum, and future work will investigate full signal recovery with such invariants.  

\section*{Funding} 

This  work was supported by: the Alfred P. Sloan Foundation [Sloan Fellowship FG-2016-6607 to M.H.]; the Defense Advanced Research Projects Agency [Young Faculty Award D16AP00117 to M.H.]; and the National Science Foundation [grant 1620216 and CAREER award 1845856 to M.H].

\section*{Acknowledgements}

We would like to thank the reviewers for their detailed comments and insights which greatly improved the manuscript. We would also like to thank Stephanie Hickey for providing useful references on flexible regions of macromolecular structures. 

\appendix

\section{\edit{Wavelet admissibility conditions}}
\label{app:wavelet_admissability}

\edit{This appendix describes the wavelet admissibility conditions which are needed for the main results in this article, namely Propositions \ref{prop:RandomDilations} and \ref{prop:DilationAndAdditiveNoise}.}
The wavelet $\psi$ is $\bm{k}$\textbf{-admissable} if $\widehat{\psi} \in \Cb^k (\R)$ and $\Psi_k < \infty, \Theta_k < \infty$ where
\begin{align}
\label{equ:Psik}
\Psi_k &:= \frac{1}{2\pi} \sum_{i=0}^k {k \choose i} \frac{k!}{i!} \, \norm{\omega^{i}(P\psi)^{(i)}(\omega)}_1 \, , \\
\label{equ:Thetak}
\Theta_k &:= \frac{1}{2\pi} \sum_{i=0}^k {k \choose i} \frac{k!}{i!} \, \norm{\omega^{i-2}(P\psi)^{(i)}(\omega)}_1\, .
\end{align} 

For $\psi$ to be $k$-admissable, it is sufficient for $\widehat{\psi} \in \Cb^{k} (\R)$, $(P\psi)^{(i)}$ to decay faster than $\omega^{i+1}$, and \\ $\int  \frac{|\widehat{\psi}(\omega)|^2}{\omega^2} \, d\omega < \infty$
(see Lemma \ref{lem:kadmissable} in Appendix \ref{app:WSC_props}). The condition $\int  \frac{|\widehat{\psi}(\omega)|^2}{\omega^2} \, d\omega < \infty$ is slightly stronger than the classic admissability condition $C_{\psi} :=\int  \frac{|\widehat{\psi}(\omega)|^2}{\omega} \, d\omega < \infty$ \cite[Theorem 4.4]{Mallat:2008:WTS:1525499}. When $\widehat{\psi}$ is continuously differentiable, $\widehat{\psi}(0)=0$ is sufficient to guarantee $C_{\psi} < \infty$; but here we need $\widehat{\psi} (\omega) \sim \omega^{\frac{1}{2}+\epsilon}$ for some $\epsilon >0$ as $\omega \rightarrow 0$. If this condition is removed, we are not guaranteed $\Theta_k < \infty$, but all results in fact still hold, with $\Lambda_k(\lambda) = \Psi_k \norm{f}_1^2$ replacing $\Lambda_k(\lambda) = \Psi_{k} \norm{f}_1^2\wedge \frac{\Theta_{k}\norm{f'}_1^2}{\lambda^2}$ in Propositions \ref{prop:RandomDilations} and \ref{prop:DilationAndAdditiveNoise}. \edit{Any wavelet with fast decay satisfies this stronger admissibility condition, and it ensures that a smooth signal will enjoy a fast decay of wavelet invariants.}
\begin{remark}
	The Morlet wavelet $\psi(x) = g(x)(e^{i\xi x}-C)$ is $k$-admissable for any $k$, since $\widehat{\psi} \in \Cb^{\infty} (\R)$, $P\psi$ has fast decay, and $\widehat{\psi} (\omega) \sim \omega$ as $\omega \rightarrow 0$. One can also choose $\widehat{\psi}$ to be an order $k+1$-spline of compact support. 
\end{remark}

\section{Properties of wavelet invariants}
\label{app:WSC_props}

This appendix establishes several important properties of wavelet invariants. Lemma \ref{lem:kadmissable} gives sufficient conditions guaranteeing that a wavelet is $k$-admissable. Lemmas \ref{lem:WaveletScatteringDeriv} and \ref{lem:WaveletScatteringDerivHighFreq_DiffFunc} bound wavelet invariant derivatives. Lemma \ref{lem:WSC_deriv_cons} bounds terms which arise in the dilation unbiasing procedure of Sections \ref{sec:WSCdilationMRA} and \ref{sec: noisy dilation MRA model}.

\begin{lemma}[$k$-admissable]
	\label{lem:kadmissable}
	If $\widehat{\psi}\in\Cb^{k} (\R)$, $(P\psi)^{(i)}$ decays fast than $\omega^{i+1}$, and $\int  \frac{|\widehat{\psi}(\omega)|^2}{\omega^2}\ d\omega <\infty$, then $\psi$ is $k$-admissable.
\end{lemma}
\begin{proof}
	We first note that $\widehat{\psi} \in \Cb^{k} (\R)$ guarantees $P\psi \in \Cb^{k} (\R)$. Since $(P\psi)^{(i)}$ decays faster than $\omega^{i+1}$ and $P\psi \in \Cb^{k} (\R)$, $\omega^i(P\psi)^{(i)}(\omega) \in \Lb^1 (\R)$ for $0 \leq i\leq k$, so $\Psi_k < \infty$. Also $P\psi \in \Cb^{k} (\R)$ and $\omega^{i}(P\psi)^{(i)} \in \Lb^1 (\R)$ implies $\omega^{i-2}(P\psi)^{(i)} \in \Lb^1 (\R)$ for $2 \leq i \leq k$. In addition, $\omega^{-2}(P\psi)(\omega) \in \Lb^1 (\R)$ by assumption. Thus to conclude $\Theta_k<\infty$, it only remains to show $\omega^{-1}(P\psi)^{'}(\omega) \in \Lb^1 (\R)$. Since $(P\psi)^{'}$ is continuous and decays faster than $\omega^2$, only the integrability around the origin needs to be verified.  We note that $\int  \frac{|\widehat{\psi}(\omega)|^2}{\omega^2}\ d\omega <\infty$ and $P\psi$ continuous implies $P\psi \sim \omega^{1+\epsilon}$ for some $\epsilon >0$ as $\omega \rightarrow 0$. Thus $(P\psi)' \sim \omega^{\epsilon}$ as $\epsilon \rightarrow 0$, so that $\omega^{-1}(P\psi)' \sim \omega^{\epsilon-1}$; the function is thus integrable around the origin since $\epsilon-1 > -1$. 
\end{proof}

\lemWaveletScatteringDeriv*

\begin{proof}
	Let $g(\omega) = (P\psi)(\omega) = |\widehat{\psi}(\omega)|^2$, and let
	\begin{align*}
	g_\lambda(\omega) := \frac{1}{\lambda}g\left(\frac{\omega}{\lambda}\right) = |\widehat{\psi}_\lambda(\omega)|^2\, .
	\end{align*}
	Utilizing Definition \ref{def:WSCderiv} we obtain
	\begin{align*}
	\lambda^m(Sf)^{(m)}(\lambda) &= \frac{1}{2\pi}\int |\widehat{f}(\omega)|^2 \left[  \lambda^m \frac{d^m}{d\lambda^m} g_\lambda(\omega)\right]\ d\omega	\, .
	\end{align*}
	Expanding the derivative gives:
	\begin{align*}
	\lambda^m \frac{d^m}{d\lambda^m} g_\lambda(\omega) &= C_{m,0}\,g_\lambda(\omega)+C_{m,1}\,\omega g'_\lambda(\omega)+C_{m,2}\,\omega^2 g''_\lambda(\omega) + \ldots C_{m,m}\,\omega^m g^{(m)}_\lambda(\omega) \, ,\\
	C_{m,i} &= (-1)^m {m \choose i} \frac{m!}{i!} \, .
	\end{align*}
	Utilizing $\norm{\widehat{f}}_{\infty} \leq \norm{f}_1$ and $g^{(i)}_\lambda(\omega) = \frac{1}{\lambda^{i+1}}g^{(i)}\left(\frac{\omega}{\lambda}\right)$, one obtains:
	\begin{align*}
	|\lambda^m (Sf)^{(m)}(\lambda)| &\leq \sum_{i=0}^m \frac{|C_{m,i}|}{2\pi} \int  |\widehat{f}(\omega)|^2 |\omega^i g_\lambda^{(i)}(\omega)| \ d\omega \\
	&\leq \norm{f}_1^2 \sum_{i=0}^m \frac{|C_{m,i}|}{2\pi} \int |\omega^i g_\lambda^{(i)}(\omega)| \ d\omega \\
	&= \norm{f}_1^2 \sum_{i=0}^m \frac{|C_{m,i}|}{2\pi} \int |\omega^i g^{(i)}(\omega)| \ d\omega \\
	&= \norm{f}_1^2 \sum_{i=0}^m \frac{|C_{m,i}|}{2\pi} \cdot \norm{\omega^i g^{(i)}(\omega)}_1 \\
	&= \Psi_m \norm{f}_1^2\, .
	\end{align*}
\end{proof}

\lemWaveletScatteringDerivHighFreqDiffFunc*

\begin{proof}
	Recall from the proof of Lemma \ref{lem:WaveletScatteringDeriv} that:
	\begin{align*}
	|\lambda^m (Sf)^{(m)}(\lambda)| &\leq \sum_{i=0}^m \frac{|C_{m,i}|}{2\pi} \int  |\widehat{f}(\omega)|^2 |\omega^i g_\lambda^{(i)}(\omega)| \ d\omega 
	\end{align*}
	where $g_\lambda(\omega) = \frac{1}{\lambda}g\left(\frac{\omega}{\lambda}\right) = |\widehat{\psi}_\lambda(\omega)|^2$ and $C_{m,i} = (-1)^m {m \choose i} \frac{m!}{i!}$. Since $\norm{\omega\widehat{f}(\omega)}_{\infty} \leq \norm{f'}_1$ and $g^{(i)}_\lambda(\omega) = \frac{1}{\lambda^{i+1}}g^{(i)}\left(\frac{\omega}{\lambda}\right)$, we obtain:
	\begin{align*}
	|\lambda^m (Sf)^{(m)}(\lambda)| &\leq \sum_{i=0}^m \frac{|C_{m,i}|}{2\pi} \int  |\omega\widehat{f}(\omega)|^2 |\omega^{i-2} g_\lambda^{(i)}(\omega)| \ d\omega \\
	&\leq \norm{f'}_1^2 \sum_{i=0}^m\frac{|C_{m,i}|}{2\pi} \int  |\omega^{i-2} g_\lambda^{(i)}(\omega)| \ d\omega \\
	&= \frac{\norm{f'}_1^2}{\lambda^2} \sum_{i=0}^m \frac{|C_{m,i}|}{2\pi} \int  |\omega^{i-2} g^{(i)}(\omega)| \ d\omega \\
	&= \frac{\norm{f'}_1^2}{\lambda^2} \sum_{i=0}^m \frac{|C_{m,i}|}{2\pi} \cdot \norm{\omega^{i-2}g^{(i)}(\omega)}_1 \\
	&= \frac{\Theta_m}{\lambda^2} \norm{f'}_1^2\, .
	\end{align*}
\end{proof}

\begin{lemma} 
	\label{lem:WSC_deriv_cons}
	Assume $Pf\in\Cb^0 (\R)$ and $\psi$ is $m$-admissable, and let $B_m, E, \Psi_m, \Theta_m$ be as defined in (\ref{equ:MomentBasedConstants}), (\ref{equ:E}), (\ref{equ:Psik}) (\ref{equ:Thetak}). Then:
	\begin{align*}
	\frac{1}{2\pi}\int |\widehat{f}(\omega)|^2 \cdot \left|B_m\eta^m\lambda^m \frac{d^m}{d\lambda^m} |\widehat{\psi}_\lambda(\omega)|^2\right|\ d\omega &\leq (E\eta)^m \Lambda_m(\lambda)\, ,
	\end{align*}
	where
	\begin{align*}
	\Lambda_m(\lambda) &= \left(\norm{f}_1^2\Psi_m \wedge \frac{\norm{f'}_1^2\Theta_m}{\lambda^2}\right)\, .
	\end{align*}
\end{lemma}
\begin{proof}
	From the proof of Lemma \ref{lem:WaveletScatteringDeriv}:
	\begin{align*}
	\frac{1}{2\pi}\int |\widehat{f}(\omega)|^2 \cdot \left|\lambda^m \frac{d^m}{d\lambda^m} |\widehat{\psi}_\lambda(\omega)|^2\right|\ d\omega &\leq \Psi_m \norm{f}_1^2\, .
	\end{align*}
	From the proof of Lemma \ref{lem:WaveletScatteringDerivHighFreq_DiffFunc}:
	\begin{align*}
	\frac{1}{2\pi}\int |\widehat{f}(\omega)|^2 \cdot \left|\lambda^m \frac{d^m}{d\lambda^m} |\widehat{\psi}_\lambda(\omega)|^2\right|\ d\omega &\leq \Theta_m \frac{\norm{f'}_1^2}{\lambda^2}\, .
	\end{align*}
	Utilizing $|B_m|\leq E^m$ gives
	\begin{align*}
	\frac{1}{2\pi}\int |\widehat{f}(\omega)|^2 \cdot \left|B_m\eta^m\lambda^m \frac{d^m}{d\lambda^m} |\widehat{\psi}_\lambda(\omega)|^2\right|\ d\omega \leq (E\eta)^m\left(\norm{f}_1^2\Psi_m \wedge \frac{\norm{f'}_1^2\Theta_m}{\lambda^2}\right)\, .
	\end{align*}
\end{proof}

The following Corollary is obtained from Lemma \ref{lem:WSC_deriv_cons} when $f$ is a dirac-delta function.
\begin{corollary}
	\label{cor:WSC_deriv_cons_fdirac}
	Assume $\psi$ is $m$-admissable, and let $B_m, E, \Psi_m$ be as defined in (\ref{equ:MomentBasedConstants}), (\ref{equ:E}), (\ref{equ:Psik}). Then:
	\begin{align*}
	\frac{1}{2\pi}\int\left|B_m\eta^m\lambda^m \frac{d^m}{d\lambda^m} |\widehat{\psi}_\lambda(\omega)|^2\right|\ d\omega &\leq (E\eta)^m \Psi_m\, .
	\end{align*}
\end{corollary}

\section{PS and wavelet invariant equivalence}
\label{app:WSC_PS_equivalence}

This appendix contains supporting results for demonstrating the equivalence of the power spectrum and wavelet invariants. Lemma \ref{lem:WSC_PS_equivalence} establishes that wavelet invariants uniquely determine any \edit{bandlimited} $\Lb^2$ function, as long as the wavelet satisfies the linear independence Condition \ref{cond: linear indep wavelet} \edit{and a mild integrability condition}. Proposition \ref{prop:linear_indep} gives two criteria which are sufficient to guarantee Condition \ref{cond: linear indep wavelet}. Finally, Lemma \ref{lem:Morlet_lin_indep} establishes that the Morlet wavelet satisfies Condition \ref{cond: linear indep wavelet}.

\lemWSCPSequivalence*

\begin{proof}
	
	\edit{Since $p$ is continuous, there exists an $\epsilon>0$ such that on $(0,\epsilon)$ one either has $p=0$, $p>0$, or $p<0$. Claim: one must have $p=0$. Suppose not, and without loss of generality assume $p>0$ on $(0,\epsilon)$ and that the support of $|\widehat{\psi}^{+}(\omega)|^2$ is contained in the interval $[1,2]$. Now choose $\lambda_0$ small enough so that $|\widehat{\psi}^{+}_{\lambda_0}(\omega)|^2$ is supported on $[\epsilon/2,\epsilon]$, i.e. $\lambda_0=\epsilon/2$. Clearly there must exist a subset $\mathcal{M}\subseteq [\epsilon/2,\epsilon]$ of positive measure such that $|\widehat{\psi}^{+}_{\lambda_0}(\omega)|^2>0$ on $\mathcal{M}$. Then:
	\begin{align*}
	0 &= \int_0^\infty p(\omega)|\widehat{\psi}^{+}_{\lambda_0}(\omega)|^2\ d\omega = \int_{\epsilon/2}^{\epsilon} p(\omega)|\widehat{\psi}^{+}_{\lambda_0}(\omega)|^2\ d\omega \geq \int_{\mathcal{M}} p(\omega)|\widehat{\psi}^{+}_{\lambda_0}(\omega)|^2\ d\omega \geq 0
	\end{align*}
	We conclude 
	\begin{align*}
	\int_{\mathcal{M}} p(\omega)|\widehat{\psi}^{+}_{\lambda_0}(\omega)|^2\ d\omega &= 0 \, ,
	\end{align*}
	but this is impossible since the integrand is strictly positive on $\mathcal{M}$. We thus conclude that $p=0$ on $(0, \epsilon)$. Thus it is sufficient to only consider frequencies $[\epsilon, \infty)$.}
	
	Assume $\int  p(\omega) |\widehat{\psi}_\lambda(\omega)|^2\ d\omega = 0$ for all $\lambda$. Since $p(\omega)=p(-\omega)$,
	\begin{align*}
	\int  p(\omega) |\widehat{\psi}_\lambda(\omega)|^2\ d\omega &= \int_{0}^{\infty} p(\omega) |\widehat{\psi}^{+}_\lambda(\omega)|^2\ d\omega = \edit{\int_{\epsilon}^{\infty} p(\omega) |\widehat{\psi}^{+}_\lambda(\omega)|^2\ d\omega =\langle p, |\widehat{\psi}^{+}_\lambda|^2 \rangle_{\edit{I}}= } 0 \quad \forall \, \lambda\, ,
	\end{align*}
	\edit{where $I=[\epsilon,\infty)$.}
	\edit{We now define $|\widehat{\phi}^{+}_\lambda(\omega)|^2 := \lambda^{-\beta} |\widehat{\psi}^{+}_\lambda(\omega)|^2$ for some $\beta>0$, and observe that
		\begin{align*}
		\int_{0}^{\infty} p(\omega) |\widehat{\phi}^{+}_\lambda(\omega)|^2\ d\omega &= 0 \quad \forall \, \lambda \quad
		\implies \quad \int_0^\infty |\langle p, |\edit{\widehat{\phi}}^{+}_\lambda|^2 \rangle_{\mathbb{R}^+}|^2 \ d\lambda = \int_0^\infty |\langle p, |\edit{\widehat{\phi}}^{+}_\lambda|^2 \rangle_{I}|^2 \ d\lambda=0 \, .
		\end{align*} }
	Note:
	\begin{align*}
	\int_0^\infty |\langle p, |\edit{\widehat{\phi}}^{+}_\lambda|^2 & \rangle_{\edit{I}}|^2 \ d\lambda = \int_0^\infty \langle p, |\edit{\widehat{\phi}}^{+}_\lambda|^2\rangle_{\edit{I}} \langle \overline{p}, |\edit{\widehat{\phi}}^{+}_\lambda|^2\rangle_{\edit{I}}\ d\lambda \\
	&= \int_0^\infty \left(\int_{\edit{I}} p(\omega_1)|\edit{\widehat{\phi}}^{+}_\lambda(\omega_1)|^2\ d\omega_1\right)\left(\int_{\edit{I}} \overline{p(\omega_2)}|\edit{\widehat{\phi}}^{+}_\lambda(\omega_2)|^2\ d\omega_2\right) d\lambda \\
	&=\int_{\edit{I}}  \overline{p(\omega_2)}\left(\int_{\edit{I}} p(\omega_1)\left(\int_0^\infty |\edit{\widehat{\phi}}^{+}_\lambda(\omega_1)|^2 |\edit{\widehat{\phi}}^{+}_\lambda(\omega_2)|^2 d\lambda\right) d\omega_1\right) d\omega_2\, .
	\end{align*}	
	\edit{We now apply the change of variable $\omega_i=1/\xi_i$, and let $g(\xi_i) =  p(1/\xi_i)$. We obtain:
	\begin{align}
	\label{equ:new_IP}
	0 &= \int_0^{1/\epsilon} \overline{g(\xi_2)}\left(\int_0^{1/\epsilon} g(\xi_1)\left(\int_0^{\infty}\frac{1}{\xi_1^2\xi_2^2}\left|\widehat{\phi}_\lambda^{+}\left(\frac{1}{\xi_1}\right)\right|^2\left|\widehat{\phi}_\lambda^{+}\left(\frac{1}{\xi_2}\right)\right|^2\ d\lambda\right)\ d\xi_1 \right)\ d\xi_2
	\end{align}}
	Now consider the kernel
	\begin{align*}
	k(\xi_1,\xi_2) &=\int_0^{\infty}\frac{1}{\xi_1^2\xi_2^2}\left|\widehat{\phi}_\lambda^{+}\left(\frac{1}{\xi_1}\right)\right|^2\left|\widehat{\phi}_\lambda^{+}\left(\frac{1}{\xi_2}\right)\right|^2\ d\lambda \, .
	\end{align*}
	\edit{Note that $k$ is a strictly positive definite kernel function} if for any finite sequence \edit{$\{\xi_i\}_{i=1}^n$ in $[0,1/\epsilon]$}, the $n$ by $n$ matrix $A$ defined by
	\begin{align*}
	A_{ij} &= \edit{k(\xi_i,\xi_j)}
	\end{align*}
	is \edit{strictly} positive definite \cite{winkler2002uncertainty}. Viewing \edit{$\tilde{\xi}_i(\lambda) = \xi_i^{-2}|\widehat{\phi}^{+}_\lambda(1/\xi_i)|^2$} as functions of $\lambda$, we see that
	\begin{align*}
	A_{ij} = \langle \edit{\tilde{\xi}}_i(\lambda), \edit{\tilde{\xi}}_j(\lambda)\rangle_{\mathbb{R}^{+}}
	\end{align*}
	and $A$ is thus a Gram matrix. Since the \edit{$\tilde{\xi}_i(\lambda)$ are linearly independent if and only if the $|\widehat{\psi}^{+}_\lambda(\omega_i)|^2$ are linearly independent, and the $|\widehat{\psi}^{+}_\lambda(\omega_i)|^2$} are linearly independent by assumption, we can conclude that $A$ and thus \edit{$k$ are strictly} positive definite. 
	\edit{Now consider} the corresponding integral operator on $\edit{[0,1/\epsilon]}$:
	\begin{align*}
	Kg(\xi_2) &= \int_0^{1/\epsilon} g(\xi_1)k(\xi_1,\xi_2)\ d\xi_1\,.
	\end{align*}
	\edit{Since $\psi\in\Lb^1(\R)$, $|\widehat{\psi}_\lambda^{+}|^2$ and thus $|\widehat{\phi}_\lambda^{+}|^2$ are continuous, and $k$ will thus be continuous as long as it remains bounded. To check boundedness we observe that $k(\xi_1,\xi_2)^2\leq k(\xi_1,\xi_1)k(\xi_2,\xi_2)$ \cite{buescu2004positive}, and 
		\begin{align*}
		k(\xi,\xi) &= \int_0^{\infty} \frac{1}{\xi^4} \left|\widehat{\phi}_{\lambda}^{+}\left(\frac{1}{\xi}\right)\right|^4\ d\lambda \\
			&= \int_0^{\infty} \frac{1}{\xi^4} \frac{1}{\lambda^{2+2\beta}}\left|\widehat{\psi}^{+}\left(\frac{1}{\lambda\xi}\right)\right|^4\ d\lambda  \\
			&= \int_0^{\infty} \frac{1}{\xi^4} (\omega\xi)^{2+2\beta} |\widehat{\psi}^{+}(\omega)|^4 \frac{d\omega}{\xi\omega^2} \\
			&= \xi^{2\beta-3} \int_0^{\infty} \omega^{2\beta}|\widehat{\psi}^{+}(\omega)|^4\ d\omega  \\
			&\leq 3\xi^{2\beta-3} \int_0^{\infty} \omega^{2\beta}|\widehat{\psi}(\omega)|^4\ d\omega \\
			&\leq 3\xi^{2\beta-3} \norm{\omega^{\beta}P\psi}_2^2\, .
		\end{align*}
		Since $\widehat{\psi}$ has a compact support, clearly $\norm{\omega^{\beta}P\psi}_2^2<\infty$, and $k$ is thus bounded on the compact interval $[0,1/\epsilon]$ as long as $\beta\geq 3/2$.}
	\edit{Since $k$ is continuous and $[0,1/\epsilon]$ is compact, $K:\Lb^2\edit{[0,1/\epsilon]} \rightarrow \Lb^2\edit{[0,1/\epsilon]}$ is a compact, self-adjoint operator and by Mercer's Theorem $K$ is also strictly positive definite \cite{winkler2002uncertainty}. Since $\langle Kg, g \rangle_{[0,1/\epsilon]}=0$ by (\ref{equ:new_IP}), we conclude $g=0$ in $\Lb^2[0,1/\epsilon]$. Thus $p(1/\xi)=0$ for almost every $\xi\in(0,1/\epsilon]$, which implies $p(\omega)=0$ for almost every $\omega\in[\epsilon,\infty)$.}
	Since $p(\omega)=p(-\omega)$ \edit{and $p=0$ on $(0,\epsilon)$}, $p=0$ for almost every $\omega\in\mathbb{R}$. \\

\end{proof}	

\proplinearindep*

\begin{proof}
Let $\{\omega_i\}_{i=1}^n$ be a finite sequence of distinct positive frequencies, and let $\tilde{\omega}_i(\lambda) = \frac{1}{|\lambda|}|\widehat{\psi}^{+}\left(\frac{\omega_i}{\lambda}\right)|^2$ denote the corresponding functions of $\lambda$. 

First assume (i).  Without loss of generality we assume that $[a,b]$ is a positive interval and that $|\widehat{\psi}(\omega)|^2 >0$ on  $(a,a+\epsilon)$ for some $\epsilon>0$. Clearly $|\widehat{\psi}^{+}(\omega)|^2 = |\widehat{\psi}(\omega)|^2$. A simple calculation shows that the support of $\tilde{\omega}_i(\lambda)$ is contained in the interval $\left[ \frac{\omega_i}{b}, \frac{\omega_i}{a}\right]$, and $\tilde{\omega}_i(\lambda)>0$ in a neighborhood of $\frac{\omega_i}{a}$. Assume we have ordered the $\omega_i$ so that $\omega_1 > \ldots >\omega_n >0$.
Now suppose
\begin{align*}
c_1\tilde{\omega}_1(\lambda) + \cdots + c_n\tilde{\omega}_n(\lambda) &= 0\, .
\end{align*}
Note $\tilde{\omega}_1(\lambda)$ is the only function in the above collection with support in a neighborhood of $\frac{\omega_1}{a}$; thus we must have $c_1=0$, so that 
\begin{align*}
c_2\tilde{\omega}_2(\lambda) + \cdots + c_n\tilde{\omega}_n(\lambda) &= 0\, .
\end{align*}
But now  $\tilde{\omega}_2(\lambda)$ is the only function in the above collection with support in a neighborhood of $\frac{\omega_2}{a}$, so we must have $c_2=0$, and proceeding iteratively we conclude that $c_1 = \ldots = c_n=0$. Thus $\{\tilde{\omega}_i(\lambda)\}_{i=1}^n$ is a linearly independent set, and Condition \ref{cond: linear indep wavelet} holds.

Now assume (ii). Since $\frac{d^{n}}{d\omega^n}\left( |\widehat{\psi}^{+}(\omega)|^2\right) \big\vert_{\omega=0} = 2\frac{d^{n}}{d\omega^n}\left( |\widehat{\psi}(\omega)|^2\right) \big\vert_{\omega=0}$, $|\widehat{\psi}^{+}(\omega)|^2$ is $\Cb^{\infty} (\R)$ and all derivatives of order at least $N$ are nonzero at $\omega=0$. Note $\{\tilde{\omega}_i(\lambda)\}_{i=1}^n=\{|\lambda|^{-1}|\widehat{\psi}^{+}(\omega_i/\lambda)|^2\}_{i=1}^n$ are linearly independent if and only if $\{|\widehat{\psi}^{+}(\omega_i/\lambda)|^2\}_{i=1}^n$ are linearly independent. Defining $\tilde{\lambda} = 1/\lambda$, this holds if and only if $\{|\widehat{\psi}^{+}(\omega_i\tilde{\lambda})|^2\}_{i=1}^n = \{g(\omega_i\tilde{\lambda})\}_{i=1}^n$ are linearly independent as functions of $\tilde{\lambda}$, where we define $g(\omega) = |\widehat{\psi}^{+}(\omega)|^2$. Assume
\begin{align*}
c_1g(\omega_1\tilde{\lambda}) + c_2g(\omega_2\tilde{\lambda})+\cdots + c_ng(\omega_n\tilde{\lambda}) &=0\, .
\end{align*}
Differentiating $m$ times for $N \leq m\leq N+n-1$, we obtain:
\begin{small}
\begin{align*}
c_1\omega_1^Ng^{(N)}(\omega_1\tilde{\lambda}) 
+\cdots + c_n\omega_n^Ng^{(N)}(\omega_n\tilde{\lambda}) &=0 \\
&\vdots \\
c_1\omega_1^{N+n-1}g^{(N+n-1)}(\omega_1\tilde{\lambda}) 
+\cdots + c_n\omega_n^{N+n-1}g^{(N+n-1)}(\omega_n\tilde{\lambda}) &=0	
\end{align*}
\end{small}
The above holds for all $\tilde{\lambda}$. We now take the limit as $\tilde{\lambda} \rightarrow 0$ to obtain:
\begin{align*}
g^{(N)}(0)(\omega_1^Nc_1+\omega_2^Nc_2+\ldots \omega_n^Nc_n) &= 0  \\
g^{(N+1)}(0)(\omega_1^{N+1}c_1+\omega_2^{N+1}c_2+\ldots \omega_n^{N+1}c_n) &= 0 \\
&\vdots \\
g^{(N+n-1)}(0)(\omega_1^{N+n-1}c_1+\omega_2^{N+n-1}c_2+\ldots \omega_n^{N+n-1}c_n) &= 0
\end{align*}
Since $g^{(m)}(0) \ne 0$, we obtain:
\begin{align*}
\begin{bmatrix}
\omega_1^N & \ldots & \omega_n^N \\
\omega_1^{N+1} & \ldots & \omega_n^{N+1} \\
\vdots & & \vdots \\
\omega_1^{N+n-1} & \ldots & \omega_n^{N+n-1}
\end{bmatrix}
\begin{bmatrix}
c_1 \\ c_2 \\ \vdots \\ c_n
\end{bmatrix}
=\begin{bmatrix}
0 \\ 0 \\ \vdots \\ 0
\end{bmatrix}
\end{align*}

\begin{align*}
\underbrace{\begin{bmatrix}
	1 & \ldots & 1 \\
	\omega_1 & \ldots & \omega_n \\
	\vdots & & \vdots \\
	\omega_1^{(n-1)} & \ldots & \omega_n^{(n-1)}
	\end{bmatrix}}_{:=A}
\underbrace{\begin{bmatrix}
	\omega_1^N & 0 & \ldots & 0 \\
	0 & \omega_2^N & \ldots & 0 \\
	\vdots & & \vdots \\
	0 & 0 & \ldots & \omega_n^{N}
	\end{bmatrix}}_{:=B}
\begin{bmatrix}
c_1 \\ c_2 \\ \vdots \\ c_n
\end{bmatrix}
=\begin{bmatrix}
0 \\ 0 \\ \vdots \\ 0
\end{bmatrix}
\end{align*}
Since $A$ is a Vandermonde matrix constructed from distinct $\omega_i$, $\det(A) \ne 0$. Since the $\omega_i$ are nonzero, $\det(B)\ne 0$. Thus $\det(AB)=\det(A)\det(B)\ne 0$. We conclude $AB$ is invertible and so all $c_i=0$, which gives Condition \ref{cond: linear indep wavelet}.

\end{proof}

\begin{lemma}
	\label{lem:Morlet_lin_indep}
	Suppose we construct a Morlet wavelet with parameter $\xi$, that is $\psi(x) = C_\xi \pi^{-1/4} e^{-x^2/2}(e^{i\xi x}-e^{-\xi^2/2})$ for $C_\xi =(1-e^{-\xi^2}-2e^{-3\xi^2/4})^{-1/2}$. Then for almost all $\xi \in \mathbb{R}^{+}$, the wavelet satisfies Condition \ref{cond: linear indep wavelet}.	
\end{lemma}	

\begin{proof}
	The Fourier transform $\widehat{\psi}$ has form
	\begin{align*}
	\widehat{\psi}(\omega) &= \widetilde{C}_\xi e^{-\omega^2/2}(e^{\xi\omega}-1)
	\end{align*}
	for some constant $\widetilde{C}_\xi $ depending on $\xi$, so that
	\begin{align*}
	g(\omega) &:= \widetilde{C}_\xi^{-2}|\widehat{\psi}(\omega)|^2 =e^{-\omega^2}(e^{\xi\omega}-1)^2\, .
	\end{align*}
	From direct calculation or a computer algebra system (CAS), one obtains:
	\begin{align*}
	g^{(n)}(0) &= \begin{cases} H_n(\xi) - 2H_n(\xi/2) & n \text{ odd} \\
	H_n(\xi) - 2H_n(\xi/2) +\frac{(-1)^{\frac{n}{2}} n!}{\left(\frac{n}{2}\right)!} & n \text{ even} \end{cases} 
	\end{align*}
	where $H_n(\xi)$ is the $n^{\text{th}}$ degree physicist's Hermite polynomial. 
	We have $g'(0)=0$, but for $n>1$, $g^{(n)}(0)=0$ only when $\xi$ is a root of the above polynomial. Since the set of roots of the polynomials $\{g^{(n)}(0)\}_{n=1}^\infty$ is countable, if $\xi$ is selected at random from $\mathbb{R}$, it is not a root of any of these polynomials with probability 1, and $g^{(n)}(0) \ne 0$ for all $n$. Thus the wavelet satisfies criterion (ii) of Proposition \ref{prop:linear_indep}, and thus the linear independence Condition \ref{cond: linear indep wavelet}.
\end{proof}

\section{Supporting results: classic MRA}
\label{app:AddNoise}

This appendix contains supporting results for Section \ref{sec:AddNoise}. The first two lemmas (Lemmas \ref{lem:PS_addnoise} and Lemma \ref{lem:SignalAddNoisePS}) establish additive noise bounds for the power spectrum and are needed to prove Proposition \ref{prop:AddNoisePS}. The next two lemmas (Lemmas \ref{lem:WhiteNoise} and Lemma \ref{lem:SignalAddNoiseWSC}) establish additive noise bounds for wavelet invariants and are needed to prove Propostion \ref{prop:AddNoiseWSC}.

\begin{lemma}
	\label{lem:PS_addnoise}
	Let $\varepsilon(x)$ be a white noise processes on $[-\frac{1}{2},\frac{1}{2}]$ with variance $\sigma^2$. Then for all frequencies $\omega, \xi$:
	\begin{align}
	\Ex\left[\,|\widehat{\varepsilon}(\omega)|^2\right] &= \sigma^2\label{equ:PS_expec_addnoise} \\
	\Ex\left[\, |\widehat{\varepsilon}(\omega)|^4\right] &\leq 3\sigma^4  \label{equ:sqPS_expec_addnoise} \\
	\Ex\left[\, |\widehat{\varepsilon}(\omega)|^2|\widehat{\varepsilon}(\xi)|^2 \right] &\leq 3\sigma^4  \label{equ:productPS_expec_addnoise}\, . 
	\end{align} 
\end{lemma}

\begin{proof}
	By Proposition \ref{prop:genItoIso},	
	\begin{align*}
	\Ex\left[\, |\widehat{\varepsilon}(\omega)|^2\right] &= \mathbb{E}\left[\, \widehat{\varepsilon}(\omega)\overline{\widehat{\varepsilon}(\omega)} \right] \\
	&= \mathbb{E} \left[ \left(\int_{-1/2}^{1/2}e^{-i \omega x} \ dB_x\right)\left(\int_{-1/2}^{1/2}e^{i \omega x} \ dB_x\right) \right] \\
	&= \sigma^2\int_{-1/2}^{1/2}\ dx \\
	&=\sigma^2,
	\end{align*}
	which shows (\ref{equ:PS_expec_addnoise}). By Proposition \ref{prop:GenFourthMoment},
	\begin{align*} 
	\Ex \left[\, |\widehat{\varepsilon}({\omega})|^4\right] &= \Ex\left[\,\widehat{\varepsilon}({\omega})^2\left(\overline{\widehat{\varepsilon}(\omega)}\right)^2\right] \\
	&= \Ex \left[\left(\int_{-1/2}^{1/2} e^{-i\omega x}\ dB_x\right)^2\left(\int_{-1/2}^{1/2} e^{i\omega x}\ dB_x\right)^2\right] \\
	&=2\sigma^4\left(\int_{-1/2}^{1/2}\ dx\right)^2+\sigma^4\left(\int_{-1/2}^{1/2}e^{-2i\omega x}\ dx\right)\left(\int_{-1/2}^{1/2}e^{2i\omega x}\ dx\right) \\
	&\leq  2\sigma^4 +\sigma^4\left(\int_{-1/2}^{1/2}|e^{-2i\omega x}|\ dx\right)\left(\int_{-1/2}^{1/2}|e^{2i\omega x}|\ dx\right) \\
	&= 3\sigma^4,
	\end{align*}
	which shows (\ref{equ:sqPS_expec_addnoise}). Finally, by Proposition \ref{prop:MostGenFourthMoment}, we have
	\begin{align*}
	&\Ex\left[\,|\widehat{\varepsilon}(\omega)|^2|\widehat{\varepsilon}(\xi)|^2\right] \\
	&= \Ex \left[ \left( \int_{-1/2}^{1/2} e^{-i\omega x}\ dB_x\right) \left( \int_{-1/2}^{1/2} e^{i\omega x}\ dB_x\right)\left( \int_{-1/2}^{1/2} e^{-i\xi x}\ dB_x\right) \left( \int_{-1/2}^{1/2} e^{i\xi x}\ dB_x\right)\right] \\
	&= \sigma^4 \left[ \left(\int_{-1/2}^{1/2} e^{-i(\omega+\xi)x}\ dx \right)\left(\int_{-1/2}^{1/2} e^{i(\omega+\xi)x}\ dx \right)\right]  \\
	&\qquad + \sigma^4\left[\left(\int_{-1/2}^{1/2} e^{i(\xi-\omega)x}\ dx \right)\left(\int_{-1/2}^{1/2} e^{i(\omega-\xi)x}\ dx \right) + \left(\int_{-1/2}^{1/2}\ dx \right)\left(\int_{-1/2}^{1/2}\ dx \right)  \right] \\
	&\leq \sigma^4\left[3\left(\int_{-1/2}^{1/2}\ dx \right)\left(\int_{-1/2}^{1/2}\ dx \right)  \right] \\
	&= 3\sigma^4,
	\end{align*}
	which gives (\ref{equ:productPS_expec_addnoise}).
\end{proof}

\begin{lemma} \label{lem:SignalAddNoisePS}
	Let $\varepsilon(x)$ be a white noise processes on $[-\frac{1}{2},\frac{1}{2}]$ with variance $\sigma^2$. Then for any signal $f \in \Lb^1(\R)$:
	\begin{align*}
	\Ex\left[(P(f+\varepsilon))(\omega)\right] &= (Pf)(\omega) + \sigma^2 \\
	\Var\left[(P(f+\varepsilon))(\omega)\right] &\leq 4\sigma^2(Pf)(\omega)+2\sigma^4 \, .
	\end{align*}	
\end{lemma}	

\begin{proof}
	Since $\mathbb{E}\left[\widehat{\varepsilon}(\omega)\right] = \mathbb{E}\left[\overline{\widehat{\varepsilon}(\omega)}\right] =0$ and $\mathbb{E}\left[|\widehat{\varepsilon}(\omega)|^2\right] =\sigma^2$ by Lemma \ref{lem:PS_addnoise},
	\begin{align*}
	\Ex\left[(P(f+\varepsilon))(\omega)\right] 
	&= \Ex\left[\left( \widehat{f}(\omega)+\widehat{\varepsilon}(\omega) \right)\left( \overline{\widehat{f}(\omega)}+\overline{\widehat{\varepsilon}(\omega)} \right)\right] \\
	&= \Ex\left[ |\widehat{f}(\omega)|^2+\widehat{f}(\omega)\overline{\widehat{\varepsilon}(\omega)}+\widehat{\varepsilon}(\omega) \overline{\widehat{f}(\omega)} + |\widehat{\varepsilon}(\omega)|^2\right] \\
	&= (Pf)(\omega) + \sigma^2. 
	\end{align*}
	We now control $\Var[(P(f+\varepsilon))(\omega)]$. 	
	Note that:
	\begin{align*}
	\left[(P(f+\varepsilon))(\omega)\right]^2 &=\left(|\widehat{f}(\omega)|^2+\widehat{f}(\omega)\overline{\widehat{\varepsilon}(\omega)}+\widehat{\varepsilon}(\omega) \overline{\widehat{f}(\omega)} + |\widehat{\varepsilon}(\omega)|^2\right)^2
	\end{align*}
	and that
	\begin{align*}
	 \Ex\left[|\widehat{\varepsilon}(\omega)|^2\, \widehat{\varepsilon}(\omega) \right] &=  \mathbb{E} \left[ \left(\int_{-1/2}^{1/2}e^{-i \omega x} \ d B_x\right)\left(\int_{-1/2}^{1/2}e^{i \omega s} \ dB_s\right)\left(\int_{-1/2}^{1/2}e^{-i \omega p} \ dB_p\right) \right] \\
	 &=0,
	 \end{align*}
	since even when $x=s=p$, $\Ex [(\Delta B_x)^3] =0$. Ignoring the terms with zero expectation, we thus get:
	\begin{align*}
	\Ex[(P(f+\varepsilon))(\omega)^2]
	&= \Ex \left( |\widehat{f}(\omega)|^4 + 4|\widehat{f}(\omega)|^2|\widehat{\varepsilon}(\omega)|^2+|\widehat{\varepsilon}(\omega)|^4 + \widehat{f}(\omega)^2\overline{\widehat{\varepsilon}({\omega})}^2+\widehat{\varepsilon}({\omega})^2 \overline{\widehat{f}({\omega})}^2\right) \\
	&\leq  \Ex \left( |\widehat{f}(\omega)|^4 + 6|\widehat{f}(\omega)|^2|\widehat{\varepsilon}(\omega)|^2+|\widehat{\varepsilon}(\omega)|^4 \right) \\
	&= [(Pf)(\omega)]^2 + 6\sigma^2(Pf)(\omega)+3\sigma^4
	\end{align*}
	where the last line follows from Lemma \ref{lem:PS_addnoise}. Thus
	\begin{align*}
	\Var[(P(f+\varepsilon))(\omega)] &= \Ex[(P(f+\varepsilon))(\omega)^2] - (\Ex[(P(f+\varepsilon))(\omega)])^2  \\
	&\leq   [(Pf)(\omega)]^2 + 6\sigma^2(Pf)(\omega)+3\sigma^4 - ((Pf)(\omega) + \sigma^2)^2 \\
	&=  4\sigma^2(Pf)(\omega)+2\sigma^4. 
	\end{align*}
\end{proof}

\propAddNoisePS*

\begin{proof}
	Let $f^{\transvar_j}(x)=f(x-\transvar_j)$ so that $y_j = f^{\transvar_j}+\varepsilon_j$. We first note since $\widehat{f^{\transvar_j}}(\omega) = e^{-i\omega\transvar_j}\widehat{f}(\omega)$, the power spectrum is translation invariant, that is $(Pf^{\transvar_j})(\omega) = (Pf)(\omega)$ for all $\omega, \transvar_j$.	
	Thus by Lemma \ref{lem:SignalAddNoisePS},
	\begin{align*}
	\Ex[(Py_j)(\omega)] &= \Ex[(P(f^{\transvar_j}+\varepsilon_j))(\omega)] = (Pf^{\transvar_j})(\omega) + \sigma^2 = (Pf)(\omega) + \sigma^2
	\end{align*}
	and
	\begin{align*}
	\Var[(Py_j)(\omega)] &= \Var[(P(f^{\transvar_j}+\varepsilon_j))(\omega)] 
	\leq 4\sigma^2(Pf^{\transvar_j})(\omega)+2\sigma^4 
	= 4\sigma^2(Pf)(\omega)+2\sigma^4.
	\end{align*}
	Since the $y_j$ are independent,
	\begin{align*}
	\Var\left(\frac{1}{M}\sum_{j=1}^M (Py_j)(\omega)\right) &\leq \frac{1}{M}\left(4\sigma^2(Pf)(\omega)+2\sigma^4\right).
	\end{align*}
	Applying Chebyshev's inequality to the random variable $X=\frac{1}{M}\sum_{j=1}^M (Py_j)(\omega)$, we obtain: 
	\[ \Prob\left(\ \left|\frac{1}{M}\sum_{j=1}^M (Py_j)(\omega) - \left((Pf)(\omega)+\sigma^2\right)\right| \geq  \frac{t(2\sigma\sqrt{(Pf)(\omega)}+\sqrt{2}\sigma^2)}{\sqrt{M}} \right) \leq \frac{1}{t^2}. \]
	Observing that $\sqrt{(Pf)(\omega)}=|\widehat{f}(\omega)|\leq \norm{f}_1$ gives (\ref{equ:AddNoisePS}).
	
\end{proof}

\begin{lemma}\label{lem:WhiteNoise}
	Let $\varepsilon(x)$ be a white noise processes on $[-\frac{1}{2},\frac{1}{2}]$ with variance $\sigma^2$. Then:
	\begin{align*}
	\Ex[({\Sc}\varepsilon)(\lambda)] &= \sigma^2 \\
	\Ex \,[({\Sc}\varepsilon)(\lambda)^2] &\leq 3\sigma^4\, .
	\end{align*}	
\end{lemma}

\begin{proof}
	Since $\Ex [\,|\widehat{\varepsilon}(\omega)|^2] = \sigma^2$ by Lemma \ref{lem:PS_addnoise}, we have:
	\begin{align*}
	\Ex[({\Sc}\varepsilon)(\lambda)] &= \Ex\left[ \norm{\varepsilon*\psi_\lambda}_2^2 \right]\\
	&= \Ex\left[ \frac{1}{2\pi}\  \norm{\widehat{\varepsilon}\cdot \widehat{\psi}_\lambda}_2^2\right] \\
	&=\Ex\left[ \frac{1}{2\pi} \int  |\widehat{\varepsilon}(\omega)|^2|\widehat{\psi}_\lambda(\omega)|^2\ d\omega\right] \\
	&=\frac{\sigma^2}{2\pi} \int  |\widehat{\psi}_\lambda(\omega)|^2\ d\omega \\
	&=\sigma^2\norm{\psi_\lambda}_2^2 \\
	&=\sigma^2.
	\end{align*}
	Since by Lemma \ref{lem:PS_addnoise}, $\Ex\left[\,|\widehat{\varepsilon}(\omega)|^2|\widehat{\varepsilon}(\xi)|^2\right] \leq 3\sigma^4$, we also have:
	\begin{align*}
	\Ex \,[({\Sc}\varepsilon)(\lambda)^2] &=  \Ex\left[ \norm{\varepsilon*\psi_\lambda}_2^4\right] \\
	&= \Ex\left[ \frac{1}{(2\pi)^2}\ \norm{\widehat{\varepsilon}\cdot \widehat{\psi}_\lambda}_2^2\ \norm{\widehat{\varepsilon}\cdot \widehat{\psi}_\lambda}_2^2\right] \\
	&=  \Ex\left[ \frac{1}{(2\pi)^2}\int \int |\widehat{\varepsilon}(\omega)|^2|\widehat{\varepsilon}(\xi)|^2|\widehat{\psi}_\lambda(\omega)|^2|\widehat{\psi}_\lambda(\xi)|^2\ d\omega\ d\xi\right] \\
	&\leq \frac{3\sigma^4}{(2\pi)^2} \int \int |\widehat{\psi}_\lambda(\omega)|^2|\widehat{\psi}_\lambda(\xi)|^2\ d\omega\ d\xi \\
	&=3\sigma^4 \left(\norm{\psi_\lambda}_2^2\right)^2 \\
	&=3\sigma^4.
	\end{align*}
\end{proof}

\begin{lemma}\label{lem:SignalAddNoiseWSC}
	Let $\varepsilon(x)$ be a white noise processes on $[-\frac{1}{2},\frac{1}{2}]$ with variance $\sigma^2$. Then for any signal $f \in \Lb^1 (\R)$:
	\begin{align*}
	\Ex[({\Sc}(f+\varepsilon))(\lambda)] &= ({\Sc}f)(\lambda) + \sigma^2 \\
	\Var[({\Sc}(f+\varepsilon))(\lambda)] &\leq 4\sigma^2({\Sc}f)(\lambda)+2\sigma^4 \, .
	\end{align*}	
\end{lemma}

\begin{proof}
	Utilizing $\Ex[\varepsilon] = \Ex\left[\overline{\varepsilon}\right]=0$ and Lemma \ref{lem:WhiteNoise}, we have:
	\begin{align*}
	\Ex[({\Sc}(f+\varepsilon))(\lambda)]
	&= \Ex \left[\int  |(f+\varepsilon)*\psi_\lambda(u)|^2 \ du \right]\\
	&= \int  |f*\psi_\lambda(u)|^2+\Ex\left[\int | \varepsilon*\psi_\lambda(u)|^2 \ du\right] \\
	&=({\Sc}f)(\lambda) + \Ex [({\Sc}\varepsilon)(\lambda)] \\
	&=({\Sc}f)(\lambda) + \sigma^2.
	\end{align*}	
	To bound $\Ex[({\Sc}(f+\varepsilon))(\lambda)^2]$, note that:
	\begin{small}
		\begin{align*}
		&[({\Sc}(f+\varepsilon))(\lambda)]^2 \\
		&=\left(\int  |f*\psi_\lambda(u_1)|^2+(\varepsilon*\psi_\lambda(u_1))(\overline{f}*\overline{\psi_\lambda(u_1)})+  (f*\psi_\lambda(u_1))(\overline{\varepsilon}*\overline{\psi_\lambda(u_1)})+| \varepsilon*\psi_\lambda(u_1))|^2 \ du_1\right) \\
		&\quad \cdot\left(\int  |f*\psi_\lambda(u_2)|^2+(\varepsilon*\psi_\lambda(u_2))(\overline{f}*\overline{\psi_\lambda(u_2)})+  (f*\psi_\lambda(u_2))(\overline{\varepsilon}*\overline{\psi_\lambda(u_2)})+| \varepsilon*\psi_\lambda(u_2))|^2 \ du_2\right)
		\end{align*}
	\end{small}
	When we take expecation, any term involving one or three $\varepsilon$ terms disappear, so that:
		\begin{align*}
		\Ex[({\Sc}(f+\varepsilon))(\lambda)^2]
		&=\Ex \left[ \int \int  |f*\psi_\lambda(u_1)|^2|f*\psi_\lambda(u_2)|^2\ du_1\ du_2  \right. \\
		&\qquad+\int \int  |f*\psi_\lambda(u_1)|^2| \varepsilon*\psi_\lambda(u_2))|^2\ du_1\ du_2 \\
		&\qquad+\int \int  (\varepsilon*\psi_\lambda(u_1))(\overline{f}*\overline{\psi_\lambda(u_1)})(\varepsilon*\psi_\lambda(u_2))(\overline{f}*\overline{\psi_\lambda(u_2)})  \ du_1\ du_2 \\
		&\qquad+ \int \int  (\varepsilon*\psi_\lambda(u_1))(\overline{f}*\overline{\psi_\lambda(u_1)})(f*\psi_\lambda(u_2))(\overline{\varepsilon}*\overline{\psi_\lambda(u_2)})   \ du_1\ du_2    \\
		&\qquad+\int \int (f*\psi_\lambda(u_1))(\overline{\varepsilon}*\overline{\psi_\lambda(u_1)})(\varepsilon*\psi_\lambda(u_2))(\overline{f}*\overline{\psi_\lambda(u_2)})    \ du_1\ du_2 \\
		&\qquad+\int \int (f*\psi_\lambda(u_1))(\overline{\varepsilon}*\overline{\psi_\lambda(u_1)})(f*\psi_\lambda(u_2))(\overline{\varepsilon}*\overline{\psi_\lambda(u_2)})    \ du_1\ du_2      \\
		&\qquad +\int \int  | \varepsilon*\psi_\lambda(u_1))|^2|f*\psi_\lambda(u_2)|^2   \ du_1\ du_2 \\
		&\qquad\left. + \int \int  | \varepsilon*\psi_\lambda(u_1))|^2| \varepsilon*\psi_\lambda(u_2))|^2   \ du_1\ du_2 \right]   \\
		&\leq \Ex  \left[ \int \int  |f*\psi_\lambda(u_1)|^2|f*\psi_\lambda(u_2)|^2\ du_1\ du_2 \right. \\
		&\qquad+ 6 \int \int  |f*\psi_\lambda(u_1)|^2| \varepsilon*\psi_\lambda(u_2))|^2\ du_1\ du_2 \\ 
		&\qquad +\left.\int \int  | \varepsilon*\psi_\lambda(u_1))|^2| \varepsilon*\psi_\lambda(u_2))|^2   \ du_1\ du_2 \right]\\
		&= \Ex \left[  [({\Sc}f)(\lambda)]^2 + 6({\Sc}f)(\lambda)({\Sc}\varepsilon)(\lambda)+[({\Sc}\varepsilon)(\lambda)]^2\right] \\
		&= [({\Sc}f)(\lambda)^2] + 6\sigma^2({\Sc}f)(\lambda)+3\sigma^4\, ,
		\end{align*}
	where the last line follows from Lemma \ref{lem:WhiteNoise}. Thus
	\begin{align*}
	\Var[({\Sc}(f+\varepsilon))({\lambda})] &= \Ex[({\Sc}(f+\varepsilon))({\lambda})^2] - \left(\Ex[({\Sc}(f+\varepsilon))({\lambda})]\right)^2 \\
	&\leq  [({\Sc}f)(\lambda)]^2 + 6\sigma^2({\Sc}f)(\lambda)+3\sigma^4 - [({\Sc}f)(\lambda)+\sigma^2]^2 \\
	&= 4\sigma^2({\Sc}f)(\lambda)+2\sigma^4.
	\end{align*}
\end{proof}

\propAddNoiseWSC*

\begin{proof}
	Let $f^{\transvar_j}(x)=f(x-\transvar_j)$ so that $y_j = f^{\transvar_j}+\varepsilon_j$. We first note that the wavelet invariants are translation invariant, that is ${\Sc}f^{\transvar_j} = {\Sc}f$ for all $\transvar_j$. We now compute the mean and variance of the coefficients $({\Sc}y_j)(\lambda)$.
	By Lemma \ref{lem:SignalAddNoiseWSC}:
	\begin{align*}
	\Ex[({\Sc}y_j)(\lambda)] &= \Ex[({\Sc}(f^{\transvar_j}+\varepsilon_j))(\lambda)] =({\Sc}f^{\transvar_j})(\lambda) + \sigma^2 =({\Sc}f)(\lambda) + \sigma^2
	\end{align*}
	and
	\begin{align*}
	\Var[({\Sc}y_j)(\lambda)]  = \Var[({\Sc}(f^{\transvar_j}+\varepsilon_j))(\lambda)] 
	\leq 4\sigma^2({\Sc}f^{\transvar_j})(\lambda)+2\sigma^4
	=4\sigma^2({\Sc}f)(\lambda)+2\sigma^4.
	\end{align*}
	Since the $y_j$ are independent,
	\begin{align*}
	\Var\left[\frac{1}{M}\sum_{j=1}^M ({\Sc}y_j)(\lambda)\right] &\leq  \frac{1}{M}\left[4\sigma^2({\Sc}f)(\lambda)+2\sigma^4\right].
	\end{align*}
	Applying Chebyshev's inequality to the random variable $X=\frac{1}{M}\sum_{j=1}^M ({\Sc}y_j)({\lambda})$ gives:
	\[ \Prob\left(\ \left|\frac{1}{M}\sum_{j=1}^M ({\Sc}y_j)({\lambda}) - \left[({\Sc}f)(\lambda)+\sigma^2\right]\right| \geq  \frac{t(2\sigma\sqrt{ ({\Sc}f)(\lambda)}+\sqrt{2}\sigma^2)}{\sqrt{M}} \right) \leq \frac{1}{t^2}. \]
	By Young's convolution inequality, $({\Sc}f)(\lambda) = \norm{f*\psi_\lambda}_2^2 \leq \norm{f}_1^2 \norm{\psi_\lambda}_2^2 = \norm{f}_1^2$, which gives (\ref{equ:AddNoiseWSC}). 
\end{proof}

\section{Supporting results: dilation MRA}
\label{app:WSC}

This appendix contains the technical details of the dilation unbiasing procedure which is central to Propositions \ref{prop:RandomDilations_PS}, \ref{prop:RandomDilations}, and \ref{prop:DilationAndAdditiveNoise}. Lemma \ref{lem:RD_GeneralUnbiasing_MeanVar} bounds the bias and variance of the estimator and Lemma \ref{lem:RD_GeneralUnbiasing_DevOfEstimator} bounds the error of the estimator given $M$ independent samples.

\lemRDGeneralUnbiasingMeanVar*

\begin{proof}	
We Taylor expand $F_\lambda(\tau)$ about $\tau=0$: 
\begin{align*}
F_\lambda(\tau) &= F_\lambda(0) +F'_\lambda(0)\tau +\frac{F''_\lambda(0)}{2}\tau^2+\ldots + \frac{F^{(k+1)}_\lambda(0)}{(k+1)!}\tau^{k+1} \\ &\qquad+\underbrace{\int_0^\tau\frac{F^{(k+2)}_\lambda(t)}{(k+1)!}(\tau-t)^{k+1}\ dt}_{:=R_{0}(\tau,\lambda)}\, .
\end{align*}
We note:
\begin{align*}
\Ex\left[F_\lambda(\tau)\right] &= F_\lambda(0)+\frac{F''_\lambda(0)}{2}\eta^2+\ldots+\frac{F_\lambda^{k}(0)}{k!}C_k\eta^{k} + \Ex\left[ R_{0}(\tau,\lambda)\right]
\end{align*}
which motivates an unbiasing with the first $k/2$ even derivatives, and thus a Taylor expansion of these derivatives:
\begin{align*}
F_\lambda(\tau) &= F_\lambda(0) +F'_\lambda(0)\tau +\ldots + \frac{F^{(k+1)}_\lambda(0)}{(k+1)!}\tau^{k+1}  +\underbrace{\int_0^\tau\frac{F^{(k+2)}_\lambda(t)}{(k+1)!}(\tau-t)^{k+1}\ dt}_{:=R_{0}(\tau,\lambda)} \\
F''_\lambda(\tau) &= F''_\lambda(0) + F^{(3)}_\lambda(0)\tau + \ldots + \frac{F_\lambda^{(k+1)}(0)}{(k-1)!}\tau^{k-1}+\underbrace{\int_0^\tau\frac{F^{(k+2)}_\lambda(t)}{(k-1)!}(\tau-t)^{k-1}\ dt}_{:=R_2(\tau,\lambda)} \\
F^{(4)}_\lambda(\tau) &= F^{(4)}_\lambda(0)+F^{(5)}_\lambda(0)\tau + \ldots + \frac{F_\lambda^{(k+1)}(0)}{(k-3)!}\tau^{k-3}+\underbrace{\int_0^\tau\frac{F^{(k+2)}_\lambda(t)}{(k-3)!}(\tau-t)^{k-3}\ dt}_{:=R_4(\tau,\lambda)} \\
&\vdots \\
F^{(k)}_\lambda(\tau) &= F^{(k)}_\lambda(0)+F^{(k+1)}_\lambda(0)\tau + \underbrace{\int_0^\tau F^{(k+2)}_\lambda(t)(\tau-t)\ dt}_{:=R_k(\tau,\lambda)}\, .
\end{align*}
Multiplication of the $i^{\text{th}}$ even derivative by $B_i\eta^i$ gives:
\begin{small}
\begin{align*}
F_\lambda(\tau) &= F_\lambda(0) +F'_\lambda(0)\tau +\ldots + \frac{F^{(k+1)}_\lambda(0)}{(k+1)!}\tau^{k+1}  +R_{0}(\tau,\lambda)\\
B_2\eta^2 F''_\lambda(\tau) &= B_2\eta^2 F''_\lambda(0) + B_2\eta^2 F^{(3)}_\lambda(0)\tau + \ldots + B_2\eta^2\frac{F_\lambda^{(k+1)}(0)}{(k-1)!}\tau^{k-1}+B_2\eta^2R_2(\tau,\lambda)\\
B_4\eta^4F^{(4)}_\lambda(\tau) &= B_4\eta^4F^{(4)}_\lambda(0)+B_4\eta^4F^{(5)}_\lambda(0)\tau + \ldots + B_4\eta^4\frac{F_\lambda^{(k+1)}(0)}{(k-3)!}\tau^{k-3}+B_4\eta^4R_4(\tau,\lambda) \\
&\vdots \\
B_k\eta^k F^{(k)}_\lambda(\tau) &= B_k\eta^kF^{(k)}_\lambda(0)+B_k\eta^kF^{(k+1)}_\lambda(0)\tau + B_k\eta^kR_k(\tau,\lambda)\, .
\end{align*}
\end{small}
We want an estimator that targets $F_\lambda(0)=L(\lambda)$. We thus consider the following variable as an estimator: 
\begin{align*}
G_\lambda(\tau) &:= F_\lambda(\tau) - B_2\eta^2F''_\lambda(\tau)-B_4\eta^4F^{(4)}_\lambda(\tau)- \ldots -B_k\eta^kF^{(k)}_\lambda(\tau)
\end{align*}
and show that $\Ex\left[G_\lambda(\tau)\right] = F_\lambda(0)+O(\eta^{k+2})$ for constants $B_i$ chosen according to (\ref{equ:MomentBasedConstants}).
We have:
\begin{small}
\begin{align*}
\Ex\left[F_\lambda(\tau)\right] &= F_\lambda(0) +F''_\lambda(0)\frac{C_2}{2}\eta^2 +\ldots + F^{(k)}_\lambda(0)\frac{C_k}{k!}\eta^k  + \Ex\left[R_{0}(\tau,\lambda)\right] \\
\Ex\left[B_2\eta^2 F''_\lambda(\tau)\right] &=  F''_\lambda(0)B_2\eta^2 + F^{(4)}_\lambda(0)\frac{B_2C_2}{2}\eta^4 + \ldots + F_\lambda^{(k)}(0)\frac{B_2C_{k-2}}{(k-2)!}\eta^{k}+\Ex\left[B_2\eta^2R_2(\tau,\lambda)\right]\\
\Ex\left[B_4\eta^4F^{(4)}_\lambda(\tau)\right] &= F^{(4)}_\lambda(0)B_4\eta^4+F^{(6)}_\lambda(0)\frac{B_4C_2}{2}\eta^6 + \ldots + F_\lambda^{(k)}(0)\frac{B_4C_{k-4}}{(k-4)!}\eta^{k}+\Ex\left[B_4\eta^4R_4(\tau,\lambda)\right] \\
&\vdots \\
\Ex\left[B_{k-2}\eta^{k-2} F^{(k-2)}_\lambda(\tau)\right] &= F^{(k-2)}_\lambda(0)B_{k-2}\eta^{k-2} + F^{(k)}_\lambda(0)\frac{B_{k-2}C_2}{2}\eta^k + \Ex\left[ B_{k-2}\eta^{k-2}R_{k-2}(\tau,\lambda)\right] \\
\Ex\left[B_k\eta^k F^{(k)}_\lambda(\tau)\right] &= F^{(k)}_\lambda(0)B_k\eta^k + \Ex\left[B_k\eta^kR_k(\tau,\lambda)\right]
\end{align*}
\end{small}
That is:
\begin{align*}
\Ex\left[G_\lambda(\tau)\right] &= F_\lambda(0) + F''_\lambda(0)\left(\frac{C_2}{2!}-B_2\right)\eta^2 + F^{(4)}_\lambda(0)\left(\frac{C_4}{4!}-\frac{B_2C_2}{2!}-B_4\right)\eta^4 \\
&\qquad +F^{(6)}_\lambda(0)\left(\frac{C_6}{6!} - \frac{B_2C_4}{4!}-\frac{B_4C_2}{2!}-B_6\right)\eta^6 \\
&\qquad \ldots + F^{(k)}_\lambda(0)\left(\frac{C_k}{k!} - \frac{B_2C_{k-2}}{(k-2)!}  - \ldots - \frac{B_{k-2}C_2}{2!}-B_k \right)\eta^k + H_1(\lambda)
\end{align*}
where 
\begin{align*}
H_1(\lambda) &= \Ex \left[ R_0(\lambda,\tau) - B_2\eta^2R_2(\tau,\lambda) - \ldots -B_k\eta^k R_k(\lambda,\tau) \right].
\end{align*}
Since (\ref{equ:MomentBasedConstants}) guarantees that
\begin{align*}
B_2 &= \frac{C_2}{2!} \\
B_4 &= \frac{C_4}{4!} - \left(\frac{C_2}{2!}\right)^2 \\
B_6 &= \frac{C_6}{6!} -\frac{C_2C_4}{2!4!} - \left(\frac{C_4}{4!} - \left(\frac{C_2}{2!}\right)^2\right)\frac{C_2}{2!} \\
&\vdots \\
B_k &= \frac{C_k}{k!} - \frac{B_2C_{k-2}}{(k-2)!}  - \ldots - \frac{B_{k-2}C_2}{2!},
\end{align*}
the coefficients of $\eta^2, \eta^4, \ldots, \eta^k$ vanish, and we obtain:
\begin{align*}
\Ex\left[G_\lambda(\tau)\right] &= F_\lambda(0) + H_1(\lambda)\, .
\end{align*}

First we bound the bias $H_1(\lambda)$. In the remainder of the proof we let $B_0=-1$ to simplify notation, so that: 
\begin{align*}
H_1(\lambda) &= \sum_{i=0,2,\ldots,k} -B_iR_i(\lambda,\tau)\eta^i \, .
\end{align*}
We first obtain a bound for $|B_iR_i(\lambda,\tau)\eta^i|$. Note:
\begin{align*}
(k+1-i)!\, \eta^iR_i(\lambda,\tau) &= \eta^i\int_0^\tau F^{(k+2)}_\lambda(t)(\tau-t)^{k+1-i}\ dt \\
&=  \eta^i\int_0^\tau \lambda^{k+2}L^{(k+2)}((1-t)\lambda)(\tau-t)^{k+1-i}\ dt\, .
\end{align*}
We observe that: 
\begin{align*}
\left|((1-t)\lambda)^{k+2} L^{({k+2})}((1-t)\lambda)\right|& \leq \Lambda_{k+2}((1-t)\lambda) \\
\left|\lambda^{k+2} L^{({k+2})}((1-t)\lambda)\right|& \leq \frac{1}{(1-t)^{k+2}}\frac{\Lambda_{k+2}((1-t)\lambda)}{\Lambda_{k+2}(\lambda)}	\Lambda_{k+2}(\lambda) \\
\left|\lambda^{k+2} L^{({k+2})}((1-t)\lambda)\right|& \leq \frac{\edit{R(\lambda)}\Lambda_{k+2}(\lambda)}{(1-t)^{k+2}}
\end{align*}
so that
\begin{align*}
- \frac{\edit{R(\lambda)}\Lambda_{k+2}(\lambda)}{(1-t)^{k+2}} &\leq \lambda^{k+2}L^{(k+2)}((1-t)\lambda) \leq \frac{\edit{R(\lambda)}\Lambda_{k+2}(\lambda)}{(1-t)^{k+2}} \, .
\end{align*}
Now assume first of all that $\tau$ is positive. We have:
\begin{align*}
\left|(k+1-i)!\, \eta^iR_i(\lambda,\tau)\right| &\leq \eta^i \edit{R(\lambda)}\Lambda_{k+2}(\lambda) \int_0^\tau \frac{(\tau-t)^{k+1-i}}{(1-t)^{k+2}}\ dt \\
&\leq \eta^i \edit{R(\lambda)}\Lambda_{k+2}(\lambda) \int_0^\tau \frac{\tau^{k+1-i}}{(1-t)^{k+2}}\ dt \\
&= \eta^i \tau^{k+1-i} \edit{R(\lambda)}\Lambda_{k+2}(\lambda)\frac{1}{(k+1)}\left(\frac{1}{(1-\tau)^{k+1}}-1\right) \\
&\leq \frac{2^{k+2}\edit{R(\lambda)}}{k+1}\eta^i \tau^{k+2-i} \Lambda_{k+2}(\lambda)
\end{align*}
where the last line follows since $\frac{1}{(1-\tau)^{k+1}} \leq 2\cdot 2^{k+1}\tau$ for $\tau \in [0,\frac{1}{2}]$. A similar argument can be applied when $\tau$ is negative, and we can conclude
\begin{align}
\label{equ:RemainderBound}
\left|B_i\eta^i R_i(\lambda,\tau)\right| &\leq \frac{2^{k+2}\edit{R(\lambda)}}{(k+1)(k+1-i)!}\Lambda_{k+2}(\lambda) |B_i| \eta^i|\tau|^{k+2-i}\,. 
\end{align}
which gives
\begin{align*}
\Ex \left|B_i\eta^i R_i(\lambda,\tau)\right| &\leq \frac{2^{k+2}\edit{R(\lambda)}}{(k+1)(k+1-i)!}\Lambda_{k+2}(\lambda) T^{k+2-i}|B_i| \eta^{k+2} \\
&= \frac{2^{k+2}(k+2-i)\edit{R(\lambda)}}{k+1}\Lambda_{k+2}(\lambda) \frac{T^{k+2-i}}{(k+2-i)!}|B_i| \eta^{k+2} \, .
\end{align*}
We thus obtain
\begin{align*}
\left|\Ex\left[G_\lambda(\tau)\right] - L(\lambda) \right| &=|H_1(\lambda)| 
\leq \frac{\edit{R(\lambda)}\Lambda_{k+2}(\lambda)}{k+1}  (2E  \eta)^{k+2} \sum_{i=0,2,\ldots,k} (k+2-i)  \\
&\lesssim \edit{R(\lambda)}k \Lambda_{k+2}(\lambda) (2E\eta)^{k+2}\, ,
\end{align*}
which establishes the bound on the bias. We now bound the variance. We note:
\begin{align*}
G_\lambda(\tau) &= \underbrace{\sum_{i=0,2,\ldots,k}\ \sum_{j=0,1,\ldots,k+1-i} \frac{-B_i}{j!}F_\lambda^{(i+j)}(0)\eta^i\tau^j}_{:=(\mathbf{I})} + \underbrace{\sum_{i=0,2,\ldots,k}-B_iR_i(\lambda,\tau)\eta^i}_{:=(\mathbf{II})}\, .
\end{align*}
Thus:
\begin{align*}
\Var\left[G_\lambda(\tau)\right] &= \Ex\left[G_\lambda(\tau)^2\right] - \Ex\left[G_\lambda(\tau)\right]^2  \\
&= \Ex\left[(\mathbf{I})(\mathbf{I})\right]+2\Ex\left[(\mathbf{I})(\mathbf{II})\right] + \Ex\left[(\mathbf{II})(\mathbf{II})\right] - F_\lambda(0)^2-2F_\lambda(0)H_1(\lambda) - H_1(\lambda)^2 \\
&\leq \underbrace{\left(\Ex\left[(\mathbf{I})(\mathbf{I})\right] - F_\lambda(0)^2\right)}_{:=(\textbf{A})} + \underbrace{\left(2\Ex\left[(\mathbf{I})(\mathbf{II})\right] - 2F_\lambda(0)H_1(\lambda) \right)}_{:=(\textbf{B})} + \underbrace{\Ex\left[(\mathbf{II})(\mathbf{II})\right]}_{:=(\textbf{C})}
\end{align*}
and we proceed to bound each term.
\begin{align*}
(\mathbf{I})(\mathbf{I})-F_\lambda(0)^2 &= \sum_{i=0,2,\ldots,k} \sum_{\ell=0,2,\ldots,k}\sum_{j=0}^{k+1-i} \sum_{s=0}^{k+1-\ell} \frac{B_iB_\ell}{j!\ell!}F_\lambda^{(i+j)}(0)F_\lambda^{(\ell+s)}(0)\eta^{i+\ell}\tau^{j+s}\, \mathbf{1}_E\,
\end{align*}
where $\mathbf{1}_E$ is an indicator function indicating that $i,j,\ell,s$ are not all zero.
We have
\begin{align*}
\Ex \left| \frac{B_iB_\ell}{j!\ell!}F_\lambda^{(i+j)}(0)F_\lambda^{(\ell+s)}(0)\eta^{i+\ell}\tau^{j+s} \right| &\leq  \frac{|B_iB_\ell|}{j!\ell!} C_{j+s}\Lambda_{i+j}(\lambda) \Lambda_{\ell+s}(\lambda)\eta^{i+\ell+j+s} \\
&\leq  \frac{|B_iB_\ell|}{j!\ell!} T^jT^s\Lambda_{i+j}(\lambda) \Lambda_{\ell+s}(\lambda)\eta^{i+\ell+j+s} \\
&\leq  E^{i+j}E^{\ell+s}\Lambda_{i+j}(\lambda) \Lambda_{\ell+s}(\lambda)\eta^{i+\ell+j+s} \\
&= \left(\Lambda_{i+j}(\lambda)(E\eta)^{i+j}\right)\left(\Lambda_{\ell+s}(\lambda)(E\eta)^{\ell+s} \right)\, .
\end{align*}
Noting that only terms where $j+s$ is even survive expectation,
and letting $\tilde{i} = i+j$ and $\tilde{\ell}=\ell+s$, we obtain
\begin{align*}
&\Ex\left[(\mathbf{I})(\mathbf{I})\right] -F_\lambda(0)^2 \\
&\leq \sum_{i=0,2,\ldots,k} \sum_{\ell=0,2,\ldots,k}\sum_{j=0}^{k+1-i} \sum_{s=0}^{k+1-\ell}\Lambda_{i+j}(\lambda)(4T\eta)^{i+j} \Lambda_{\ell+s}(\lambda)(4T\eta)^{\ell+s}\mathbf{1}_E \mathbf{1}(j+s\text{ even})\\
&= \sum_{\tilde{i}=0}^{k+1} \, \sum_{\tilde{\ell}=0}^{k+1}\, C_{\tilde{i},\tilde{\ell}}\Lambda_{\tilde{i}}(\lambda)(E\eta)^{\tilde{i}} \Lambda_{\tilde{\ell}}(\lambda)(E\eta)^{\tilde{\ell}}
\end{align*}
for coefficients $C_{\tilde{i},\tilde{\ell}}$ such that $C_{0,0}=0$, $C_{\tilde{i},\tilde{\ell}}=0$ if $\tilde{i}+\tilde{\ell}$ is odd, and $C_{\tilde{i},\tilde{\ell}}\leq k^2$. Thus:
\begin{align*}
\Ex\left[(\mathbf{I})(\mathbf{I})\right] -F_\lambda(0)^2 &\leq k^2 \sum_{ \substack{ 2\leq \tilde{i}+\tilde{\ell} \leq 2k+2 \\ \tilde{i}+\tilde{\ell} \text{ even} }}\, \Lambda_{\tilde{i}}(\lambda) \Lambda_{\tilde{\ell}}(\lambda)(E\eta)^{\tilde{i}+\tilde{\ell}} \leq k^2 \LambdaConstant(\lambda)^2 \, . 
\end{align*}

Next we bound $\Ex\left[(\mathbf{II})(\mathbf{II})\right]$.
\begin{align*}
(\mathbf{II})(\mathbf{II}) &= \sum_{i=0,2,\ldots,k}\ \sum_{\ell=0,2,\ldots k} B_iB_\ell R_i(\lambda,\tau)R_\ell(\lambda,\tau) \eta^{i+\ell}
\end{align*}
Utilizing Equation (\ref{equ:RemainderBound}), we have:
\begin{align*}
\left| B_iB_\ell R_i(\lambda,\tau)R_\ell(\lambda,\tau) \eta^{i+\ell} \right| &\leq \frac{2^{2k+4}\edit{R(\lambda)}^2|B_iB_\ell|}{(k+1)^2(k+1-i)!(k+1-\ell)!}\Lambda_{k+2}(\lambda)^2 \eta^{i+\ell}|\tau|^{2k+4-i-\ell}
\end{align*}
which gives
\begin{align*}
&\Ex \left| B_iB_\ell R_i(\lambda,\tau)R_\ell(\lambda,\tau) \eta^{i+\ell} \right| \leq \frac{2^{2k+4}\edit{R(\lambda)}^2T^{2k+4-i-\ell}|B_iB_\ell|}{(k+1)^2(k+1-i)!(k+1-\ell)!}\Lambda_{k+2}(\lambda)^2 \eta^{2k+4} \\
&\qquad\leq \frac{\edit{R(\lambda)}^2(k+2-i)(k+2-\ell)}{(k+1)^2}\left(\frac{T^{k+2-i}|B_i|}{(k+2-i)!}\right)\left(\frac{T^{k+2-\ell}|B_\ell|}{(k+2-\ell)!}\right)\Lambda_{k+2}(\lambda)^2 (2\eta)^{2k+4} \\
&\qquad\leq \frac{\edit{R(\lambda)}^2(k+2-i)(k+2-\ell)}{(k+1)^2}\Lambda_{k+2}(\lambda)^2 (2E\eta)^{2k+4}
\end{align*}
so that
\begin{align*}
\Ex\left[(\mathbf{II})(\mathbf{II})\right] &\leq \frac{\edit{R(\lambda)}^2}{(k+1)^2} \Lambda_{k+2}(\lambda)^2(2E\eta)^{2k+4} \sum_{i=0,2,\ldots,k}\ \sum_{\ell=0,2,\ldots k} (k+1-i)(k+2-\ell)  \\
&\lesssim k^2 \edit{R(\lambda)}^2\Lambda_{k+2}(\lambda)^2(2E\eta)^{2k+4} \\
&\leq k^2 \edit{R(\lambda)}^2\LambdaConstant(\lambda)^2 \, .
\end{align*}
Finally we bound the cross term $2\Ex\left[(\mathbf{I})(\mathbf{II})\right] - 2F_\lambda(0)H_1(\lambda)$.
\begin{align}
\label{equ:cross_term}
(\mathbf{I})(\mathbf{II}) &= \sum_{i=0,2,\ldots,k} \sum_{j=0}^{k+1-i} \sum_{\ell=0,2,\ldots,k} \frac{B_i}{j!}F_\lambda^{(i+j)}(0)\eta^i\tau^j B_\ell R_\ell(\lambda,\tau)\eta^\ell
\end{align}
Since $\left|F_\lambda^{(i+j)}(0)\right| \leq \Lambda_{i+j}(\lambda)$ and $|B_\ell R_\ell(\lambda,\tau)\eta^\ell| \leq \frac{2^{k+2}\edit{R(\lambda)}|B_\ell|}{(k+1)(k+1-\ell)!} \Lambda_{k+2}(\lambda) \eta^\ell \tau^{k+2-\ell}$ from (\ref{equ:RemainderBound}), we have
\begin{align*}
\left| \frac{B_i}{j!}F_\lambda^{(i+j)}(0)\eta^i\tau^j B_\ell R_\ell(\lambda,\tau)\eta^\ell \right| &\leq \frac{2^{k+2}\edit{R(\lambda)}|B_iB_\ell|}{(k+1)j!(k+1-\ell)!}\Lambda_{i+j}(\lambda)  \Lambda_{k+2}(\lambda) \eta^{i+\ell} \tau^{k+2+j-\ell}
\end{align*}
so that
\begin{align*}
&\Ex \left| \frac{B_i}{j!}F_\lambda^{(i+j)}(0)\eta^i\tau^j B_\ell R_\ell(\lambda,\tau)\eta^\ell \right| \\
&\qquad\leq \frac{2^{k+2}\edit{R(\lambda)}T^{k+2+j-\ell}|B_iB_\ell|}{(k+1)j!(k+1-\ell)!}\Lambda_{i+j}(\lambda)  \Lambda_{k+2}(\lambda) \eta^{i+j+k+2} \\
&\qquad= \frac{2^{k+2}\edit{R(\lambda)}(k+2-\ell)}{(k+1)} \left(\frac{T^j|B_i|}{j!}\right)\left(\frac{T^{k+2-\ell}|B_\ell|}{(k+2-\ell)!}\right)\Lambda_{i+j}(\lambda)  \Lambda_{k+2}(\lambda) \eta^{i+j+k+2} \\
&\qquad= \frac{\edit{R(\lambda)}(k+2-\ell)}{(k+1)}\left[(E\eta)^{i+j}\Lambda_{i+j}(\lambda)\right] \cdot \left[(2E\eta)^{k+2} \Lambda_{k+2}(\lambda)\right]\,.
\end{align*}
The same bound holds for the terms of $F_\lambda(0)H_1(\lambda)$, which arise from $i=0, j=0$ in (\ref{equ:cross_term}), so that
\begin{align*}
2\Ex&\left[(\mathbf{I})(\mathbf{II})\right] -2F_\lambda(0)H_1(\lambda)\\
&\lesssim  \left(\sum_{i=0,2,\ldots,k} \sum_{j=0}^{k+1-i} (E\eta)^{i+j}\Lambda_{i+j}(\lambda) \right)\left(\sum_{\ell=0,2,\ldots,k} \frac{\edit{R(\lambda)}(k+2-\ell)}{(k+1)}(2E\eta)^{k+2} \Lambda_{k+2}(\lambda)\right)\\
&\lesssim \left(k \sum_{\tilde{i}=0}^{k+1} \Lambda_{\tilde{i}}(\lambda)(E\eta)^{\tilde{i}} \right)\left(k\edit{R(\lambda)}(2E\eta)^{k+2} \Lambda_{k+2}(\lambda)\right) \\
&\leq k^2 \edit{R(\lambda)}\sum_{\tilde{i}=0}^{k+1} \Lambda_{\tilde{i}}(\lambda)\Lambda_{k+2}(\lambda)(2E\eta)^{\tilde{i}+k+2} \\
&\leq k^2 \edit{R(\lambda)}\LambdaConstant(\lambda)^2
\end{align*}
Thus $\Var[G_\lambda(\tau)] \lesssim k^2\edit{R(\lambda)}^2\LambdaConstant(\lambda)^2$ and the lemma is proved.

\end{proof}

\lemRDGeneralUnbiasingDevOfEstimator*	

\begin{proof}
	By Lemma \ref{lem:RD_GeneralUnbiasing_MeanVar} and the independence of the $\tau_j$, we have
	\begin{align*}
	|L(\lambda) - \Ex\ \widetilde{L}(\lambda)| &\lesssim k\edit{R(\lambda)}\Lambda_{k+2}(\lambda)(2E\eta)^{k+2} \\
	\Var\ \widetilde{L}(\lambda) &\lesssim \frac{1}{M}k^2\LambdaConstant(\lambda)^2
	\end{align*}
	so by Chebyshev's Inequality we can conclude that with probability at least $1-1/t^2$, we have:
	\begin{align*}
	| \widetilde{L}(\lambda) - \Ex[ \widetilde{L}(\lambda)]| &\leq \frac{tk\edit{R(\lambda)}\LambdaConstant(\lambda)}{\sqrt{M}}
	\end{align*}
	which gives
	\begin{align*}
	|L(\lambda) - \widetilde{L}(\lambda)| &\leq 	|L(\lambda) -  \Ex[\widetilde{L}(\lambda)]| + | \Ex[\widetilde{L}(\lambda)] -  \widetilde{L}(\lambda)| \\
	&\lesssim k\edit{R(\lambda)}\Lambda_{k+2}(\lambda)(2E\eta)^{k+2} +\frac{tk\edit{R(\lambda)}\LambdaConstant(\lambda)}{\sqrt{M}} \, .
	\end{align*}
\end{proof}	

\section{Supporting results: \edit{noisy dilation MRA}}
\label{app:gen_mra}

This appendix contains supporting results needed to prove Proposition \ref{prop:DilationAndAdditiveNoise}, which defines a wavelet invariant estimator for \edit{noisy dilation MRA}. Lemma \ref{lem:GenMRA_AddNoise} controls the additive noise error and Lemma \ref{lem:GenMRA_CrossTerm} controls the cross-term error. Lemma \ref{lem:VanishingIntegrals} guarantees that the dilation unbiasing procedure applied to the additive noise still has mean $\sigma^2$, which is needed to prove Lemma  \ref{lem:GenMRA_AddNoise}.

\lemGenMRAAddNoise*	

\begin{proof}
	Let
	\begin{align*}
	D(\varepsilon_{j},\lambda) &:= \frac{1}{2\pi}\int |\widehat{\epsilon_j}(\omega)|^2 A_\lambda |\widehat{\psi}_\lambda(\omega)|^2 \ d\omega\, .
	\end{align*}
	By Lemma \ref{lem:PS_addnoise}, $\Ex_\varepsilon\left[ |\widehat{\varepsilon_j}(\omega)|^2 \right]= \sigma^2$, and
	we thus obtain:
	\begin{align*}
	\Ex_\varepsilon \left[D(\varepsilon_{j},\lambda)\right]
	&= \Ex_\varepsilon \left[\frac{1}{2\pi} \int |\widehat{\varepsilon_j}(\omega)|^2 A_\lambda|\widehat{\psi}_\lambda(\omega)|^2\ d\omega\right] \\
	&= \Ex_\varepsilon\left[ \frac{1}{2\pi} \int |\widehat{\varepsilon_j}(\omega)|^2 |\widehat{\psi}_\lambda(\omega)|^2\ d\omega -\frac{1}{2\pi} \int |\widehat{\varepsilon_j}(\omega)|^2 B_2\eta^2\lambda^2\frac{d}{d\lambda^2}|\widehat{\psi}_\lambda(\omega)|^2\ d\omega - \ldots \right. \\
	&\qquad \left.- \frac{1}{2\pi} \int |\widehat{\varepsilon_j}(\omega)|^2 B_k\eta^k\lambda^k\frac{d}{d\lambda^k}|\widehat{\psi}_\lambda(\omega)|^2\ d\omega\right] \\
	&=\sigma^2\left( \frac{1}{2\pi} \int|\widehat{\psi}_\lambda(\omega)|^2\ d\omega -\frac{B_2\eta^2}{2\pi} \int \lambda^2\frac{d}{d\lambda^2}|\widehat{\psi}_\lambda(\omega)|^2\ d\omega  - \ldots  \right. \\ 
	&\qquad \left. - \frac{B_k\eta^k}{2\pi} \int \lambda^k\frac{d}{d\lambda^k}|\widehat{\psi}_\lambda(\omega)|^2\ d\omega\right) \\
	&=\sigma^2 (1-0-\ldots -0) \\
	&=\sigma^2\, ,
	\end{align*} 
	where we have used Lemma \ref{lem:VanishingIntegrals} to conclude $\int  \lambda^m\left(\frac{d^m}{d\lambda^m}|\widehat{\psi}_\lambda(\omega)|^2\right)\ d\omega = 0$ for $m=2,\ldots,k$. 
	Also since $(a_1+\ldots+a_n)^2 \leq n(a_1^2+\ldots+a_n^2)$ by the Cauchy-Schwarz inequality, we obtain:
	\begin{align*}
	\Ex_\varepsilon \left[ D(\varepsilon_{j},\lambda)^2 \right] \leq  \Ex_\varepsilon \left[k\sum_{m=0,2,..,k}\left( \frac{B_m\eta^m}{2\pi} \int |\widehat{\varepsilon_j}(\omega)|^2 \lambda^m\frac{d^m}{d\lambda^m}|\widehat{\psi}_\lambda(\omega)|^2\ d\omega\right)^2\right] 
	\end{align*}
	where we let $\frac{d}{d\lambda^0}|\widehat{\psi}_\lambda(\omega)|^2$ denote $|\widehat{\psi}_\lambda(\omega)|^2$ and $B_0=1$.
	By Lemma \ref{lem:PS_addnoise}, we have $\Ex_\varepsilon\left[ |\varepsilon_j(\omega)|^2|\varepsilon_j(\xi)|^2\right] \leq 3\sigma^4$ for all frequencies $\omega, \xi$, so that
	\begin{align*}
	\Ex_\varepsilon &\left[\left(\frac{B_m\eta^m}{2\pi}\int |\widehat{\varepsilon_j}(\omega)|^2 \lambda^m\frac{d^m}{d\lambda^m}|\widehat{\psi}_\lambda(\omega)|^2\ d\omega\right)^2\right] \\
	&\quad \leq  \Ex_\varepsilon \left[\frac{B_m^2\eta^{2m}}{4\pi^2}\int\int |\widehat{\varepsilon_j}(\omega)|^2|\widehat{\varepsilon_j}(\xi)|^2 \left|\lambda^m\frac{d^m}{d\lambda^m}|\widehat{\psi}_\lambda(\omega)|^2\right| \cdot \left|\lambda^m\frac{d^m}{d\lambda^m}|\widehat{\psi}_\lambda(\xi)|^2\right|\ d\omega\ d\xi \right]\\
	&\quad \leq 3\sigma^4\left(\frac{1}{2\pi}\int \left|B_m\eta^m\lambda^m\frac{d^m}{d\lambda^m}|\widehat{\psi}_\lambda(\omega)|^2\right|\ d\omega\right)^2 \\
	&\quad \leq 3\sigma^4\Psi_m^2(E\eta)^{2m}\, ,
	\end{align*}
	where the last line follows from Corollary \ref{cor:WSC_deriv_cons_fdirac} in Appendix \ref{app:WSC_props}.
	We thus obtain:
	\begin{align*}
	\Ex_\varepsilon \left[D(\varepsilon_{j},\lambda)^2\right] &\leq
	k\sum_{m=0,2,..,k}\Ex_\varepsilon\left[ \left(\frac{B_m\eta^m}{2\pi} \int |\widehat{\varepsilon_j}(\omega)|^2 \lambda^m\frac{d^m}{d\lambda^m}|\widehat{\psi}_\lambda(\omega)|^2\ d\omega\right)^2\right] \\
	&\leq 3k\sigma^4 \sum_{m=0,2,..,k}\Psi_m^2(E\eta)^{2m} := (\mathbf{I})
	\end{align*}
	so that
	\begin{align*}
	\Ex_\varepsilon \left[D(\varepsilon_{j},\lambda) - \sigma^2\right] &= 0 \\
	\Var_\varepsilon  \left[D(\varepsilon_{j},\lambda) - \sigma^2 \right] &= \Var_\varepsilon  \left[D(\varepsilon_{j},\lambda) \right] \leq \Ex_\varepsilon \left[(D(\varepsilon_{j},\lambda))^2\right] \leq (\mathbf{I}).
	\end{align*}
	Thus
	\begin{align*}
	\Var_\varepsilon \left(\frac{1}{M}\sum_{j=1}^M D(\varepsilon_{j},\lambda) - \sigma^2 \right) \leq \frac{(\mathbf{I})}{M}
	\end{align*}
	so that by Chebyshev's Inequality with probability at least $1-1/t^2$
	\begin{align*}
	\left|\frac{1}{M}\sum_{j=1}^M D(\varepsilon_{j},\lambda) - \sigma^2\right| &\leq \frac{t\sqrt{(\mathbf{I})}}{\sqrt{M}} \leq t\sqrt{3k}\left(\sum_{m=0,2,\ldots,k}\Psi_m(E\eta)^m\right)\frac{\sigma^2}{\sqrt{M}} 
	= 2t\sqrt{k}\PsiConstant\frac{\sigma^2}{\sqrt{M}} \, .
	\end{align*}
\end{proof}

\lemGenMRACrossTerm*

\begin{proof}
	We have
	\begin{align*}
	\frac{1}{M}\sum_{j=1}^M \frac{1}{2\pi} \int \left(\widehat{f}_{\tau_j}(\omega)\overline{\widehat{\varepsilon}_j}(\omega)+\overline{\widehat{f}_{\tau_j}}(\omega)\widehat{\varepsilon}_j(\omega)\right) A_\lambda |\widehat{\psi}_\lambda(\omega)|^2\ d\omega &=\frac{1}{M}\sum_{j=1}^M Y_j + \overline{Y_j}
	\end{align*}
	where 
	\begin{align*}
	Y_j :=\frac{1}{2\pi} \int \left(\overline{\widehat{f}_{\tau_j}}(\omega)\widehat{\varepsilon}_j(\omega)\right)A_\lambda|\widehat{\psi}_\lambda(\omega)|^2\ d\omega \, .
	\end{align*}
	The random variable $Y_j $ has randomness depending on both $\varepsilon_j$ and $\tau_j$. Note that
	\begin{align*}
	\Ex_{\varepsilon,\tau} [Y_j] &= \Ex_{\varepsilon,\tau}\left[\Ex_{\varepsilon,\tau}[Y_j | \tau_j ]\right]
	\end{align*}
	since $Y_j$ is integrable. Thus since $\Ex_{\varepsilon,\tau}[\widehat{\varepsilon}_j(\omega)] = 0$, we obtain $\Ex_{\varepsilon,\tau}[Y_j|\tau_j] = 0$, which yields $\Ex_{\varepsilon,\tau}[Y_j]=0$. We also have:
	\begin{align*}
	\Var_{\varepsilon,\tau}[Y_j] &= \Ex_{\varepsilon,\tau}[Y_j^2] \\
	&\leq \Ex_{\varepsilon,\tau}\left[\left(\frac{1}{2\pi} \int |\overline{\widehat{f}_{\tau_j}}(\omega)|\cdot|\widehat{\varepsilon}_j(\omega)|\cdot\left|A_\lambda|\widehat{\psi}_\lambda(\omega)|^2\right|\ d\omega\right)^2 \right]\\
	&\leq \Ex_{\varepsilon,\tau} \left[ \left(\frac{1}{2\pi} \int |\overline{\widehat{f}_{\tau_j}}(\omega)|^2\cdot\left|A_\lambda|\widehat{\psi}_\lambda(\omega)|^2\right|\ d\omega\right)\left(\frac{1}{2\pi} \int|\widehat{\varepsilon}_j(\omega)|^2\cdot\left|A_\lambda|\widehat{\psi}_\lambda(\omega)|^2\right|\ d\omega\right) \right]\\
	&= \Ex_{\tau}\left[\frac{1}{2\pi} \int |\overline{\widehat{f}_{\tau_j}}(\omega)|^2\cdot\left|A_\lambda|\widehat{\psi}_\lambda(\omega)|^2\right|\ d\omega\right]\Ex_{\varepsilon}\left[\frac{1}{2\pi} \int|\widehat{\varepsilon}_j(\omega)|^2\cdot\left|A_\lambda|\widehat{\psi}_\lambda(\omega)|^2\right|\ d\omega\right]\, .
	\end{align*}
	Letting $B_0=1$ and applying Lemma \ref{lem:WSC_deriv_cons}, we have:
	\begin{align*}
	\Ex_{\tau}\left[\frac{1}{2\pi} \int |\overline{\widehat{f}_{\tau_j}}(\omega)|^2\cdot\left|A_\lambda|\widehat{\psi}_\lambda(\omega)|^2\right|\ d\omega\right]
	&\leq \Ex_{\tau}\left[\sum_{m=0,2,\ldots,k}\frac{1}{2\pi} \int |\widehat{f}_{\tau_j}(\omega)|^2\cdot\left|B_m\eta^m\lambda^m\frac{d^m}{d\lambda^m}|\widehat{\psi}_\lambda(\omega)|^2\right|\ d\omega\right] \\
	&\leq \Ex_{\tau}\left[\sum_{m=0,2,\ldots,k} (E\eta)^m\left(\norm{f_{\tau_j}}_1^2\Psi_m \wedge \frac{\norm{f_{\tau_j}'}_1^2\Theta_m}{\lambda^2}\right)\right] \\
	&\leq \sum_{m=0,2,\ldots,k} (E\eta)^m\left(\norm{f}_1^2\Psi_m \wedge \frac{4\norm{f'}_1^2\Theta_m}{\lambda^2}\right) \\
	&\leq 4\sum_{m=0,2,\ldots,k} (E\eta)^m \Lambda_m(\lambda) \\
	&\lesssim \Lambda_0(\lambda) + \LambdaConstant(\lambda)\, 
	\end{align*}
	since $\norm{\tau_j}_\infty \leq \frac{1}{2}$ guarantees $\norm{f_{\tau_j}'}_1 = \frac{1}{1-\tau_j}\norm{f'}_1 \leq 2\norm{f'}_1$. Also:
	\begin{align*}
	\Ex_{\varepsilon} \left[\frac{1}{2\pi} \int|\widehat{\varepsilon}_j(\omega)|^2\cdot\left|A_\lambda|\widehat{\psi}_\lambda(\omega)|^2\right|\ d\omega\right]
	&\leq \Ex_{\varepsilon}  \left[\sum_{m=0,2,\ldots,k}\frac{1}{2\pi} \int |\widehat{\varepsilon}_j(\omega)|^2\cdot\left|B_m\eta^m\lambda^m\frac{d^m}{d\lambda^m}|\widehat{\psi}_\lambda(\omega)|^2\right|\ d\omega\right] \\&= \sigma^2 \left(\sum_{m=0,2,\ldots,k}\frac{1}{2\pi} \int\left|B_m\eta^m\lambda^m\frac{d^m}{d\lambda^m}|\widehat{\psi}_\lambda(\omega)|^2\right|\ d\omega\right) \\
	&\leq \sigma^2 \sum_{m=0,2,\ldots,k}(E\eta)^m\Psi_m \\
	&= \sigma^2 \PsiConstant
	\end{align*}
	where the second line follows from Lemma \ref{lem:PS_addnoise} in Appendix \ref{app:AddNoise} and the next to last line from Corollary \ref{cor:WSC_deriv_cons_fdirac} in Appendix \ref{app:WSC_props}. We thus have:
	\begin{align*}
	\Ex_{\varepsilon,\tau}[Y_j] &= 0 \\
	\Var_{\varepsilon,\tau}[Y_j] &\lesssim \sigma^2 \PsiConstant(\Lambda_0(\lambda)+\LambdaConstant(\lambda))
	\end{align*}
	and an identical argument can be applied to the $\overline{Y_j}$ so that by Chebyshev's Inequality with probability at least $1-1/t^2$:
	\begin{align*}
	\left|\frac{1}{M}\sum_{j=1}^M  Y_j+\overline{Y_j}\right| &\leq \left|\frac{1}{M}\sum_{j=1}^M Y_j\right| + \left|\frac{1}{M}\sum_{j=1}^M \overline{Y_j}\right|
	\lesssim t\sqrt{\PsiConstant}\sqrt{\Lambda_0(\lambda)+\LambdaConstant(\lambda)}\frac{\sigma}{\sqrt{M}}.
	\end{align*}
\end{proof}

\begin{lemma}
	\label{lem:VanishingIntegrals}
	Assume $\psi$ is $k$-admissable. Then:
	\begin{align}
	\label{equ:VanishingIntegrals}
	\int  \lambda^m\left(\frac{d^m}{d\lambda^m}|\widehat{\psi}_\lambda(\omega)|^2\right)\ d\omega &= 0
	\end{align}
	for all $1\leq m\leq k$.
\end{lemma}

\begin{proof}
	We recall that since $\psi$ is $k$-admissable, $|\widehat{\psi}_\lambda(\omega)|^2 \in \Cb^k (\R)$, and to simplify notation we let $g = |\widehat{\psi}|^2$ and 
	\begin{align*}
	g_\lambda(\omega) &= \frac{1}{\lambda}g\left(\frac{\omega}{\lambda}\right) = |\widehat{\psi}_\lambda(\omega)|^2\, .
	\end{align*}
	We first establish that:
	\begin{align}
	\label{equ:IterativeCondition}
	\lambda^k\left(\frac{d}{d\lambda^k}g_\lambda(\omega)\right) &= \frac{d}{d\omega}\left(-\omega\lambda^{k-1}\frac{d}{d\lambda^{k-1}}g_\lambda(\omega)\right)-(k-1)\lambda^{k-1}\frac{d}{d\lambda^{k-1}}g_\lambda(\omega)\, .
	\end{align}
	The proof is by induction. When $k=1$, 	we obtain
	\begin{align*}
	\text{LHS of Eqn. (\ref{equ:IterativeCondition})} &= \lambda\frac{d}{d\lambda}\left(\frac{1}{\lambda}g\left(\frac{\omega}{\lambda}\right)\right) 
	= -\frac{\omega}{\lambda^2}g'\left(\frac{\omega}{\lambda}\right)-\frac{1}{\lambda}g\left(\frac{\omega}{\lambda}\right) 
	= -\omega g_\lambda'(\omega) - g_\lambda(\omega)
	\end{align*}	
	and 
	\begin{align*}
	\text{RHS of Eqn. (\ref{equ:IterativeCondition})} &= \frac{d}{d\omega}\left(-\omega g_\lambda(\omega)\right) = -\omega g_\lambda'(\omega) - g_\lambda(\omega)\, ,
	\end{align*}
	so the base case is established. We now assume that Equation (\ref{equ:IterativeCondition}) holds and show it also holds for $k+1$ replacing $k$. By the inductive hypothesis:
	\begin{align*}
	\frac{d}{d\lambda^k}g_\lambda(\omega) &= \frac{d}{d\omega}\left(-\omega\lambda^{-1}\frac{d}{d\lambda^{k-1}}g_\lambda(\omega)\right)-(k-1)\lambda^{-1}\frac{d}{d\lambda^{k-1}}g_\lambda(\omega) \\
	\frac{d}{d\lambda^{k+1}}g_\lambda(\omega) &=\frac{d}{d\omega}\left(-\omega\lambda^{-1}\frac{d}{d\lambda^{k}}g_\lambda(\omega) + \frac{d}{d\lambda^{k-1}}g_\lambda(\omega)\omega\lambda^{-2}\right) \\
	&\qquad -(k-1)\left(\lambda^{-1}\frac{d}{d\lambda^k}g_\lambda(\omega)+\frac{d}{d\lambda^{k-1}}g_\lambda(\omega)(-\lambda^{-2})\right) \\
	&=\frac{d}{d\omega}\left(-\omega\lambda^{-1}\frac{d}{d\lambda^k}g_\lambda(\omega)\right) - (k-1)\lambda^{-1}\frac{d}{d\lambda^k}g_\lambda(\omega) \\
	&\qquad + \underbrace{\frac{d}{d\omega}\left(\omega\lambda^{-2}\frac{d}{d\lambda^{k-1}}g_\lambda(\omega)\right) +(k-1)\lambda^{-2}\frac{d}{d\lambda^{k-1}}g_\lambda(\omega)}_{=-\lambda^{-1}\frac{d}{d\lambda^k}g_\lambda(\omega) \text{ by inductive hypothesis}} \\
	&=\frac{d}{d\omega}\left(-\omega\lambda^{-1}\frac{d}{d\lambda^k}g_\lambda(\omega)\right) - k\lambda^{-1}\frac{d}{d\lambda^k}g_\lambda(\omega)
	\end{align*}
	so that
	\begin{align*}
	\lambda^{k+1}\frac{d}{d\lambda^{k+1}}g_\lambda(\omega) &=\frac{d}{d\omega}\left(-\omega\lambda^{k}\frac{d}{d\lambda^k}g_\lambda(\omega)\right) - k\lambda^{k}\frac{d}{d\lambda^k}g_\lambda(\omega)\, .
	\end{align*}
	Thus (\ref{equ:IterativeCondition}) is established. We now use integration by parts to show (\ref{equ:IterativeCondition}) implies (\ref{equ:VanishingIntegrals}) in the Lemma. The proof of (\ref{equ:VanishingIntegrals}) is once again by induction. When $k=1$, we have already shown
	\begin{align}
	\label{equ:lambda_deriv_base_case}
	\lambda\left(\frac{d}{d\lambda} g_\lambda(\omega)\right) &= -\omega g_\lambda'(\omega) - g_\lambda(\omega) \, .
	\end{align}
	Integration by parts gives
	\begin{align*}
	\int \omega g_\lambda'(\omega)\ d\omega &= \left(\omega g_\lambda(\omega)\right)\biggr|_{-\infty}^{\infty} - \int g_\lambda(\omega)\ d\omega = \int g_\lambda(\omega)\ d\omega \,.
	\end{align*}
	Note $\omega g_\lambda(\omega)$ vanishes at $\pm \infty$ since $g \in \Lb^1 (\R)$ guarantees $g_\lambda \in \Lb^1 (\R)$, and thus $g_\lambda$ must decay faster that $\omega^{-1}$. Utilizing (\ref{equ:lambda_deriv_base_case}), 
	\begin{align*}
	\int \omega g_\lambda'(\omega) - g_\lambda(\omega)\ d\omega &= 0 \quad \implies \quad \int \lambda\left(\frac{d}{d\lambda} g_\lambda(\omega)\right) d\omega =0
	\end{align*}
	and the base case is established. We now assume
	\begin{align*}
	\int \lambda^{k-1}\left(\frac{d}{d\lambda^{k-1}}g_\lambda(\omega)\right)\ d\omega &= 0\, .
	\end{align*}
	By integrating Equation (\ref{equ:IterativeCondition}), we obtain:
	\begin{align*}
	\int \lambda^k & \left(\frac{d}{d\lambda^k}g_\lambda(\omega)\right)\ d\omega \\
	&= \int \frac{d}{d\omega}\left(-\omega\lambda^{k-1}\frac{d}{d\lambda^{k-1}}g_\lambda(\omega)\right)\ d\omega-(k-1)\underbrace{\int \lambda^{k-1}\frac{d}{d\lambda^{k-1}}g_\lambda(\omega)\ d\omega}_{=0 \text{ by induc. hypo.}} \\
	&=\int -\omega\frac{d}{d\omega}\left(\lambda^{k-1}\frac{d}{d\lambda^{k-1}} g_\lambda(\omega)\right)\ d\omega -\underbrace{\int \lambda^{k-1}\frac{d}{d\lambda^{k-1}}g_\lambda(\omega)\ d\omega}_{=0 \text{ by induc. hypo.}} \\
	&=-\omega\lambda^{k-1}\frac{d}{d\lambda^{k-1}}g_\lambda(\omega)\biggr|_{-\infty}^{\infty} + \underbrace{\int \lambda^{k-1}\frac{d}{d\lambda^{k-1}}g_\lambda(\omega)\ d\omega}_{=0 \text{ by induc. hypo.}} \\
	&=0 \, .
	\end{align*}
	We are guaranteed $-\omega\lambda^{k-1}\frac{d}{d\lambda^{k-1}}g_\lambda(\omega)$ vanishes at $\pm \infty$ since in the proof of Lemma \ref{lem:WaveletScatteringDeriv} we showed \\$\lambda^{k-1}\frac{d}{d\lambda^{k-1}}g_\lambda(\omega) = \sum_{j=0}^{k-1} C_j \omega^{j}g_\lambda^{(j)}(\omega)$, and $\omega^j g_\lambda^{(j)} \in \Lb^1 (\R)$ implies $\omega^{j+1}g_\lambda^{(j)}$ vanishes at $\pm \infty$.
\end{proof}	
 
\section{Moment estimation for \edit{noisy dilation MRA}}
\label{app:numerical_implementation}

In this appendix we outline a moment estimation procedure for \edit{noisy dilation MRA} (Model \ref{model:genMRA}) in the special case $t=0$, i.e. signals are randomly dilated and subjected to additive noise but are not translated. This procedure is a generalization of the method presented in Section \ref{sec:EmpMomentEst}. 

Given the additive noise level, the moments of the dilation distribution $\tau$ can be empirically estimated from the mean and variance of the random variables $\beta_{m}(y_j)$ defined by
\begin{align}
\label{equ:beta_m_def}
\beta_{m}(y_j) &= \int_0^{2^{\ell}\pi} \omega^m \widehat{y}_j(\omega)\ d\omega
\end{align}
for integer $m \geq 0$. To account for the effect of additive noise on the above random variables, we define:
\begin{equation}
\begin{array}{r@{}l}
\vspace{.5em}
g_m(\ell,\sigma)&=\displaystyle \int_{0}^{2^{\ell}\pi} \int_{0}^{2^{\ell}\pi} \frac{2\sigma^2\xi^m\omega^m\sin(\frac{1}{2}(\xi-\omega))}{(\xi-\omega)}\ d\omega\ d\xi \label{equ:g_m_def}
\end{array}
\end{equation}
and an order $m$ additive noise adjusted squared coefficient of variation by:
\begin{align}
\label{equ:CV_m}
CV_m &:= \frac{\Var[\beta_{m}(y_j)] - g_m(\ell, \sigma)}{|\Ex[\beta_{m}(y_j)]|^2}.
\end{align}
\begin{remark}
	If the noisy signals are supported in $[-\frac{N}{2}, \frac{N}{2}]$ instead of $[-\frac{1}{2}, \frac{1}{2}]$, (\ref{equ:g_m_def}) is replaced with:
	\begin{align*}
	g_m(N,\ell,\sigma)&=\int_{0}^{2^{\ell}\pi} \int_{0}^{2^{\ell}\pi} \frac{2\sigma^2\xi^m\omega^m\sin(\frac{N}{2}(\xi-\omega))}{(\xi-\omega)}\ d\omega\ d\xi.
	\end{align*}
\end{remark}
The following proposition mirrors Proposition \ref{prop:emp_mom_est_dilMRA} for dilation MRA; its proof appears at the end of Appendix \ref{app:numerical_implementation}.
\begin{proposition}	
	\label{prop:emp_mom_est}
	Assume Model \ref{model:genMRA} with $t=0$ and $CV_0, CV_1$ defined by (\ref{equ:beta_m_def}), (\ref{equ:g_m_def}), and (\ref{equ:CV_m}). Then
	\begin{align*}
	CV_0 &=\eta^2+(3C_4-3)\eta^4+O(\eta^6) \\
	CV_1 &= 4\eta^2 + (25C_4-33)\eta^4+O(\eta^6) \, .
	\end{align*}	
\end{proposition}	
Once again we cannot compute $CV_m$ exactly, but by replacing $\Var, \Ex$ with their finite sample estimators, we obtain approximations $\widetilde{CV}_m$  which can be used to define estimators of the dilation moments. 

\begin{definition}
	\label{def:emp_dil_mom_est}
	Assume Model \ref{model:dilMRA} with $t=0$ and $\widetilde{CV}_0, \widetilde{CV}_1$ the empirical counterparts of (\ref{equ:CV_m}).
	Define the second order estimator of $\eta^2$ by $\widetilde{\eta}^2 = \widetilde{CV}_0.$
	Define the fourth order estimators of $(\eta^2, C_4\eta^4)$ by the unique positive solution $(\widetilde{\eta}^2, \widetilde{C}_4)$ of
	\begin{align*}
	\widetilde{CV}_0 &=\eta^2+(3C_4-3)\eta^4 \\
	\widetilde{CV}_1 &= 4\eta^2 + (25C_4-33)\eta^4.
	\end{align*}
\end{definition}

As $M \rightarrow \infty$, the second order moment estimator is accurate up to $O(\eta^4)$ and the fourth order moment estimators are accurate up to $O(\eta^6)$. 
However in the finite sample regime, the $g_m(\ell, \sigma)$ appearing in (\ref{equ:CV_m}) will be replaced with $g_m(\ell, \sigma) \pm O(\sigma^2/\sqrt{M})$, so that the estimators given in Definition \ref{def:emp_dil_mom_est} are subject to an error of order $O(\sigma^2/\sqrt{M})$. More generally, the additive noise fluctuations imply that to estimate the first $k/2$ even moments of $\tau$ up to an $O(\eta^{k+1})$ error will require $\sigma^2/\sqrt{M} \leq \eta^{k+1}$, or $M \geq \sigma^4\eta^{-2(k+1)}$. 

Having established an empirical moment estimation procedure for \edit{noisy dilation MRA} when $t=0$, we repeat the simulations of Section \ref{sec:GenMRANoiseSim} on the restricted model, but estimate the additive and dilation moments empirically. Since accurately estimating the moments of $\tau$ is difficult for $\sigma$ large, we make three modifications to the oracle set-up. First, we lower the additive noise level by a factor of 2 from the oracle simulations, and consider all parameter combinations resulting from $\sigma = 2^{-5}, 2^{-4}$ (\edit{giving $\snr=9.0,2.2$}) and $\eta = 0.06, 0.12$. Secondly, we take $M$ substantially larger than for the oracle simultions, with $16,384 \leq M \leq 370,727$. Thirdly, we compute WSC $k=4$ only for large dilations. For large dilations $(\eta^2, C_4\eta^4)$ are approximated with fourth order estimators, while for small dilations $\eta^2$ is approximated with a second order estimator (see Definition \ref{def:emp_dil_mom_est}). 

Results are shown in Figure \ref{fig:GenMRAModelEmpirical}, and the same overall behavior observed in the oracle simulations for large $M$ holds. The additive noise level was estimated empirically as described in Section \ref{sec:AddNoiseLevel}. For the medium and high frequency signal, WSC $k=2$ has substantially smaller error than both PS $k=0$ and WSC $k=0$; for the large frequency signal, the error is decreased by at least a factor of 2 for large dilations and a factor of 4 for small dilations relative to both zero order estimators. When WSC $k=4$ is defined, it has a smaller error than WSC $k=2$ for the high frequency signal, while WSC $k=2$ is preferable for the low and medium frequency signal. We observe that for the oracle simulations WSC $k=4$ is preferable for all frequencies, so this is most likely due to error in the moment estimation degrading the WSC $k=4$ estimator. For the low frequency signal, PS $k=0$ once again achieves the smallest error for small dilations, while for large dilations the higher order wavelet methods appear to surpass PS $k=0$ for $M$ large enough.

\begin{figure}
	\centering
	\begin{subfigure}[b]{0.32\textwidth}
		\centering
		\includegraphics[width=.85\textwidth]{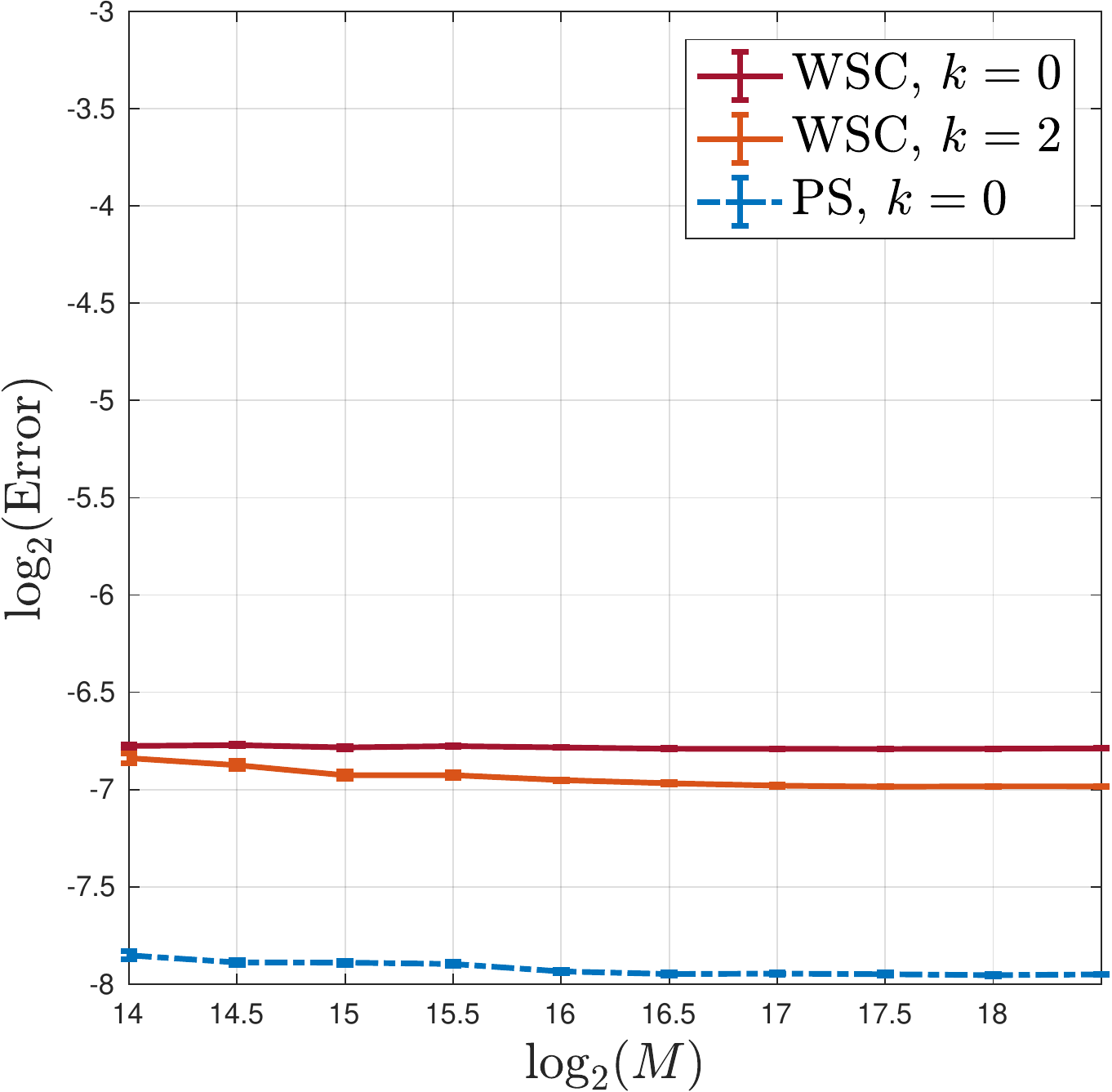}
		\caption{$f_1,\edit{\snr=9.0}, \eta=0.06$}
		\vspace*{.1cm}
		\label{fig:sim_E_7_low}
	\end{subfigure}
	\hfill
	\begin{subfigure}[b]{0.32\textwidth}
		\centering
		\includegraphics[width=.85\textwidth]{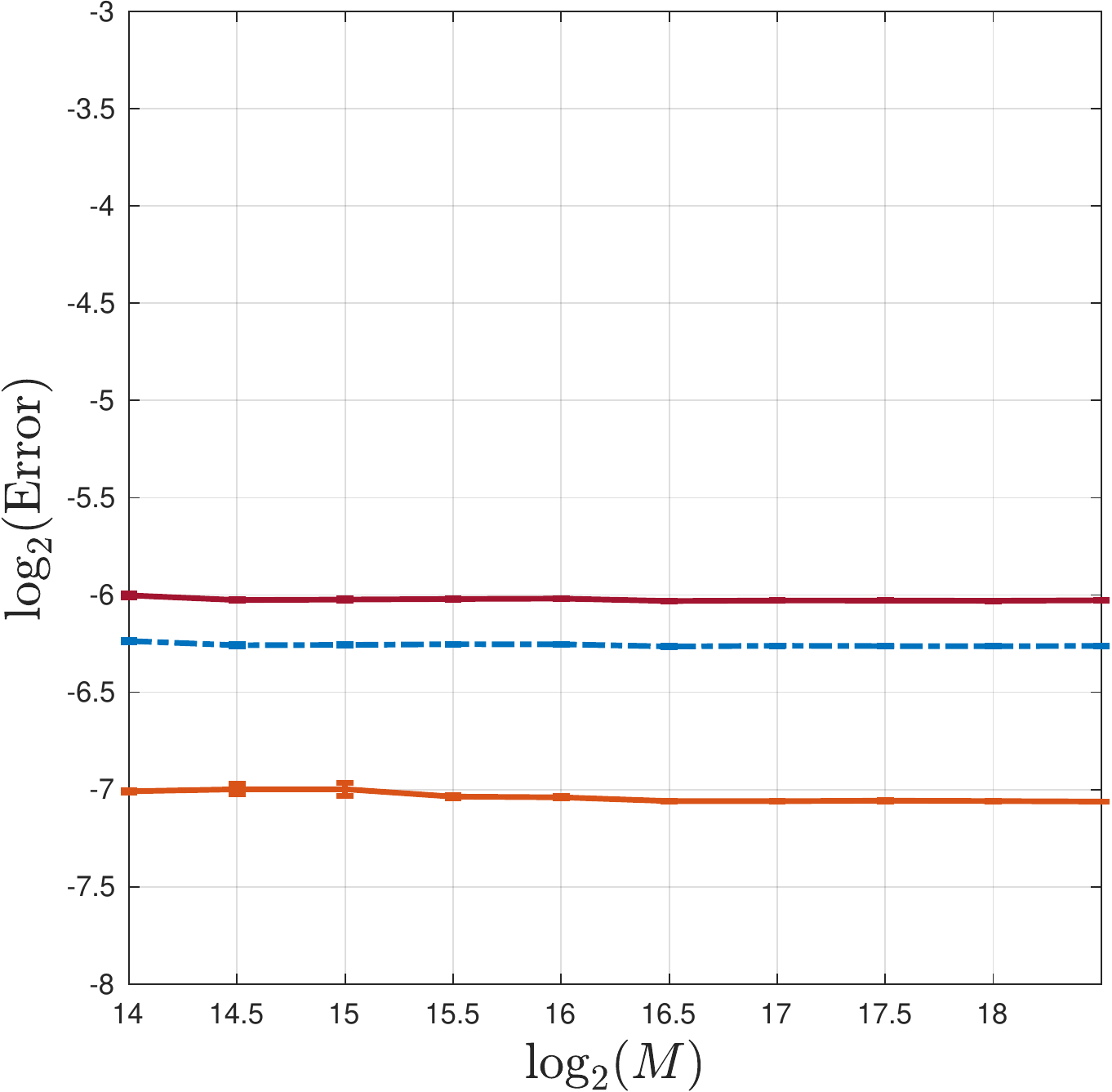}
		\caption{$f_2,\edit{\snr=9.0}, \eta=0.06$}
		\vspace*{.1cm}
		\label{fig:sim_E_7_med}
	\end{subfigure}
	\hfill
	\begin{subfigure}[b]{0.32\textwidth}
		\centering
		\includegraphics[width=.85\textwidth]{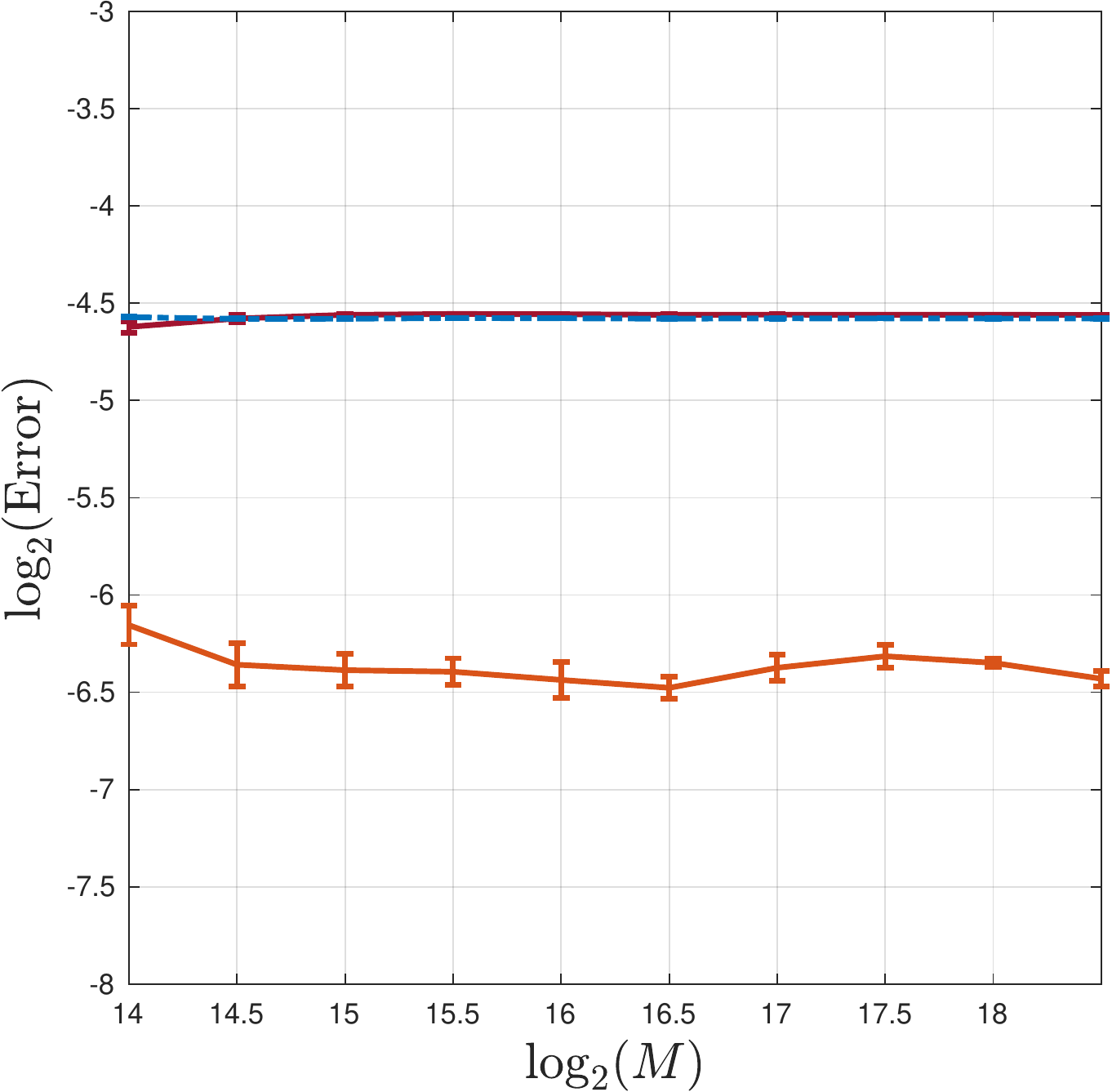}
		\caption{$f_3,\edit{\snr=9.0}, \eta=0.06$}
		\vspace*{.1cm}
		\label{fig:sim_E_7_high}
	\end{subfigure}
	\begin{subfigure}[b]{0.32\textwidth}
		\centering
		\includegraphics[width=.85\textwidth]{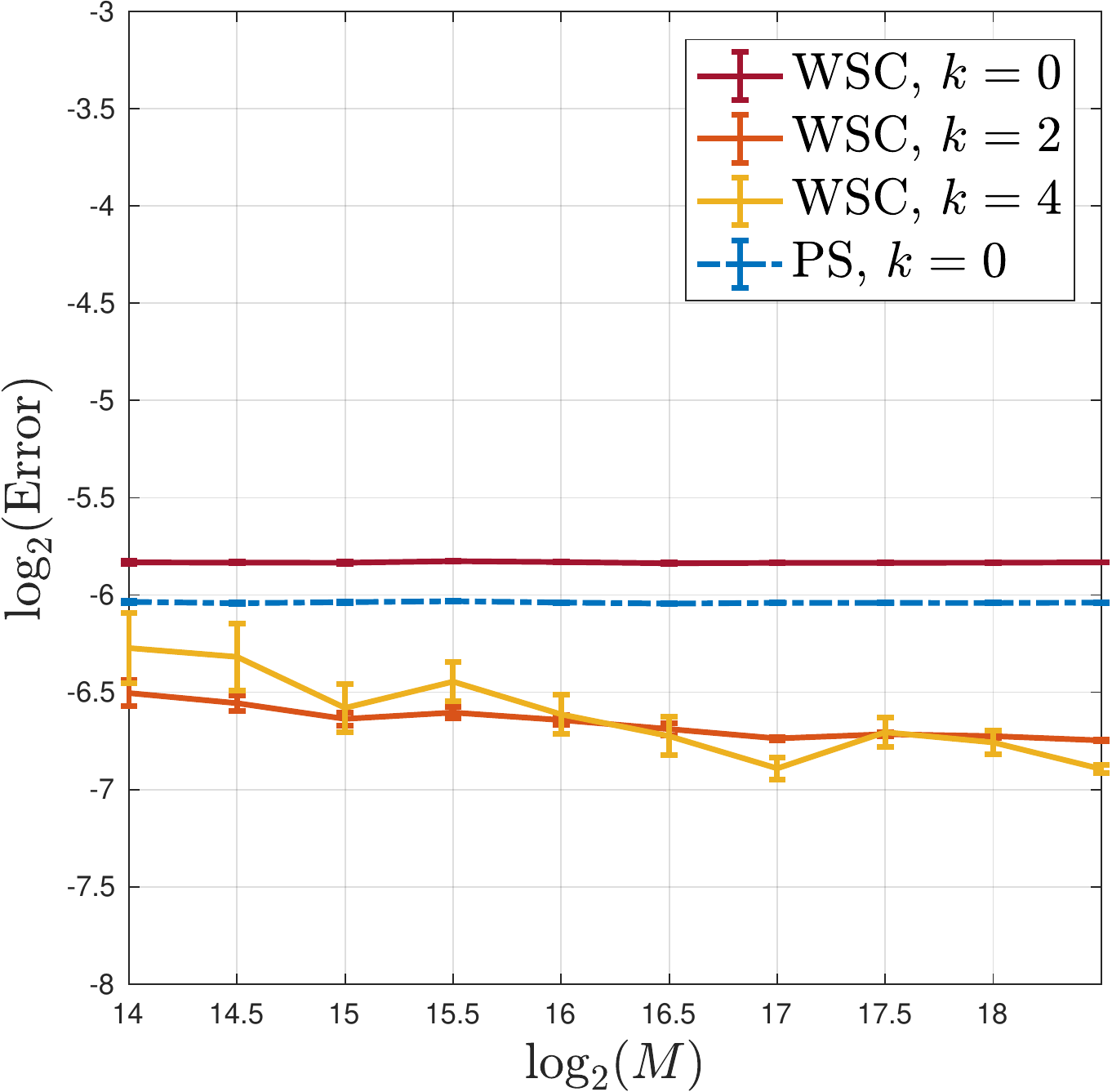}
		\caption{$f_1,\edit{\snr=9.0}, \eta=0.12$}
		\label{fig:sim_E_8_low}
	\end{subfigure}
	\hfill
	\begin{subfigure}[b]{0.32\textwidth}
		\centering
		\includegraphics[width=.85\textwidth]{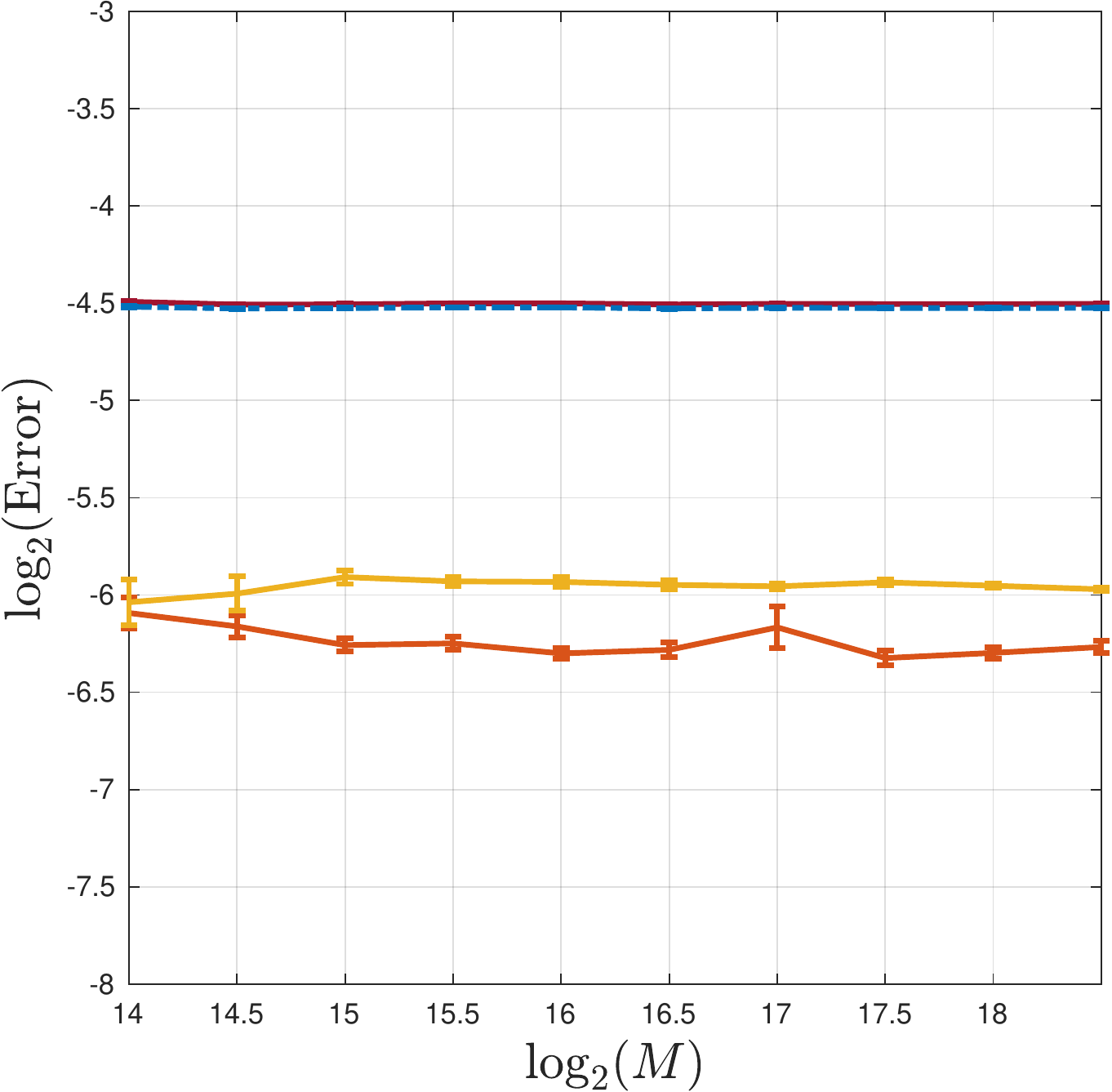}
		\caption{$f_2,\edit{\snr=9.0}, \eta=0.12$}
		\label{fig:sim_E_8_med}
	\end{subfigure}
	\hfill
	\begin{subfigure}[b]{0.32\textwidth}
		\centering
		\includegraphics[width=.85\textwidth]{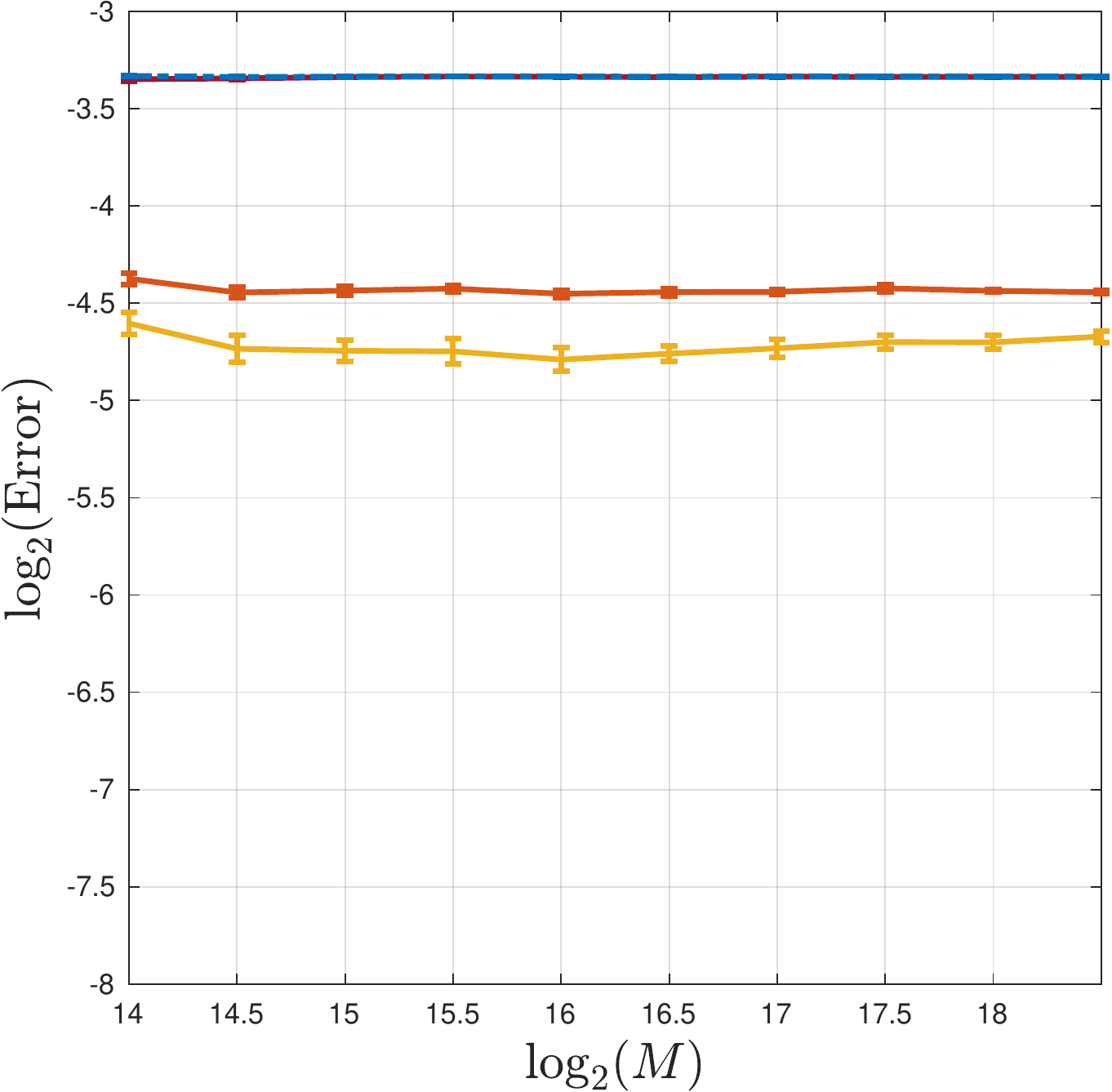}
		\caption{$f_3,\edit{\snr=9.0}, \eta=0.12$}
		\label{fig:sim_E_8_high}
	\end{subfigure}
	\begin{subfigure}[b]{0.32\textwidth}
		\centering
		\includegraphics[width=.85\textwidth]{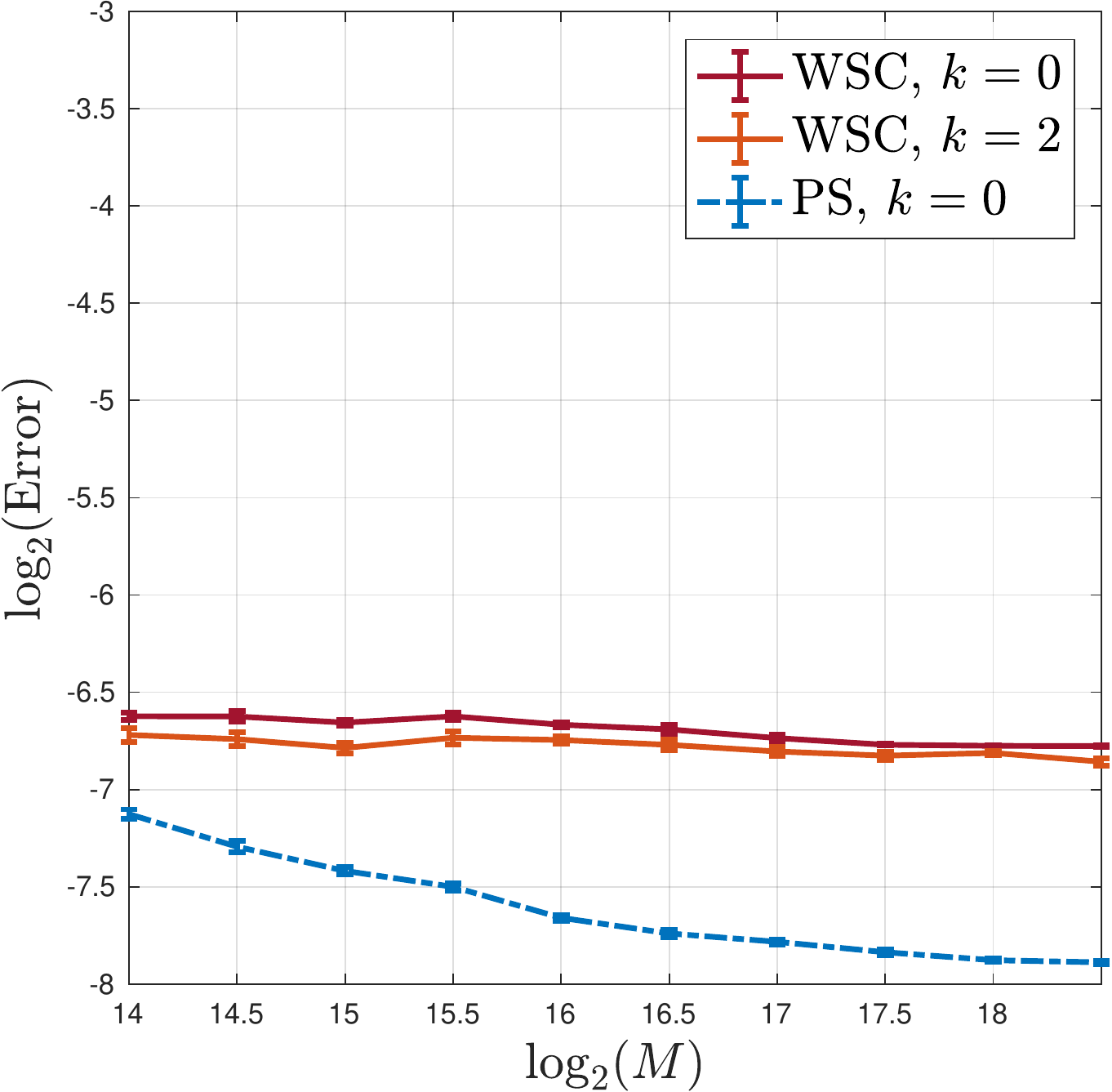}
		\caption{$f_1,\edit{\snr=2.2}, \eta=0.06$}
		\vspace*{.1cm}
		\label{fig:sim_E_3_low}
	\end{subfigure}
	\hfill
	\begin{subfigure}[b]{0.32\textwidth}
		\centering
		\includegraphics[width=.85\textwidth]{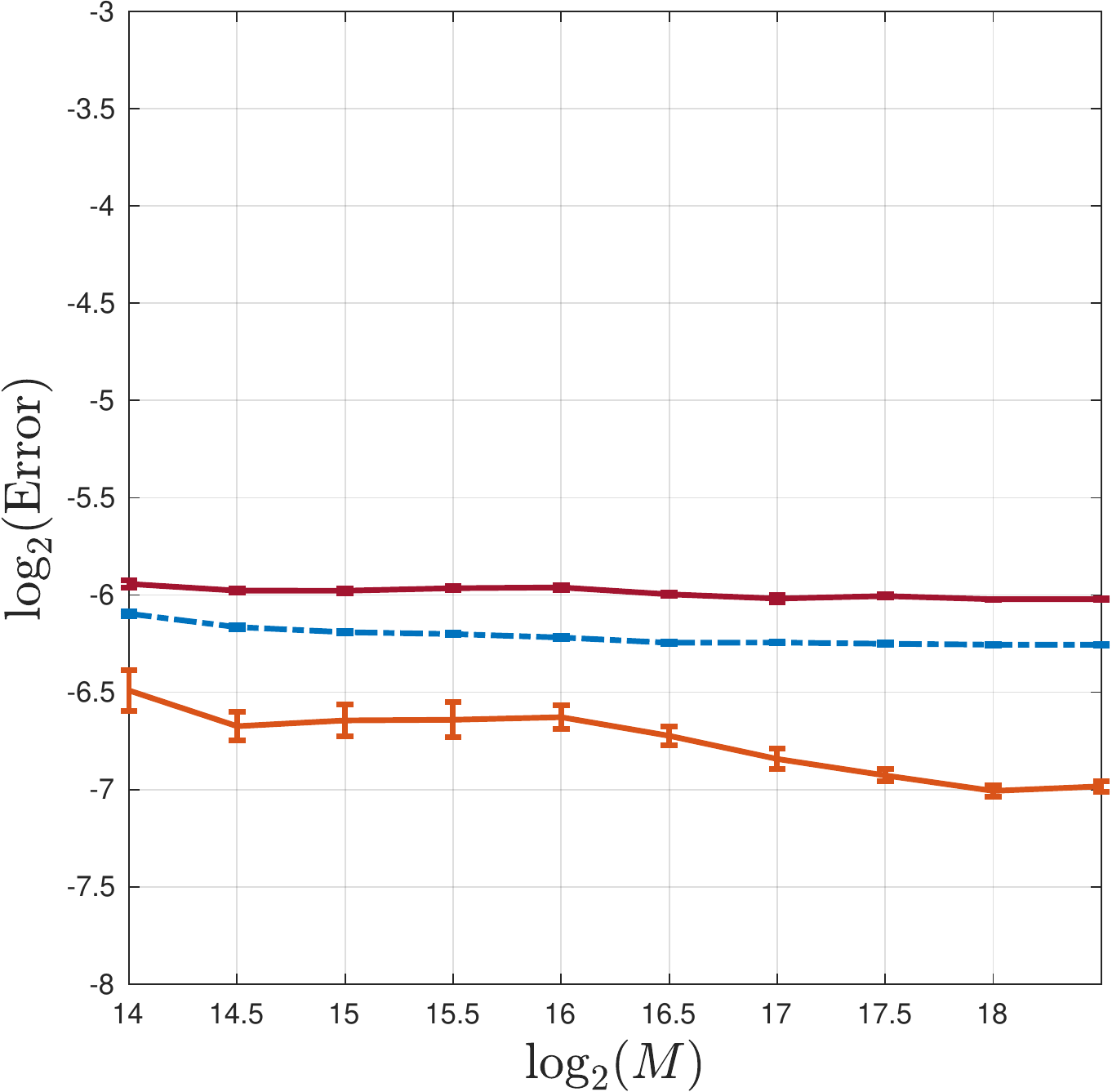}
		\caption{$f_2,\edit{\snr=2.2}, \eta=0.06$}
		\vspace*{.1cm}
		\label{fig:sim_E_3_med}
	\end{subfigure}
	\hfill
	\begin{subfigure}[b]{0.32\textwidth}
		\centering
		\includegraphics[width=.85\textwidth]{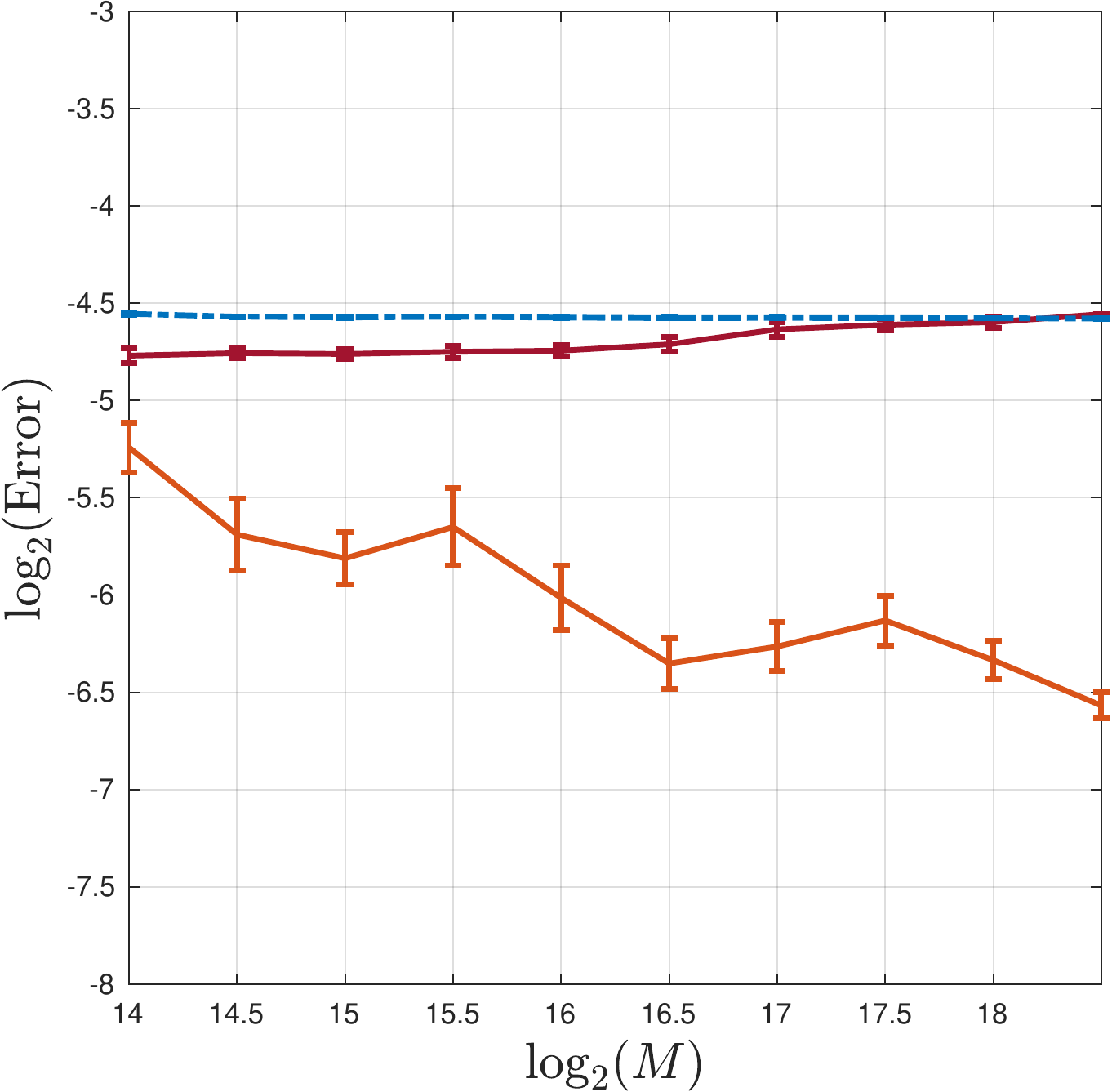}
		\caption{$f_3,\edit{\snr=2.2}, \eta=0.06$}
		\vspace*{.1cm}
		\label{fig:sim_E_3_high}
	\end{subfigure}
	\begin{subfigure}[b]{0.32\textwidth}
		\centering
		\includegraphics[width=.85\textwidth]{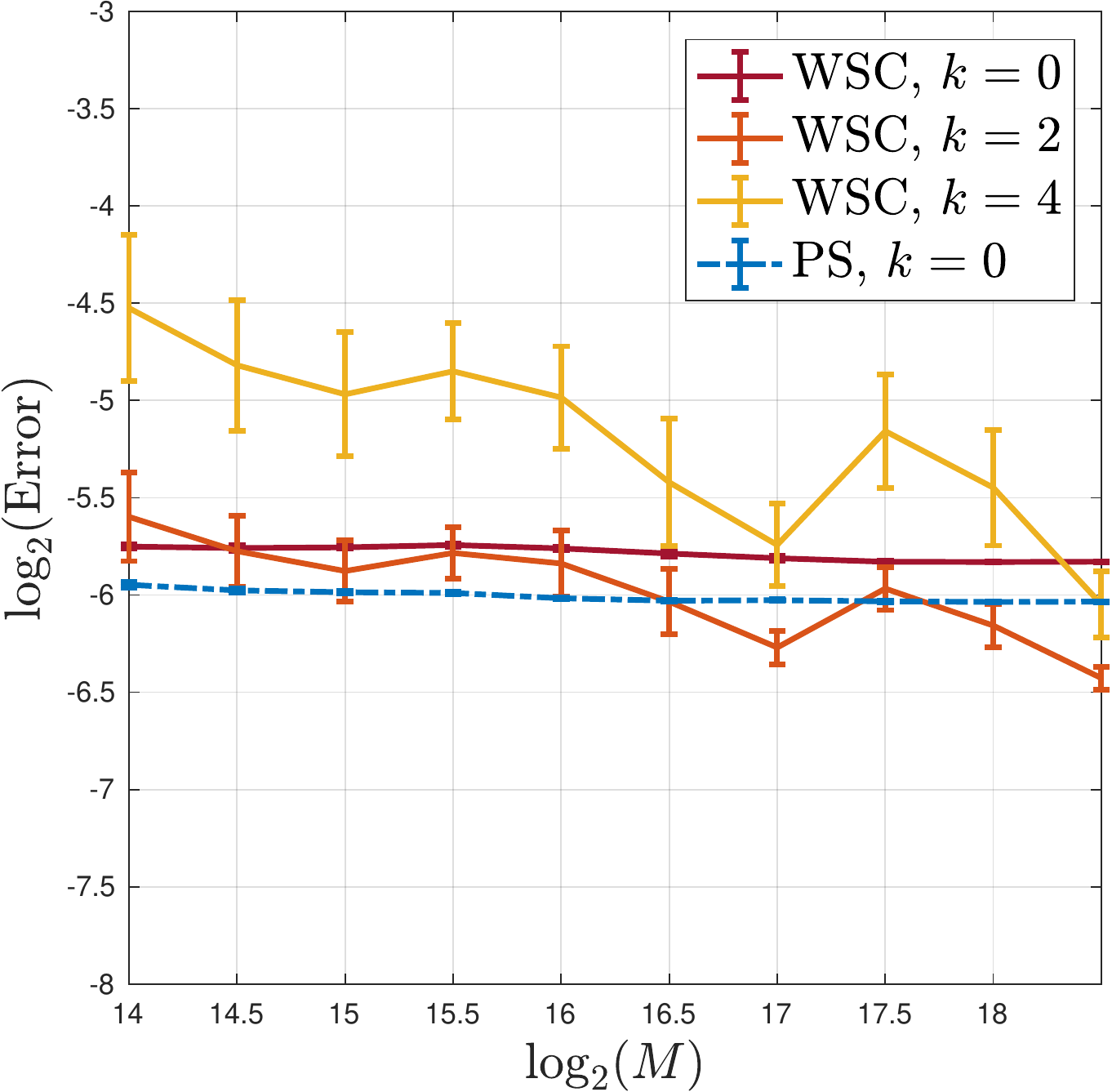}
		\caption{$f_1,\edit{\snr=2.2}, \eta=0.12$}
		\vspace*{.1cm}
		\label{fig:sim_E_4_low}
	\end{subfigure}
	\hfill
	\begin{subfigure}[b]{0.32\textwidth}
		\centering
		\includegraphics[width=.85\textwidth]{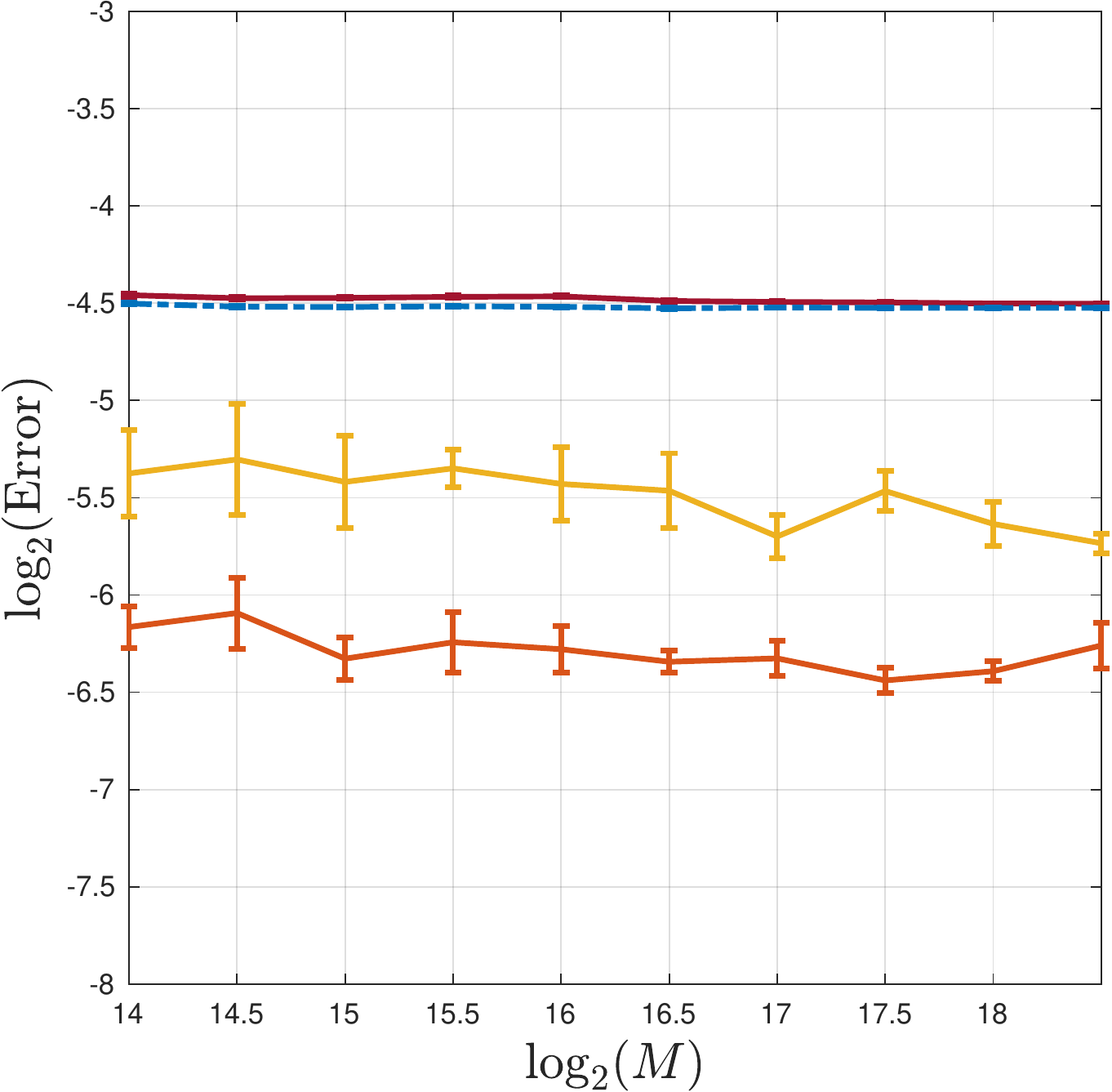}
		\caption{$f_2,\edit{\snr=2.2}, \eta=0.12$}
		\vspace*{.1cm}
		\label{fig:sim_E_4_med}
	\end{subfigure}
	\hfill
	\begin{subfigure}[b]{0.32\textwidth}
		\centering
		\includegraphics[width=.85\textwidth]{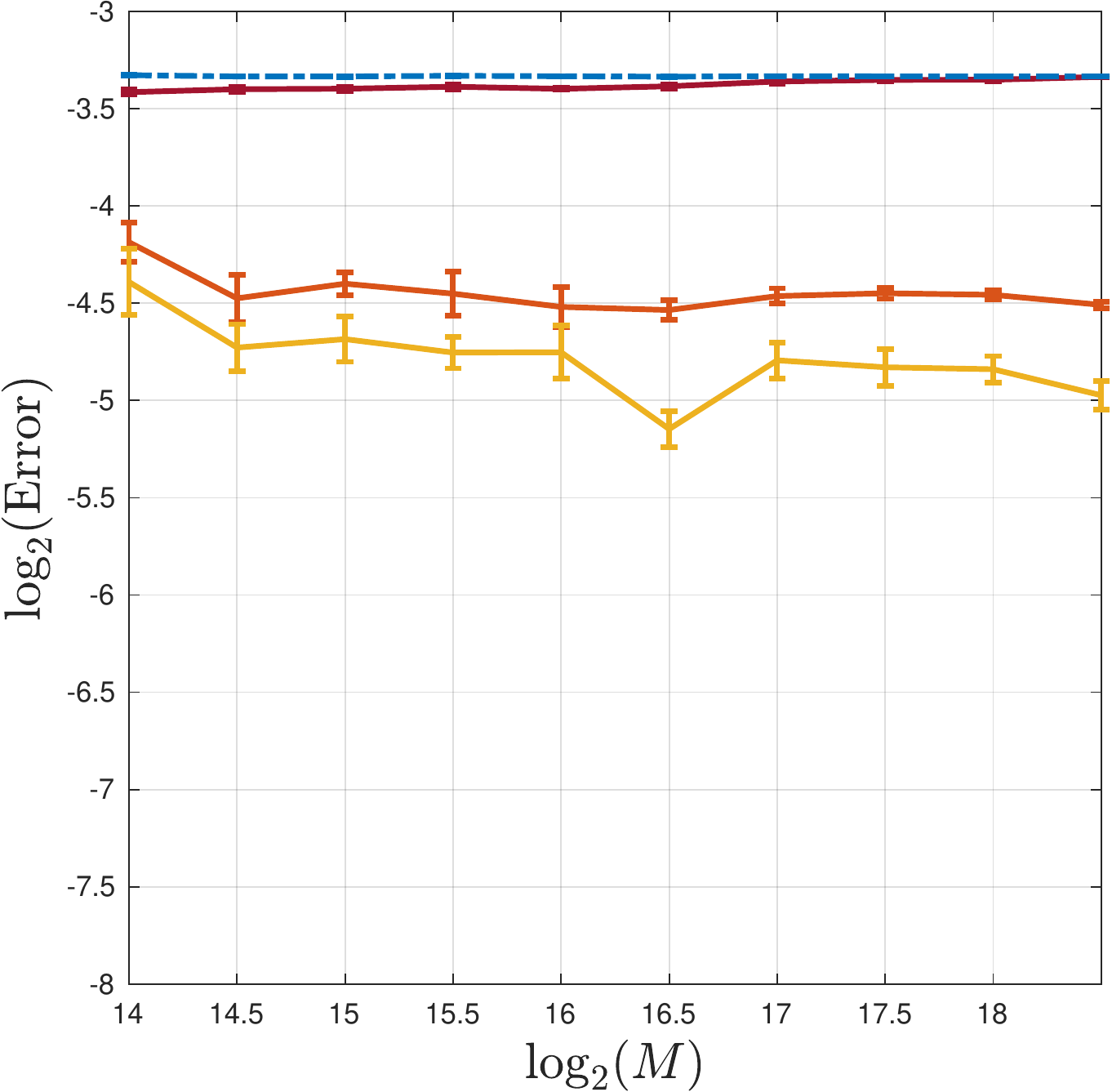}
		\caption{$f_3,\edit{\snr=2.2}, \eta=0.12$}
		\vspace*{.1cm}
		\label{fig:sim_E_4_high}
	\end{subfigure}
	\caption{$\Lb^2$ error with standard error bars for \edit{noisy dilation MRA} model ($t=0$, empirical moment estimation). First, second, third column shows results for low, medium, high frequency signals. All plots have the same axis limits.}
	\label{fig:GenMRAModelEmpirical}
\end{figure}

\begin{proof}[Proof of Proposition \ref{prop:emp_mom_est}]
	Since $y_j = L_{\tau_j}f+\varepsilon_j$, we have
	\begin{align*}
	\Ex[ \edit{\beta}_{m}(y_j)] &=\Ex \left[\int_0^{2^{\ell}\pi} \omega^m (\widehat{f}_{\tau_j}(\omega)+\widehat{\varepsilon}_j(\omega))\ d\omega \right]\\
	&= \Ex \left[\int_0^{2^{\ell}\pi} \omega^m \widehat{f}_{\tau_j}(\omega)\ d\omega\right] \\
	&= \Ex \left[\int_0^{2^{\ell}\pi} \omega^m \widehat{f}((1-\tau_j)\omega)\ d\omega\right] \\
	&= \Ex \left[\int_0^{2^{\ell}\pi(1-\tau_j)} \frac{\xi^m}{(1-\tau_j)^m} \widehat{f}(\xi)\ \frac{d\xi}{(1-\tau_j)}\right] \\
	&= \edit{\beta}_{m}(f) \Ex\left[(1-\tau_j)^{-(m+1)}\right].
	\end{align*} 
	We now compute the variance. We first establish that
	\begin{align*}
	g_m(\ell,\sigma)&=\Ex\left[\left(\int_0^{2^{\ell}\pi} \omega^m\widehat{\varepsilon}_{j}(\omega)\ d\omega\right)\left(\int_0^{2^{\ell}\pi} \omega^m\overline{\widehat{\varepsilon}_{j}(\omega)}\ d\omega\right) \right] \, .
	\end{align*}
	By Thm 4.5 of \cite{klebaner2012introduction}
	\begin{align*}
	\Ex\left[\widehat{\varepsilon}_{j}(\omega)\overline{\widehat{\varepsilon}_{j}(\xi)}\right] &= \Ex \left[ \left(\int_{-1/2}^{1/2} e^{-i\omega t}\ dB_t\right)\left(\int_{-1/2}^{1/2} e^{i\xi t}\ dB_t\right) \right] \\
	&= \sigma^2 \int_{-1/2}^{1/2} e^{i(\xi-\omega)t} \ dt \\
	&= \frac{2\sigma^2\sin(\frac{1}{2}(\xi-\omega))}{(\xi-\omega)}
	\end{align*}
	so that	
	\begin{align*}
	\Ex\left[\left(\int_0^{2^{\ell}\pi} \widehat{\varepsilon}_{j}(\omega)\ d\omega\right)\left(\int_0^{2^{\ell}\pi} \overline{\widehat{\varepsilon}_{j}(\omega)}\ d\omega\right) \right]
	&= \int_0^{2^{\ell}\pi} \int_0^{2^{\ell}\pi} \ \omega^m \xi^m \Ex\left[\widehat{\varepsilon}_{j}(\omega)\overline{\widehat{\varepsilon}_{j}(\xi)}\right]\ d\omega\ d\xi \\
	&= \int_0^{2^{\ell}\pi} \int_0^{2^{\ell}\pi} \omega^m \xi^m \frac{2\sigma^2\sin(\frac{1}{2}(\xi-\omega))}{(\xi-\omega)}\ d\omega\ d\xi \\
	&= g_m(\ell,\sigma)\, .
	\end{align*}
	We thus obtain:
	\begin{align*}
	\left[|\edit{\beta}_{m}(y_j)|^2\right] 
	&= \Ex \left[  \left(\int_0^{2^{\ell}\pi} \omega^m (\widehat{f}_{\tau_j}(\omega)+\widehat{\varepsilon}_j(\omega))\ d\omega\right)\left(\int_0^{2^{\ell}\pi} \omega^m (\overline{\widehat{f}_{\tau_j}(\omega)}+\overline{\widehat{\varepsilon}_j(\omega)})\ d\omega\right) \right] \\
	&=\Ex\left[ \left(\int_0^{2^{\ell}\pi} \omega^m\widehat{f}((1-\tau_j)\omega)\ d\omega\right)\left(\int_0^{2^{\ell}\pi} \omega^m\overline{\widehat{f}((1-\tau_j)\omega)}\ d\omega\right) \right. \\
	&\qquad \left. + \left(\int_0^{2^{\ell}\pi} \omega^m\widehat{\varepsilon}_{j}(\omega)\ d\omega\right)\left(\int_0^{2^{\ell}\pi} \omega^m\overline{\widehat{\varepsilon}_{j}(\omega)}\ d\omega\right) \right] \\
	&= \Ex\left[(1-\tau_j)^{-2(m+1)}\edit{\beta}_{m}(f)\overline{\edit{\beta}_{m}(f)} \right] +g_m(\ell,\sigma)\\
	&= |\edit{\beta}_{m}(f)|^2\,  \Ex\left[(1-\tau_j)^{-2(m+1)} \right] + g_m(\ell,\sigma) \, .
	\end{align*}
	Thus:
	\begin{align*}
	\Var[\edit{\beta}_{m}(y_j)] - g_m(\ell,\sigma)
	&= \Ex\left[|\edit{\beta}_{m}(y_j)|^2\right] - g_m(\ell,\sigma) - |\Ex\left[\edit{\beta}_{m}(y_j)\right]|^2 \\
	&= |\edit{\beta}_{m}(f)|^2  \Ex\left[(1-\tau_j)^{-2(m+1)} \right] -|\edit{\beta}_{m}(f)|^2 \left(\Ex\left[(1-\tau_j)^{-(m+1)}\right]\right)^2\, .
	\end{align*}
	Dividing by $|\Ex\left[\edit{\beta}_{m}(y_j)\right]|^2$ gives:
	\begin{align*}
	CV_m &= \frac{ \Ex[(1-\tau_j)^{-2(m+1)} ] }{(\Ex[(1-\tau_j)^{-(m+1)}])^2} - 1\, ,
	\end{align*}
	and the remainder of the proof is identical to the proof of Proposition \ref{prop:emp_mom_est_dilMRA}.
\end{proof}

\section{Additional simulations for \edit{noisy dilation MRA}}
\label{sec:additional_sim_results}

We investiagte the $\Lb^2$ error of estimating the power spectrum using PS ($k=0$) and WSC ($k=0,2,4$) for three additional high frequency functions:
\begin{align*}
f_4(x) &= \edit{1.175}\cos(32x)\cdot \ind(x\in[-0.2,0.2])  \\ 
f_5(x) &=  \edit{0.299}\exp^{-0.04x^2}\cos(30x+1.5x^2) \\
f_6(x) &=  (\edit{2.304}/\pi)\cos(35x)\text{sinc}(3x)  \, .
\end{align*}
\edit{The multiplicative constants were chosen so that the $\Lb^2$ norms of $f_4, f_5, f_6$ are comparable with the $\Lb^2$ norms of the Gabor signals $f_1, f_2, f_3$ defined in Section \ref{sec:DilationNoiseSims}.} The signal $f_4$ is not continuous and has compact support, with a slowly decaying, oscillating Fourier transform given by $\widehat{f}_4(\omega)/\edit{0.47} =$ $\text{sinc}\left(0.2(\omega-32)\right)+\text{sinc}\left(0.2(-\omega-32)\right)$. The signal $f_5$ is a linear chirp with a constantly varying instantaneous frequency. The signal $f_6$ is slowly decaying in space, with a discontinuous Fourier transform of compact support given by $ \widehat{f}_6(\omega)/\edit{0.384} = \ind(\omega \in [-38, -32])+\ind(\omega \in [32, 38])$.  

Implementation details were as described in Section \ref{sec:num_imp}, and simulations were run with oracle moment estimation on the full model (parameter values as described in Section \ref{sec:GenMRANoiseSim}). \edit{Figure \ref{fig:GenMRAMoreExsOracle} shows the $\Lb^2$ error.}
As for the high frequency Gabor in Section \ref{sec:GenMRANoiseSim}, WSC ($k=2$) and WSC ($k=4$) significantly outperformed the zero order estimators. In addition for large dilations, the WSC ($k=4$) outperformed WSC ($k=2$) on $f_4$ and $f_6$. 

\begin{figure}
	\centering
	\begin{subfigure}[b]{0.32\textwidth}
		\centering
		\includegraphics[width=.85\textwidth]{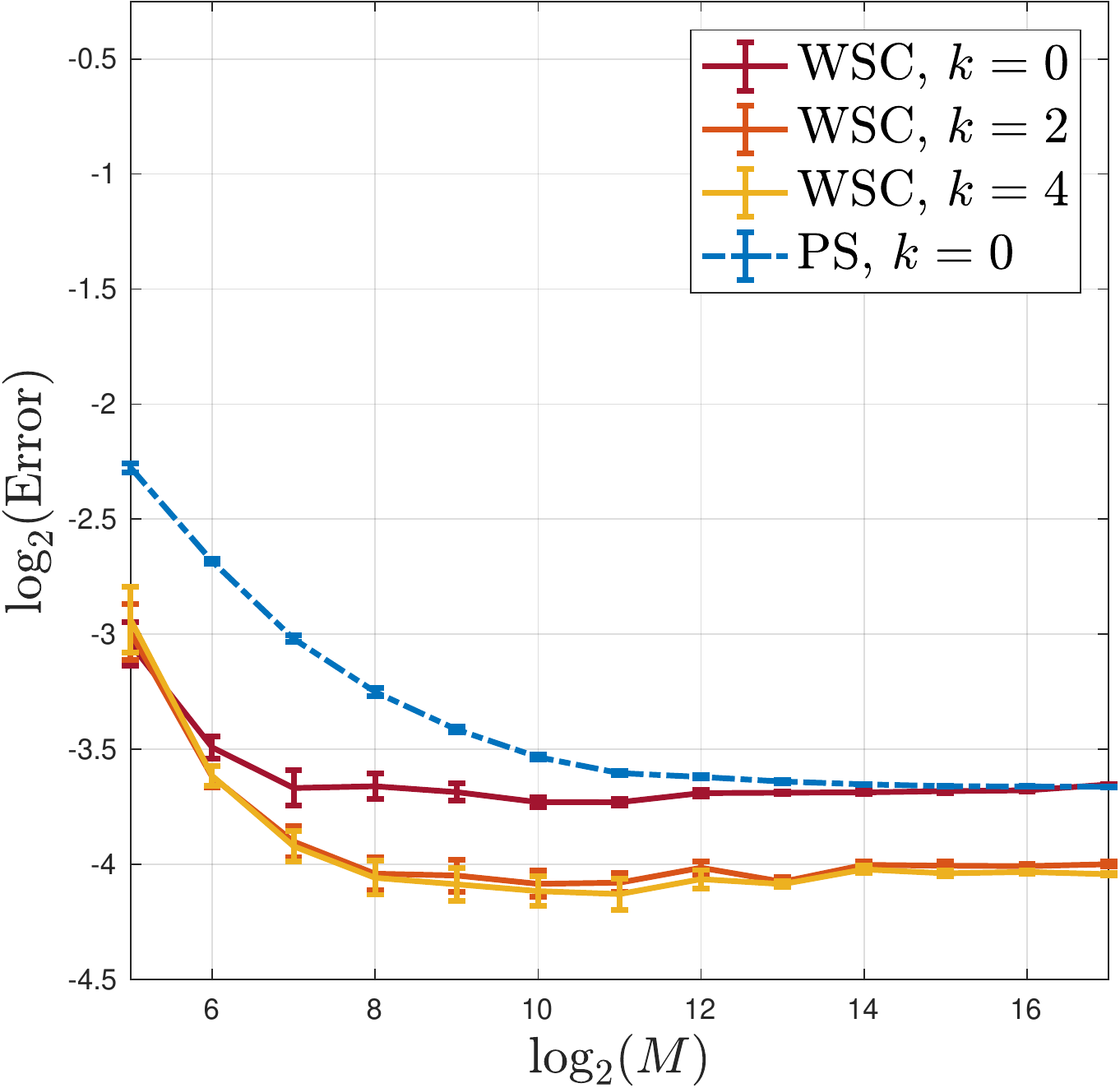}
		\caption{$f_4,\edit{\snr=2.2}, \eta=0.06$}
		\vspace*{.1cm}
		\label{fig:sim_O_3_sinc}
	\end{subfigure}
	\hfill
	\begin{subfigure}[b]{0.32\textwidth}
		\centering
		\includegraphics[width=.85\textwidth]{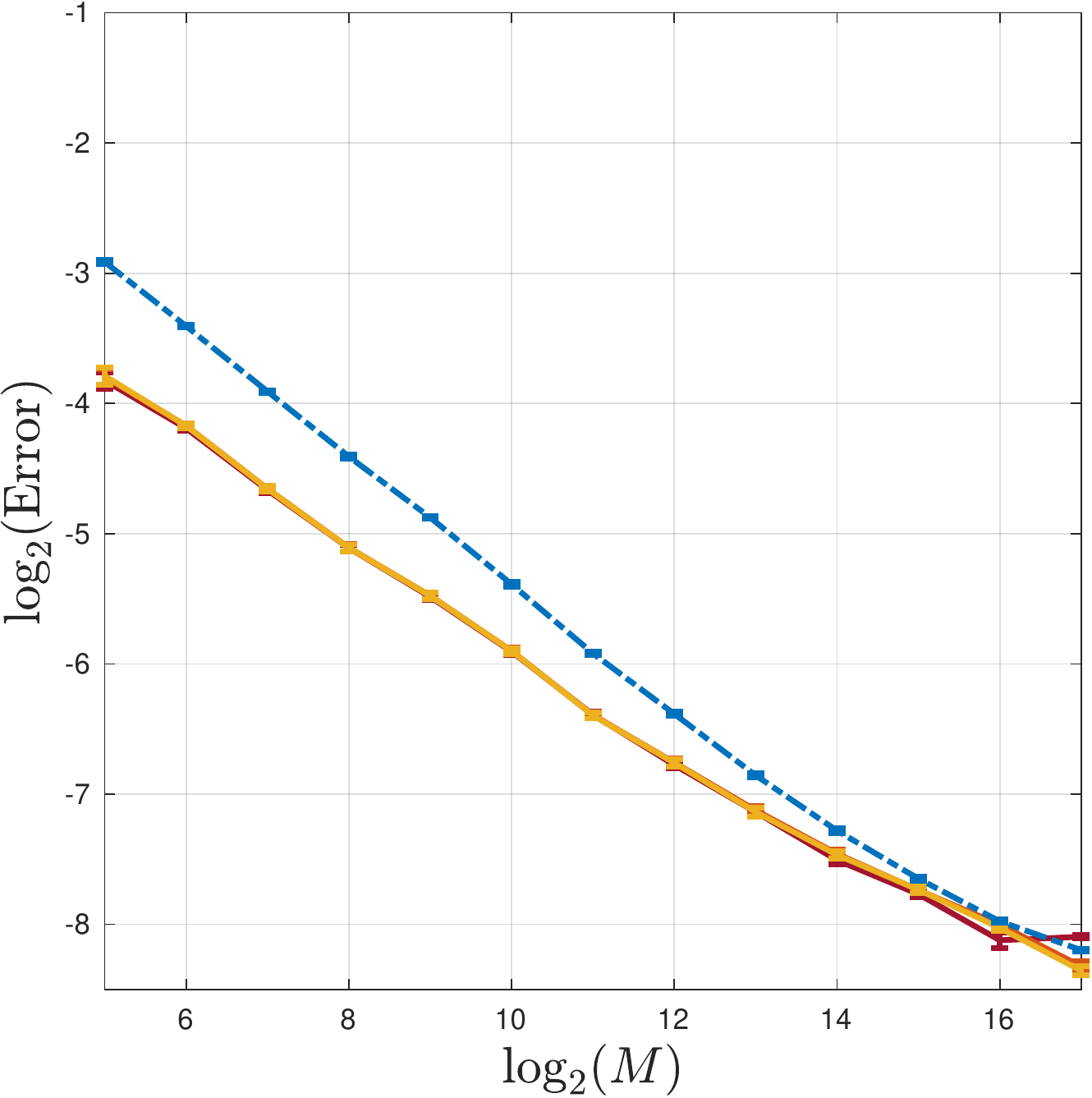}
		\caption{$f_5,\edit{\snr=2.2}, \eta=0.06$}
		\vspace*{.1cm}
		\label{fig:sim_O_3_chirp}
	\end{subfigure}
	\hfill
	\begin{subfigure}[b]{0.32\textwidth}
		\centering
		\includegraphics[width=.85\textwidth]{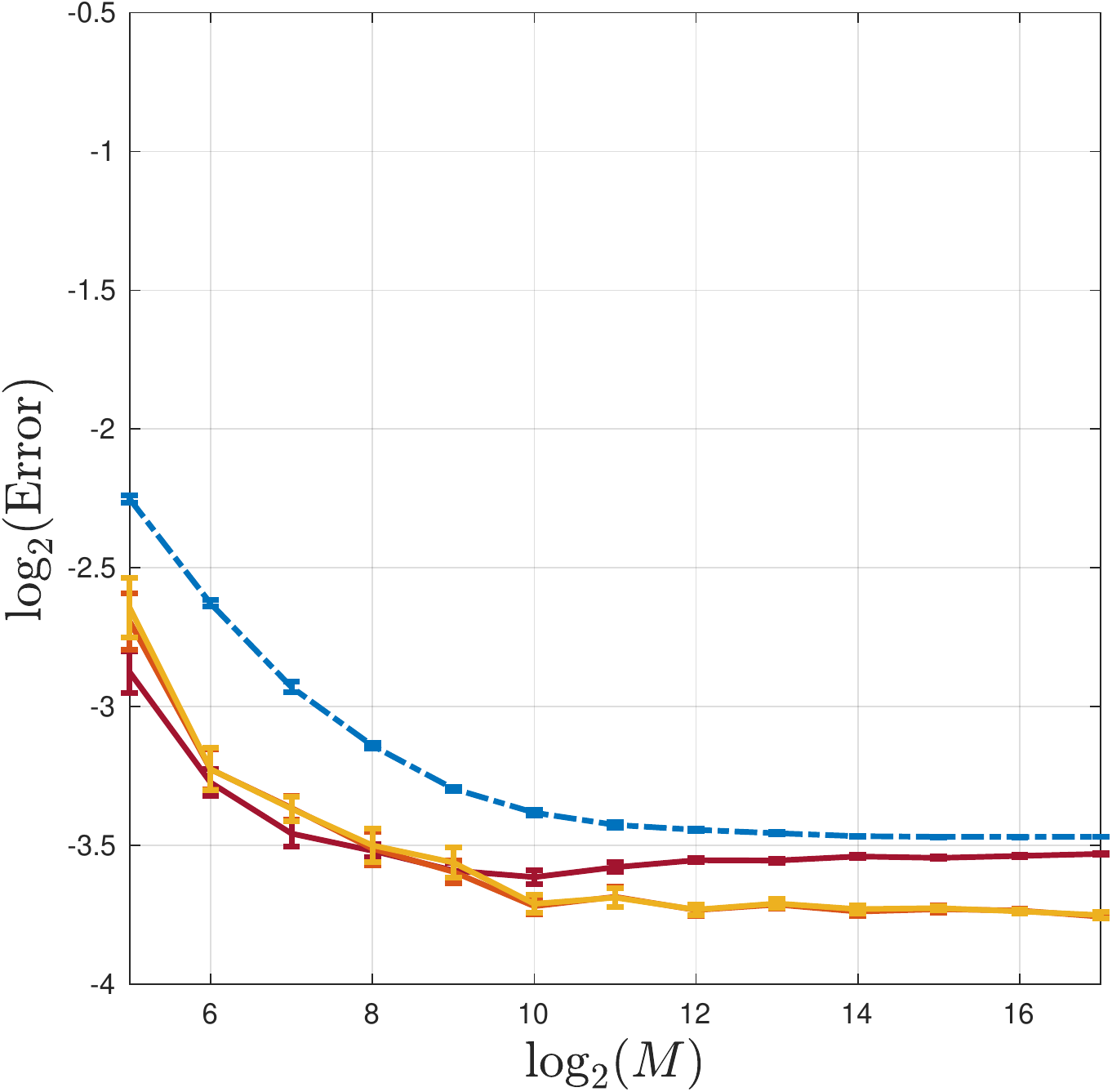}
		\caption{$f_6,\edit{\snr=2.2}, \eta=0.06$}
		\vspace*{.1cm}
		\label{fig:sim_O_3_step}
	\end{subfigure}
	\begin{subfigure}[b]{0.32\textwidth}
		\centering
		\includegraphics[width=.85\textwidth]{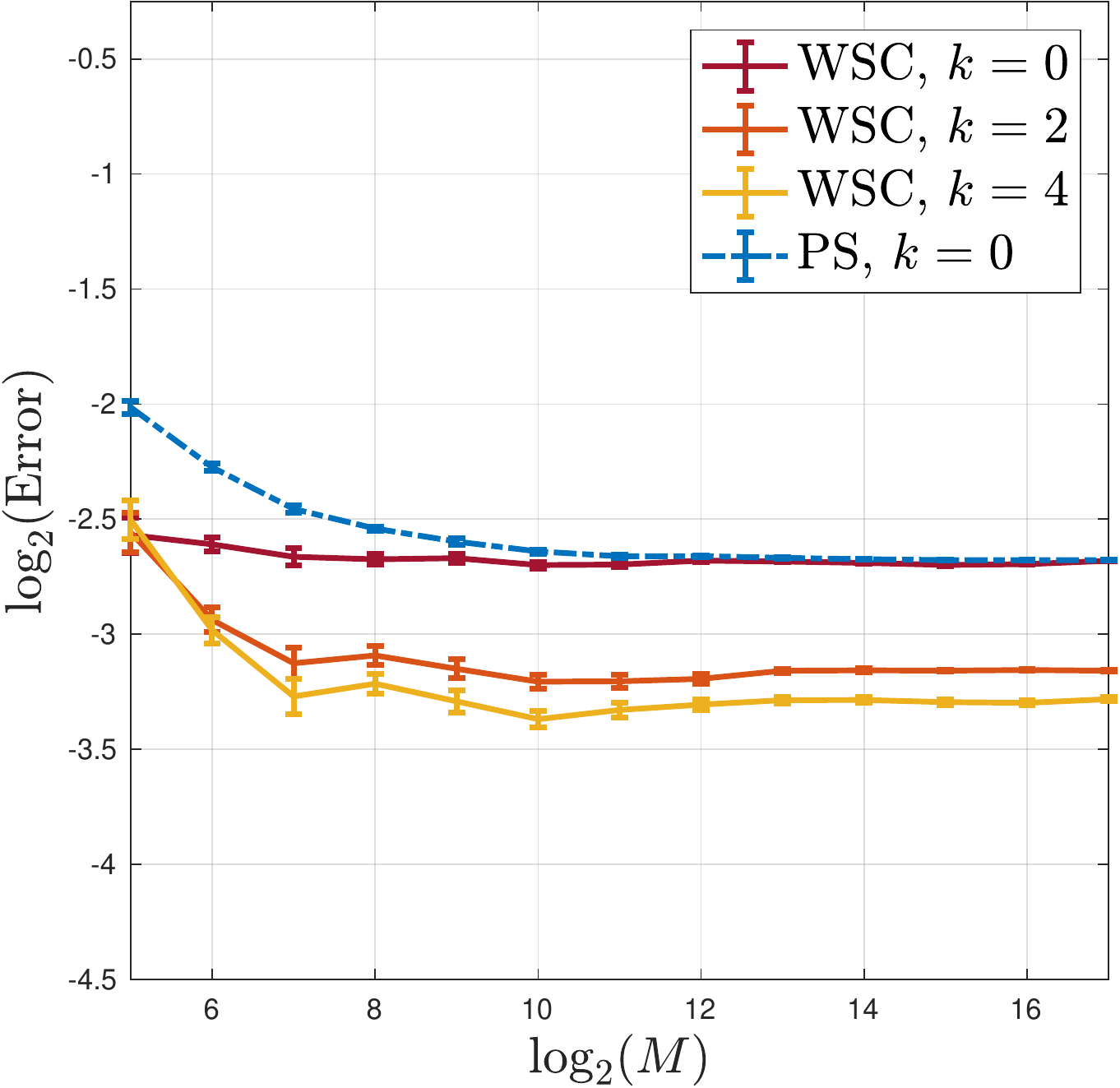}
		\caption{$f_4,\edit{\snr=2.2}, \eta=0.12$}
		\vspace*{.1cm}
		\label{fig:sim_O_4_sinc}
	\end{subfigure}
	\hfill
	\begin{subfigure}[b]{0.32\textwidth}
		\centering
		\includegraphics[width=.85\textwidth]{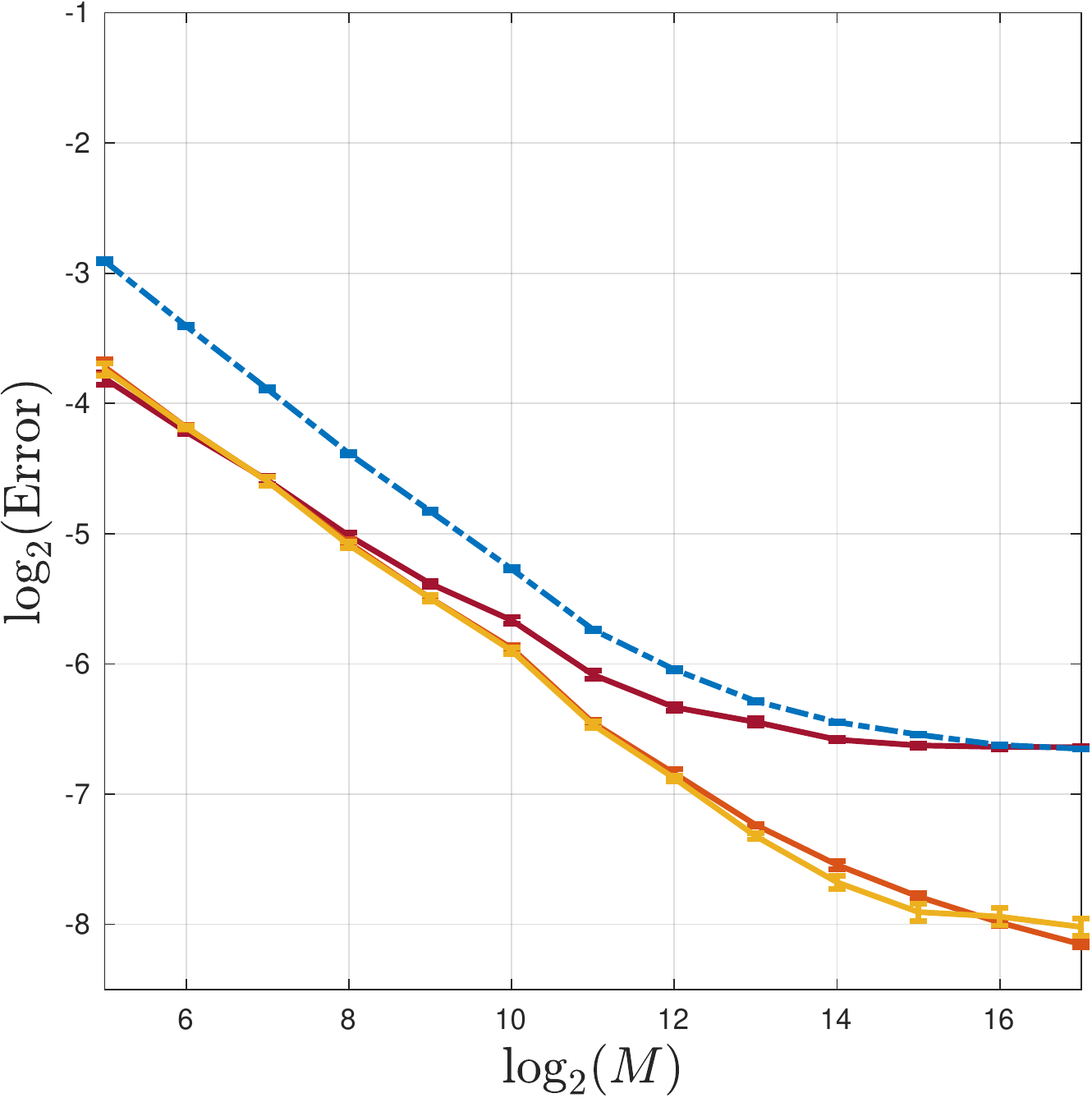}
		\caption{$f_5,\edit{\snr=2.2}, \eta=0.12$}
		\vspace*{.1cm}
		\label{fig:sim_O_4_chirp}
	\end{subfigure}
	\hfill
	\begin{subfigure}[b]{0.32\textwidth}
		\centering
		\includegraphics[width=.85\textwidth]{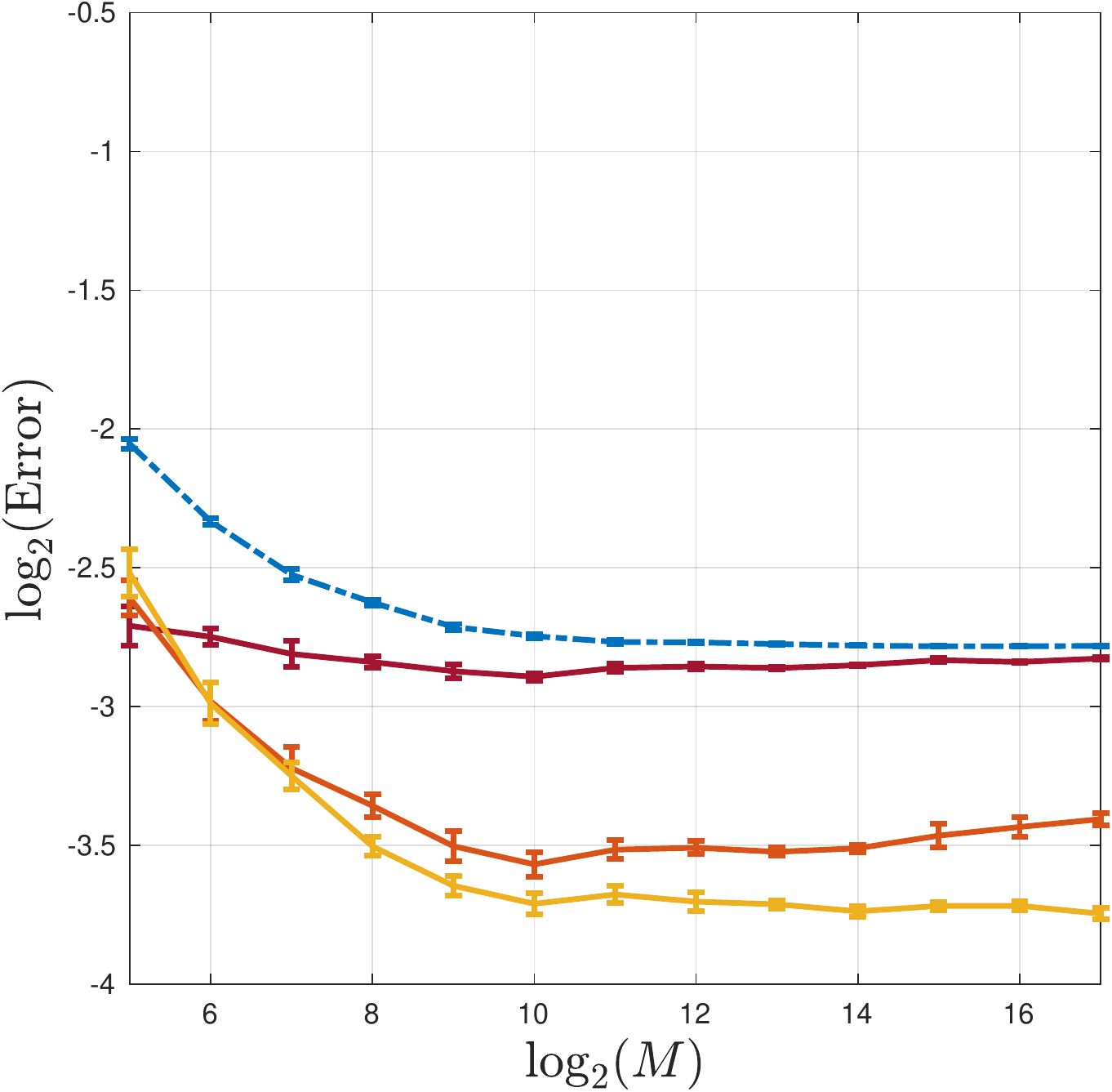}
		\caption{$f_6,\edit{\snr=2.2}, \eta=0.12$}
		\vspace*{.1cm}
		\label{fig:sim_O_4_step}
	\end{subfigure}
	\begin{subfigure}[b]{0.32\textwidth}
		\centering
		\includegraphics[width=.85\textwidth]{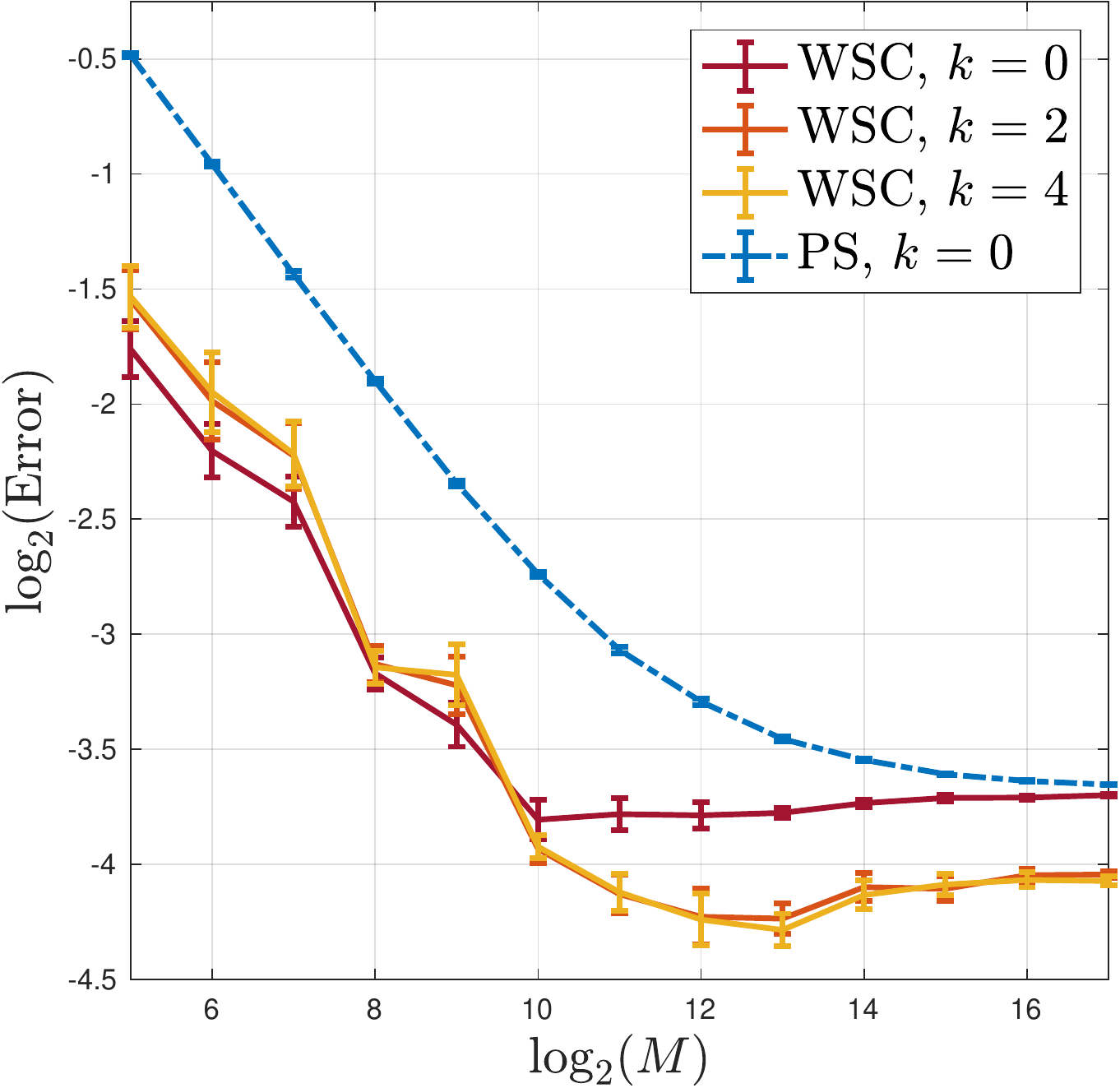}
		\caption{$f_4,\edit{\snr=0.56}, \eta=0.06$}
		\vspace*{.1cm}
		\label{fig:sim_O_5_sinc}
	\end{subfigure}
	\hfill
	\begin{subfigure}[b]{0.32\textwidth}
		\centering
		\includegraphics[width=.85\textwidth]{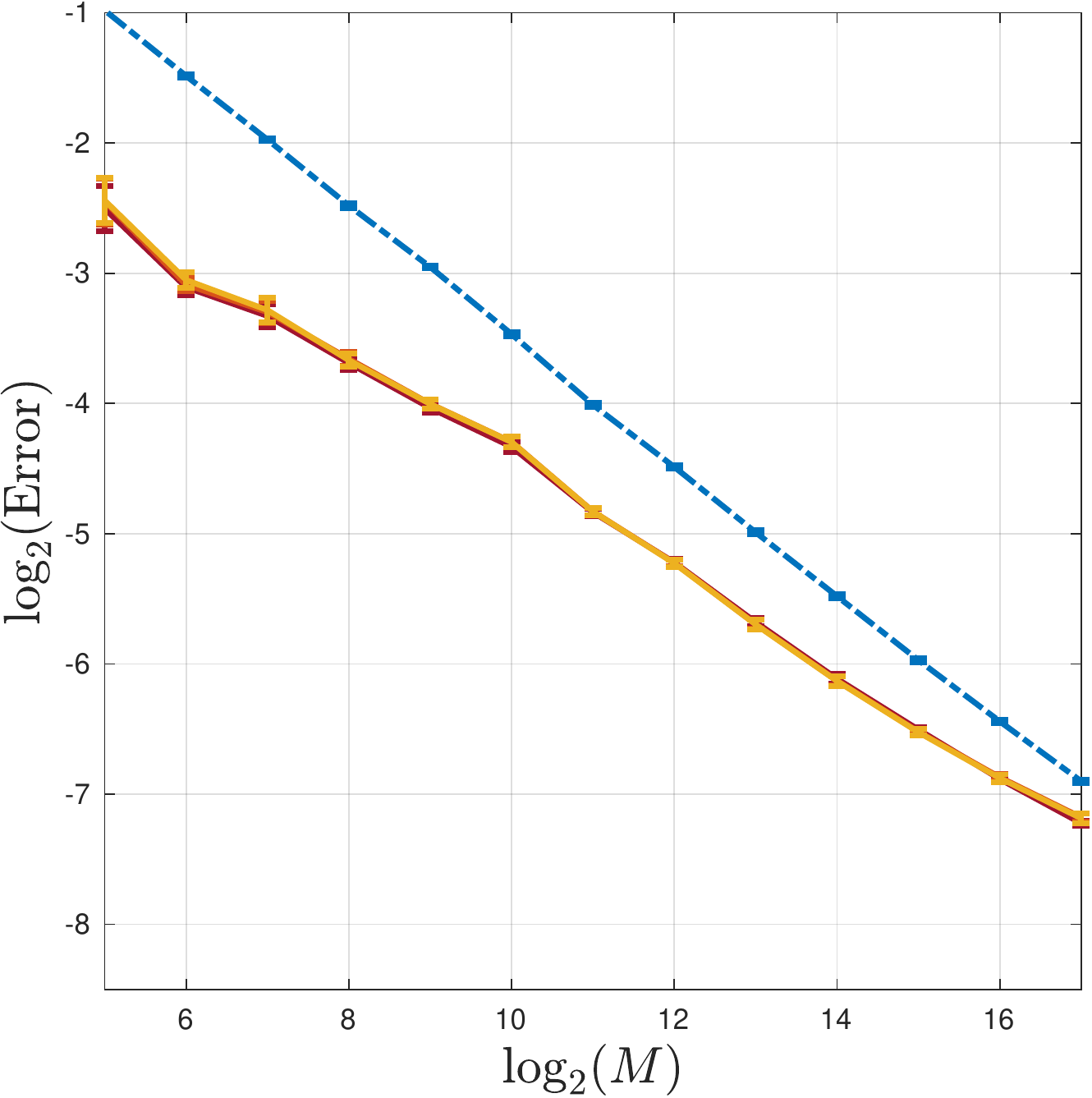}
		\caption{$f_5,\edit{\snr=0.56}, \eta=0.06$}
		\vspace*{.1cm}
		\label{fig:sim_O_5_chirp}
	\end{subfigure}
	\hfill
	\begin{subfigure}[b]{0.32\textwidth}
		\centering
		\includegraphics[width=.85\textwidth]{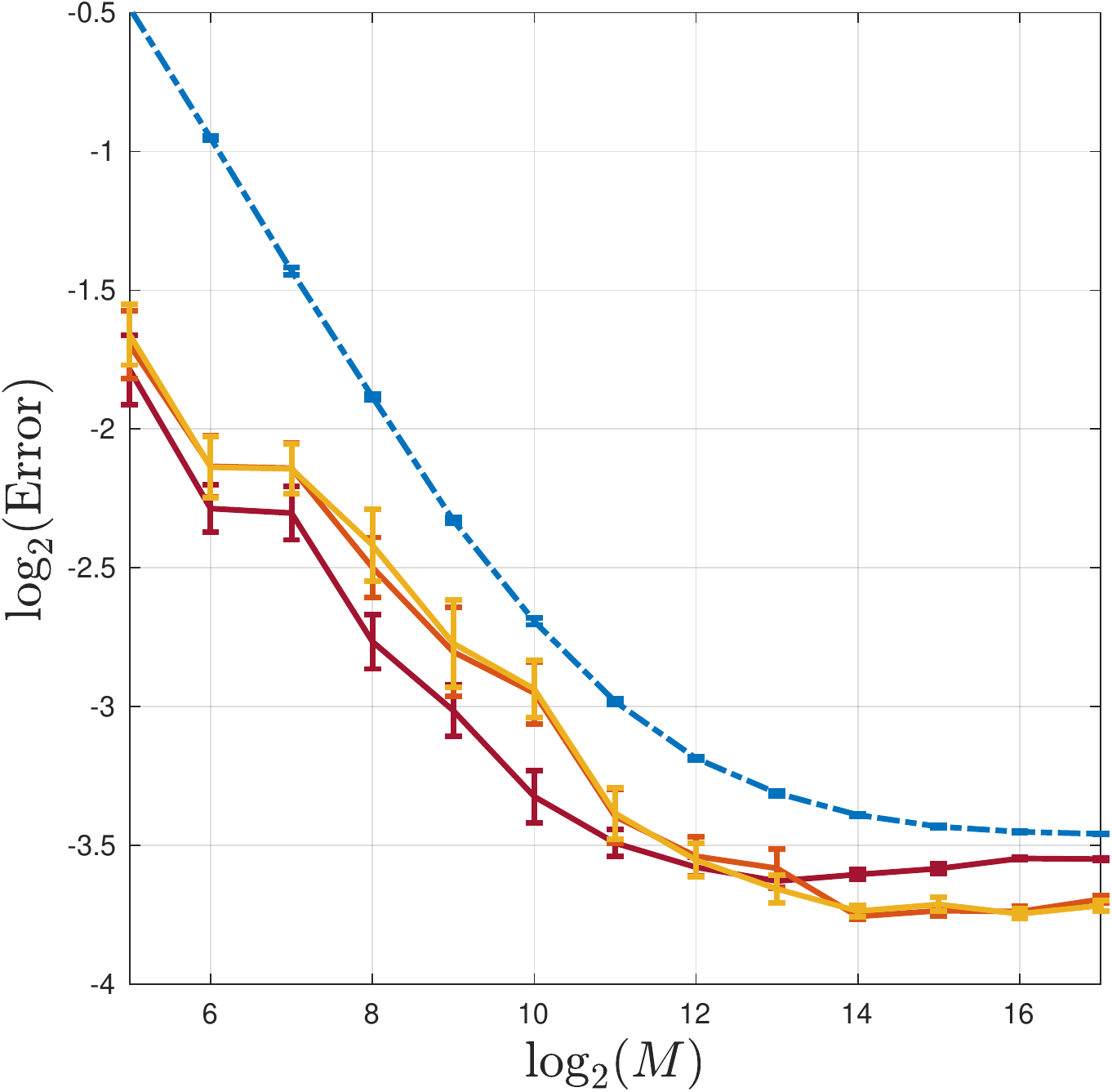}
		\caption{$f_6,\edit{\snr=0.56}, \eta=0.06$}
		\vspace*{.1cm}
		\label{fig:sim_O_5_step}
	\end{subfigure}
	\begin{subfigure}[b]{0.32\textwidth}
		\centering
		\includegraphics[width=.85\textwidth]{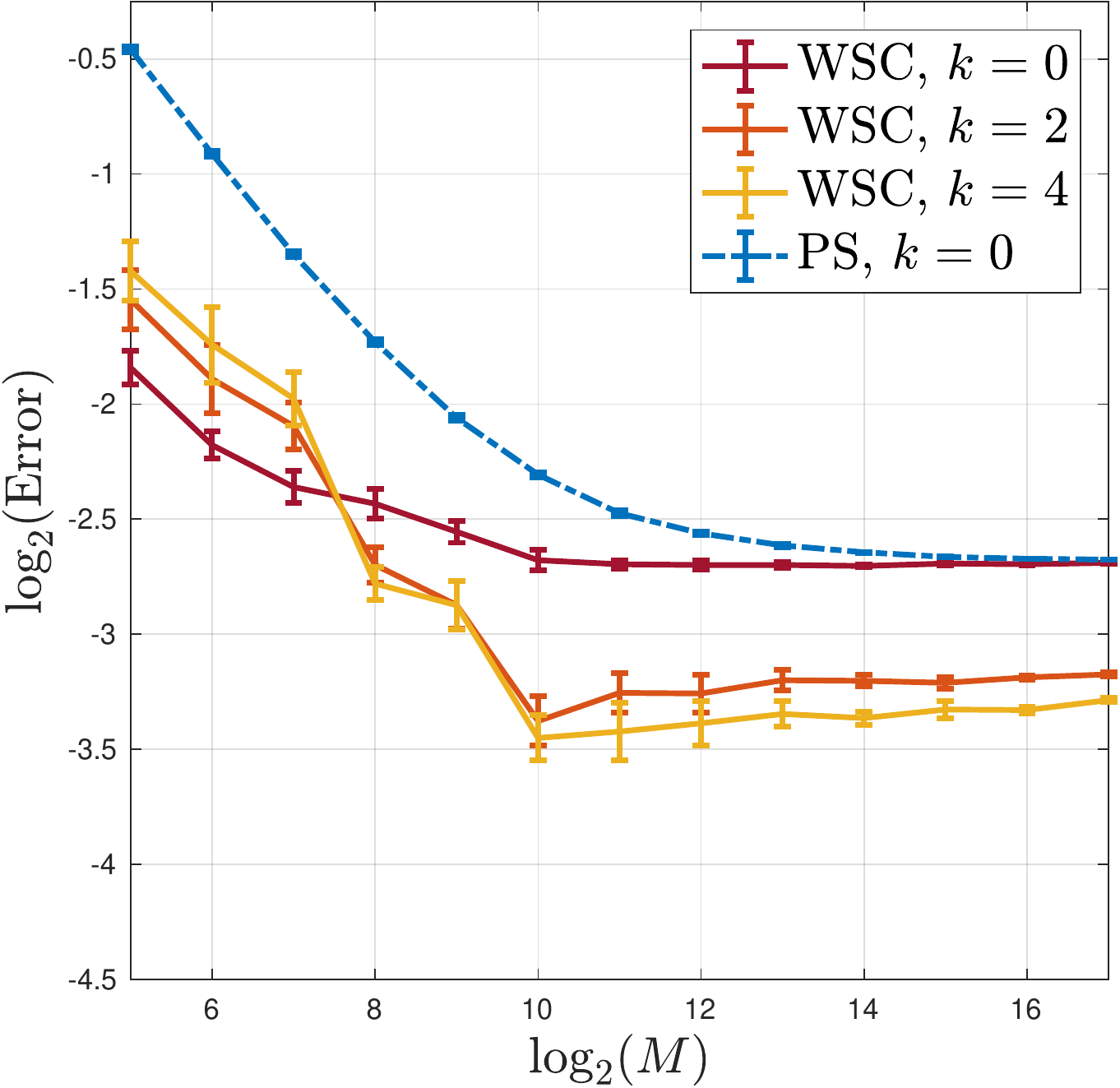}
		\caption{$f_4,\edit{\snr=0.56}, \eta=0.12$}
		\label{fig:sim_O_6_sinc}
	\end{subfigure}
	\hfill
	\begin{subfigure}[b]{0.32\textwidth}
		\centering
		\includegraphics[width=.85\textwidth]{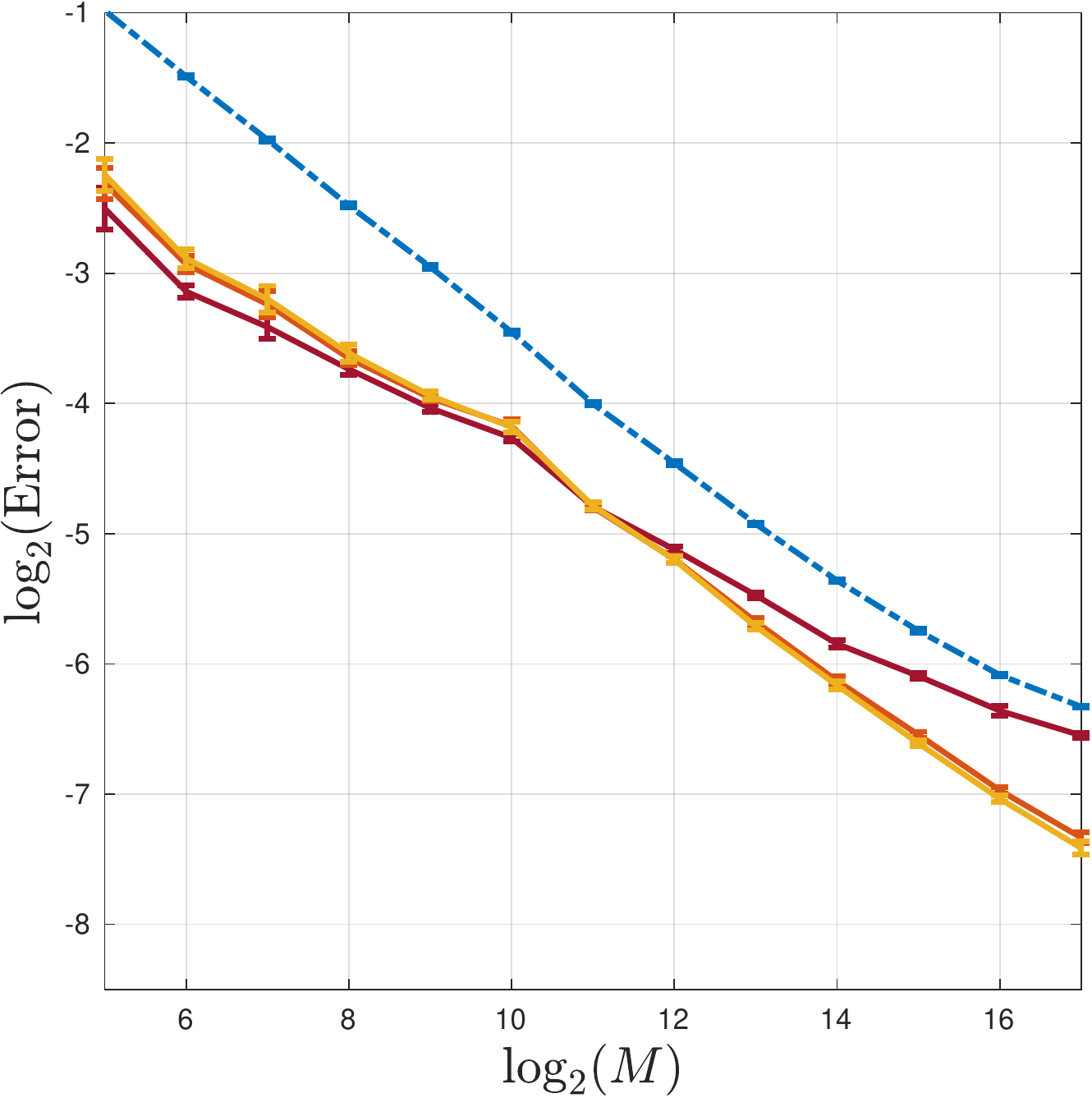}
		\caption{$f_5,\edit{\snr=0.56}, \eta=0.12$}
		\label{fig:sim_O_6_chirp}
	\end{subfigure}
	\hfill
	\begin{subfigure}[b]{0.32\textwidth}
		\centering
		\includegraphics[width=.85\textwidth]{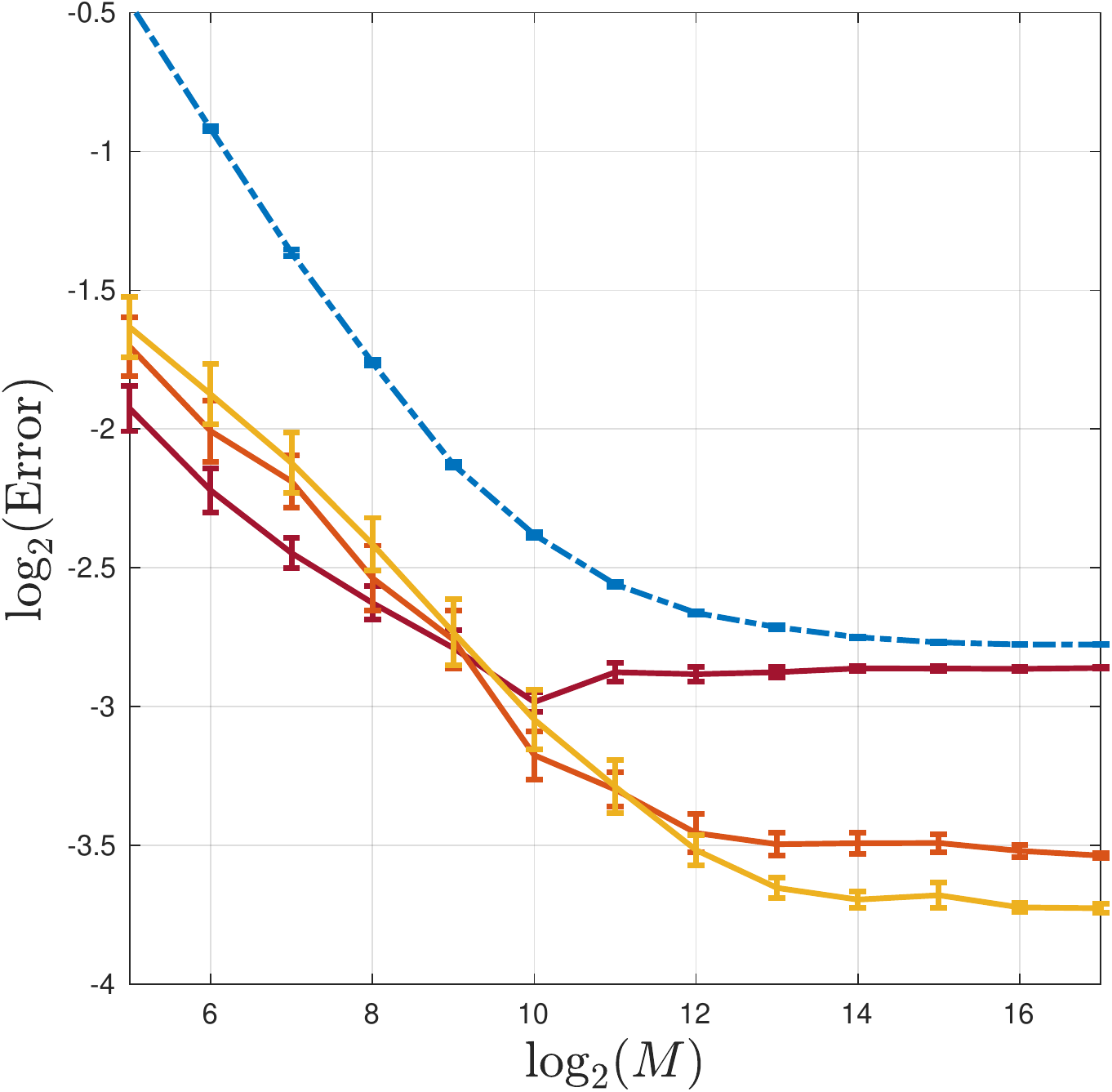}
		\caption{$f_6,\edit{\snr=0.56}, \eta=0.12$}
		\label{fig:sim_O_6_step}
	\end{subfigure}
	\caption{$\Lb^2$ error with standard error bars for \edit{noisy dilation MRA} model (oracle moment estimation). First, second, third column shows results for $f_4$, $f_5$, $f_6$.  All plots for the same signal have the same axis limits.}
	\label{fig:GenMRAMoreExsOracle}
\end{figure}

\edit{\section{Expectation maximization algorithm for noisy dilation MRA}}
\label{app:EMalg}

\edit{In this appendix we discuss how the expectation-maximization (EM) algorithm proposed in \cite{abbe2018multireference} can be extended to solve noisy dilation MRA. We first summarize the EM framework, which differentiates between observed data $y = \{y_j\}_{j=1}^M$, latent variables $s = \{s_j\}_{j=1}^M$, and model parameters $x$.  The goal is to produce the $x$ which 
	maximizes the marginalized likelihood function
	\begin{align*}
	p(y \big| x ) &= \int p(y,s \big| x) \ ds \, .
	\end{align*}
	Maximizing $p(y \big| x )$ directly is generally not tenable because enumerating the various values for $s$ is too costly. However EM algorithms can be used to find local maxima of the above function, by iterating between estimating the conditional distribution of latent variables given the current estimate of parameters (E-step) and estimating parameters given the current estimate of the conditional distribution of latent variables (M-step). Specifically the iterative procedure updates $x^{k}$, the current estimate of $x$, by:
	\begin{align}
	\label{equ:genQstep}
	Q(x \big| x^k) &= \mathbb{E}_{s | y, x^k} \left[ \log p(y,s \big| x) \right] \quad &\text{E-step} \\
	\label{equ:genMstep}
	x^{k+1} &= \text{arg}\,\max\limits_{x} Q(x \big| x^k)  \quad &\text{M-step}
	\end{align}
	Since (under certain conditions) $\log p(y \big| x)$ improves at least as much as $Q$ at each iteration \cite{dempster1977maximum}, the algorithm converges to a local maximum of $p(y \big| x)$. \\
	\indent This framework can be applied to noisy dilation MRA, and explicit formulas for both the E-step and M-step can be derived. Assume for simplicity that signals have been discretized to have length $n$, and that the translation distribution $\rho_t$ and dilation distribution $\rho_\tau$ are unkown and also discrete with $n$ possible values $\{t^{\ell}\}_{\ell=1}^n$, $\{\tau^{q}\}_{q=1}^n$ respectively. 
	Letting $x=(f, \rho_t, \rho_\tau)$ denote the parameters, $s_j=(t_j,\tau_j)$ denote the latent/nuisance variables, and $p_x$ denote conditioning on $x$, the likelihood function has form:
	\begin{align*}
	p(y,s \big| x ) 
	&=p_x(y \big| s )p_x(s) 
	=\prod_{j=1}^M \frac{1}{(2\pi\sigma^2)^{\frac{n}{2}}} \exp\left(-\frac{1}{2\sigma^2} \norm{ L_{\tau_j}T_{t_j}f - y_j }_2^2 \right) \rho_t(t_j)\rho_\tau(\tau_j) \, .
	\end{align*}
	Thus (up to a constant) the log likelihood has form
	\begin{align}
	\label{equ:log_likelihood}
	\log p(y,s \big| x ) &= \sum_{j=1}^M-\frac{1}{2\sigma^2} \norm{ L_{\tau_j}T_{t_j}f - y_j }_2^2 + \sum_{j=1}^M \log \rho_t(t_j) +\sum_{j=1}^M \log \rho_\tau(\tau_j) \, .
	\end{align}
	Given the current estimate $x^k =(f^k, \rho^k_t,\rho^k_\tau)$ of parameters, the E-step is performed by first computing the conditional distribution of the latent variables: 
	\begin{align}
	\label{equ:Estep}
	w_k^{\ell,q,j} &= \Prob\left(t_j = t^\ell, \tau_j=\tau^q \big| x^k\right)= C_k^j \exp\left(-\frac{1}{2\sigma^2} \norm{ L_{\tau_j}T_{t_j}f^k - y_j }_2^2 \right)\rho^k_t(t^\ell)\rho^k_\tau(\tau^q) \, ,
	\end{align}
	where $C_k^j$ is a normalizing constant so that $\sum_{\ell, q} w_k^{\ell,q,j} = 1$. These weights are then used to compute $Q$, that is, by combining (\ref{equ:genQstep}), (\ref{equ:log_likelihood}), and (\ref{equ:Estep}):
	\begin{align}
	\label{equ:Qdef}
	Q(f,\rho_t, \rho_\tau \big| f^k,\rho^k_t, \rho^k_\tau) 
	&= \sum_{j=1}^M \sum_{\ell=1}^n \sum_{q=1}^n w_k^{\ell,q,j}\left(-\frac{1}{2\sigma^2} \norm{ L_{\tau_j}T_{t_j}f - y_j }_2^2 + \log \rho_t(t^\ell) + \log \rho_\tau(\tau^q)\right)\, , 
	\end{align}
	up to a constant.
	The M-step is then computed by:
	\begin{align}
	\label{equ:Mstep}
	(f^{k+1}, \rho^{k+1}_t, \rho^{k+1}_\tau) &= \text{arg}\,\max\limits_{f, \rho_t, \rho_\tau} Q(f,\rho_t, \rho_\tau \big| f^k,\rho^k_t, \rho^k_\tau) \, .
	\end{align}
	Since $f,\rho_t, \rho_\tau$ all appear in distinct sums in (\ref{equ:Qdef}), performing the maximization in (\ref{equ:Mstep}) is straightforward. Since $\norm{ L_{\tau_j}T_{t_j}f - y_j}_2^2 = \frac{1}{1-\tau_q} \norm{ f - T_\ell^{-1} L_\tau^{-1}y_j}_2^2$, 
	it is easy to check that:
	\begin{align}
	\label{equ:fupdate}
	f^{k+1}&= \frac{1}{C} \sum_{j=1}^M \sum_{\ell=1}^n \sum_{q=1}^n \frac{w_k^{\ell,q,j}}{(1-\tau_q)}T_\ell^{-1} L_\tau^{-1}y_j \quad, \quad
	C = \sum_{j=1}^M \sum_{\ell=1}^n \sum_{q=1}^n \frac{w_k^{\ell,q,j}}{(1-\tau_q)} \, .
	\end{align}
	Using Lemma 15 in  \cite{abbe2018multireference} , one can also obtain closed form expressions for the updates to $\rho^k_t, \rho^k_\tau$:
	\begin{align*}
	\rho^{k+1}_t(t^\ell) &= \frac{\widetilde{w}_k^\ell}{ \sum_{\ell'} \widetilde{w}_k^{\ell'}} \ \text{for}\ \widetilde{w}_k^\ell = \sum_j \sum_q w^{\ell,q,j}_k \quad, \quad
	\rho^{k+1}_\tau(\tau^q) = \frac{\widetilde{v}_k^q}{ \sum_{q'} \widetilde{v}_k^{q'}} \ \text{for}\   \widetilde{v}_k^q =  \sum_j \sum_\ell w^{\ell,q,j}_k \, .
	\end{align*}
	Note when a discrete signal defined on some fixed grid is dilated, its dilation is defined on a different grid. Thus computing (\ref{equ:Estep}) and (\ref{equ:fupdate}) will involve off-grid interpolation, a subtlety not arising in classic MRA, and this interpolation may contribute additional error. We also note that one can always force the translation distribution to be uniform by retranslating the signals uniformly, and in this case all sums over $\ell$ in this section could be eliminated. This would improve the computational complexity of the algorithm but may be disadvantageous in terms of sample complexity, as in classic MRA a uniform translation distribution requires a larger sample size for accurate estimation than an aperiodic translation distribution \cite{abbe2018multireference}.}

\section{Supporting results: stochastic calculus}

This appendix contains several stochastic calculus results which are used to control the statistics of the additive noise. Proposition \ref{prop:genItoIso} is
a simple generalization of Thm 4.5 of \cite{klebaner2012introduction}. Proposition \ref{prop:GenFourthMoment} controls the second moment of the stochastic quantity in Proposition \ref{prop:genItoIso}, and is in fact a special case of Proposition \ref{prop:MostGenFourthMoment}. Both Propositions \ref{prop:GenFourthMoment} and \ref{prop:MostGenFourthMoment} are proved with standard techniques from stochastic calculus, and for brevity we omit the proofs.

\begin{proposition}
	\label{prop:genItoIso}
	Assume $\int_0^T f(t)^2\ dt < \infty$, $\int_0^T \overline{f(t)}^2\ dt < \infty$, and let $B_t$ be a Brownian motion with variance $\sigma^2$. Then:
	\[ \Ex \left[\left(\int_0^T f(t)\ dB_t\right)\left(\int_0^T \overline{f(t)}\ dB_t\right)\right] = \sigma^2 \int_0^T f(t)\overline{f(t)}\ dt \, . \]
\end{proposition}

\begin{proposition}
	\label{prop:GenFourthMoment}
	Let $f(t)$ be a bounded and continuous complex deterministic function on $[0,T]$, and let $B_t$ be a Brownian motion with variance $\sigma^2$. Then for a fixed nonrandom time $T$, we have:
	\begin{align*} 
	\Ex \left[\left(\int_0^T f(t)\ dB_t\right)^2\left(\int_0^T \overline{f(t)}\ dB_t\right)^2\right]&= 2\sigma^4\left(\int_0^T|f(t)|^2\ dt\right)^2
	+ \sigma^4\left(\int_0^Tf(t)^2\ dt\right)\left(\int_0^T\overline{f(t)}^2\ dt\right)\, .
	\end{align*}
\end{proposition}

\begin{corollary}
	When $f(t)$ is real, the above reduces to:
	\[ \Ex \left[\left(\int_0^T f(t)\ dB_t\right)^4\right]= 3\sigma^4\left(\int_0^Tf(t)^2\ dt\right)^2 \, .\] 	
\end{corollary}	

\begin{proposition}
	\label{prop:MostGenFourthMoment}
	Let $f(t), g(t)$ be bounded and continuous complex deterministic functions on $[0,T]$, and let $B_t$ be a Brownian motion with variance $\sigma^2$. Then for a fixed nonrandom time $T$, we have:
	\begin{small}
		\begin{align*} &\Ex\left[ \left(\int_0^Tf(t)\ dB_t\right)\left(\int_0^T\overline{f(t)}\ dB_t\right)\left(\int_0^Tg(t)\ dB_t\right)\left(\int_0^T\overline{g(t)}\ dB_t\right) \right] \\
		\qquad&=\sigma^4\left[\left(\int_0^Tf(t)g(t)\ dt\right)\left(\int_0^T\overline{f(t)}\overline{g(t)}\ dt\right)+\left(\int_0^Tf(t)\overline{g(t)}\ dt\right)\left(\int_0^T\overline{f(t)}g(t)\ dt\right) \right. \\
		&\qquad\left.+\left(\int_0^T|f(t)|^2\ dt\right)\left(\int_0^T|g(t)|^2\ dt\right)\right] \, .
		\end{align*}
	\end{small}
\end{proposition}

\bibliographystyle{unsrt}
\bibliography{MainBib.bib}

\end{document}